\documentclass[acmsmall,screen,authorversion]{acmart}
\pdfoutput=1

\acmJournal{PACMPL}
\acmVolume{5}
\acmNumber{ICFP} \acmArticle{69}
\acmYear{2021}
\acmMonth{8}
\acmDOI{10.1145/3473574}
\startPage{1}
\acmPrice{}
\acmSubmissionID{icfp21main-p46-p}

\setcopyright{rightsretained}

\usepackage{booktabs}   \usepackage{subcaption}

\makeatletter

\edef\bold@wt{b} 

\SetSymbolFont{operators}{bold}{\tx@enc}{\rmdefaultB}{\bold@wt}{n}

\DeclareMathAlphabet{\mathbf}{\tx@enc}{\rmdefaultB}{\bold@wt}{n}
\SetMathAlphabet{\mathit}{bold}{\tx@enc}{\rmdefaultB}{\bold@wt}{it}
\makeatother

\RequirePackage{amsfonts}   \RequirePackage{mathtools}

\RequirePackage{dsfont}

\RequirePackage{stmaryrd}

\RequirePackage{tikz}
\usetikzlibrary{arrows,fit,matrix,positioning,calc}

\makeatletter

\theoremstyle{plain}\@ifundefined{theorem}{\newtheorem{theorem}{Theorem}[section]}{}
\@ifundefined{corollary}{\newtheorem{corollary}[theorem]{Corollary}}{}
\@ifundefined{prop}{\newtheorem{prop}{Proposition}[section]}{}
\@ifundefined{lemma}{\newtheorem{lemma}[theorem]{Lemma}}{}
\@ifundefined{conjecture}{}{}
\@ifundefined{definition}{}{}
\@ifundefined{remark}{}{}
\@ifundefined{note}{}{}
\newtheorem*{lemma*}{Lemma}
\theoremstyle{definition}\@ifundefined{fact}{}{}
\@ifundefined{hypothesis}{}{}
\@ifundefined{example}{}{}
\@ifundefined{exercise}{}{}
\@ifundefined{problem}{}{}
\@ifundefined{notation}{}{}

\makeatother

\renewcommand{\emptyset}{\varnothing}

\newcommand{\set}[1]{\{ #1 \}}

\DeclareMathOperator*{\dom}{dom}

\newcommand{\ifempty}[3]{\ifx\\#1\\ #2\else #3\fi}

\renewcommand{\to}[1][]{\ifempty{#1}{\rightarrow}{\xrightarrow{\, #1 \,}}}

\newcommand{\ntrans}[1]{\stackrel{\bullet}{\rightarrow}}

\renewcommand{\vec}[1]{\boldsymbol{\mathrm{#1}}}

\usepackage[scaled=0.8]{beramono}
\usepackage{url}
\usepackage{flushend} \usepackage[olditem,oldenum,defblank]{paralist} \usepackage[inline]{enumitem}

\makeatletter

\newif\ifsmallrulenames  \smallrulenamesfalse

\newcommand{\choosernsize}[2]{\ifsmallrulenames#1\else#2\fi}

\newcommand{\rn}[1]{\ifmmode 
    \mathchoice
      {\mbox{\choosernsize{\small}{}\sc #1}}
      {\mbox{\choosernsize{\small}{}\sc #1}}
      {\mbox{\choosernsize{\tiny}{\small}\sc #1}}
      {\mbox{\choosernsize{\tiny}{\tiny}\uppercase{#1}}}\else
    \hbox{\choosernsize{\small}{}\sc #1}\fi}

\newif\ifsuppressrulenames  
\suppressrulenamesfalse

\newif\ifbcprulessavespace
\bcprulessavespacefalse

\newif\ifbcprulestwocol
\bcprulestwocolfalse

\newcommand{\inflabel}[1]{\ifsuppressrulenames\else
    \def\lab{#1}\ifx\lab\empty
      \relax
    \else
      (\rn{\lab})\fi\fi
}

\newlength{\afterruleskip}
\newlength{\labelminsep}
\newdimen\labelcolwidth

\newcommand{\typicallabel}[1]{
  \setbox \@tempboxa \hbox{\inflabel{#1}}
  \labelcolwidth \wd\@tempboxa
  }
\typicallabel{}

\newif  \ifindexrules   \indexrulesfalse

\newbox\@labelbox
\newbox\rulebox
\newdimen\ruledim
\newdimen\labeldim

\newcommand{\layoutruleverbose}[2]{\unvbox\voidb@x  \addvspace{\afterruleskip}

   \setbox \rulebox \hbox{$\displaystyle #2$}

   \setbox \@labelbox \hbox{#1}
   \ruledim \wd \rulebox
   \labeldim \wd \@labelbox

\@tempdima \linewidth
   \advance \@tempdima -\labelcolwidth
   \ifdim \@tempdima < \ruledim
     \@tempdima \ruledim
   \else
     \advance \@tempdima by \ruledim
     \divide \@tempdima by 2
   \fi
   \advance \@tempdima by \labelminsep
   \advance \@tempdima by \labeldim
   \ifdim \@tempdima < \linewidth
\@tempdima \linewidth
     \advance \@tempdima -\labelcolwidth
     \hbox to \linewidth{\hbox to \@tempdima{\hfil
         \box\rulebox
         \hfil}\hfill
       \hbox to 0pt{\hss\box\@labelbox}}\else
\@tempdima 0pt
   \advance \@tempdima by \ruledim
   \advance \@tempdima by \labelminsep
   \advance \@tempdima by \labeldim
   \ifdim \@tempdima < \linewidth
\hbox to \linewidth{\hfil
         \box\rulebox
         \hskip \labelminsep
         \box\@labelbox}\else
\@tempdima \linewidth
     \advance \@tempdima -\labelcolwidth
     \hbox to \linewidth{\hbox to \@tempdima{\hfil
          \box\rulebox
          \hfil}
       \hfil}\penalty10000
     \hbox to \linewidth{\hfil
         \box\@labelbox}\fi\fi

   \addvspace{\afterruleskip}\@doendpe  \ignorespaces
   }

\newcommand{\layoutrulenolabel}[1]{\unvbox\voidb@x  \addvspace{\afterruleskip}

   \setbox \rulebox \hbox{$\displaystyle #1$}

   \@tempdima \linewidth
   \advance \@tempdima -\labelcolwidth
   \hbox to \@tempdima{\hfil 
      \box\rulebox
      \hfil}

   \addvspace{\afterruleskip}\@doendpe  \ignorespaces
   }

\newcommand{\layoutruleterse}[2]{\setbox \rulebox \hbox{$\displaystyle #2$}
   \noindent
   \parbox[b]{0.5\linewidth}
    {\vspace*{0.4em} \hfill\box\rulebox\hfill~}
   }

\newcommand{\layoutrule}[2]{\ifbcprulessavespace 
        \layoutruleterse{#1}{#2}
      \else
        \layoutruleverbose{#1}{#2}
      \fi
}

\newif\ifnewrule   \newrulefalse
\newcommand{\setrulebody}[1]{\ifnewrule
     \@ifundefined{HIGHLIGHT}{\fbox{\ensuremath{#1}}}{\HIGHLIGHT{#1}}\else
     #1
  \fi
}

\newcommand{\typesetax}[1]{\setrulebody{\begin{array}{@{}c@{}}#1\end{array}}}
\newcommand{\typesetrule}[2]{\setrulebody{\frac{\begin{array}{@{}c@{}}#1\end{array}}{\begin{array}{@{}c@{}}#2\end{array}}}}

\newcommand{\ruleindexprefix}[1]{\gdef\ruleindexprefixstring{#1}}
\ruleindexprefix{}
\newcommand{\maybeindex}[1]{\ifindexrules
    \index{\ruleindexprefixstring#1@\rn{#1}}\fi}

\def\infax{\@ifnextchar[{\@infaxy}{\@infaxx}}
\def\@infaxx#1{\ifbcprulessavespace $\typesetax{#1}$\else \layoutrulenolabel{\typesetax{#1}}\fi\newrulefalse\ignorespaces}
\def\@infaxy[#1]{\maybeindex{#1}\@infax{\inflabel{#1}}}
\def\@infax#1#2{\layoutrule{#1}{\typesetax{#2}}\ignorespaces}

\def\infrule{\@ifnextchar[{\@infruley}{\@infrulex}}
\def\@infrulex#1#2{\ifbcprulessavespace $\typesetrule{#1}{#2}$\else \layoutrulenolabel{\typesetrule{#1}{#2}}\fi\newrulefalse\ignorespaces}
\def\@infruley[#1]{\maybeindex{#1}\@infrule{\inflabel{#1}}}
\def\@infrule#1#2#3{\layoutrule{#1}{\typesetrule{#2}{#3}}\ignorespaces}

\newcommand{\andalso}{\quad\quad}

\makeatother
\usepackage{bussproofs}

\usepackage{multicol}

\usepackage{xifthen}

\usepackage{tikz-cd}
\tikzcdset{arrow style=tikz, diagrams={>=latex}}  

\RequirePackage{listings}  \RequirePackage{xcolor}

\lstdefinelanguage{Scala}{morekeywords={abstract,case,catch,char,class,def,do,else,extends,false,final,finally,for,if,import,implicit,match,module,new,null,object,override,package,private,protected,public,for,public,return,scala,super,sealed,this,throw,trait,true,try,type,val,var,with,while,yield,static,void,let,in},otherkeywords={=>,<-,<:,>:,\#,\&},
  sensitive,morecomment=[l]//,morecomment=[s]{/*}{*/},morestring=[b]",morestring=[b]',morestring=[b]""",
  moredelim=**[is][\btHL]{`}{`},
  showstringspaces=false}[keywords,comments,strings]

\lstdefinestyle{scala-plain}{language=Scala,mathescape=true,aboveskip=1em,
  belowskip=1em,
basicstyle=\normalsize\ttfamily,keywordstyle=\bfseries,columns=fullflexible,keepspaces=true,numbers=none}

\definecolor{lstscalabg}{rgb}{0.95,0.95,1}     \definecolor{lstscalastr}{rgb}{0.6,0,0}        \definecolor{lstscalacmt}{rgb}{0.25,0.5,0.35}  \definecolor{lstscalakw}{rgb}{0.5,0,0.35}      \definecolor{lstscaladoc}{rgb}{0.25,0.35,0.75} 

\lstdefinestyle{scala-color}{style=scala-plain,keywordstyle=\color{lstscalakw}\bfseries,
  stringstyle=\color{lstscalastr},
  commentstyle=\color{lstscalacmt},
  morecomment=[s][\color{lstscaladoc}]{/**}{*/},
}

\lstnewenvironment{lstscala}{\lstset{style=scala-plain}}{}
\lstnewenvironment{lstscalac}{\lstset{style=scala-color}}{}
\lstnewenvironment{lstscalasmall}{\lstset{style=scala-color,basicstyle=\footnotesize\ttfamily}}{}

\usepackage{microtype}

\newcommand{\secfnt}[1]{\textbf{\textsf{#1}}}

\newcommand{\fig}[4][tbp]{\begin{figure}[#1]#4\Description{#3}\caption{\label{#2}#3}\end{figure}}

\definecolor{dark-gray}{gray}{0.5}
\definecolor{shade}{gray}{0.8}

\newcommand{\enableLstShortInline}{\lstMakeShortInline[style=scala-plain,flexiblecolumns=false,mathescape=true,basicstyle=\tt]@
}
\newcommand{\disableLstShortInline}{\lstDeleteShortInline@}

\newcommand{\F}{f\kern.1em}

\newcommand{\pref}[2]{\hyperref[#2]{#1\ref*{#2}}}

\newcommand{\Sec}[2][\S]{\pref{#1}{sec:#2}}
\newcommand{\Fig}[2][Fig.]{\pref{#1~}{fig:#2}}

\newcommand{\Lem}[2][Lemma]{\pref{#1~}{lem:#2}}
\newcommand{\Thm}[2][Theorem]{\pref{#1~}{thm:#2}}
\newcommand{\Cor}[2][Corollary]{\pref{#1~}{cor:#2}}
\newcommand{\Prop}[2][Prop.]{\pref{#1~}{prop:#2}}

\newcommand{\Item}[2][]{\pref{#1}{item:#2}}
\newcommand{\Feat}[2][F]{\pref{#1}{feat:#2}}

        \newcommand{\ie}{i.e.\ }
\newcommand{\Eg}{For example, }    \newcommand{\eg}{e.g.\ }

\newcommand{\cf}{cf.\ }
        
\newcommand{\wrt}{w.r.t.\ }

\newcommand{\FSub}{\ensuremath{F_{\leq}}}
\newcommand{\FSubC}{\ensuremath{F_{<:}}}
\newcommand{\FOmega}{\ensuremath{F_\omega}}
\newcommand{\KFOmegaSub}{\ensuremath{F^\omega_{<:}}}
\newcommand{\FOmegaSubC}{\ensuremath{\mathcal{F}^\omega_{\leq}}}

\newcommand{\FOmegaInt}{\ensuremath{F^\omega_{\displaystyle \cdot \cdot}}}

\newcounter{RuleRef}
\newcommand{\ruledef}[1]{\refstepcounter{RuleRef}#1\label{rule:#1}}
\newcommand{\ruledefN}[2]{\refstepcounter{RuleRef}#2\label{rule:#1}}
\newcommand{\ruleref}[1]{\textsc{\hyperref[rule:#1]{#1}}}
\newcommand{\rulerefP}[1]{\textsc{(\ruleref{#1})}}
\newcommand{\rulerefN}[2]{\textsc{\hyperref[rule:#1]{#2}}}
\newcommand{\rulerefPN}[2]{\textsc{(\rulerefN{#1}{#2})}}

\newcommand{\ts}[1][]{\vdash_{\sf #1}}
\newcommand{\nts}[1][]{\nvdash_{\sf #1}}
\newcommand{\tsNe}{\ts[ne]}
\newcommand{\tsVar}{\ts[var]}
\newcommand{\tsD}{\ts[d]}
\newcommand{\tsE}{\ts[e]}

\newcommand{\tsC}{\ts[c]}
\newcommand{\tsTf}{\ts[tf]}
\newcommand{\gap}{\quad\quad}
\newcommand{\orElse}{\hspace{1ex}\big|\hspace{1ex}}

\newcommand{\judgment}[2]{\smallskip \secfnt{#1} \hfill #2}

\newcommand{\hlbox}[1]{\colorbox{shade}{$#1$}}

\newcommand{\infruleLeft}[3]{\typicallabel{#1}\infrule[\ruledef{#1}]{#2}{#3}\typicallabel{}}

\newcommand{\infaxSimp}[2][]{$\typesetax{#2}$\ifthenelse{\isempty{#1}}{}{\hspace{\labelminsep}\sc(\ruledef{#1})}}
\newcommand{\infruleSimpBase}[3]{$\typesetrule{#2}{#3}$\ifthenelse{\isempty{#1}}{}{\hspace{\labelminsep}\sc #1}}
\newcommand{\infruleSimp}[3][]{\infruleSimpBase{\ifthenelse{\isempty{#1}}{}{(\ruledef{#1})}}{#2}{#3}}
\newcommand{\infruleSimpN}[4]{\infruleSimpBase{(\ruledefN{#1}{#2})}{#3}{#4}}

\newlength{\savedcolumnsep}
\newlength{\savedmulticolsep}
\newenvironment{nosepmulticols}[3][0pt]{\setlength{\savedcolumnsep}{\columnsep}\setlength{\savedmulticolsep}{\multicolsep}\setlength{\multicolsep}{#1}\setlength{\columnsep}{#2}\begin{multicols}{#3}}{\end{multicols}\setlength{\columnsep}{\savedcolumnsep}\setlength{\multicolsep}{\savedmulticolsep}}

\newenvironment{multicolenum}[2][]{\begin{nosepmulticols}[\topsep]{-2\parindent}{#2}\begin{enumerate}[#1]}{\end{enumerate}\end{nosepmulticols}}

\newcommand{\fv}{\mathsf{fv}}
\renewcommand{\dom}{\mathsf{dom}}

\newcommand{\kin}{:}
\newcommand{\tin}{:}

\newcommand{\ksn}{\,\rightrightarrows\,}
\newcommand{\kck}{\,\leftleftarrows\,}

\newcommand{\kas}{\mskip1mu{:}\mskip2mu}
\newcommand{\tas}{\mskip1mu{:}\mskip2mu}

\newcommand{\genjudg}{\!\mathcal{J}}
\newcommand{\judg}[2][\ts \genjudg]{#2 #1}
\newcommand{\judgD}[2][\genjudg]{\judg[\ts #1]{#2}}

\newcommand{\judgS}[2][\genjudg]{\judg[\ts #1]{#2}}
\newcommand{\judgC}[2][\genjudg]{\judg[\ts #1]{#2}}
\newcommand{\judgDD}[2][\genjudg]{\judg[\tsD #1]{#2}}
\newcommand{\judgEE}[2][\genjudg]{\judg[\tsE #1]{#2}}

\newcommand{\judgCC}[2][\genjudg]{\judg[\tsC #1]{#2}}

\newcommand{\ctx}[2][]{{#2 \; \mathsf{ctx}_{\sf #1}}}
\newcommand{\kind}[3][]{#2 \ts[#1] {#3 \; \mathsf{kd}}}
\newcommand{\type}[3][]{#2 \ts[#1] #3 \kin \kstar}
\newcommand{\ctxD}[1]{\ctx{#1}}
\newcommand{\ctxDD}[1]{\ctx[d]{#1}}
\newcommand{\kindD}[2]{\kind{#1}{#2}}
\newcommand{\kindDD}[2]{\kind[d]{#1}{#2}}
\newcommand{\typeD}[2]{\type{#1}{#2}}

\newcommand{\ctxE}[1]{\ctx[e]{#1}}
\newcommand{\kindE}[2]{\kind[e]{#1}{#2}}

\newcommand{\ctxS}[1]{{#1 \; \mathsf{ctxs}}}
\newcommand{\kindS}[2]{#1 \ts {#2 \; \mathsf{kds}}}
\newcommand{\typeS}[2]{\type{#1}{#2}}
\newcommand{\spS}[4]{#1 \ts #2 \kin #3 \kin #4}
\newcommand{\ctxC}[1]{\ctx{#1}}
\newcommand{\ctxCC}[1]{\ctx[c]{#1}}
\newcommand{\kindC}[2]{\kind{#1}{#2}}
\newcommand{\kindCC}[2]{\kind[c]{#1}{#2}}
\newcommand{\typeC}[2]{#1 \ts #2 \ksn #2 \intv #2}
\newcommand{\typeCCk}[2]{#1 \ts #2 \kck \kstar}

\newcommand{\typeCP}[2]{#1 \ts #2 \ksn (#2) \intv\mskip1mu (#2)}
\newcommand{\spC}[4]{#1 \ts #2 \ksn #3 \ksn #4}
\newcommand{\spCC}[4]{#1 \ts[c] #2 \ksn #3 \ksn #4}

\newcommand{\op}[1]{#1_{\sf n}}
\newcommand{\cst}[1]{#1_{\sf c}}

\newcommand{\alEq}{\equiv}

\newcommand{\ksub}{\leq}
\newcommand{\tsub}{\leq}
\newcommand{\tsup}{\geq}

\newcommand{\keq}{=}
\newcommand{\teq}{=}

\newcommand{\ctxEqD}[2]{\ctx{#1 \keq #2}}
\newcommand{\ctxEqC}[2]{\ctx{#1 \keq #2}}

\newcommand{\wkEq}{\approx}

\newcommand{\subst}[3]{#1[#3/#2]}  

\newcommand{\hsubst}[4]{#1[#4/{#2}^{\,#3}]} \newcommand{\rapp}[3]{\app{{#1} \,\cdot^{#2}}{#3}}

\newcommand{\etaExpRaw}[2]{\eta_{#1} #2}
\newcommand{\etaExp}[2]{\etaExpRaw{#1}{(#2)}}

\newcommand{\weakEtaExpRaw}[2]{\bar{\eta}_{#1} #2}
\newcommand{\weakEtaExp}[2]{\weakEtaExpRaw{#1}{(#2)}}

\newcommand{\nfRaw}{\mathsf{nf}}
\newcommand{\nf}[2]{\nfRaw_{#1}(#2)}
\newcommand{\nfCtx}[1]{\nf{}{#1}}

\newcommand{\reduces}[1][]{\longrightarrow#1}
\newcommand{\betastep}{\reduces[_\beta]}
\newcommand{\betared}{\reduces[_{\smash{\beta}}^*]}

\newcommand{\cbvstep}{\reduces[_{\sf v}]}
\newcommand{\cbvred}{\reduces[_{\sf v}^*]}

\newcommand{\skEq}{=_{\rm SK}}                        \newcommand{\skSub}{\leq_{\rm SK}}                    \newcommand{\skRed}[1][]{\Rrightarrow_{#1}}           \newcommand{\skExp}[1][]{\Lleftarrow_{#1}}            \newcommand{\skSig}{\Gamma_{\rm SK}}                  \newcommand{\skS}[1][]{\mathsf{S}}                    \newcommand{\skK}[1][]{\mathsf{K}}                    \newcommand{\skApp}[2]{#1 \odot #2}                   \newcommand{\skSRedVar}{S_{\sf r}}                    \newcommand{\skSExpVar}{S_{\sf e}}                    \newcommand{\skKRedVar}{K_{\sf r}}                    \newcommand{\skKExpVar}{K_{\sf e}}                        \newcommand{\skKRed}[2]{K_{\sf r} \, #1 \, #2}                \newcommand{\encode}[1]{\llbracket #1 \rrbracket}

\newcommand{\cempty}{\emptyset}

\newcommand{\kstar}{{*}}

\newcommand{\kempty}{\emptyset}

\DeclareMathOperator{\intv}{{.}{.}}
\newcommand{\intvP}{\intv\mskip1mu{}}

\newcommand{\hointv}[1]{\intv_{#1}}
\newcommand{\hointvP}[1]{\intv_{#1}\mskip1mu{}}

\newcommand{\ksingp}[1]{S(#1)}                \newcommand{\ksing}[2]{\ksingp{#1 \kin #2}}   

\newcommand{\kpowp}[1]{P(#1)}                  

\newcommand{\dfun}[3]{({#1}\kas{#2}) \to {#3}}

\newcommand{\sempty}{\epsilon}
\newcommand{\scons}[2]{#1 , #2}

\renewcommand{\Top}{\top}
\renewcommand{\Bot}{\bot}

\newcommand{\tmax}[1]{\Top_{#1}}
\newcommand{\tmin}[1]{\Bot_{#1}}

\newcommand{\kmax}[1]{\kstar_{#1}}

\newcommand{\call}[2]{\forall #1. \, {#2}}     \newcommand{\all}[3]{\call{{#1}\kas{#2}}{#3}}  

\newcommand{\fun}[2]{#1 \to #2}

\newcommand{\clam}[2]{\lambda #1. \, #2}       \newcommand{\lam}[3]{\clam{{#1}\tas{#2}}{#3}}  \newcommand{\Lam}[3]{\clam{{#1}\kas{#2}}{#3}}
\newcommand{\app}[2]{{#1} \, {#2}}
\newcommand{\appp}[3]{\app{\app{#1}{#2}}{#3}}
\newcommand{\apppp}[4]{\app{\appp{#1}{#2}{#3}}{#4}}

\newcommand{\ksimp}[1]{\lvert #1 \rvert}

\newcommand{\simpEq}[2]{\ksimp{#1} \alEq \ksimp{#2}}

\newcommand{\toElim}[1]{#1} \newcommand{\toType}[1]{#1}

\newcommand{\elet}{\mathbf{let} \; }
\newcommand{\ein}{\mathbf{in} \;}

\newcommand{\eletin}[2]{\elet #1 \; \ein #2}

\newcommand{\record}[1]{\set{\, #1 \,}}

\newcommand{\SupSec}[2][Appendix~]{\Sec[#1]{#2}}
\newcommand{\SupItem}[1]{\ref*{item:#1}}

\newcommand{\isInTheAppendix}{is included \hyperref[sec:modules]{in the appendix}}

\begin{document}

\title[A Theory of Higher-Order Subtyping with Type Intervals
(Extended Version)]{A Theory of Higher-Order Subtyping with Type Intervals}
\subtitle{Extended Version Including Supplementary Material}

\author{Sandro Stucki}
\orcid{0000-0001-5608-8273}             \affiliation{
\department{Computer Science and Engineering}
\institution{Chalmers University of Technology}
\city{Gothenburg}
\country{Sweden}                      }
\email{sandros@chalmers.se}             

\author{Paolo G. Giarrusso}
\affiliation{
\institution{Bedrock Systems Inc.}    \city{Berlin}
\country{Germany}                     }
\email{p.giarrusso@gmail.com}

\begin{abstract}
  The calculus of Dependent Object Types (DOT) has enabled a more
  principled and robust implementation of Scala, but its support for
  type-level computation has proven insufficient.
As a remedy, we propose \FOmegaInt{}, a rigorous theoretical
  foundation for Scala's higher-kinded types.
\FOmegaInt{} extends \KFOmegaSub{} with \emph{interval kinds}, which
  afford a unified treatment of important type- and kind-level
  abstraction mechanisms found in Scala, such as bounded
  quantification, bounded operator abstractions, translucent type
  definitions and first-class subtyping constraints.
The result is a flexible and general theory of higher-order subtyping.
We prove type and kind safety of \FOmegaInt{}, as well as weak
  normalization of types and undecidability of subtyping.
All our proofs are mechanized in Agda using a fully syntactic
  approach based on hereditary substitution.
\end{abstract}



\maketitle

\enableLstShortInline

\section{Introduction}
\label{sec:intro}

Modern statically typed programming languages provide powerful
type-level abstraction facilities such as parametric polymorphism,
subtyping, generic datatypes, and varying degrees of type-level
computation.
When several of these features are present in the same language, new
and more expressive combinations arise,
such as
\begin{inparaenum}[(F1)]
\item\label{feat:bd-quant} bounded quantification,
\item\label{feat:bd-op} bounded type operators, and
\item\label{feat:typ-def} translucent type definitions.
\end{inparaenum}
Such mechanisms further increase the expressivity of the language, but
also the complexity of its type system and, ultimately, its implementation.

A case in point is the Scala programming language.
Scala is a multi-paradigm language that integrates functional and
object-oriented concepts.
It features all of the aforementioned type-level constructs and many
more.
While programmers enjoy the expressivity of Scala's type system, the
language developers have struggled for years to manage its complexity
--
tracing down elusive compiler bugs and soundness issues in a
seemingly ad-hoc fashion.
The development of the calculus of Dependent Object Types
(DOT)~\cite{AminGORS16wf} marked a turning point.
By providing a solid theoretical foundation for a large part of the
Scala language, DOT inspired a complete re-design of the Scala
compiler as well as a substantial redesign of the language itself --
giving rise to Scala~3 \cite{Dotty20code}.
But despite initial hopes to the contrary, DOT has proven insufficient
to express Scala's \emph{higher-kinded} (HK) types,
which are used pervasively throughout the Scala standard
library~\cite{MoorsPO08fool}.
The lack of a proper theory of HK types severely complicates their
implementation in Scala~3 \cite{OderskyMP16scala}.
To address this problem, we introduce \FOmegaInt{},
a rigorous theoretical foundation for Scala's HK types with a
machine-checked safety proof.

We start by illustrating the use of HK types through an example.
\vspace{-.5\baselineskip}
\begin{lstscala}
  type Ordering[A] = (A, A) => Boolean
  class SortedView[A, B >: A](xs: List[A], ord: Ordering[B]) {
    def foldLeft[C](z: C, op: (C, A) => C): C = //...
    def concat[C >: A <: B](ys: List[C]): SortedView[C, B] = //...
    // definitions of further operations such as 'map', 'flatMap', etc.
  }
\end{lstscala}
\vspace{-.5\baselineskip}
The example is inspired by the definition of the class @SeqView.Sorted@
from the standard library.\footnote{See
  \url{https://www.scala-lang.org/api/2.13.6/scala/collection/SeqView$$Sorted.html}}
It is heavily simplified for conciseness and to avoid the use of Scala
idioms irrelevant to this paper.

The @SortedView@ class allows one to iterate over the elements of an
underlying list @xs@ in increasing order (via @foldLeft@) without
modifying the original list.
The second constructor argument @ord@ provides the comparison
operation used for sorting.
Note that @ord@ compares data of type @B@, which is declared a
supertype of the element type @A@ via the annotation @B >: A@.
Thus we may use a comparison function for a less precise type to sort
the elements in the view.\footnote{Expert readers may notice that \lstinline{Ordering} is contravariant
  in its argument, hence \lstinline{SortedView} could simply
take a parameter of type \lstinline{Ordering[A]} and still accept instances of \lstinline{Ordering[B]}.
But this simplification does not extend to the Scala standard
  library, where \lstinline{Ordering} is invariant due to additional
  methods; hence the lower bound is necessary.
For readability, we
use the simplified version of
\lstinline{Ordering} but keep the bound in the definition of
  \lstinline{SeqView}.}

In Scala, parametrized classes like @SortedView@ are first-class type
operators.
Hence, @SortedView@ is an instance of a lower-bounded type
operator~(\Feat{bd-op}).
The operator @Ordering[A]@ is just a convenient type alias for the
function type @(A, A) => Boolean@ of binary operators.
It is a \emph{transparent} type definition~(\Feat{typ-def}),
in that instances of @Ordering[A]@ can be
replaced by its definition anywhere in the program.
The method @concat@ illustrates the use of lower- and upper-bounded
polymorphism~(\Feat{bd-quant}) to combine views on lists of different
but related element types.
Views are covariant:
if @A@ is a subtype of @C@ then a view @v@ on an @A@-list is also a
view on a @C@-list and can be extended as such using @v.concat(ys)@.
The upper bound on @C@ ensures that the ordering remains applicable.
This pattern of using lower-bounds to implement operations on
covariant data structures is common in the Scala standard library.

Although bounded quantification~(\Feat{bd-quant}) and bounded
operators~(\Feat{bd-op}) may seem conceptually different from
translucent type definitions~(\Feat{typ-def}), all three are closely
related.
A type alias declaration of the form @X = A@ (such as @Ordering@ above) declares
that the type @X@ has type @A@ as upper and lower bound.
Because subtyping in Scala is antisymmetric, this effectively
identifies @X@ with @A@.
In general, a type declaration of the form
@X >: A <: B@,
introduces an abstract type @X@ that is bounded by @A@ from below and
by @B@ from above.
In other words, the declaration
@X >: A <: B@
specifies a \emph{type interval} in which @X@ must be contained.
Type aliases
take the form of \emph{singleton
  intervals}
@X >: A <: A@,
where the lower and upper bounds coincide.

In our example, the types @Ordering@ and @SortedView@ can be seen as abstract types.
We can write their declarations as\footnote{
This code is accepted by the Scala~2.13.6 compiler.
Ironically, it is rejected by Scala~3,
which is built on DOT.
}
\vspace{-.5\baselineskip}
\begin{lstscala}
  type Ordering[A] >: (A, A) => Boolean <: (A, A) => Boolean
  type SortedView[A, B >: A] <: {
    def foldLeft[C](z: C, op: (C, A) => C): C
    def concat[C >: A <: B](ys: List[C]): SortedView[C, B]   /* ... */   }
\end{lstscala}
\vspace{-.5\baselineskip}
This version is closer to how type definitions are represented in DOT.
The class @SortedView@ is now represented as an abstract type declaration
with only an upper bound.
This means that its interface remains exposed -- any instance of
@SortedView[A, B]@ can be up-cast to a record with fields @foldLeft@,
@concat@, etc.
But not every record (or class) containing those fields is
automatically an instance of the type @SortedView[A, B]@.
This ensures that @SortedView@ continues to behave like a nominal
type, as Scala classes are supposed to.\footnote{For details on how Scala classes may be encoded in DOT, we refer the reader to \citet[Chap.~2]{Amin16thesis}.}
Unlike @SortedView@, the abstract type @Ordering@ is both upper- and
lower-bounded.

Intervals can also encode abstract types missing a bound.
Scala features a pair of extremal types @Any@ and @Nothing@:
the maximal type @Any@ is a supertype of every other type;
the minimal type @Nothing@ is a subtype of every other type.
Abstract types @X@ with only an upper or lower bound @A@ thus inhabit
the degenerate intervals
@X >: Nothing <: A@
and
@X >: A <: Any@,
respectively.

This suggests a uniform treatment of \Feat{bd-quant}--\Feat{typ-def}
through type intervals, and indeed, this is essentially how bounded
quantification~(\Feat{bd-quant}) and type definitions~(\Feat{typ-def})
are modeled in DOT.
Unfortunately, DOT lacks intrinsics for higher-order type computation,
such as type operator abstractions and applications, preventing it
from encoding simple HK type definitions such as the identity operator
@type Id[X] = X@ \cite{OderskyMP16scala}.
Traditionally, \Feat{bd-quant}--\Feat{typ-def} have been studied
through orthogonal extensions of \citeauthor{Girard72thesis}'s higher-order
polymorphic $\lambda$-calculus \FOmega{}~\citeyearpar{Girard72thesis}.
Bounded higher-order subtyping (\Feat{bd-quant} and \Feat{bd-op}) has
been formalized in variants of~\KFOmegaSub{}
\cite{PierceS97tcs,CompagnoniG03ic}, translucent type definitions~(\Feat{typ-def}) through singleton kinds~\cite{StoneH00popl}.
The treatment of~\Feat{bd-quant} and~\Feat{typ-def} in DOT
suggest a different, unified approach to
studying~\Feat{bd-quant}--\Feat{typ-def}:
via a formal theory of higher-order subtyping with type
intervals, which we pursue with \FOmegaInt{}.

Our goal in doing so is twofold.
\begin{inparablank} \item First, we want to establish a theoretical foundation for Scala's
  HK types, with full support for type-level computations
  (including~\Feat{bd-quant}--\Feat{typ-def}).
A principled theoretical understanding of HK types is crucial
  because potential safety issues are hard to identify and fix by
  ``trial and error'' alone, especially when they arise from feature
  interactions.
\item Second, we want to study the concept of type intervals
in its own right.
Despite their apparent simplicity, adding type intervals
  to~\KFOmegaSub{} leads to a surprisingly rich theory of higher-order
  subtyping that goes beyond previous treatments
  of \Feat{bd-quant}--\Feat{typ-def}.
That is because type intervals encode first-class subtyping
  constraints, or \emph{type inequations}, similar to extensional
  identity types in Martin-L\"of type theory (MLTT).
\end{inparablank}

In DOT, type intervals are baked into abstract type members and
therefore tied to the use of \emph{path-dependent types};
in \FOmegaInt{}, we break this bond.
A DOT type member declaration $\record{ X \colon A \intv B }$ roughly
corresponds to a Scala-style type member declaration
@class C { type X >: A <: B }@.
It combines two separate type-system features:
the declaration of an abstract type member $X$ in a record or class,
and the declaration of subtyping constraints on $X$ via the bounds $A$
and $B$.
These two features are independent.
For example, the Agda programming language features unbounded abstract
type members via record types while the $F_{<:}$ calculus features
subtyping constraints via bounded quantification but no type members.
The notion of path-dependent types in DOT and Scala is intimately
linked to abstract type members.
Given an instance $z \kas \record{ X \colon A \intv B }$, the type
expression $z.X$ denotes the type value assigned to $X$ in $z$.
The type $z.X$ is \emph{path-dependent} because it depends on the
term-level expression $z$ (the ``path'' to $X$).
In DOT (but not Scala), type members are also used to model bounded
quantification.

In \FOmegaInt{}, we deliberately separate the notion of type intervals
from that of abstract type members (and path-dependent types) and drop
the latter.
This simplifies the theory and allows us to study the power of type
intervals in the context of higher-order subtyping:
for bounded quantification, bounded operator abstraction and
translucent type definitions -- all of which are independent of
path-dependent types, yet commonly used when working with Scala's HK
types.

We
leave the development of a combined theory of HK and
path-dependent types for future work, and focus here on the
theoretical and practical insights afforded by the novel combination
of \emph{higher-order subtyping with type intervals}.
Concretely, we make the following contributions.
\begin{enumerate}
\item We propose \emph{type intervals} as a unifying concept for
  expressing bounded quantification, bounded operator abstractions,
  and translucent type definitions.
Going beyond the status~quo, we show that type intervals are
  expressive enough to also cover less familiar constructs, such as
  lower-bounded operator abstractions and first-class inequations
  (\Sec{background}).

\item We introduce \FOmegaInt{} --~an extension of \FOmega{} with
  \emph{interval kinds}~-- as a formal calculus of higher-order
  subtyping with type intervals (\Sec{declarative}).
\FOmegaInt{} is the first formalization of Scala's higher-kinded
  types with a rigorous, machine-checked type safety proof.
As such, it
provides a theoretical foundation for
several important features of Scala-like type systems.

\item We establish important metatheoretic properties of our theory:
\begin{inparablank} \item kind safety (\Sec{declarative}),
  \item weak normalization of types (\Sec{normalization}),
  \item type safety (\Sec{canonical}), and
  \item undecidability of subtyping (\Sec{undec}).
  \end{inparablank}
The metatheoretic proofs are complicated substantially by the
  interaction of advanced type system features such as dependent
  kinds, subtyping and subkinding,
  and (in)equality reflection.
As others have
  recognized~\cite{AspinallC01tcs,Zwanenburg99tlca,YangO17oopsla}, the
  combination of dependent types (or kinds) and subtyping poses a
  particular challenge in metatheoretic developments.
The usefulness of our proof techniques thus extends beyond the scope
  of \FOmegaInt{} to other systems combining dependent types and
  subtyping.

\item The metatheoretic development is entirely syntactic (it involves
  no model constructions) and has been fully mechanized
  using the Agda proof assistant~\cite{Norell07thesis}.
The main technical device is a purely syntactic, bottom-up
  normalization procedure based on a novel variant of \emph{hereditary
    substitution} that computes the $\beta\eta$-normal forms of types
  and kinds (\Sec{normalization}).

\end{enumerate}
We outline our proof strategy in \Sec{strategy},
review related work in \Sec{related} and
give concluding remarks in \Sec{conclusions}.

Because of space constraints, we omit most proofs and many details of
the metatheory from the paper and focus instead on the big picture:
the design and expressiveness of \FOmegaInt{} as well as the many
challenges involved in proving its type safety and our strategies for
addressing them.
However, the complete metatheory, including full proofs of all lemmas
and theorems stated in the paper, has been mechanized in Agda, and the
source code is freely available as an artifact~\cite{artifact}.
An overview of the Agda formalization, establishing the connection to
the theory presented in the paper, \isInTheAppendix{}, along
with detailed human-readable descriptions of metatheoretic results
that have been omitted from the paper.
Yet more details can be found in the first author's
PhD~dissertation~\cite{Stucki17thesis}.

\newcommand{\allmeta}{\textit{All}}
\newcommand{\bndmeta}{\textit{Bounded}}
\newcommand{\allvar}{\texttt{All}}
\newcommand{\bndvar}{\texttt{Bounded}}

\section{Subtyping with Type Intervals}
\label{sec:background}

Before we define our formal theory of type intervals, let us illustrate the
core ideas
in a bit more detail.
Consider the following Scala type definitions.
\vspace{-.5\baselineskip}
\begin{lstscala}
  abstract class Bounded[B, F[_ <: B]] { def apply[X <: B]: F[X] }
  type All[F[_]] = Bounded[Any, F]
\end{lstscala}
\vspace{-.3\baselineskip}
The class @Bounded@ and the type alias @All@ are Scala encodings of
the bounded and unbounded universal quantifiers found in
\FSub{}~\cite{CurienG92mscs}. We chose this example for its brevity and because it exemplifies
the type-level mechanisms
found in more realistic definitions, such as those given in
\Sec{intro}.
In particular, it features bounded quantification (of @X <: B@ in
@apply@), a bounded operator (the parameter @F[_ <: B]@ of @Bounded@)
and a transparent type alias (@All@).

We want to translate these two definitions into a typed
$\lambda$-calculus.
The challenge is to find a type system that is expressive enough to do
so.
The example involves type-level computations, so our first candidate
is \citeauthor{Girard72thesis}'s higher-order polymorphic
$\lambda$-calculus \FOmega{}~\citeyearpar{Girard72thesis}, but it
lacks even basic support for subtyping.
Our next candidate is \KFOmegaSub{}, which extends \FOmega{} with
higher-order subtyping and bounded quantification.
But most variants of \KFOmegaSub{} lack support for bounded
operators~\cite[cf.][]{Pierce02tapl,PierceS97tcs}.
Thankfully, \citeauthor{CompagnoniG03ic} have developed \FOmegaSubC{},
a variant of \KFOmegaSub{} with bounded
operators~\citeyearpar{CompagnoniG03ic}.
\FOmegaSubC{} has four type variable binders:
\begin{alignat*}{5}
  \clam{X \tsub A &\kas K}{t} &\quad& \text{term-level type abstraction}&\qquad
  \call{X \tsub A &\kas K}{B} &\quad& \text{type-level bounded quantifier}\\
  \clam{X \tsub A &\kas K}{B} && \text{type-level type abstraction}&
  \fun{(X \tsub A &\kas K)}{J} && \text{kind-level dependent arrow}
\end{alignat*}
The type $A$ in a binding $X \tsub A \colon K$ is called the \emph{upper
  bound} of $X$, and must be of kind $K$.
To represent unconstrained bindings, \FOmegaSubC{} features a
\emph{top} type $\Top$, which is a supertype of every other type (like
@Any@ in Scala).
The bound $\Top$ in $X \tsub \Top \kin \kstar$ is thus trivially
satisfied and can be omitted.

Since operator abstractions carry bounds in \FOmegaSubC{}, so must
arrow kinds.
This makes arrow kinds type-dependent, which substantially complicates
the meta theory of \FOmegaSubC{} when compared to other variants of
\KFOmegaSub{}.
As usual, we abbreviate $\dfun{X}{J}{K}$ to $\fun{J}{K}$ when $X$ does
not occur freely in $K$.

A possible translation of @Bounded@ and @All@ to \FOmegaSubC{} is
\begin{align*}
  \bndmeta &\; \coloneq \;
  \Lam{B}{\kstar}{\Lam{F}{\dfun{X \tsub B}{\kstar}{\kstar}}{\all{X \tsub B}{\kstar}{\app{F}{X}}}} &
  \allmeta &\; \coloneq \; \app{\bndmeta}{\Top}
\end{align*}
The named, parametrized class @Bounded[B, F[_ <: B]]@ has been
replaced by a pair of nested anonymous operator abstractions taking
arguments
$B \colon \kstar$ and $F \colon \dfun{X \tsub B}{\kstar}{\kstar}$;
the signature @apply[X <: B]: F[X]@ of the method @apply@ by the bounded universal
$\all{X \tsub B}{\kstar}{\app{F}{X}}$.

The declared kinds of the variables $B$, $F$ and $X$ in the definition
of $\bndmeta$ indicate what sort of type they represent:
$B$ and $X$ are proper types, while~$F$ is a unary bounded
operator. This makes $\bndmeta$ itself a higher-order type operator of kind
$\fun{\kstar}{\fun{(\dfun{X \tsub B}{\kstar}{\kstar})}{\kstar}}$.
For example, we obtain the type of the polymorphic identity function
by applying $\bndmeta$
as follows:
\begin{align*}
  \appp{\bndmeta}{\Top}{(\Lam{X}{\kstar}{\fun{X}{X}})} \; &= \;
  \appp{(\Lam{B}{\kstar}{\Lam{F}{\dfun{X \tsub B}{\kstar}{\kstar}}{\all{X \tsub B}{\kstar}{\app{F}{X}}}})}{\Top}{(\Lam{X}{\kstar}{\fun{X}{X}})} \\
  &= \; \all{X}{\kstar}{\fun{X}{X}}.
\end{align*}
The translation of the type alias @All[F[_]]@ is then just the
partial application $\allmeta = \app{\bndmeta}{\Top}$.

The above definitions of $\bndmeta$ and $\allmeta$ are
meta-definitions, \ie they are just convenient shorthands for the type
expressions
$\Lam{B}{\kstar}{\Lam{F}{\dfun{X \tsub B}{\kstar}{\kstar}}{\all{X \tsub
      B}{\kstar}{\app{F}{X}}}}$ and $\app{\bndmeta}{\Top}$.
But we can also give object-level definitions of $\bndmeta$ and
$\allmeta$ in~\FOmegaSubC{}, using standard syntactic sugar for
let-binding type and term variables:
\begin{align*}
  \eletin{X \kin K = A}{t} \; &\coloneq \; \app{(\Lam{X}{K}{t})}{A}, &
  \eletin{x \tin B = s}{t} \; &\coloneq \; \app{(\Lam{x}{B}{t})}{s}.
\end{align*}
We can use $\bndmeta$ as an abstract type operator in a term $t$ by
let-binding it to a type variable:
\[
  \eletin{\bndvar \kin \dfun{B}{\kstar}{\fun{(\dfun{X \tsub B}{\kstar}{\kstar})}{\kstar}}
\; = \;
\Lam{B}{\kstar}{\Lam{F}{\dfun{X \tsub B}{\kstar}{\kstar}}{\all{X \tsub A}{\kstar}{\app{F}{X}}}}}{t}.
\]
This definition is \emph{opaque}, \ie the term $t$ sees the signature
of $\bndvar$, but not its definition.
Consider \begin{alignat*}{4}
  \elet &\bndvar \kin
  \dfun{B}{\kstar}{\fun{(\dfun{X \tsub B}{\kstar}{\kstar})}{\kstar}}
  \; &&= \; \dotsc &&\ein\\
  \elet &x \tin \all{X}{\kstar}{\fun{X}{X}} \;
  &&= \; \lam{X}{\kstar}{\lam{z}{X}{z}} \quad &&\ein
  &&\qquad \text{\em --- OK}\\
  \elet &y \tin \appp{\bndvar}{\Top}{(\lam{X}{\kstar}{\fun{X}{X}})} \;
  &&= \; \lam{X}{\kstar}{\lam{z}{X}{z}} &&\ein \dotsc
  &&\qquad \text{\em --- type error}
\end{alignat*}
The third definition does not type check because
$\bndvar \not\teq \Lam{B}{\kstar}{\Lam{F}{\dfun{X \tsub
      B}{\kstar}{\kstar}}{\all{X}{\kstar}{\app{F}{X}}}}$ as types,
despite the binding.
Indeed, \FOmegaSubC{} cannot express \emph{transparent} type
definitions.

Furthermore, \FOmegaSubC{} also lacks support for \emph{lower-bounded}
definitions.
As discussed in \Sec{intro}, these have important applications \eg in
the Scala standard library.
Both transparent and lower-bounded definitions, and all the features
of \FOmegaSubC{}, can be uniformly expressed using \emph{interval
  kinds}.

\subsection{Intervals and Singletons}

As the name implies, an interval kind $A \intv B$ is inhabited by a
range of proper types $C$, namely those that are supertypes
$C \tsup A$ of its lower bound $A$ and subtypes $C \tsub B$ of its
upper bound $B$.
Hence, kinding statements of the form $C \kin A \intv B$ are
equivalent to pairs of subtyping statements $A \tsub C$ and
$C \tsub B$.
We make this equivalence formal in the next section.

Since every proper type is a subtype of $\Top$, intervals of the form
$A \intv \Top$ are effectively unconstrained from above, and can thus
be used to encode \emph{lower-bounded} definitions.
Similarly, upper-bounded definitions can be expressed using intervals
of the form $\Bot \intv A$ where the lower bound is the minimum or
\emph{bottom} type $\Bot$, our equivalent of Scala's @Nothing@ type.
For example, we recover \FSub{}-style bounded quantifiers
$\call{X \tsub B}{A}$ as $\all{X}{\Bot \intv B}{A}$.
Interval kinds of the form $A \intv A$, where the lower and upper
bounds coincide, are called \emph{singleton kinds} or simply
\emph{singletons}.
Given $B \kin A \intv A$, we have both
$A \tsub B$ and $B \tsub A$, which, assuming an antisymmetric
subtyping relation, implies $A \teq B$.
Singleton kinds can thus encode transparent definitions and have been
studied for that purpose by \citet{StoneH00popl}.
We adopt their notation $\ksingp{A} = A \intv A$ for the singleton containing just $A$.

Using interval kinds, we refine our definitions of $\bndvar$ and
$\allvar$ to make them transparent.
\begin{alignat*}{4}
  &\elet &\bndvar &\kin \dfun{B}{\kstar}{\dfun{F}{\,\fun{\Bot \intv B}{\kstar}}{\ksingp{\all{X}{\Bot \intv B}{\app{F}{X}}}}}\\
&&&= \; \Lam{B}{\kstar}{\Lam{F}{\fun{\,\Bot \intv B}{\kstar}}{\all{X}{\Bot \intv B}{\app{F}{X}}}} &\; &\ein \\
  &\elet &\allvar &\kin \dfun{F}{\fun{\kstar}{\kstar}}{\ksingp{\appp{\bndvar}{\Top}{F}}}
\; = \; \app{\bndvar}{\Top}
  &&\ein \dotsc
\end{alignat*}
The signature of $\bndvar$ tells us that, when we apply it to suitable
type arguments $B$ and $F$, the result is both a subtype and a
supertype of $\all{X}{\Bot \intv B}{\app{F}{X}}$.
In other words, we have
$\appp{\bndvar}{B}{F} \teq \all{X}{\Bot \intv B}{\app{F}{X}}$ in the
body of the let-binding.
Similarly, we have $\app{\allvar}{F} \teq \appp{\bndvar}{\Top}{F}$, as
desired.

Interval kinds $A \intv B$ are only well-formed if $A$ and $B$ are
proper types, \ie of kind $\kstar$.
To express all of the binders found in \FOmegaSubC{}, we need a way to
encode bindings of the form $X \tsub A \kin K$, for arbitrary kinds
$K$.
As we will see in \Sec{declarative}, this is indeed possible because
\FOmegaInt{} can encode \emph{higher-order interval kinds}
$A \hointv{K} B$ for arbitrary $K$ using its other kind- and
type-level constructs.

\subsection{First-Class Inequations}
\label{sec:inconsistent}

Instances $C \kin A \intv B$ of an interval kind $A \intv B$ represent
types bounded by $A$ and $B$ respectively.
But they also represent proofs that $A \tsub C$ and $C \tsub B$, and
--~by transitivity of subtyping~-- that $A \tsub B$.
In other words, the inhabitants of interval kinds $A \intv B$
represent first-class \emph{type inequations} $A \tsub B$.
Similarly, higher-order intervals represent type operator inequations.
Interval kinds thus provide us with a mechanism for \emph{(in)equality
  reflection}, \ie a way to extend the subtyping relation via
assumptions made at the term- or type-level (via type abstractions).

Among other things, this allows us to postulate type operators with
associated subtyping rules through type variable bindings.
We will see an example of this in \Sec{undec}; other examples are
intersection types or equi-recursive types and their associated
subtyping theories (see \SupSec{custom_sub_theories} for a detailed
example).
This is possible because we do not impose any \emph{consistency
  constraints} on the bounds of intervals.
That is, an interval kind $A \intv B$ is well-formed, irrespective of
whether $A \tsub B$ is actually provable or not.
If we can prove that $A \tsub B$, we say that the bounds of
$A \intv B$ are \emph{consistent}.

Having both (in)equality reflection and inconsistent bounds makes
\FOmegaInt{} very expressive, but breaks subject reduction of open
terms and decidability of (sub)typing.
This is common in type theories with equality reflection (\eg
extensional MLTT~\citep{Nordstrom90mltt}) because they allow the
reflection of absurd assumptions.
For example, in a context where $Z \kin \Top \intv \Bot$, we have
$\Top \tsub Z \tsub \Bot$, \ie the subtyping relation becomes trivial,
and we can type non-terminating and stuck terms.
\begin{align*}
  &\app{(\lam{x}{\Top}{\app{x}{x}})}{(\lam{x}{\Top}{\app{x}{x}})}
  \tin \Top &
  &\text{\em --- non-terminating, but } \;
  \Top \; \tsub \; Z \; \tsub \; \Bot \; \tsub \; \fun{\Top}{\Top}\\
  &\app{(\Lam{X}{\kstar}{\lam{x}{X}{x}})}{(\lam{x}{\Top}{x})}
  \tin \Top &
  &\text{\em --- stuck, but } \;
  \all{X}{\kstar}{\fun{X}{X}} \; \tsub \; \Top \; \tsub \; Z \; \tsub
    \; \Bot \; \tsub \; \fun{\Top}{\Top}
\end{align*}
It is therefore unsafe to reduce terms under absurd assumptions in
general.
Note that these examples do not break type safety of \FOmegaInt{}
overall though.
The absurd assumption $Z \kin \Top \intv \Bot$ can never be
instantiated, and hence reduction of \emph{closed} terms remains
perfectly safe.
We discuss this point in more detail at the end of \Sec{declarative}.
The use of inconsistent bounds to prove undecidability of subtyping is
more subtle; we return to it in \Sec{undec}.

Seeing the trouble inconsistent bounds can cause, one may wonder why
we do not just enforce consistency of interval bounds statically.
There are several reasons.
\begin{itemize}[left=1pt .. \parindent]
\item \emph{Statically enforcing consistent bounds in Scala is hard.}
While we could statically enforce consistent bounds in \FOmegaInt{},
  \citet{AminRO14oopsla} have shown that this would not extend to
  systems closer to Scala;
a detailed explanation is given by
  \citet[Sec.~3.4.4,~4.2.3]{Amin16thesis}.

\item \emph{Inconsistent bounds are useful.}
Intervals with unconstrained bounds are useful, \eg to encode
  generalized algebraic datatypes (GADTs) via first-class inequality
  constraints~\cite{ParreauxBG19scala,CretinR14lics}.
Consistency of such constraints cannot be established when a GADT is
  defined, only when it is instantiated.

\item \emph{Decidability of subtyping could be recovered.}
Even if we enforced consistent bounds in \FOmegaInt{}, subtyping
  would likely remain undecidable because \FOmegaInt{}, like full
  \FSub{} and all variants of the DOT calculus, use a strong subtyping
  rule for universals that is a known sources of
  undecidability~\cite{Pierce92popl}.
Recent work by~\citet{HuL20popl} suggests a novel approach for
  algorithmic subtyping that handles both inconsistent bounds and strong
  subtyping for universals.
Whether or not their approach can be generalized to higher-order
  subtyping is a question we leave for future work.
\end{itemize}

\disableLstShortInline

\section{The Declarative System}
\label{sec:declarative}

In this section, we introduce \FOmegaInt{} -- our formal theory of
higher-order subtyping with type intervals.
We present its syntax and its type system, and establish some basic
metatheoretic properties
-- just enough to show that subject reduction holds for well-kinded
open types.
Finally, we
discuss the challenges involved in proving type safety, and outline
our strategy for doing so.

\subsection{Syntax}
\label{sec:syntax}

\fig[tb]{fig:syntax}{Syntax of \FOmegaInt.}{\begin{align*}
x, y, z&, \dotsc             \hspace{6em}\secfnt{Term variable}\hspace{3.2em}
X, Y, Z, \dotsc              &&\secfnt{Type variable}\\
s, t \; &::= \; x
  \orElse \lam{x}{A}{t}
  \orElse \app{s}{t}
  \orElse \Lam{X}{K}{t}
  \orElse \app{t}{A}         &&\secfnt{Term}\\
u, v \; &::= \; \lam{x}{A}{t}
  \orElse \Lam{X}{K}{t}      &&\secfnt{Value}\\
A, B, C \; &::= \; X
  \orElse \Top
  \orElse \Bot
  \orElse \fun{A}{B}
  \orElse \all{X}{K}{A}
  \orElse \lam{X}{K}{A}
  \orElse \app{A}{B}         &&\secfnt{Type}\\
\Gamma, \Delta \; &::= \; \cempty
  \orElse \Gamma, x \tas A
  \orElse \Gamma, X \kas K   &&\secfnt{Typing context}\\
J, K, L \; &::= \; A \intv B
  \orElse \dfun{X}{J}{K}     \hspace{1.7em}\secfnt{Kind}\hspace{3.2em}
j , k , l \; ::= \; *
  \orElse \fun{j}{k}         &\quad&\secfnt{Shape (simple kind)}
\end{align*}\vspace{-1.5em}
}
 
The syntax of \FOmegaInt{} is given in~\Fig{syntax}.
The syntax of terms and types is identical to that of \FOmega{} except
for the extremal type constants $\Top$ and $\Bot$.
The \emph{top type} $\Top$ is the maximal proper type:
any other proper type is a subtype of $\Top$.
Dually, the \emph{bottom type} $\Bot$ is the minimal proper type.
Following \citet{Pierce02tapl},
$\lambda$s
carry domain annotations.
This will become important in \Sec{decl_rules}, \Sec{raw_nf}
and \Sec{commute}.

The main differences between \FOmegaInt{} and other variants of
\FOmega{} are reflected in its kind language.
First, the usual kind of proper types~$\kstar$ is replaced by the
\emph{interval kind} former~$A \intv B$.
The interval $A \intv B$ is inhabited by exactly those proper types
that are supertypes of $A$ and subtypes of $B$.
The degenerate interval~$\Bot \intv \Top$ spans \emph{all} proper
types.
Hence we use $\kstar$ as a shorthand for $\Bot \intv \Top$.
Second, most variants of \FOmega{} have a \emph{simple} kind
language (as described by the non-terminal $k$ in~\Fig{syntax}).
In contrast, \FOmegaInt{} has a \emph{dependent} kind language.
The arrow kind $\dfun{X}{J}{K}$ acts as a binder for the type variable
$X$ which may appear freely in the codomain $K$.
Dependent kinds play an important role when modeling \emph{bounded
  type operators}.
For example, consider a binary type operator of kind
$\dfun{X}{\kstar}{\dfun{Y}{\Bot \intv X}{\kstar}}$.
The upper-bounded kind $\Bot \intv X$ of $Y$ ensures that the operator
can only be applied to types $A$, $B$ if $B$ is a subtype of $A$.
This idea goes back to \citeauthor{CompagnoniG03ic}'s~\FOmegaSubC,
which features both upper-bounded type operators and dependent arrow
kinds~\citeyearpar{CompagnoniG03ic}.

We abbreviate $\Bot \intv \Top$ by $\kstar$ and $\dfun{X}{J}{K}$ by
$\fun{J}{K}$ when $X$ is not free in $K$.
This allows us to treat
\emph{simple kinds} or \emph{shapes} $k$ as a subset of
(dependent) kinds $K$.
In the opposite direction, we define an \emph{erasure} map $\ksimp{K}$
which forgets any dependencies in
$K$
(see~\Fig{encodings}).
Given a kind~$K$, we say \emph{$K$~has shape~$\ksimp{K}$}.
Unlike kinds, shapes are stable under
substitution, \ie $\ksimp{\subst{K}{X}{A}} \alEq \ksimp{K}$.

Following \citet{Barendregt93lct}, we identify expressions $e$ (terms,
types and kinds), up to \emph{$\alpha$-equivalence} and assume that
the names of bound and free variables are distinct.
We write $e \alEq e'$
to stress that
$e$ and $e'$ are $\alpha$-equivalent.
The set of free variables of $e$ is denoted by $\fv(e)$,
and we write $\subst{e}{x}{t}$ and $\subst{e}{X}{A}$ for
capture-avoiding term and type substitutions in $e$, respectively.
We require that the variables bound in a typing context~$\Gamma$ be
distinct so that we may think of $\Gamma$ as a finite map
and use function notation, such as $\dom(\Gamma)$, $\Gamma(x)$,
$\Gamma(X)$. We write $(\Gamma, \Delta)$ for the concatenation of two
contexts~$\Gamma$ and~$\Delta$ with disjoint domains,
and we often omit the empty context $\cempty$, writing \eg
$\Gamma = x \tas A, Y \tas K$ instead of
$\Gamma = \cempty, x \tas A, Y \tas K$.

\subsubsection{Encodings}
\label{sec:encodings}

\fig[tb]{fig:encodings}{Syntactic shorthands and encodings}{\begin{minipage}[t]{0.18\linewidth}
  \judgment{Kind constants}{\vphantom{\fbox{$A \hointv{K} B$}}}
  \begin{align*}
    \kstar  & \; := \; \Bot \intv \Top \\
    \kempty & \; := \; \Top \intv \Bot
  \end{align*}
\end{minipage}\hfill
\begin{minipage}[t]{0.75\textwidth}
  \judgment{Higher-order type intervals}{\fbox{$A \hointv{K} B$} \,
    \fbox{\vphantom{$A$}$\kmax{K}$}}
  \[
    \begin{alignedat}{2}
      A &\hointv{A' \intv B'}    &&B \; := \; A \intv B \\
      A &\hointv{\dfun{X}{J}{K}} &&B \; := \;
      \dfun{X}{J}{\app{A}{X}\hointv{K}\app{B}{X}}\\
        &&&\phantom{B \; := \; {}} \text{for } X \notin \fv(A) \cup \fv(B)
    \end{alignedat}\hspace{2.8em}
    \begin{alignedat}{2}
      &\kmax{A \intv B}      &&\; := \; \Bot \intv \Top \\
      &\kmax{\dfun{X}{J}{K}} &&\; := \; \dfun{X}{J}{\kmax{K}}\\
      &~
    \end{alignedat}
  \]
\end{minipage}\bigskip

\begin{minipage}[t]{0.55\textwidth}
  \judgment{Higher-order extrema}{\fbox{$\tmax{K}$} \,
      \fbox{$\tmin{K}$}}
  \[
    \begin{alignedat}{2}
      &\tmax{A \intv B}      &&\; := \; \Top \\
      &\tmax{\dfun{X}{J}{K}} &&\; := \; \lam{X}{J}{\tmax{K}}
    \end{alignedat}\hspace{3em}
    \begin{alignedat}{2}
      &\tmin{A \intv B}      &&\; := \; \Bot \\
      &\tmin{\dfun{X}{J}{K}} &&\; := \; \lam{X}{J}{\tmin{K}}
    \end{alignedat}
  \]
\end{minipage}\hfill \begin{minipage}[t]{0.36\linewidth}
  \judgment{Type erasure}{\fbox{$\ksimp{K}$}}
  \begin{alignat*}{2}
      &\ksimp{A \intv B}      &&\; := \; \kstar \\
      &\ksimp{\dfun{X}{J}{K}} &&\; := \; \fun{\ksimp{J}}{\ksimp{K}}
  \end{alignat*}
\end{minipage}\bigskip

\judgment{Bounded quantification and type operators}{\hphantom{\fbox{$\ksimp{K}$}}}
\begin{align*}
    \all{X \tsub A}{K}{B}  \; &:= \; \all{X}{(\tmin{K}) \hointv{K} A}{B} &
    \dfun{X \tsub A}{J}{K} \; &:= \; \dfun{X}{(\tmin{J}) \hointv{J} A}{K}\\
    \Lam{X \tsub A}{K}{t}  \; &:= \; \Lam{X}{(\tmin{K}) \hointv{K} A}{t} &
    \lam{ X \tsub A}{K}{B} \; &:= \; \lam{X}{ (\tmin{K}) \hointv{K} A}{B}
\end{align*}\vspace{-1.5em}}
 
Together with the extremal types $\Top$ and $\Bot$, interval kinds
allow us to express \emph{bounded quantification} and \emph{bounded
  operators} over proper types.
For example, the \FSub{}-style universal type $\call{X \tsub A}{B}$
can be expressed as $\all{X}{\Bot \intv A}{B}$ in \FOmegaInt.
To extend this principle to \emph{higher-order} bounded quantification
and type operators, we define encodings for higher-order interval
kinds and extremal types via type abstraction and dependent kinds in
\Fig{encodings}.
The encoding of \emph{higher-order maxima}~$\tmax{K}$ is
standard~\citep[cf.][]{Pierce02tapl,CompagnoniG03ic}; that of
\emph{higher-order minima}~$\tmin{K}$ follows the same principle.
The encoding of \emph{higher-order interval kinds}~$A \hointv{K} B$
resembles that of higher-order \emph{singleton kinds} given
by~\citet{StoneH00popl}.
Indeed,
singleton kinds are just interval kinds where the upper and lower bounds coincide.
Encodings of higher-order \FOmegaSubC{}-style bounded operators and
universal quantifiers are also given in \Fig{encodings}.

\subsubsection{Structural Operational Semantics}
\label{sec:sos}

For computations in terms, we adopt the standard call-by-value~(CBV)
semantics given by~\citet[Fig.~30-1]{Pierce02tapl},
writing $t \cbvred t'$ when the term $t$ CBV-reduces in one or more
steps to $t'$.
For types and kinds, we define the \emph{one-step $\beta$-reduction}
relation $\betastep$ as the compatible closure of $\beta$-contraction
of type operators \wrt all the type and kind formers.
We write $\betared$ for its reflexive, transitive closure,
\emph{$\beta$-reduction}.

\subsection{Declarative Typing and Kinding}
\label{sec:decl_rules}

The static semantics of \FOmegaInt{} are summarized in
\Fig[Figs.]{decl_rules1} and~\ref{fig:decl_rules2}.
We refer to this set of judgments as the \emph{declarative} system, as
opposed to the \emph{canonical} system introduced in \Sec{canonical}.
We sometimes write $\judg{\Gamma}$ to denote an arbitrary judgment of
the declarative system. Throughout the paper, we silently assume that judgments are
\emph{well-scoped}, \ie
if $\judg{\Gamma}$ then $\fv(\genjudg) \subseteq \dom(\Gamma)$.
We now discuss each of the judgments, emphasizing novel rules.

\fig{fig:decl_rules1}{Declarative presentation of \FOmegaInt{} -- part 1.
Premises \hlbox{\text{in gray}} are validity conditions
  (see~\Sec{decl_basics_short}).}{\judgment{Context formation}{\fbox{$\ctxD{\Gamma}$}}
\begin{center}
  \begin{minipage}[t]{0.2\linewidth}
    \infruleLeft{C-Empty}{\quad}{\ctxD{\cempty}}
  \end{minipage}\hspace{1em}
  \begin{minipage}[t]{0.36\linewidth}
    \infruleLeft{C-TmBind}
    {\ctxD{\Gamma} \andalso \kindD{\Gamma}{K}}
    {\ctxD{\Gamma, X \kas K}}
  \end{minipage}\hspace{1em}
  \begin{minipage}[t]{0.35\linewidth}
    \infruleLeft{C-TpBind}
    {\ctxD{\Gamma} \andalso \typeD{\Gamma}{A}}
    {\ctxD{\Gamma, x \tas A}}
  \end{minipage}
\end{center}

\judgment{Kind formation}{\fbox{$\kindD{\Gamma}{K}$}}
\begin{multicols}{2}
  \typicallabel{K-Abs}
\infrule[\ruledef{Wf-Intv}]
  {\typeD{\Gamma}{A} \andalso \typeD{\Gamma}{B}}
  {\kindD{\Gamma}{A \intv B}}

\infrule[\ruledef{Wf-DArr}]
  {\kindD{\Gamma}{J} \andalso \kindD{\Gamma, X \kas J}{K}}
  {\kindD{\Gamma}{\dfun{X}{J}{K}}}
\end{multicols}

\judgment{Kinding}{\fbox{$\Gamma \ts A \kin K$}}
\begin{multicols}{3}
  \typicallabel{}
\infrule[\ruledef{K-Var}]
  {\ctxD{\Gamma}  \andalso  \Gamma(X) = K}
  {\Gamma \ts X \kin K}

\infrule[\ruledef{K-Top}]
  {\ctxD{\Gamma}}
  {\Gamma \ts \Top \kin \kstar}

\infrule[\ruledef{K-Bot}]
  {\ctxD{\Gamma}}
  {\Gamma \ts \Bot \kin \kstar}
\end{multicols}
\begin{multicols}{2}
  \typicallabel{K-Abs}
\infrule[\ruledef{K-Arr}]
  {\typeD{\Gamma}{A}  \andalso  \typeD{\Gamma}{B}}
  {\Gamma \ts \fun{A}{B} \kin \kstar}

\infrule[\ruledef{K-Abs}]
  {\kind{\Gamma}{J}  \andalso  \Gamma, X \kas J \ts A \kin K\\
   \hlbox{\kind{\Gamma, X \kas J}{K}}}
  {\Gamma \ts \lam{X}{J}{A} \kin \dfun{X}{J}{K}}

\infrule[\ruledef{K-Sing}]
  {\Gamma \ts A \kin B \intv C}
  {\Gamma \ts A \kin A \intv A}

\infrule[\ruledef{K-All}]
  {\kindD{\Gamma}{K}  \andalso  \typeD{\Gamma, X \kas K}{A}}
  {\Gamma \ts \all{X}{K}{A} \kin \kstar}

\infrule[\ruledef{K-App}]
  {\Gamma \ts A \kin \dfun{X}{J}{K}  \andalso  \Gamma \ts B \kin J\\
   \hlbox{\kind{\Gamma, X \kas J}{K}       \andalso
          \kind{\Gamma}{\subst{K}{X}{B}}}}
  {\Gamma \ts \app{A}{B} \kin \subst{K}{X}{B}}

\infrule[\ruledef{K-Sub}]
  {\Gamma \ts A \kin J  \andalso  \Gamma \ts J \ksub K}
  {\Gamma \ts A \kin K}
\end{multicols}

\judgment{Typing}{\fbox{$\Gamma \ts t \tin A$}}
\begin{multicols}{2}
  \typicallabel{T-Var}
\infrule[\ruledef{T-Var}]
  {\ctxD{\Gamma}  \andalso  \Gamma(x) = A}
  {\Gamma \ts x \tin A}

\infrule[\ruledef{T-Abs}]
  {\typeD{\Gamma}{A}  \andalso  \typeD{\Gamma}{B}\\
   \Gamma, x \tas A \ts t \tin B}
  {\Gamma \ts \lam{x}{A}{t} \tin \fun{A}{B}}

\infrule[\ruledef{T-App}]
  {\Gamma \ts s \tin \fun{A}{B}  \andalso  \Gamma \ts t \tin A}
  {\Gamma \ts \app{s}{t} \tin B}

\infrule[\ruledef{T-TAbs}]
  {\kindD{\Gamma}{K}  \andalso  \Gamma, X \kas K \ts t \tin A}
  {\Gamma \ts \Lam{X}{K}{t} \tin \all{X}{K}{A}}

\infrule[\ruledef{T-TApp}]
  {\quad\\
   \Gamma \ts t \tin \all{X}{K}{A}  \andalso  \Gamma \ts B \kin K}
  {\Gamma \ts \app{t}{B} \tin \subst{A}{X}{B}}

\infrule[\ruledef{T-Sub}]
  {\Gamma \ts t \tin A  \andalso  \Gamma \ts A \tsub B \kin \kstar}
  {\Gamma \ts t \tin B}
\end{multicols}\vspace{-1em}

}
 \fig{fig:decl_rules2}{Declarative presentation of \FOmegaInt{} -- part 2.
Premises \hlbox{\text{in gray}} are validity conditions
  (see~\Sec{decl_basics_short}).}{\judgment{Subkinding}{\fbox{$\Gamma \ts J \ksub K$}}
\begin{multicols}{2}
  \typicallabel{SK-Intv}
\infrule[\ruledef{SK-Intv}]
  {\quad\\
   \Gamma \ts A_2 \tsub A_1 \kin \kstar  \andalso
   \Gamma \ts B_1 \tsub B_2 \kin \kstar}
  {\Gamma \ts A_1 \intv B_1 \ksub A_2 \intv B_2}

\infrule[\ruledef{SK-DArr}]
  {\kind{\Gamma}{\dfun{X}{J_1}{K_1}}\\
   \Gamma \ts J_2 \ksub J_1  \andalso
   \Gamma, X \kas J_2 \ts K_1 \ksub K_2}
  {\Gamma \ts \dfun{X}{J_1}{K_1} \ksub \dfun{X}{J_2}{K_2}}
\end{multicols}

\judgment{Subtyping}{\fbox{$\Gamma \ts A \ksub B \kin K$}}
\begin{multicols}{2}
  \typicallabel{ST-Refl}
\infrule[\ruledef{ST-Refl}]
  {\Gamma \ts A \kin K}
  {\Gamma \ts A \tsub A \kin K}

\infrule[\ruledef{ST-Top}]
  {\Gamma \ts A \kin B \intv C}
  {\Gamma \ts A \tsub \Top \kin \kstar}

\infrule[\ruledefN{ST-Beta1}{ST-$\beta_1$}]
  {\Gamma, X \kas J \ts A \kin K  \andalso  \Gamma \ts B \kin J\\
   \hlbox{\Gamma \ts \subst{A}{X}{B} \kin \subst{K}{X}{B}}\\
   \hlbox{\kind{\Gamma, X \kas J}{K}       \andalso
          \kind{\Gamma}{\subst{K}{X}{B}}}}
  {\Gamma \ts \app{(\lam{X}{J}{A})}{B} \tsub \subst{A}{X}{B} \kin
    \subst{K}{X}{B}}

\infrule[\ruledefN{ST-Eta1}{ST-$\eta_1$}]
  {\Gamma \ts A \kin \dfun{X}{J}{K}  \andalso  X \notin \fv(A)}
  {\Gamma \ts \lam{X}{J}{\app{A}{X}} \tsub A \kin \dfun{X}{J}{K}}

\infrule[\ruledef{ST-Arr}]
  {\Gamma \ts A_2 \tsub A_1 \kin \kstar  \andalso
   \Gamma \ts B_1 \tsub B_2 \kin \kstar}
  {\Gamma \ts \fun{A_1}{B_1} \tsub \fun{A_2}{B_2} \kin \kstar}

\infrule[\ruledef{ST-Abs}]
  {\Gamma \ts \lam{X}{J_1}{A_1} \kin \dfun{X}{J}{K}\\
   \Gamma \ts \lam{X}{J_2}{A_2} \kin \dfun{X}{J}{K}\\
   \Gamma, X \kas J \ts A_1 \tsub A_2 \kin K}
  {\Gamma \ts \lam{X}{J_1}{A_1} \tsub \lam{X}{J_2}{A_2} \kin \dfun{X}{J}{K}}

\infrule[\ruledefN{ST-Bnd1}{ST-Bnd$_1$}]
  {\Gamma \ts A \kin B_1 \intv B_2}
  {\Gamma \ts B_1 \tsub A \kin \kstar}

\infrule[\ruledef{ST-Intv}]
  {\Gamma \ts A_1 \tsub A_2 \kin B \intv C}
  {\Gamma \ts A_1 \tsub A_2 \kin A_1 \intv A_2}

\infrule[\ruledef{ST-Trans}]
  {\Gamma \ts A \tsub B \kin K  \andalso  \Gamma \ts B \tsub C \kin K}
  {\Gamma \ts A \tsub C \kin K}

\infrule[\ruledef{ST-Bot}]
  {\Gamma \ts A \kin B \intv C}
  {\Gamma \ts \Bot \tsub A \kin \kstar}

\infrule[\ruledefN{ST-Beta2}{ST-$\beta_2$}]
  {\Gamma, X \kas J \ts A \kin K  \andalso  \Gamma \ts B \kin J\\
   \hlbox{\Gamma \ts \subst{A}{X}{B} \kin \subst{K}{X}{B}}\\
   \hlbox{\kind{\Gamma, X \kas J}{K}       \andalso
          \kind{\Gamma}{\subst{K}{X}{B}}}}
  {\Gamma \ts \subst{A}{X}{B} \tsub \app{(\lam{X}{J}{A})}{B} \kin
    \subst{K}{X}{B}}

\infrule[\ruledefN{ST-Eta2}{ST-$\eta_2$}]
  {\Gamma \ts A \kin \dfun{X}{J}{K}  \andalso  X \notin \fv(A)}
  {\Gamma \ts A \tsub \lam{X}{J}{\app{A}{X}} \kin \dfun{X}{J}{K}}

\infrule[\ruledef{ST-All}]
  {\typeD{\Gamma}{\all{X}{K_1}{A_1}}\\
   \Gamma \ts K_2 \ksub K_1  \quad\,
   \Gamma, X \kas K_2 \ts A_1 \tsub A_2 \kin \kstar}
  {\Gamma \ts \all{X}{K_1}{A_1} \tsub \all{X}{K_2}{A_2} \kin \kstar}

\infrule[\ruledef{ST-App}]
  {\Gamma \ts A_1 \tsub A_2 \kin \dfun{X}{J}{K}  \andalso
   \Gamma \ts B_1 \teq B_2 \kin J\\
   \hlbox{\Gamma \ts B_1 \kin J              \quad \kind{\Gamma, X \kas J}{K}         \quad \kind{\Gamma}{\subst{K}{X}{B_1}}}}
  {\Gamma \ts \app{A_1}{B_1} \tsub \app{A_2}{B_2} \kin \subst{K}{X}{B_1}}

\infrule[\ruledefN{ST-Bnd2}{ST-Bnd$_2$}]
  {\Gamma \ts A \kin B_1 \intv B_2}
  {\Gamma \ts A \tsub B_2 \kin \kstar}

\infrule[\ruledef{ST-Sub}]
  {\Gamma \ts A_1 \tsub A_2 \kin J  \andalso  \Gamma \ts J \ksub K}
  {\Gamma \ts A_1 \tsub A_2 \kin K}
\end{multicols}\bigskip

\begin{minipage}[2]{0.43\linewidth}
  \judgment{Kind equality}{\fbox{$\Gamma \ts J \keq K$}}

\infruleLeft{SK-AntiSym}
  {\Gamma \ts J \tsub K  \andalso  \Gamma \ts K \tsub J}
  {\Gamma \ts J \teq K}
\end{minipage}\hfill
\begin{minipage}[2]{0.51\linewidth}
  \judgment{Type equality}{\fbox{$\Gamma \ts A \keq B \kin K$}}

\infruleLeft{ST-AntiSym}
  {\Gamma \ts A \tsub B \kin K  \andalso  \Gamma \ts B \tsub A \kin K}
  {\Gamma \ts A \teq B \kin K}
\end{minipage}
\smallskip

}
 
\subsubsection{Context and Kind Formation}

The rules for context formation $\ctxD{\Gamma}$ are standard.
They ensure that the type and kind annotations of all bindings are
well-formed.
The rules of the remaining judgments are set up so that they can only
be derived in well-formed contexts.

Our kind formation judgment $\kindD{\Gamma}{K}$ ensures that
\begin{inparaenum}[(a)]
\item all types appearing in $K$ are well-kinded,
and
\item that the bounds of intervals are proper types (not $\lambda$s),
  forbidding for instance $\Bot \intv \lam{X}{K}{A}$.
\end{inparaenum}
The formation rule \ruleref{Wf-DArr} for dependent arrows is standard.
An interval $A \intv B$ is well-formed if
$A$ and $B$ are proper types.
As discussed in \Sec{inconsistent}, we choose not to enforce
$A \tsub B$;
\eg the empty kind $\kempty = \Top \intv \Bot$ is well-formed.
We say that the bounds of an interval $A \intv B$ are
\emph{consistent} in $\Gamma$ if $\Gamma \ts A \tsub B \kin \kstar$
and \emph{inconsistent} otherwise.
If $A$ and $B$ are closed types and $\nts A \tsub B \kin \kstar$ in
the empty context, we say that the bounds of $A \intv B$ are
\emph{absurd}.
For example, the bounds of $\kstar$ are always consistent
while those of $\Top \intv \Bot$ are absurd.\footnote{See~\Lem{st-inv} in \Sec{canon-inversion}.}
The bounds of $X \intv Y$ are inconsistent in
$\Gamma = X \kas \kstar, Y \kas \kstar$
but not absurd because $X$ and $Y$ are open types.

\subsubsection{Kinding and Typing}

The kinding rules \ruleref{K-Var}, \ruleref{K-Top}, \ruleref{K-Bot},
\ruleref{K-Arr} and~\ruleref{K-All} are all standard.
The rule~\ruleref{K-All} resembles
that found in~\FOmega{}:
no bound annotations are needed because the bounds of $X$ are
internalized in the kind $K$.
Similarly, \ruleref{K-Abs} and~\ruleref{K-App}, resemble those in~\FOmega{} more
than those in \KFOmegaSub{}. The rules \ruleref{K-Sing} and~\ruleref{K-Sub} are used to adjust the
kind of a type:
\ruleref{K-Sub} is the kind-level analog of~\ruleref{T-Sub}
(subsumption);
\ruleref{K-Sing} resembles \citeauthor{StoneH00popl}'s singleton
introduction rule~\citeyearpar{StoneH00popl}.
Note that \ruleref{K-Sing} only \emph{narrows} the kind of a type
whereas \ruleref{K-Sub} only \emph{widens} it.
The premise of \ruleref{K-Sing} may look a bit surprising:
why use $\Gamma \ts A \kin B \intv C$ instead of
$\Gamma \ts A \kin \kstar$~?
This extra flexibility is necessary to prove that types inhabiting
intervals are proper types, \ie that $\Gamma \ts A \kin B \intv C$
implies $\Gamma \ts A \kin \kstar$.
Thus the relaxed premise justifies itself.

The typing rules are again entirely standard, with the possible
exception of some additional context and kind formation premises that
would be redundant in variants of \FOmega{} with simple kinds.

\subsubsection{Subkinding and Subtyping}

\FOmegaInt{} features both subtyping and \emph{subkinding}.
The use of subkinding in~\FOmegaInt{} is both natural and essential.
If a type is contained in an interval $A \intv B$, then one naturally
expects it to also be contained in a \emph{wider} interval
$A' \intv B'$ where $A' \tsub A$ and $B \tsub B'$.
This is captured in the rule~\ruleref{SK-Intv}.
Subkinding is also essential.
It is thanks to \ruleref{SK-Intv} that we can express bounded polymorphism and bounded
type operators in~\FOmegaInt{}.
Consider \eg a polymorphic term
$t \tin \all{X}{\Bot \intv A}{B}$ and a type argument $C$ such that
$C \tsub A \kin \kstar$.
We can apply $t$ to $C$ because $C$ has kind $C \intv C$
(by~\ruleref{K-Sing}) which in turn is a subkind of $\Bot \intv A$
(by~\ruleref{SK-Intv}).
The rule~\ruleref{SK-DArr} lifts the interval containment order
through dependent arrow kinds.
It resembles \citeauthor{AspinallC01tcs}'s subtyping
rule~\textsc{(s-$\pi$)} for dependent product
types~\citet{AspinallC01tcs}.

Subtyping judgments $\Gamma \ts A \tsub B \kin K$ are indexed by the
common kind $K$ in which $A$ and $B$ are related.
Note that two types may be related in some kinds but not others.
For example, the extremal types $\Top$ and $\Bot$ are related as
proper types, \ie $\ts \Bot \tsub \Top \kin \kstar$, but not as
inhabitants of their respective singleton kinds $\Bot \intv \Bot$ and
$\Top \intv \Top$.
For a given context $\Gamma$ and kind $K$, the subtyping relation
$\Gamma \ts A \tsub B \kin K$ is a preorder, as witnessed by the rules~\ruleref{ST-Refl} and~\ruleref{ST-Trans}.
Note that there are no such rules for subkinding,
but it is easy to prove them admissible.
In \Sec{canonical}, we will see that some (but not all) instances of
\ruleref{ST-Refl} and~\ruleref{ST-Trans} can be eliminated too.

The rules~\ruleref{ST-Top} and~\ruleref{ST-Bot} establish $\Top$
and~$\Bot$ as the maximum and minimum proper types \wrt subtyping.
The premise $\Gamma \ts A \kin B \intv C$ ensures that the extremal
types are only related to other proper types.
As for \ruleref{K-Sing}, the
kind $B \intv C$ in the
premise allows us to prove that types inhabiting intervals are proper
types:
\vspace{-\baselineskip}
\begin{prooftree}
  \AxiomC{$\Gamma \ts A \kin B \intv C$}
  \LeftLabel{\rulerefP{K-Sing}}
  \UnaryInfC{$\Gamma \ts A \kin A \intv A$}
    \AxiomC{$\Gamma \ts A \kin B \intv C$}
    \LeftLabel{\rulerefP{ST-Bot}}
    \UnaryInfC{$\Gamma \ts \Bot \tsub A \kin \kstar$}
      \AxiomC{$\Gamma \ts A \kin B \intv C$}
      \RightLabel{\rulerefP{ST-Top}}
      \UnaryInfC{$\Gamma \ts A \tsub \Top \kin \kstar$}
    \RightLabel{\rulerefP{SK-Intv}}
    \BinaryInfC{$\Gamma \ts A \intv A \ksub \Bot \intv \Top$}
  \RightLabel{\rulerefP{K-Sub}}
  \BinaryInfC{$\Gamma \ts A \kin \Bot \intv \Top$}
\end{prooftree}
This derivation would not be possible if the rules~\ruleref{K-Sing},
\ruleref{ST-Bot} or~\ruleref{ST-Top} had premise
$\Gamma \ts A \kin \kstar$.

The rules~\rulerefN{ST-Beta1}{ST-$\beta_1$}
and~\rulerefN{ST-Beta2}{ST-$\beta_2$} correspond to
$\beta$-contraction and expansion, respectively.
Two separate rules are needed because subtyping is not symmetric.
We could have combined them into a single type equality rule but that
would have complicated the definition of type equality.
Similarly, the rules~\rulerefN{ST-Eta1}{ST-$\eta_1$}
and~\rulerefN{ST-Eta2}{ST-$\eta_2$} relate $\eta$-convertible types.
The rule for universals resembles
\ruleref{SK-DArr}.

Most variants of \KFOmegaSub{} separate subtyping of type operator
applications into a subtyping rule that only compares the heads of
applications, and a congruence rule for type equality \wrt
application.
Here we fuse these two rules into a single subtyping
rule~\ruleref{ST-App}.
Since we do not track the variance of type operators,
the arguments must be equal types.
Because arrow kinds are dependent,
either $B_1$ or $B_2$
must be substituted for $X$ in $K$ in the conclusion.
Both are equally suitable; we pick $B_1$.

The rule for subtyping operator abstractions, \ruleref{ST-Abs}, is
maybe the most unusual when compared to other variants of \FOmega{}
since it allows abstractions to be subtypes even if their domain
annotations $J_1$ and $J_2$ are not subkinds.
Other systems adopt weaker versions of this rule where $J \alEq J_2$ or even $J \alEq J_1 \alEq J_2$.
But such rules are not suitable for a theory featuring both subkinding
and $\eta$-equality.
Let $A \kin \fun{\kstar}{\kstar}$ be an operator and $B \kin \kstar$ a
type in $\Gamma$.
By \ruleref{K-Sub}, \ruleref{ST-Sub}, \rulerefN{ST-Eta1}{ST-Eta$_{1,2}$}
and antisymmetry, \[
  \Gamma \; \ts \;
  \lam{X}{\Bot \intv B}{\app{A}{X}}
  \; \teq \;
  A
  \; \teq \;
  \lam{X}{B \intv \Top}{\app{A}{X}}
  \; \kin \; \dfun{X}{B \intv B}{\kstar}
\]
\ie the two $\eta$-expansions of $A$ are equal as types, despite
having distinct domain annotations.
Because the $\eta$-rules allow such equations, we adopt a compatible
subtyping rule for abstractions.
The first two premises of~\ruleref{ST-Abs} ensure that both
abstractions --~irrespective of their domain annotations~-- inhabit
the common arrow kind $\dfun{X}{J}{K}$;
the remaining premise ensures that the bodies of the two abstractions
are pointwise subtypes assuming the common domain $X \kas J$.
Note that systems without domain annotations
avoid such complications~\citep[cf.][]{Abel08mscs}.

So far, we have seen how a type interval $A \intv B$ can be formed
using~\ruleref{Wf-Intv}, and introduced using~\ruleref{K-Sing}, but we
have yet to see how the bounds $A$ and $B$ of the interval can be put
to use.
Type intervals are ``eliminated'' by turning them into subtyping
judgments via a pair of \emph{bound projection}
rules~\rulerefN{ST-Bnd1}{ST-Bnd$_1$}
and~\rulerefN{ST-Bnd2}{ST-Bnd$_2$}.
Given a type $A \kin B \intv C$, the
rules~\rulerefN{ST-Bnd1}{ST-Bnd$_1$}
and~\rulerefN{ST-Bnd2}{ST-Bnd$_2$} assert that $B$ and $C$ are indeed
lower and upper bounds, respectively, of $A$.
When $A$ is a variable, we may use rule~\rulerefN{ST-Bnd2}{ST-Bnd$_2$}
to derive judgments of the form $\Gamma \ts X \tsub C$, similar to
those obtained using the variable subtyping rule from \FSub{}.
More generally, the bound projection rules allow us to \emph{reflect}
any well-formed assumption $\Gamma(X) = A \intv B$
--~consistent or not~--
into a corresponding subtyping judgment
$\Gamma \ts A \tsub B \kin \kstar$.
We discuss the ramifications this has for type safety in
\Sec{decl_typesafety}.

As for kinding judgments, there are two subtyping rules that allow us
to adjust the kinds of subtyping judgments, \ruleref{ST-Sub}
and~\ruleref{ST-Intv}.
The former is the analog of~\ruleref{K-Sub} for subtyping, whereas the
latter is the subtyping counterpart of the interval introduction
rule~\ruleref{K-Sing}:
if $A$ and $B$ are subtypes in some interval $C \intv D$, then surely
they are still subtypes in the interval $A \intv B$ bounded by
those very same types.
Indeed, $A \intv B$ is the smallest interval in which the two types
are related.
The \ruleref{ST-Intv} rule plays an important role in the proof of
subject reduction for types and kinds (\Thm{beta_red_teq}) because it
allows us to relate $\beta$-equal types inhabiting singleton kinds.

\subsubsection{Kind and Type Equality}

The kind and type equality judgments are each generated by exactly one
rule:
\ruleref{SK-AntiSym} for kind equality and \ruleref{ST-AntiSym}
for type equality.
In most variants of \FOmega{} with subtyping, the subtyping relation
is not defined to be antisymmetric.
Instead antisymmetry may or may not be an admissible property that has
to be proven~\cite[cf.][]{CompagnoniG99csl}.
In \FOmegaInt{}, antisymmetry is not an admissible property, however.
To see this, consider the context $\Gamma = X \kas A \intv A$ for some
proper type $A$.
Then $\Gamma \ts X \kin A \intv A$, and
we can derive that $X$ and $A$ are mutual subtypes
using~\rulerefN{ST-Bnd1}{ST-Bnd$_1$}
and~\rulerefN{ST-Bnd2}{ST-Bnd$_2$}.
But without antisymmetry we have no way to derive
$\Gamma \ts X \teq A \kin \kstar$.
Faced with this issue, we could have chosen to add a
\emph{singleton reflection} rule
for deriving $\Gamma \ts X \teq A \kin \kstar$ from
$\Gamma \ts X \kin A \intv A$ directly,
such as the one due to~\citet{StoneH00popl}.
Interestingly, antisymmetry for proper types is derivable from
\citeauthor{StoneH00popl}'s
rule and other rules about type intervals.
We conjecture that antisymmetry of subtyping under arbitrary kinds
would also have been admissible in such a system, albeit at the cost
of a more complicated type equality judgment.
We prefer the simpler judgment with an explicit antisymmetry rule.

\subsection{Basic Metatheoretic Properties}
\label{sec:decl_basics_short}

With the dynamics and statics in place, we can begin our work on the
metatheory of~\FOmegaInt.
Our system enjoys the usual basic metatheoretic properties, such as
preservation of all the judgments under weakening, substitution and
narrowing of contexts, as well as admissibility of the missing
order-theoretic and congruence rules for subkinding, kind equality and
type equality.
Although these properties constitute the foundation on which we build
the remainder of our metatheory, they are also entirely standard, and
little insight is gained by spelling them out in detail.
We therefore relegate them to
\SupSec{decl_details}.
There, the reader will also find a collection of
admissible rules that justify the encodings of the higher-order
extrema and interval kinds given in~\Sec{encodings}.
These include formation, subtyping and subkinding rules for the
encoded kinds, and typing rules for introducing and eliminating the
more familiar forms of bounded universals.
Unlike the other admissible rules, they are not important for the
remainder of the metatheoretic development.

There are two standard properties of the declarative system that are
exceptional in that their proofs are not routine inductions,
namely \emph{validity} of the various judgments and
\emph{functionality} of substitutions.
Roughly, a judgment is valid if all its parts are well-formed.
\begin{lemma}[validity]\label{lem:decl_validity}
~
\begin{enumerate}[leftmargin=12em,labelsep=1em,topsep=.5ex]\item[(kinding validity)]
If $\; \Gamma \ts A \kin K$,
    then $\; \ctxD{\Gamma}$ and $\; \kindD{\Gamma}{K}$.

\item[(typing validity)]
If $\; \Gamma \ts t \tin A$,
    then $\; \ctxD{\Gamma}$ and $\; \typeD{\Gamma}{A}$.

\item[(kind (in)equation validity)]
If $\; \Gamma \ts J \keq K$ or $\; \Gamma \ts J \ksub K$,
    then $\; \kindD{\Gamma}{J}$ and $\; \kindD{\Gamma}{K}$.

\item[(type (in)equation validity)]
    If $\; \Gamma \ts A \teq B \kin K$
    or $\; \Gamma \ts A \tsub B \kin K$,
    then $\; \Gamma \ts A \kin K$ and
    $\; \Gamma \ts B \kin K$.
  \end{enumerate}
\end{lemma}\noindent
The validity lemma provides a ``sanity check'' for the static
semantics, but it also plays a crucial role in the proofs of other
important properties, such as subject reduction or soundness of type
normalization.
Unfortunately, it is harder to prove than one might expect.
The proofs of kinding, subkinding and subtyping validity require the
following \emph{functionality} lemma for the case of \ruleref{ST-App}.
\begin{lemma}[functionality]\label{lem:decl_funct}
  Let $\Gamma \ts A_1 \teq A_2 \kin K$.
  \begin{enumerate}\item If $\; \kindD{\Gamma, X \kas K}{J}$, then
    $\Gamma \ts \subst{J}{X}{A_1} \keq \subst{J}{X}{A_2}$.
  \item If $\; \Gamma, X \kas K \ts B \kin J$, then
    $\Gamma \ts \subst{B}{X}{A_1} \keq \subst{B}{X}{A_2} \kin
    \subst{J}{X}{A_1}$.
  \end{enumerate}
\end{lemma}\noindent
The proof of functionality, in turn, depends on kinding and subtyping
validity.
Proving the two statements by simultaneous induction
is not enough to
resolve the circular dependency because the proof of functionality
would require us to apply the IH to derivations obtained via validity.
Since these are not generally sub-derivations of the relevant premise,
the induction does not go through.

Instead, we follow \citet{HarperP05tocl} and establish validity by
``temporarily extending'' certain rules of the declarative system with
additional premises, which we call \emph{validity conditions},
shown in gray in \Fig[Figs.]{decl_rules1} and~\ref{fig:decl_rules2}.
We then prove functionality and validity for the extended system, show
that the two systems are equivalent and the validity conditions are
redundant after all, and obtain \Lem{decl_validity} for the original
system.
For details see \SupSec{decl_validity_full}.

\subsection{Subject Reduction for Well-Kinded Types}
\label{sec:kind_pres}

For most versions of \FOmega{}, subject reduction
for types is easy to prove because types are simply-kinded.
In \FOmegaInt{}, the proof is complicated by the presence of
type-dependent kinds and subkinding.
However these complications are minor since there are only two
\emph{shapes} of kinds --~intervals and arrows~-- with exactly one
subkinding rule per shape.
Hence subkinding is easy to invert.

To prove subject reduction, we show that $\beta$-reduction steps can be lifted to type and kind equality.
\begin{theorem}\label{thm:beta_red_teq}
  ~
  \begin{enumerate}\item If $\; \kindD{\Gamma}{J}$ and $J \betastep K$, then
    $\Gamma \ts J \keq K$.
  \item If $\; \Gamma \ts A \kin K$ and $A \betastep B$, then
    $\Gamma \ts A \teq B \kin K$.
\end{enumerate}
\end{theorem}\noindent
Subject reduction for kinds and types then follows immediately from
\Thm{beta_red_teq} and validity.
\begin{corollary}[subject reduction for kinding]\label{cor:pres_kind}
~
  \begin{enumerate}\item If $\; \kindD{\Gamma}{J}$ and $J \betared K$, then
    $\kindD{\Gamma}{K}$.
  \item If $\; \Gamma \ts A \kin K$ and $A \betared B$, then
    $\Gamma \ts B \kin K$.
\end{enumerate}
\end{corollary}

\subsection{The Long Road to Type Safety}
\label{sec:decl_typesafety}

After establishing subject reduction for well-kinded types,
we prove type safety via \emph{progress} and \emph{preservation} (aka subject
reduction)~\citep{WrightF94ic}.
But as we show in this section, we must first weaken the statement of
preservation for it to hold in \FOmegaInt{}.

Preservation typically applies to all open terms:
\begin{prop}[preservation]
  If $\; \Gamma \ts t \tin A$ and $t \cbvstep t'$, then $\Gamma \ts t' \tin A$.
\end{prop}\noindent
However, in \FOmegaInt{} this statement fails because reduction of open terms is
unsafe.
The culprit are type variable bindings with absurd bounds.
Consider the following example.
Let $v$ be the polymorphic identity function
$v = \Lam{X}{\kstar}{\lam{x}{X}{x}}$ which is of type
$A = \all{X}{\kstar}{\fun{X}{X}}$.
In~\FOmegaInt{}, closed universals are not subtypes of closed
arrows;\footnote{See~\Lem{st-inv} in \Sec{canon-inversion}.}
hence $v$ cannot be applied to itself.
For the same reason, the term application
$t = \app{\app{(\lam{x}{A}{x})}{v}}{v}$ is ill-typed as a closed term.
Yet, $t$ is well-typed in the context
$\Gamma = X \kas (\fun{A}{A}) \intvP (\fun{A}{\fun{A}{A}})$ because we
can use \rulerefN{ST-Bnd1}{ST-Bnd$_1$} and
\rulerefN{ST-Bnd2}{ST-Bnd$_2$} to derive
$\Gamma \ts \fun{A}{A} \tsub X \tsub \fun{A}{\fun{A}{A}} \kin \kstar$
and subsumption to derive
$\Gamma \ts \lam{x}{A}{x} \kin \fun{A}{\fun{A}{A}} \kin \kstar$.
Note that
the bounds of $X$ are proper types, so
its kind is well-formed, as is the context $\Gamma$.
But since $\nts A \tsub \fun{A}{A} \kin \kstar$, the bounds of the
interval are absurd.

Next, consider what happens when $t$ takes a reduction step.
\[
  \app{\app{(\lam{x}{A}{x})}{v}}{v}
  \; \cbvstep \; \app{(\subst{x}{x}{v})}{v}
  \; \alEq    \; \app{v}{v}.
\]
According to preservation, the application $\app{v}{v}$ should have
type $A$, but instead it is ill-typed, even in $\Gamma$.
The assumption $X \kin (\fun{A}{A}) \intvP (\fun{A}{\fun{A}{A}})$ is
useless here, since $v$ does not have type $\fun{A}{A}$.
Not only is $\app{v}{v}$ ill-typed, it is also stuck.
Hence $\app{v}{v}$ is neither a value nor can it be reduced further --
type safety clearly does not hold in $\Gamma$.

But all is not lost.
Type safety still holds for closed terms, as does a weaker form of
preservation.
\begin{prop}[preservation -- weak version]\label{prop:weak-pres}
If $\; \ts t \tin A$ and $t \cbvstep t'$, then $\ts t' \tin A$.
\end{prop}\noindent
Throughout the next two sections, we will work our way towards a proof
of this proposition.

\subsection{Challenges and Proof Strategy}
\label{sec:strategy}

To conclude the section, let us briefly explore the challenges
involved in proving weak preservation and our strategy to address them.
The complexity of the subtyping relation throws a spanner in the works
when we try to prove weak preservation for cases where $\beta$-contractions occur.
To prove these cases, one normally starts by showing that
the following rules are admissible:
\begin{center}
  \infruleSimp
  {\Gamma \ts \fun{A_1}{B_1} \tsub \fun{A_2}{B_2} \kin \kstar}
  {\Gamma \ts A_2 \tsub A_1 \kin \kstar  \andalso
   \Gamma \ts B_1 \tsub B_2 \kin \kstar}
  \hspace{2em}
  \infruleSimp
  {\Gamma \ts \all{X}{K_1}{A_1} \tsub \all{X}{K_2}{A_2} \kin \kstar}
  {\Gamma \ts K_2 \ksub K_1  \andalso
   \Gamma, X \kas K_2 \ts A_1 \tsub A_2 \kin \kstar}
\end{center}
These properties are generally known as \emph{inversion of subtyping},
and are closely related to the $\Pi$-injectivity property, which is a
well-known source of complexity in dependent type theories.
There are several features of subtyping that severely complicate the
proof of subtyping inversion.
\begin{enumerate}\item\label{item:sinv1} The rules for $\beta$ and $\eta$-conversion,
  together with transitivity, may change the shapes of related types
  in the middle of a subtyping derivation, \eg from a type former to
  a type application.
  \[
    \Gamma \; \ts \; \fun{A_1}{A_2} \; \tsub \;
    \app{(\lam{X}{\kstar}{\fun{X}{A_2}})}{A_1} \; \tsub \; \dotsm \;
    \tsub \; \app{(\lam{X}{\kstar}{\fun{X}{B_2}})}{B_1} \; \tsub \;
    \fun{B_1}{B_2} \; \kin \; \kstar
  \]
\item\label{item:sinv2} The subsumption rule~\ruleref{ST-Sub} may
  change the kinds of related types at any time.
\item\label{item:sinv3} As outlined above,
we can derive judgments of the form
  $\Gamma \ts A \tsub X \tsub B \kin \kstar$ where $A$ and $B$ need
  not be of the same shape, from absurd assumptions in $\Gamma$.
For example
  \[
    \Gamma \; \ts \; \fun{A}{B} \; \tsub \; X \; \tsub \;
    \all{X}{K}{C} \; \kin \; \kstar
  \]
\end{enumerate}
We address these points as follows.
First, we eliminate uses of the $\beta\eta$~rules (\Item{sinv1}) by
adopting an alternative presentation of subtyping which we dub
\emph{canonical subtyping} (\Sec{canonical}).
Canonical subtyping only relates types in \emph{$\eta$-long
  $\beta$-normal form} (\Sec{normalization}),
so there is no need for $\beta\eta$~rules.
This approach works in variants of \FOmega{} with
$\eta$-conversion~\citep[cf.][]{AbelR08csl}, whereas
rewriting-based proofs of
$\Pi$-injectivity~\citep[e.g.][]{Barendregt93lct,Adams06jfp} do not generalize readily to our setting.

The canonical presentation of subtyping also restricts the placement
of subsumption (\Item{sinv2}) to certain strategic positions, just
as in algorithmic or bidirectional subtyping (\Sec{canonical}).

Finally, we avoid issues caused by absurd bounds (\Item{sinv3}) by
proving subtyping inversion only in the empty context, \ie only for
closed types.
That suffices for proving weak preservation and, as the above example
illustrates, it is the best we can do in a system with inequality
reflection (\Sec{canon-inversion}).

The core challenges of the proof of subtyping inversion thus consists
in showing that well-formed kinds and well-kinded types have
normal forms (\Sec{normalization}), so that the canonical and
declarative versions of subtyping can be proven equivalent
(\Sec{canonical}).
We address these challenges in the next two sections.

\section{Normalization of Types}
\label{sec:normalization}

As discussed in the previous section,
we cannot prove inversion of subtyping directly,
because $\beta\eta$-convertible types and kinds
differ in their syntactic structure.
We address this problem in two steps:
(1) in this section, we show that types and kinds in \FOmegaInt{}
can be reduced to $\beta\eta$-normal form
via a \emph{bottom-up normalization} procedure based on
\emph{hereditary substitution};
(2) in next section, we give a canonical presentation of subtyping
that only relates types in normal form.

\subsection{Syntax}

We begin by introducing an alternative syntax for types that is better
suited to our definition of hereditary substitution.
The key difference is that the type arguments of applications
are grouped together in sequences called \emph{spines}.
Hence, we refer to this presentation of types as \emph{spine form}.
\begin{align*}
D, E &::= \app{F}{\vec{E}}     & F, G &::= X
\orElse \Top
\orElse \Bot
\orElse \fun{D}{E}
\orElse \all{X}{K}{E}
\orElse \lam{X}{K}{E}          & \vec{D}, \vec{E} &::= \sempty
\orElse \scons{D}{\vec{E}}     \end{align*}

\noindent
Applications form a separate syntactic category,
\emph{eliminations}~$E$, and consist of a \emph{head}~$F$ and a
spine~$\vec{E}$.
A head is any type former that is not an application.
We adopt vector notation for spines, writing $\vec{E}$ for the
sequence $\vec{E} = E_1, E_2, \dotsc, E_n$ and $(\scons{\vec{D}}{\vec{E}})$ for the concatenation of
$\vec{D}$ and $\vec{E}$.

The two representations of types are isomorphic, so we mix them
freely, knowing that explicit conversions can always be inserted where
necessary.

\subsection{Hereditary Substitution in Raw Types}
\label{sec:raw_hsubst}

In the \Sec{canonical}, we will introduce a system of canonical
judgments defined directly on normal forms.
Since kinds in \FOmegaInt{} are dependent, some of the kinding and
subtyping rules involve substitutions in kinds, e.g.~\ruleref{K-App}
or~\rulerefN{ST-Beta1}{ST-$\beta_{1,2}$}.
Unfortunately, substitutions do not preserve normal forms because
substituting an operator abstraction for the head of a neutral type
introduces a new redex.
For example
$\subst{(\app{Y}{A})}{Y}{\lam{X}{K}{B}} \alEq
\app{(\lam{X}{K}{B})}{A}$ is not a normal form, even if $\app{Y}{A}$
and $\lam{X}{K}{B}$ are.
To define a canonical counterpart of e.g.~\ruleref{K-App} directly on
normal kinds and types, we need a variant of substitution that
immediately eliminates the $\beta$-redexes it introduces.
This type of substitution operation is known as \emph{hereditary
  substitution}~\cite{WatkinsCPW04types}.
The challenge in defining a hereditary substitution function
is, of course, ensuring its totality.

\fig{fig:hsubst}{Hereditary substitution and reducing application}{\small\raggedright \judgment{Hereditary substitution}{\fbox{$\hsubst{D}{X}{k}{E}$}}
\begin{alignat*}{3}
\hsubst{(\app{Y}{\vec{D}}) &}{X}{k}{E} &&\; = \;
  \begin{cases}
    \; \rapp{E}{k}{(\hsubst{\vec{D}}{X}{k}{E})} & \text{ if }\, Y = X,\\
    \; \app{Y}{(\hsubst{\vec{D}}{X}{k}{E})}     & \text{ otherwise,}
  \end{cases}\\
\hsubst{(\app{\Top}{\vec{D}}) &}{X}{k}{E} &&\; = \; \app{\Top}{(\hsubst{\vec{D}}{X}{k}{E})}\\
\hsubst{(\app{\Bot}{\vec{D}}) &}{X}{k}{E} &&\; = \; \app{\Bot}{(\hsubst{\vec{D}}{X}{k}{E})}\\
\hsubst{(\app{(\fun{D_1}{D_2})}{\vec{D}}) &}{X}{k}{E} &&\; = \; \app{(\fun{\hsubst{D_1}{X}{k}{E}}{\hsubst{D_2}{X}{k}{E}})}{(\hsubst{\vec{D}}{X}{k}{E})}\\
\hsubst{(\app{(\all{Y}{K}{D'})}{\vec{D}}) &}{X}{k}{E} &&\; = \; \app{(\all{Y}{\hsubst{K}{X}{k}{E}}{\hsubst{D'}{X}{k}{E}})}{(\hsubst{\vec{D}}{X}{k}{E})} &&\qquad \text{ for } Y \neq X, Y \notin \fv(E),\\
\hsubst{(\app{(\lam{Y}{K}{D'})}{\vec{D}}) &}{X}{k}{E} &&\; = \; \app{(\lam{Y}{\hsubst{K}{X}{k}{E}}{\hsubst{D'}{X}{k}{E}})}{(\hsubst{\vec{D}}{X}{k}{E})} &&\qquad \text{ for } Y \neq X, Y \notin \fv(E).\\[-1ex]
\intertext{\judgment{}{\fbox{$\hsubst{\vec{D}}{X}{k}{E}$}\vspace{-2.25em}}}
\hsubst{\sempty &}{X}{k}{E} &&\; = \; \sempty \\
\hsubst{(\scons{D'}{\vec{D}}) &}{X}{k}{E} &&\; = \; \scons{(\hsubst{D'}{X}{k}{E})}{(\hsubst{\vec{D}}{X}{k}{E})}\\[-1ex]
\intertext{\judgment{}{\fbox{$\hsubst{K}{X}{k}{E}$}\vspace{-1.75em}}}
\hsubst{(D_1 \intv D_2) &}{X}{k}{E} &&\; = \; \hsubst{D_1}{X}{k}{E} \intv \hsubst{D_2}{X}{k}{E} \\
\hsubst{(\dfun{Y}{J}{K}) &}{X}{k}{E} &&\; = \; \dfun{Y}{\hsubst{J}{X}{k}{E}}{\hsubst{K}{X}{k}{E}}
      &&\qquad \text{ for } Y \neq X, Y \notin \fv(E).
\end{alignat*}

\begin{minipage}[t]{.45\linewidth}
\judgment{Reducing application}{\fbox{$\rapp{D}{k}{E}$}}
\begin{alignat*}{3}
D &\rapp{}{\kstar}{} &&E \; = \; \app{D}{E} \\
D &\rapp{}{\fun{k}{l}}{} &&E \; = \;
  \begin{cases}
    \; \hsubst{D'}{X}{k}{E} & \text{ if } \, D = \lam{X}{J}{D'},\\
    \; \app{D}{E}           & \text{ otherwise.}
  \end{cases}
\end{alignat*}
\end{minipage}
\hfill
\begin{minipage}[t]{.52\linewidth}
\judgment{}{\fbox{$\rapp{D}{k}{\vec E}$}}
\begin{alignat*}{3}
D &\rapp{}{k}{\sempty} && \; = \; D \\
D &\rapp{}{k}{(\scons{E}{\vec{E}})} && \; = \;
  \begin{cases}
    \; \rapp{(\rapp{D}{\fun{k_1}{k_2}}{E})}{k_2}{\vec{E}}
    & \text{ if } \, k = \fun{k_1}{k_2},\\
    \; \app{D}{(\scons{E}{\vec{E}})}
    & \text{ otherwise.}
  \end{cases}
\end{alignat*}
\end{minipage}}
 
Our definition of hereditary substitution is given in \Fig{hsubst}.
Hereditary substitution is defined mutually with \emph{reducing application} of eliminations
by recursion on the structure of the shape $k$.
The definitions of the three hereditary substitution functions
proceed by inner recursion on the structure of the parameters
$K$, $D$ and $\vec{D}$, respectively.
Note that there are some recursive calls where no parameter decreases,
but one can check (and we have done so) that at least one of the
relevant parameters decreases strictly along every cycle in the call
graph.
Hence the five functions remain structurally recursive, ensuring their
totality.
Note the crucial use of spine forms:
$\rapp{D}{k}{\vec{E}}$ simultaneously unwinds $\vec{E}$ and $k$ using
the fact that $\vec{E}$ matches the right-associative structure of
$k$.

Our presentation of hereditary substitution differs from others
in the literature.
Like \citet{KellerA10msfp}, we define hereditary substitution by
structural recursion and mutually with reducing application.
But their definition
is based on an intrinsically typed representation,
which does not readily generalize to a system with dependent types (or
kinds).
Instead, like \citet{AbelR08csl} we define hereditary substitution
directly on raw (\ie unkinded) types,
so our definition contains degenerate cases;
unlike \citeauthor{AbelR08csl}'s, our
definition is structurally recursive hence easier to mechanize.
Our approach of defining hereditary substitution by recursion on shapes rather than (dependent) kinds
was inspired by \citeauthor{HarperL07jfp}'s formalization of
Canonical~LF \citeyearpar{HarperL07jfp}.
However,
they define hereditary substitution as an inductive relation,
thus they avoid degenerate cases but must establish
functionality and termination separately.

Because the essential difference between ordinary and hereditary
substitution is that the latter reduces newly created $\beta$-redexes,
the results of the two operations are $\beta$-convertible.
\begin{lemma}\label{lem:subst-red-hsubst}
Let $E$ be an elimination, $X$ a type variable and $k$ a shape, then
\begin{multicolenum}{2}\item\label{item:red-hsubst-kind}
    $\subst{K}{X}{E} \betared \hsubst{K}{X}{k}{E}$ for any kind $K$;
  \item\label{item:red-hsubst-elim}
    $\subst{D}{X}{E} \betared \hsubst{D}{X}{k}{E} \,$ for any type $D$;
\item\label{item:red-rapp-elim}
    $\app{E}{D} \betared \rapp{E}{k}{D} \;$ for any type $D$;
  \item\label{item:red-rapp-sp}
    $\app{E}{\vec{D}} \betared \rapp{E}{k}{\vec{D}} \;$ for any spine $\vec{D}$.
  \end{multicolenum}
\end{lemma}\noindent
It is an immediate consequence of \Lem{subst-red-hsubst} and subject
reduction that ordinary and hereditary substitutions produce
judgmentally equal results.
\begin{corollary}[soundness of hereditary substitution]\label{cor:hsubst-sound}
  Let $\; \Gamma \ts A \kin K$, then
  \begin{enumerate}\item\label{item:hsubst-snd-kind} if
    $\; \kindD{\Gamma, X \kas K, \Delta}{J}$, then
    $\Gamma, \subst{\Delta}{X}{A} \ts \subst{J}{X}{A} \keq
    \hsubst{J}{X}{\ksimp{K}}{A}$;
  \item\label{item:hsubst-snd-type} if
    $\; \Gamma, X \kas K, \Delta \ts B \kin J$, then
    $\Gamma, \subst{\Delta}{X}{A} \ts \subst{B}{X}{A} \teq
    \hsubst{B}{X}{\ksimp{K}}{A} \kin \subst{J}{X}{A}$.
  \end{enumerate}
\end{corollary}

\subsection{Normalization of Raw Types}
\label{sec:raw_nf}

Based on hereditary substitution, we define a bottom-up
\emph{normalization} function $\nfRaw$ on kinds and types.
It is a straightforward extension of the normalization function by
\citet{AbelR08csl} to dependent kinds.
The function $\nfRaw$ is defined directly on raw types and kinds and
relies on a separate function for $\eta$-expanding variables.
The definition of both functions is given in \Fig{normalization}.

\fig[tb]{fig:normalization}{Normalization of types and kinds}{\judgment{$\eta$-expansion}{\fbox{$\etaExp{K}{A}$}}
\begin{align*}
&\etaExpRaw{D_1 \intv D_2}{(A)} \; = \; A &
&\etaExpRaw{\dfun{X}{J}{K}}{(A)} \; = \; \lam{X}{J}{\etaExp{K}{\app{A}{(\etaExp{J}{X})}}}
    \quad \text{ for } X \notin \fv(A).
  &&~
\end{align*}

\judgment{Normalization}{\fbox{$\nf{\Gamma}{A}$}}
\begin{alignat*}{3}
&\nf{\Gamma}{X} && \; = \; \etaExp{\Gamma(X)}{X}
    &&\qquad \text{ if } \, X \in \dom(\Gamma),\\
  &&& \phantom{{} \; = \; {}} X
    &&\qquad \text{ otherwise,}\\
&\nf{\Gamma}{A} && \; = \; A
    &&\qquad \text{ for } A \in \set{ \Bot, \Top },\\
&\nf{\Gamma}{\fun{A}{B}} && \; = \; \fun{\nf{\Gamma}{A}}{\nf{\Gamma}{B}}\\
&\nf{\Gamma}{\all{X}{K}{A}} && \; = \; \all{X}{K'}{\nf{\Gamma, X \kas K'}{A}}
    &&\qquad \text{ where } \, K' = \nf{\Gamma}{K},\\
&\nf{\Gamma}{\lam{X}{K}{A}} && \; = \; \lam{X}{K'}{\nf{\Gamma, X \kas K'}{A}}
    &&\qquad \text{ where } \, K' = \nf{\Gamma}{K},\\
&\nf{\Gamma}{\app{A}{B}} && \; = \; \hsubst{E}{X}{\ksimp{K}}{\nf{\Gamma}{B}}
    &&\qquad \text{ if } \, \nf{\Gamma}{A} = \lam{X}{K}{E},\\
  &&& \phantom{{}\; = \; {}}\app{(\nf{\Gamma}{A})}{(\nf{\Gamma}{B})}
    &&\qquad \text{ otherwise.}\\[-1ex]
\intertext{\judgment{}{\fbox{$\nf{\Gamma}{K}$}}\vspace{-1.5em}}
&\nf{\Gamma}{A \intv B} && \; = \; \nf{\Gamma}{A} \intv \nf{\Gamma}{B}\\
&\nf{\Gamma}{\dfun{X}{J}{K}} && \; = \; \dfun{X}{J'}{\nf{\Gamma, X \kas J'}{K}}
    &&\qquad \text{ where } \, J' = \nf{\Gamma}{J}.
\end{alignat*}\vspace{-1.5em}}
 
The $\eta$-expansion $\etaExp{K}{A}$ of a type $A$ of kind $K$ is
defined by recursion on the structure of~$K$.
It is used in the definition of $\nfRaw$ to expand type variables.
Note that $\etaExp{K}{A}$ immediately expands newly introduced
argument variables to produce $\eta$-long forms.
Normalization $\nf{\Gamma}{A}$ and $\nf{\Gamma}{K}$ of raw types and
kinds in context $\Gamma$ are defined by mutual recursion on $A$ and
$K$, respectively.
The case of applications uses hereditary substitution to eliminate
$\beta$-redexes.
Note the crucial use of domain-annotations:
in order to hereditarily substitute a type argument in the body of an
operator abstraction $\lam{X}{K}{E}$, we need to guess its shape, or
equivalently, the shape of $X$.
Since the normalization function is defined directly on raw, unkinded
types, the only way to obtain this information is from the kind
annotation~$K$ in the abstraction.

The context parameter $\Gamma$ of $\nfRaw$ is used to look up the
declared kinds of variables, which drive their $\eta$-expansion.
To ensure that the resulting $\eta$-expansions are normal, the context
$\Gamma$ must itself be normal.
We therefore extend normalization pointwise to contexts, defining
$\nfCtx{\Gamma}$ as
\begin{align*}
  \nfCtx{\cempty} &= \cempty &
  \nfCtx{\Gamma, x \tas A} &=
    \nfCtx{\Gamma}, x \tas \nf{\nfCtx{\Gamma}}{A} &
  \nfCtx{\Gamma, X \kas K} &=
    \nfCtx{\Gamma}, X \kas \nf{\nfCtx{\Gamma}}{K}
\end{align*}
Since $\nfRaw$ is a total function defined directly on raw types and
kinds, it necessarily contains degenerate cases, \ie the resulting
types need not be $\beta$-normal.
\Eg the case of applications relies on the domain annotations $K$ of
operator abstractions $\lam{X}{K}{A}$ in head position to be truthful.
The ill-kinded type
$\Omega = (\lam{X}{\kstar}{\app{X}{X}})(\lam{X}{\kstar}{\app{X}{X}})$
will result in
$\nf{}{\Omega} \alEq
\hsubst{(\app{X}{X})}{X}{\kstar}{\lam{X}{\kstar}{\app{X}{X}}} \alEq
\rapp{(\lam{X}{\kstar}{\app{X}{X}})}{\kstar}{(\lam{X}{\kstar}{\app{X}{X}})}
\alEq \Omega$.
However, for well-kinded types $\Gamma \ts A \kin K$, the type
$\nf{\Gamma}{A}$ is guaranteed to be an $\eta$-long $\beta$-normal
form, as we will see in \Sec{canonical} (\cf \Lem{canon-complete}).

Furthermore, $\eta$-expansion and normalization are \emph{sound}, \ie
they do not alter the meaning of types and kinds.
In particular, well-kinded types and well-formed kinds are
judgmentally equal to their normalized counterparts.
The proof relies on soundness of hereditary substitutions.

\begin{lemma}[soundness of normalization]\label{lem:nf-sound}
  ~
  \begin{multicolenum}{2}\item\label{item:nf-sound-kd} If $\; \kindD{\Gamma}{K}$, then
    $\Gamma \ts K \keq \nf{\nfCtx{\Gamma}}{K}$.
  \item\label{item:nf-sound-tp} If $\; \Gamma \ts A \kin K$, then
    $\Gamma \ts A \teq \nf{\nfCtx{\Gamma}}{A} \kin K$.
\end{multicolenum}
\end{lemma}

\subsection{Commutativity of Normalization and Hereditary Substitution}
\label{sec:commute}

We are now almost ready to introduce the canonical presentation of
\FOmegaInt{}.
Our final task in this section is to establish a series of
\emph{commutativity properties} of hereditary substitution and
normalization.
They say, roughly, that the order of these operations can be switched
without changing the result.
We require these properties to prove equivalence of the canonical and
declarative systems.
\Eg to prove that \rulerefN{ST-Beta1}{ST-Beta$_1$} is admissible in
the canonical system, we must show that
\[
  \Gamma \; \ts \;
  \hsubst{\nf{\Gamma}{A}}{X}{\ksimp{J}}{\nf{\Gamma}{B}} \; \keq \;
  \nf{\Gamma}{\subst{A}{X}{B}} \; \kin \;
  \nf{\Gamma}{\subst{K}{X}{B}}.
\]
Since normalization involves hereditary substitutions, we must further
show that these preserve the canonical judgments.
The case for applications involves
kind equations of the form
\[
  \Gamma \; \ts \; \hsubst{\hsubst{K}{Y}{k}{A}}{X}{j}{B} \; \keq \;
  \hsubst{\hsubst{K}{X}{j}{B}}{Y}{k}{\hsubst{A}{X}{j}{B}}.
\]
Unfortunately, the commutativity properties do not hold for arbitrary
raw types and kinds.
The reasons are twofold.
First, our definition of hereditary substitution contains degenerate
cases for ill-kinded inputs, which can cause inconsistencies when we
commute hereditary substitutions.
For example, it's easy to verify that for
$B = \lam{Z}{\fun{\kstar}{\kstar}}{Z}$
\[
  \hsubst{\hsubst{(\app{X}{A})}{X}{\kstar}{Y}}{Y}{\fun{\kstar}{\kstar}}{B}
  \; \alEq \;
  A
  \; \not\alEq \;
  \app{B}{A}
  \; \alEq \;
\hsubst{\hsubst{(\app{X}{A})}{Y}{\fun{\kstar}{\kstar}}{B}}{X}{\kstar}{\hsubst{Y}{Y}{\fun{\kstar}{\kstar}}{B}}.
\]
Second, normalization involves $\eta$-expansion and, as we have seen
in \Sec{decl_rules}, $\eta$-expansions of the same type variable can
differ syntactically in their domain annotations.

We address the two problems separately.
For the former, we adopt the approach taken by \citet{AbelR08csl},
namely to prove commutativity of hereditary substitutions only for
well-kinded normal forms.
To apply their technique, we first need to show that hereditary
substitutions preserve kinding (of normal forms).
This is easy in their setting, which is simply kinded, but challenging
in ours because our kinding rules involve substitutions in dependent
kinds.
A direct proof that hereditary substitutions preserve kinding would
require the very commutativity lemmas we are trying to establish.
We circumvent this issue by relaxing our requirements:
for $\hsubst{e}{Y}{k}{V}$ to be non-degenerate, the normal form $V$ need
not actually be well-kinded;
it only needs to have \emph{shape} $k$.
Using this insight, we prove commutativity of hereditary substitutions
in \ref*{item:comm-step4} steps.
\begin{enumerate}
\item\label{item:comm-step1}
We define a \emph{simple kinding judgment} that assign shapes
  (rather than kinds) to normal forms.

\item We show that hereditary substitution preserves simple kinding.
Because shapes have no type dependencies, the proof does not require
  any commutativity lemmas.

\item\label{item:comm-step3} We show that hereditary substitutions
  in simply kinded normal forms commute.

\item\label{item:comm-step4}
We show that simple kinding is sound:
every type of kind $K$ has a normal form of shape $\ksimp{K}$.

\end{enumerate}
We refer the reader to \SupSec{simple} for details.

It remains to show that normalization commutes with substitution.
To work around the issue of incongruous domain annotations,
we introduce an auxiliary equivalence $A \wkEq B$ on types and kinds,
called \emph{weak equality}, that identifies operator abstractions up
to the \emph{shape} of their domain annotations,
\ie $\lam{X}{J}{A} \wkEq \lam{X}{K}{A}$ when $\simpEq{J}{K}$.
Substitution then \emph{weakly commutes} with
normalization of well-formed kinds and well-kinded types.
\begin{lemma}\label{lem:weq-nf-comm}
  Let $\; \Gamma \ts A \kin J$ and $V = \nf{\nfCtx{\Gamma}}{A}$, then
  \begin{enumerate}\item if $\; \kindD{\Gamma, X \kas J, \Delta}{K}$, then
    $\nf{\nfCtx{\Gamma,\subst{\Delta}{X}{A}}}{\subst{K}{X}{A}}
    \wkEq
    \hsubst{(\nf{\nfCtx{\Gamma, X \kas J, \Delta}}{K})}{X}{\ksimp{J}}{V}$;
  \item if $\; \Gamma, X \kas J, \Delta \ts B \kin K$, then
    $\nf{\nfCtx{\Gamma,\subst{\Delta}{X}{A}}}{\subst{B}{X}{A}}
    \wkEq
    \hsubst{(\nf{\nfCtx{\Gamma, X \kas J, \Delta}}{B})}{X}{\ksimp{J}}{V}$;
  \end{enumerate}
\end{lemma}\noindent
The proof uses commutativity of hereditary substitutions
and a few helper lemmas, \eg that substitutions weakly commute with
$\eta$-expansion.
The full proof is given in \SupSec{comm-nf-subst}.

\section{The Canonical System}
\label{sec:canonical}

Having characterized the normal forms of kinds and types in the
previous section, we now present our \emph{canonical system} of
judgments directly defined on normal forms,
and summarize its most important metatheoretic properties:
the hereditary substitution lemma, equivalence \wrt the declarative
system, and inversion of subtyping.
We conclude the section by revisiting the type safety proof outlined
in \Sec{declarative}.

\fig{fig:canon_rules1}{Canonical presentation of \FOmegaInt{} -- part 1}{
\judgment{Canonical kinding of variables}{\fbox{$\Gamma \tsVar X \kin K$}}

\bigskip

\begin{minipage}[t]{0.35\linewidth}
\infruleLeft{CV-Var}
  {\ctxC{\Gamma}  \andalso  \Gamma(X) = K}
  {\Gamma \tsVar X \kin K}
\end{minipage}\hfill
\begin{minipage}[t]{0.56\linewidth}
\infruleLeft{CV-Sub}
  {\Gamma \tsVar X \kin J  \andalso  \Gamma \ts J \ksub K  \andalso
   \kindC{\Gamma}{K}}
  {\Gamma \tsVar X \kin K}
\end{minipage}

\bigskip

\judgment{Spine kinding}{\fbox{$\spC{\Gamma}{J}{\vec{V}}{K}$}}

\bigskip

\begin{minipage}[t]{0.32\linewidth}
\infruleLeft{CK-Empty}{\quad}{\spC{\Gamma}{K}{\sempty}{K}}
\end{minipage}\hfill
\begin{minipage}[t]{0.64\linewidth}
\infruleLeft{CK-Cons}
  {\Gamma \ts U \kck J \andalso \kindC{\Gamma}{J} \andalso
   \spC{\Gamma}{\hsubst{K}{X}{\ksimp{J}}{U}}{\vec{V}}{L} }
  {\spC{\Gamma}{\dfun{X}{J}{K}}{\scons{U}{\vec V}}{L}}
\end{minipage}

\bigskip

\begin{minipage}[t]{0.49\linewidth}
  \judgment{Kinding of neutral types}{\fbox{$\Gamma \tsNe N \kin K$}}

\infruleLeft{CK-Ne}
  {\Gamma \tsVar X \kin J  \andalso  \spC{\Gamma}{J}{\vec{V}}{K}}
  {\Gamma \tsNe \app{X}{\vec{V}} \kin K}
\end{minipage}\hfill
\begin{minipage}[t]{0.43\linewidth}
  \judgment{Kinding checking}{\fbox{$\Gamma \ts V \kck K$}}

\infruleLeft{CK-Sub}
  {\Gamma \ts V \ksn J  \andalso  \Gamma \ts J \ksub K}
  {\Gamma \ts V \kck K}
\end{minipage}

\bigskip

\judgment{Kind synthesis for normal types}{\fbox{$\Gamma \ts V \ksn K$}}
\begin{multicols}{2}
  \typicallabel{CK-Sing}
\infrule[\ruledef{CK-Top}]
  {\ctxC{\Gamma}}
  {\typeC{\Gamma}{\Top}}

\infrule[\ruledef{CK-Arr}]
  {\typeC{\Gamma}{U}  \andalso  \typeC{\Gamma}{V}}
  {\typeCP{\Gamma}{\fun{U}{V}}}

\infrule[\ruledef{CK-Abs}]
  {\kindC{\Gamma}{J}  \andalso  \Gamma, X \kas J \ts V \ksn K}
  {\Gamma \ts \lam{X}{J}{V} \ksn \dfun{X}{J}{K}}

\infrule[\ruledef{CK-Bot}]
  {\ctxC{\Gamma}}
  {\typeC{\Gamma}{\Bot}}

\infrule[\ruledef{CK-All}]
  {\kindC{\Gamma}{K}  \andalso  \typeC{\Gamma, X \kas K}{V}}
  {\typeCP{\Gamma}{\all{X}{K}{V}}}

\infrule[\ruledef{CK-Sing}]
  {\Gamma \tsNe N \kin U \intv V}
  {\Gamma \ts N \ksn N \intv N}
\end{multicols}\vspace{-1em}}
 \fig{fig:canon_rules2}{Canonical presentation of \FOmegaInt{} -- part 2}{
\judgment{Subtyping of proper types}{\fbox{$\Gamma \ts U \tsub V$}}
\begin{multicols}{2}
  \typicallabel{CST-Top}
\infrule[\ruledef{CST-Top}]
  {\typeC{\Gamma}{V}}
  {\Gamma \ts V \tsub \Top}

\infrule[\ruledef{CST-Trans}]
  {\Gamma \ts U \tsub V  \andalso  \Gamma \ts V \tsub W}
  {\Gamma \ts U \tsub W}

\infrule[\ruledef{CST-Ne}]
  {\quad\\
   \Gamma \tsVar X \kin K  \andalso
   \spC{\Gamma}{K}{\vec{V}_1 \keq \vec{V}_2}{U \intv W}}
  {\Gamma \ts \app{X}{\vec{V}_1} \tsub \app{X}{\vec{V}_2}}

\infrule[\ruledefN{CST-Bnd1}{CST-Bnd$_1$}]
  {\Gamma \tsNe N \kin V_1 \intv V_2}
  {\Gamma \ts V_1 \tsub N}

\infrule[\ruledef{CST-Bot}]
  {\typeC{\Gamma}{V}}
  {\Gamma \ts \Bot \tsub V}

\infrule[\ruledef{CST-Arr}]
  {\Gamma \ts U_2 \tsub U_1  \andalso  \Gamma \ts V_1 \tsub V_2}
  {\Gamma \ts \fun{U_1}{V_1} \tsub \fun{U_2}{V_2}}

\infrule[\ruledef{CST-All}]
  {\typeC{\Gamma}{\all{X}{K_1}{V_1}}\\
   \Gamma \ts K_2 \ksub K_1  \andalso
   \Gamma, X \kas K_2 \ts V_1 \tsub V_2}
  {\Gamma \ts \all{X}{K_1}{V_1} \tsub \all{X}{K_2}{V_2}}

\infrule[\ruledefN{CST-Bnd2}{CST-Bnd$_2$}]
  {\Gamma \tsNe N \kin V_1 \intv V_2}
  {\Gamma \ts N \tsub V_2}
\end{multicols}\vspace{-1ex}

\judgment{Checked subtyping}{\fbox{$\Gamma \ts U \tsub V \kck K$}}

\bigskip

\begin{minipage}[t]{0.39\linewidth}
\infruleLeft{CST-Intv}
  {\Gamma \ts V_1 \kck U \intv W\\ \Gamma \ts V_2 \kck U \intv W\\
   \Gamma \ts V_1 \tsub V_2}
  {\Gamma \ts V_1 \tsub V_2 \kck U \intv W}
\end{minipage}\hfill
\begin{minipage}[t]{0.58\linewidth}
\infruleLeft{CST-Abs}
  {\Gamma \ts \lam{X}{J_1}{V_1} \kck \dfun{X}{J}{K}\\
   \Gamma \ts \lam{X}{J_2}{V_2} \kck \dfun{X}{J}{K}\\
   \Gamma, X \kas J \ts V_1 \tsub V_2 \kck K}
  {\Gamma \ts \lam{X}{J_1}{V_1} \tsub \lam{X}{J_2}{V_2} \kck \dfun{X}{J}{K}}
\end{minipage}

\bigskip

\judgment{Spine equality}{\fbox{$\spC{\Gamma}{J}{\vec{U} \keq \vec{V}}{K}$}}

\bigskip

\begin{minipage}[t]{0.34\linewidth}
\infruleLeft{SpEq-Empty}{\quad}
  {\spC{\Gamma}{K}{\sempty \keq \sempty}{K}}
\end{minipage}\hfill
\begin{minipage}[t]{0.645\linewidth}
\infruleLeft{SpEq-Cons}
  {\Gamma \ts U_1 \keq U_2 \kck J \andalso
   \spC{\Gamma}{\hsubst{K}{X}{\ksimp{J}}{U_1}}{\vec{V}_1 \keq \vec{V}_2}{L}}
  {\spC{\Gamma}{\dfun{X}{J}{K}}{\scons{U_1}{\vec{V}_1} \keq \scons{U_2}{\vec{V}_2}}{L}}
\end{minipage}}
 
The rules for canonical kinding, subtyping and spine equality are given in
\Fig[Figs.]{canon_rules1} and~\ref{fig:canon_rules2}.
The canonical system also contains judgments for kind and context
formation, subkinding, and type and kind equality, but the rules of
those judgments are analogous to their declarative counterparts, so we
omit them.
We use the following naming conventions to distinguish normal forms:
$U$, $V$, $W$ denote normal types;
$M$, $N$ denote neutral types.
No special notation is used for normal kinds.

The judgments for kinding, subtyping and spine equality are
bidirectional.
Double arrows are used to indicate whether a kind is an input
($K \ksn$ or $\kck K$) or an output ($\ksn K$) of the judgment.
The contexts and subjects of judgments are always considered inputs.
The rules for kind synthesis are similar to those for declarative
kinding, except that the synthesized kinds are more precise, that
there is no subsumption rule, and that kinding of variables and
applications has been combined into a single rule~\ruleref{CK-Sing}
for kinding neutral proper types.
A quick inspection of the rules reveals that all synthesized kinds are
singletons.
In particular, $\Gamma \ts V \!\ksn V \intv V$ for proper types $V$.
The kind checking judgment $\Gamma \ts V \kck K$ has only one
inference rule:
the subsumption rule \ruleref{CK-Sub}.

The canonical kinding judgments $\Gamma \tsVar X \kin K$ for variables
and $\Gamma \tsNe N \kin K$ for neutral types are not directed because
of the presence of the subsumption rule~\ruleref{CV-Sub}.
While this rule is not actually necessary for variable kinding, it
considerably simplifies the metatheory.
Without it, the proof of context narrowing would circularly depend on
at least three other lemmas
--~transitivity of subkinding, functionality and the hereditary
substitution lemma~--
all of which use somewhat idiosyncratic, possibly incompatible
induction strategies.

Canonical spine kinding $\spC{\Gamma}{J}{\vec{V}}{K}$ differs from the
other judgments in that it features both an input kind $J$ and an
output kind $K$.
When an operator of shape $J$ is applied to the spine $\vec{V}$, the
resulting type is of shape $K$ -- as exemplified by the
rule~\ruleref{CK-Ne}.
In~\ruleref{CK-Cons}, the head~$U$ of the spine~$\scons{U}{\vec{V}}$
is hereditarily substituted for $X$ in the codomain~$K$ of the overall
input kind~$\dfun{X}{J}{K}$ to obtain the input
kind~$\hsubst{K}{X}{\ksimp{J}}{U}$ for kinding the tail~$\vec{V}$ of
the spine.
The use of hereditary (rather than ordinary) substitution ensures that
the resulting kind remains normal.

The rules for canonical subtyping are divided into two judgments:
\emph{subtyping of proper types} $\Gamma \ts U \tsub V$ and
\emph{checked subtyping} $\Gamma \ts U \tsub V \kck K$.
The separate judgment for proper subtyping $\Gamma \ts U \tsub V$
simplifies the metatheory by disentangling subtyping and kinding.
It resembles the subtyping judgment of \FSub{}.
Notable differences are the two bound projection
rules~\rulerefN{CST-Bnd1}{CST-Bnd$_1$}
and~\rulerefN{CST-Bnd2}{CST-Bnd$_2$}, which replace the variable
subtyping rule, and the rule~\ruleref{CST-Ne} for subtyping
neutrals.
The most interesting of these is~\ruleref{CST-Ne}.
It says that two neutral types $\app{X}{\vec{V}_1}$ and
$\app{X}{\vec{V}_2}$ headed by a common type variable $X$ are subtypes
if they have canonically equal spines.
Importantly, the spines $\vec{V}_1$ and $\vec{V}_2$ need not be
syntactically equal.
In \FOmegaInt{}, normal forms may be judgmentally equal yet differ
syntactically, \eg because they have different domain annotations (via
\ruleref{CST-Abs}), because one of the types is declared as a type
alias of the other (via a singleton kind),
or because one of the types can be proven to alias the other due to inconsistent bounds:
$ \Gamma, \, X \kas\, U \intv V, \, Y \kas\, V \intv U, \, \Delta \ts
\, U \teq V$.
The last example illustrates that there is no easy way for the
normalization function $\nfRaw$ to resolve such equations.
In systems without dependent kinds and equality reflection,
judgmentally equal types have syntactically equal normal forms, so
that \ruleref{CST-Ne} can be omitted \citep[see e.g.][]{AbelR08csl};
in systems with singleton kinds (but no type intervals), type aliases
can be resolved during normalization \citep[see e.g.][]{StoneH00popl}.
Neither of these apply in~\FOmegaInt{}.
Unfortunately, the presence of \ruleref{CST-Ne} complicates the
metatheory of the canonical system (\cf \Sec{canon-hsubst}).

The checked subtyping judgment $\Gamma \ts U \tsub V \kck K$ is
kind-directed.
The kind $K$ determines whether $U$ and $V$ are compared as proper
types (\ruleref{CST-Intv}) or type operators (\ruleref{CST-Abs}).
The rule~\ruleref{CST-Intv} checks that the types $V_1$ and $V_2$ have
the expected kind $U \intv W$ and are proper subtypes.
Because normal types are $\eta$-long, the only normal types of arrow
kind are operator abstractions.
They are compared using rule~\ruleref{CST-Abs}, exactly as in the
declarative system.
The rules of the \emph{spine equality} judgment
$\spC{\Gamma}{J}{\vec{U} \teq \vec{V}}{K}$ resemble those of spine
kinding.

\subsection{Metatheoretic Properties of the Canonical System}
\label{sec:canon-sound}

It is easy to show that the canonical system is sound \wrt the
declarative one, \ie that normal forms related by the canonical
judgments are also related by the declarative counterparts.
To avoid confusion, we mark canonical judgments with the subscript
``$\mathsf{c}$'' and declarative ones with ``$\mathsf{d}$'' in the
following soundness lemma.
The full statement of the lemma has \SupItem{can-snd-speq}~parts, one for each canonical
judgment;
we only show the most important ones.
\begin{lemma}[soundness of the canonical rules -- excerpt]
  \label{lem:canon-sound}
  ~
  \begin{enumerate}\item If $\; \Gamma \tsC V \ksn K$ or
    $\Gamma \tsC V \kck K$, then $\; \Gamma \tsD V \kin K$.
  \item If $\; \Gamma \tsC U \tsub V$, then
    $\; \Gamma \tsD U \tsub V \kin \kstar$.
  \item If $\; \Gamma \tsC U \tsub V \kck K$,
    then $\; \Gamma \tsD U \tsub V \kin K$.
  \end{enumerate}
\end{lemma}\noindent
Many of the basic properties of the declarative system
--~
weakening, context narrowing, admissible order-theoretic rules,
many validity properties, and the commutativity lemmas from the
previous section~--
also hold in the canonical system, and their proofs carry over with
minor modifications.
Full statements and proofs are given in \SupSec{canon-sound-full}.
Notable exceptions are the substitution and functionality lemmas.
These do not hold because ordinary substitutions do not preserve
normal forms.
Instead, we need to establish analogous lemmas for hereditary
substitutions.

\subsubsection{The Hereditary Substitution Lemma}
\label{sec:canon-hsubst}

The most important metatheoretic property of the canonical system is
the \emph{hereditary substitution lemma}, which states, roughly, that
canonical judgments are preserved by hereditary substitutions of
canonically well-kinded types.
It is key to proving completeness \wrt the declarative system, and
thus to our overall goal of establishing type safety.
Proving and even stating the hereditary substitution lemma is
challenging.
The full statement of the lemma has \SupItem{can-rapp-teq} separate
parts, all of which have to be proven simultaneously.
The large number of canonical judgments is one reason for the
complexity of the lemma.
But the foremost reason is that the proof of the hereditary
substitution lemma circularly depends on \emph{functionality of the
  canonical judgments}.
This circular dependency is caused by the subtyping
rule~\ruleref{CST-Ne}.

To illustrate this, assume we are given two normal forms $V$ and $W$
such that
$\Gamma \ts V \teq W \kck \kstar$, but $V \not\alEq W$ syntactically.
We have seen examples of such normal forms $V$ and $W$ earlier.
Assume further that there is some $X$ with
$\Gamma(X) = \fun{\kstar}{\kstar}$ and consider what happens when we
hereditarily substitute the operator $\lam{Y}{\kstar}{U}$ for $X$ in
the judgment $\Gamma \ts \app{X}{V} \tsub \app{X}{W}$ obtained
using~\ruleref{CST-Ne}.
We would like to show that hereditary substitution preserves this
inequation, \ie that
\[
  \Gamma \; \ts \;
  \hsubst{(\app{X}{V})}{X}{\fun{\kstar}{\kstar}}{\lam{Y}{\kstar}{U}}
  \; \alEq \;
  \hsubst{U}{Y}{\kstar\!}{V}
  \; \tsub \;
  \hsubst{U}{Y}{\kstar\!}{W}
  \; \alEq \;
  \hsubst{(\app{X}{W})}{X}{\fun{\kstar}{\kstar}}{\lam{Y}{\kstar}{U}},
\]
which requires functionality.
The example illustrates a second point, namely that, in order to prove
that hereditary substitutions preserve canonical kinding and
subtyping, we need to prove that kinding and subtyping of reducing
applications is admissible.
Our hereditary substitution lemma must cover all of these properties,
leading to the aforementioned grand total of
\SupItem{can-rapp-teq}~parts.
We give a small excerpt here, illustrating some of the properties
just discussed.
\begin{lemma}[hereditary substitution -- excerpt]\label{lem:canon-hsubst}
  Assume $\Gamma \ts U_1 \; \teq \; U_2 \kck K$.
  \begin{enumerate}\item Hereditary substitution preserves kind checking:

    if $\; \Gamma, X \kas K, \Delta \ts V \kck J$, then
    $\; \Gamma, \hsubst{\Delta}{X}{\ksimp{K}}{U_1} \ts
    \hsubst{V}{X}{\ksimp{K}}{U_1} \kck \hsubst{J}{X}{\ksimp{K}}{U_1}$.
  \item Functionality/monotonicity of hereditary substitution:

    if $\; \Gamma, X \kas K, \Delta \ts V \kck J$ and
    $\; \kindC{\Gamma, X \kas K, \Delta}{J}$, then
    \[
      \Gamma, \hsubst{\Delta}{X}{\ksimp{K}}{U_1} \ts
      \hsubst{V}{X}{\ksimp{K}}{U_1} \tsub
      \hsubst{V}{X}{\ksimp{K}}{U_2} \kck
      \hsubst{J}{X}{\ksimp{K}}{U_1}.
    \]
\item Hereditary substitution preserves checked subtyping:

    if $\; \Gamma, X \kas K, \Delta \ts V_1 \tsub V_2 \kck J$ and
    $\kindC{\Gamma, X \kas K, \Delta}{J}$, then
    \[
      \Gamma, \hsubst{\Delta}{X}{\ksimp{K}}{U_1} \ts
      \hsubst{V_1}{X}{\ksimp{K}}{U_1} \tsub
      \hsubst{V_2}{X}{\ksimp{K}}{U_2} \kck
      \hsubst{J}{X}{\ksimp{K}}{U_1}.
    \]
\item Canonical equality of reducing applications is admissible:

    if $\; \Gamma \ts V_1 \teq V_2 \kck \dfun{X}{K}{J}$, then
    $
      \Gamma \ts
      \rapp{V_1}{\fun{\ksimp{K}}{\ksimp{J}}}{U_1} \teq
      \rapp{V_2}{\fun{\ksimp{K}}{\ksimp{J}}}{U_2} \kck
      \hsubst{J}{X}{\ksimp{K}}{U_1}.
    $
  \end{enumerate}
\end{lemma}\noindent
The structure of the proof mirrors that of the recursive definition of
hereditary substitution itself.
All \SupItem{can-rapp-teq}~parts are proven simultaneously by
induction in the structure of the shape $\ksimp{K}$.
Most parts proceed by an inner induction on the derivations of the
judgments into which $U_1$ and $U_2$ are being substituted.
The cases involving the rules~\ruleref{CK-Cons}
and~\ruleref{SpEq-Cons}, rely on commutativity of hereditary
substitutions in kinds.
Details are given in \SupSec{canon-hsubst-full}.

Just as for the declarative system, the proof of functionality enables
us to prove some additional validity properties,
which, in turn, are necessary to prove completeness of the canonical
system.

\subsubsection{Completeness of Canonical Kinding}
\label{sec:canon-complete}

In \Sec{normalization}, we saw that every declaratively well-kinded
type has a judgmentally equal $\beta\eta$-normal form
(\Lem[Lemmas]{nf-sound}).
To establish equivalence of the canonical and declarative systems, it
remains to show that the normal forms of types related by a
declarative judgment are also canonically related.
The full statement of the completeness lemma has
\SupItem{can-cmp-beta}~parts, one for each declarative judgment plus
three auxiliary parts for dealing with hereditary substitutions and
$\beta\eta$-conversions.
The most important ones are the following, where we again use the subscripts ``$\mathsf{c}$'' and
``$\mathsf{d}$'' to distinguish canonical judgments from declarative
ones.
\begin{lemma}[completeness of the canonical rules -- excerpt]
    \label{lem:canon-complete}
  ~
  \begin{enumerate}\item If $\; \Gamma \tsD A \kin K$, then
    $\; \nfCtx{\Gamma} \tsC \nf{}{A} \kck \nf{}{K}$.
  \item If $\; \Gamma \tsD A \ksub B \kin K$, then
    $\; \nfCtx{\Gamma} \tsC \nf{}{A} \ksub \nf{}{B} \kck \nf{}{K}$.
  \end{enumerate}
\end{lemma}\noindent
The completeness proof relies on \Lem{canon-hsubst} to
show that the normal forms of applications are canonically
well-kinded, and on the weak commutativity properties established in
\Sec{commute} to show that $\beta$- and $\eta$-conversions are
admissible in the canonical system.
In addition, the proof relies on the validity conditions
discussed in \Sec{decl_basics_short}.
The full statement and proof of the completeness lemma are given in
\SupSec{canon-complete-full}.

\subsection{Inversion of Subtyping and Type Safety}
\label{sec:canon-inversion}

As we saw in \Sec{declarative}, reductions in open terms are unsafe
because variable bindings with inconsistent bounds can inject
arbitrary inequations into the subtyping relation.
Under such assumptions, subtyping cannot be inverted in any meaningful
way.
We therefore prove inversion of subtyping only in the empty context,
following the approach by \citet{RompfA16oopsla}:
\begin{enumerate}
\item We introduce a helper judgment $\tsTf U \tsub V$, which states
  that $U$ is a proper subtype of $V$ in the empty context.
It is obtained from the canonical subtyping judgment
  $\Gamma \ts U \tsub V$ by fixing $\Gamma = \cempty$ and removing
  \ruleref{CST-Trans} and any rules involving free variables
  (\ruleref{CST-Ne}, \rulerefN{CST-Bnd1}{CST-Bnd$_1$} and
  \rulerefN{CST-Bnd2}{CST-Bnd$_2$}).
Soundness of $\tsTf U \tsub V$ \wrt canonical subtyping is
  immediate.
\item We prove that transitivity is admissible in $\tsTf U \tsub V$.
This is straightforward since there are no $\beta\eta$-conversion
  rules or bound projection rules that get in the way.
\item Because transitivity is admissible, it is straightforward to
  establish completeness, and thus equivalence of $\tsTf U \tsub V$
  \wrt canonical subtyping.
\item Inversion of the canonical subtyping relation in the empty
  context then follows immediately by inspection of the rules for
  $\tsTf U \tsub V$ and equivalence of the two judgments.
\item Inversion of top-level declarative subtyping follows by
  equivalence of canonical and declarative subtyping and soundness of
  normalization.
\end{enumerate}
\begin{lemma}[inversion of top-level subtyping]
  \label{lem:st-inv}
  Let $\; \cempty \ts A_1 \tsub A_2 \kin \kstar$.
  \begin{enumerate}\item If $\; A_1 = \fun{B_1}{C_1}$ and
    $\; A_2 = \fun{B_2}{C_2}$, then
    $\; \cempty \ts B_2 \tsub B_1 \kin \kstar$ and
    $\; \cempty \ts C_1 \tsub C_2 \kin \kstar$.
  \item If $\; A_1 = \all{X}{K_1}{B_1}$ and
    $\; A_2 = \all{X}{K_2}{B_2}$, then
    $\; \cempty \ts K_2 \ksub K_1$ and
    $\; X \kas K_2 \ts B_1 \tsub B_2 \kin \kstar$.
  \item \begin{inparablank} \item $\cempty \nts \Top \tsub \Bot$, $\;$and$\;$
    \item $\cempty \nts \fun{A}{B} \tsub \all{X}{K}{C}$, $\;$and$\;$
    \item $\cempty \nts \all{X}{K}{A} \tsub \fun{B}{C}$.
    \end{inparablank}
\end{enumerate}
\end{lemma}\noindent
With subtyping inversion in place, we are finally ready to prove type
safety of \FOmegaInt{}.
\begin{theorem}[type safety] Well-typed terms do not get stuck.
  \begin{enumerate}[leftmargin=10em,labelsep=1em]\item[(progress)] If $\; \ts t \tin A$, then either $t$ is a value,
    or $t \cbvstep t'$ for some term $t'$.
  \item[(weak preservation)] If $\; \ts t \tin A$ and
    $t \cbvstep t'$, then $\ts t' \tin A$.
  \end{enumerate}
\end{theorem}\noindent
The proofs are standard.
Details are given in \SupSec{canon-inversion-full}.

\section{Undecidability of Subtyping}
\label{sec:undec}

In this section, we prove that subtyping for \FOmegaInt{} is
undecidable.
The culprit is equality reflection via the bound projection
rules~\rulerefN{ST-Bnd1}{ST-Bnd$_1$}
and~\rulerefN{ST-Bnd2}{ST-Bnd$_2$} (cf. \Sec{inconsistent}).
Following~\citet{CastellanCD15tlca}, we prove undecidability of
subtyping by reduction from convertibility of SK combinator terms.
The \emph{SK combinator calculus} has only three term formers
$t ::= \skS \orElse \skK \orElse \app{s}{t}$ and two equational
axioms:
$\app{\app{\app{\skS}{s}}{t}}{u} \skEq \app{\app{s}{u}}{(\app{t}{u})}$
and
$\app{\app{\skK}{s}}{t} \skEq s$.
Yet SK is Turing-complete, and convertibility of SK terms is undecidable.
We embed SK terms and equations into \FOmegaInt{} via the following
declarations:
\begin{align*}
  \skSig \; \coloneq \; {}&
    S \kin \kstar, \qquad K \kin \kstar, \qquad
    \skApp{}{} \kin \fun{\kstar}{\fun{\kstar}{\kstar}}, \\
  &\skSRedVar \kin \dfun{X}{\kstar}{\dfun{Y}{\kstar}{\dfun{Z}{\kstar}{\skApp{\skApp{\skApp{S}{X}}{Y}}{Z} \, \intv \,
    \skApp{\skApp{X}{Z}}{(\skApp{Y}{Z})}}}}, \\
  &\skSExpVar \kin \dfun{X}{\kstar}{\dfun{Y}{\kstar}{\dfun{Z}{\kstar}{\skApp{\skApp{X}{Z}}{(\skApp{Y}{Z})} \, \intv \,
    \skApp{\skApp{\skApp{S}{X}}{Y}}{Z}}}}, \\
  &\skKRedVar \kin \dfun{X}{\kstar}{\dfun{Y}{\kstar}{\skApp{\skApp{K}{X}}{Y}} \, \intv \, X}, \qquad
    \skKExpVar \kin \dfun{X}{\kstar}{\dfun{Y}{\kstar}{X \, \intv \, \skApp{\skApp{K}{X}}{Y}}}.\\
  \encode{\skS} \; \coloneq \; {}& S \qquad
    \encode{\skK} \; \coloneq \; K \qquad
    \encode{\app{s}{t}} \; \coloneq \; \skApp{\encode{s}}{\encode{t}}.
\end{align*}
The map $\encode{-}$ encodes SK terms as types under $\skSig$
and induces a reduction from SK convertibility $s \skEq t$ to
subtyping $\skSig \ts \encode{s} \tsub \encode{t} \kin \kstar$,
which can be used to prove undecidability of subtyping.
\begin{theorem}[undecidability]\label{thm:undec}
  Let $s$, $t$ be SK terms.
Then $\; \skSig \ts \encode{s} \tsub \encode{t} \kin \kstar\;$
  iff $\; s \skEq t$.
\end{theorem}\noindent
It is easy to verify the ``if'' direction.
\Eg the contraction law for $\skK$ corresponds to the inequation
$
  \skSig \ts
  \skApp{\skApp{K}{\encode{s}}}{\encode{t}}
  \; \tsub \;
  \skKRed{\encode{s}}{\encode{t}}
  \; \tsub \;
  \encode{s}
  \kin \kstar.
$
The ``only if'' direction is more challenging.
The complexity of subtyping derivations precludes a direct decoding
into $\skEq$ for many of the same reasons that a direct proof of
subtyping inversion is unfeasible.
In addition, types such as $\Top$ or $\all{X}{K}{A}$ that are not
encodings of SK terms can appear as intermediate expression along a
subtyping derivation.
\Eg note the fleeting appearance of $\Top$ in the following.
\[
  \skSig \;\; \ts \;\;
  \encode{\skS}
  \;\; \alEq \;\;
  S
  \;\; \tsub \;\;
  \skApp{\skApp{K}{S}}{(\skApp{\skApp{K}{\Top}}{S})}
  \;\; \tsub \;\;
  \skApp{\skApp{K}{S}}{\Top}
  \;\; \tsub \;\;
S
  \;\; \alEq \;\;
  \encode{\skS}
  \;\; \kin \;\; \kstar.
\]
To eliminate such spurious appearances of undecodable types, our proof
takes a detour through four auxiliary judgment forms.
\begin{enumerate}
\item We first prove undecidability of canonical subtyping to
  eliminate instances of $\beta\eta$-conversions.
Undecidability of declarative subtyping follows by equivalence of
  the two judgments.

\item We define a \emph{reduced} version of the canonical system from
  which we exclude any rules (and judgment forms) that are not
  relevant to the embedding shown above.
\Eg any judgments involving higher-order operators are excluded, as
  are the rules~\ruleref{CST-Arr}
  and~\ruleref{CST-All} and the variable subsumption rule~\ruleref{CV-Sub}.
The reduced system still contains~\ruleref{CST-Bot}
  and~\ruleref{CST-Top} since we cannot rule out intermediate
  occurrences of these rules a-priori.
We show that canonical subtyping derivations under $\skSig$ can be
  translated into reduced ones.

\item We extend the SK term syntax with $\Bot$ and $\Top$ and define
  an order $\skSub$ on extended terms.
The term order is an asymmetric version of $\skSub$ which features
  the rules $\Bot \skSub t$ and $t \skSub \Top$.
Thanks to these, reduced canonical subtyping derivations can be
  directly decoded into $\skSub$.

\item We introduce a pair of \emph{parallel reduction relations}
  $\skRed[\leq]$ and $\skRed[\geq]$ on the extended syntax.
These contain the rules $\Bot \skRed[\leq] t$ and
  $\Top \skRed[\geq] t$ for eliminating occurrences of $\Bot$ and
  $\Top$, along with the usual contraction rules for SK terms.
Crucially, the reduction rules can eliminate but never introduce
  instances of $\Bot$ and $\Top$.
Hence, if $s$ is a pure SK term, $s \skRed t$ implies $s \skEq t$.
The parallel reductions enjoy a confluence property \wrt the term
  order:
If $s \skSub t$, then there is a $u$ such that
  $s \skRed[\leq]^* u \skExp[\leq]^* t$.
Via confluence, $s \skSub t$ implies $s \skEq t$ for pure $s$ and
  $t$.
\end{enumerate}
Thus we have established a chain of implications from which the result
follows.
\[
  \skSig \ts \encode{s} \tsub \encode{t} \kin \kstar
  \; \implies \;
\skSig \ts[red] \encode{s} \tsub \encode{t}
  \; \implies \;
  s \skSub t
  \; \implies \;
  s \skRed[\leq]^* u \skExp[\leq]^* t
  \; \implies \;
  s \skEq t.
\]
For full details, we refer the intrepid reader to the mechanized proof
of \Thm{undec}, which is given in the \texttt{FOmegaInt.Undecidable}
module of our Agda formalization~\cite{artifact}.

\section{Related Work}
\label{sec:related}

Bounded quantification has been studied extensively through \FSub{},
a variant of System~F with bounded quantification, which comes in two
flavors:
the \emph{Kernel} variant \FSubC{}~\cite{CardelliMMS91tacs} based on
\citeauthor{CardelliW85csur}'s
\emph{Kernel~Fun}~\citeyearpar{CardelliW85csur} has decidable
subtyping, while \emph{Full}~\FSub{} \cite{CurienG92mscs} features a
more expressive subtyping rule for bounded universal quantifiers that
renders subtyping undecidable~\cite{Pierce92popl}.
Recently, \citet{HuL20popl} have shown that the $D_{<:}$ calculus --~a
simplified variant of DOT that uses an expressive $\forall$-subtyping
rule~-- suffers from the same decidability issue as Full~\FSub{}.
For compatibility with DOT, and knowing that subtyping in \FOmegaInt{}
is undecidable either way, we also adopt the more expressive rule.
The metatheory developed in
\Sec[\S\S]{declarative}--\ref{sec:undec}
is largely unaffected by this choice.

An extension of \citeauthor{Girard72thesis}'s
\FOmega{}~\citeyearpar{Girard72thesis} with higher-order subtyping and
bounded quantification was first proposed
by~\citet{Cardelli90unpublished} under the name \KFOmegaSub{}.
Basic meta theoretic properties of \KFOmegaSub{} were established by
\citet{PierceS97tcs}, \citet{Compagnoni95csl}, and
\citet{CompagnoniG99csl}.
An extension with bounded operator abstractions (\FOmegaSubC{}) has
been developed by \citet{CompagnoniG03ic}.
More recently, \citet{AbelR08csl} developed a variant of \KFOmegaSub{}
where types are identified up to $\beta\eta$-equality
and proved its decidability using hereditary
substitution.
Their work inspired the syntactic approach taken in this paper.

Many of the ideas in \KFOmegaSub{} go back to early work by
\citet{Cardelli88popl} on \emph{power types}.
Though very expressive, power types render the type language
non-normalizing, and in a later work \citet{CardelliL91jfp} replaced
them with the better behaved \emph{power kinds}.
Power kinds can be directly expressed in \FOmegaInt{} as interval
kinds that are bounded by $\Bot$ from below:
$\kpowp{A} = \Bot \intv A$.
\citet{Crary97icfp} developed an extension of \FOmega{} with power
kinds as a general calculus for higher-order subtyping.
His representations of higher-order bounded quantifiers and operators
closely resemble ours.

The notion of \emph{translucency} was introduced by
\citet{HarperL94popl} in the setting of ML-style modules with
sharing constraints.
They proposed \emph{translucent sums} as a uniform way of representing
translucent type definitions.
\citet{StoneH00popl} later proposed \emph{singleton kinds} as an
alternative mechanism for representing type definitions with sharing
constraints.
Interval kinds are closely related, conceptually and formally, to Stone and Harper's
singleton kinds.

The safety of (in)consistent subtyping constraints in \FSub{}-like
systems has been studied in depth by~\citet{CretinR14lics} and
\citet{SchererR15esop}.
They formalize two distinct types of subtyping coercions:
\emph{coherent} coercions can be erased (\ie used implicitly) while
\emph{incoherent} coercions are introduced and eliminated explicitly.
Reductions is allowed (and safe) only under coherent abstractions.
It is unclear if their results extend to \FOmegaSubC{}-like systems
such as ours.

Hereditary substitution is due to \citet{WatkinsCPW04types} and has
been used to prove weak normalization of a variety of systems.
A particularly illuminating example is provided by
\citet{KellerA10msfp} who use it to implement a normalization function
for STLC in Agda.
Other examples are the work by \citet{AbelR08csl} on \KFOmegaSub{} and
the presentation of Canonical~LF by \citet{HarperL07jfp}, both of
which inspired the metatheoretic development in this paper.
Hereditary substitution has also been used to mechanize the equational
theory of singleton kinds~\cite{Crary09lfmtp} and the semantics of the
SML language~\cite{LeeCH07popl} in Twelf.

\FOmegaInt{} belongs to a long line of calculi developed to model
Scala's type system.
One of the first to support complex type operators and a form of
interval kinds was \emph{Scalina} \cite{MoorsPO08fool}.
Type and kind safety of Scalina was never established but it inspired
an extension of Scala's type system with HK types, including
higher-order bounded polymorphism, type operators and type
definitions~\cite{MoorsPO08oopsla}.
More recently, \citet{AminGORS16wf} introduced the calculus of
\emph{Dependent Object Types (DOT)} as a theoretical foundation for
Scala and a core calculus for the \emph{Scala~3}
compiler~\cite{Dotty20code}.
Many variants of DOT have been developed, differing in expressiveness and presentation;
most come with mechanized type safety proofs~\citep[see e.g.][]{Amin16thesis,RompfA16oopsla,RapoportL19oopsla,GiarrussoSTBK20}.
Central to all is the notion of \emph{abstract type members}.
Because type members can have lower and upper bounds, they provide a
form of type intervals.
In~\FOmegaInt{}, we separate the concept of type intervals from that
of abstract type members via interval kinds.
DOT admits encodings of some type operators, but none of the DOT
calculi developed so far can express general HK types as supported by
Scala.
But the development of the Scala~3 compiler has shown the need for a principled theory of HK
types~\cite{OderskyMP16scala}.
In this paper, we have proposed such a theory.

In the future, we wish to extend our work on HK types with existing
work on DOT by recombining abstract type members with interval kinds.
A type member definition would then be of the form
$\record{ X \colon K }$ where $K$ may be a higher-order interval kind.
We expect this to cause new
feature interactions, some of which may be problematic.
A sketch of such an extension, including a brief discussion of
potential issues, can be found in the first author's
dissertation~\cite[][Ch.~6]{Stucki17thesis}.

Another direction for future work is to adapt the techniques developed
by \citet{HuL20popl} for algorithmic subtyping in the DOT-like
calculus $D_{<:}$ to our system.
\citeauthor{HuL20popl} address the problems caused by inconsistent
bounds in $D_{<:}$ by replacing the general subtyping transitivity
rule with a specialized rule that combines transitivity and bound
projection.
This isolates the problematic use of inequality reflection in a single
rule.
They then show that one can obtain a decidable system by removing this
rule and weakening the rule for subtyping universals.
Not only is this an elegant solution, it also closely reflects the
strategy implemented in the Scala compiler.
We do believe that a variant of \citeauthor{HuL20popl}'s strategy
could be applied to our system.
However, we expect transitivity elimination to be considerably more
challenging in our system than in $D_{<:}$, as one would expect in a
dependently kinded setting.

\enableLstShortInline

\section{Conclusions}
\label{sec:conclusions}

We have described \FOmegaInt{}, a formal theory of higher-order
subtyping with type intervals.
In \FOmegaInt{}, type intervals are represented through interval
kinds.
We showed how interval kinds can be used to encode bounded universal
quantification, bounded type operators and singleton
kinds. We illustrated the use of interval kinds to abstract over and reflect
type inequations and discussed the problems that arise
when the corresponding intervals have inconsistent bounds.

We established basic metatheoretic properties of \FOmegaInt{}.
We proved subject reduction in its full generality on the type level,
and in a restricted form on the term level.
We showed that types and kinds are weakly normalizing by defining a
bottom-up normalization procedure on raw kinds and types and proving
its soundness.
We gave an alternative, canonical presentation of the kind and type
level of \FOmegaInt{}, defined directly on $\beta\eta$-normal forms.
We showed that hereditary substitutions preserve canonical judgments
and used this result to establish equivalence of the declarative and
canonical presentations.
We showed that canonical and, by equivalence, declarative subtyping
can be inverted in the empty context.
Based on these results, we established type safety of \FOmegaInt{}.
We concluded our metatheoretic development by showing that subtyping
is undecidable in~\FOmegaInt{}.
The metatheory
has been fully mechanized in Agda. 

Our goal in developing \FOmegaInt{} was twofold:
study the theory of type intervals for
higher-order subtyping, and develop a foundation for Scala's
higher-kinded types.
We believe that \FOmegaInt{} fulfills this goal and constitutes an important step toward a full formalization of
Scala's expressive type system.
During the development of \FOmegaInt{}, we discovered a number of
minor flaws in Scala~3 (listed in \SupSec{scala_issues}).
None of these issues constitute critical bugs -- in particular, they do not break type safety.
But they do illustrate that subtyping in Scala~3 is slightly weaker
than necessary, suggesting that there is room for improvement.
Indeed, we hope that \FOmegaInt{} will serve as a blueprint for a more
principled implementation of higher-order subtyping in future versions
of Scala.

\disableLstShortInline

\begin{acks}
  We owe special thanks to Guillaume Martres for many discussions
  about this work, for his patience in answering our questions about
  the Scala~3 type checker, and for his striking ability to produce
  counterexamples to type system drafts.
For insightful discussions and feedback on earlier versions of this
  work we thank Andreas Abel, Nada Amin, Jesper Cockx, Martin Odersky and François
  Pottier.
We thank the anonymous reviewers for their helpful comments and
  suggestions.
This paper is based upon work supported by the
  \grantsponsor{501100000781}{European Research Council (ERC)}{
    http://dx.doi.org/10.13039/501100000781} under
  Grant~\grantnum{501100000781}{587327 DOPPLER}
  and by the
  \grantsponsor{501100004359}{Swedish Research Council (VR)}{
    http://dx.doi.org/10.13039/501100004359} under
  Grants~\grantnum{501100004359}{2015-04154 PolUser} and
  \grantnum{501100004359}{2018-04230 Perspex}.
\end{acks}

\appendix

\section{Overview of the Agda Mechanization}
\label{sec:modules}

The following is a list of the modules included in the Agda
mechanization, with short descriptions of their purpose.
The list is organized in blocks which correspond roughly to the
sections of the paper where the metatheoretic definitions and
properties corresponding to the contents of the module are (first)
described.
For more information, see the \emph{README.md} and
\emph{Correspondence.agda} files included in our
artifact~\cite{artifact}.

\subsection{The Declarative System}

Syntax of raw (\ie untyped) terms along with support for untyped
substitutions.
\begin{itemize}
\item\texttt{FOmegaInt.Syntax}
\end{itemize}
Variants of $\beta$-reduction/equivalence and properties thereof.
\begin{itemize}
\item\texttt{FOmegaInt.Reduction.Cbv}
\item\texttt{FOmegaInt.Reduction.Full}
\end{itemize}
Declarative typing, kinding, subtyping, etc.\ along with corresponding
substitution lemmas.
\begin{itemize}
\item\texttt{FOmegaInt.Typing}
\end{itemize}
An alternative presentation of kinding and subtyping that is better
suited for proving functionality and validity lemmas, and a proof that
the two presentations are equivalent.
\begin{itemize}
\item\texttt{FOmegaInt.Kinding.Declarative}
\item\texttt{FOmegaInt.Kinding.Declarative.Validity}
\item\texttt{FOmegaInt.Kinding.Declarative.Equivalence}
\end{itemize}
Encodings and properties of higher-order extremal types, interval
kinds and bounded quantifiers.
\begin{itemize}
\item\texttt{FOmegaInt.Typing.Encodings}
\end{itemize}

\subsection{Normalization of Types}

Hereditary substitutions and normalization of raw types and kinds.
\begin{itemize}
\item\texttt{FOmegaInt.Syntax.SingleVariableSubstitution}
\item\texttt{FOmegaInt.Syntax.HereditarySubstitution}
\item\texttt{FOmegaInt.Syntax.Normalization}
\end{itemize}
Weak equality of raw terms (up to kind annotations).
\begin{itemize}
\item\texttt{FOmegaInt.Syntax.WeakEquality}
\end{itemize}
Soundness of normalization \wrt to declarative kinding.
\begin{itemize}
\item\texttt{FOmegaInt.Kinding.Declarative.Normalization}
\end{itemize}
Simple kinding of types, and hereditary substitution lemmas; lemmas
about $\eta$-expansion of simply kinded types and kinds.
\begin{itemize}
\item\texttt{FOmegaInt.Kinding.Simple}
\item\texttt{FOmegaInt.Kinding.Simple.EtaExpansion}
\end{itemize}
Normalization and simultaneous simplification of declaratively kinded
types.
\begin{itemize}
\item\texttt{FOmegaInt.Kinding.Simple.Normalization}
\end{itemize}

\subsection{The Canonical System}

Canonical kinding of types along with (hereditary) substitution,
validity and inversion lemmas for canonical kinding and subtyping.
\begin{itemize}
\item\texttt{FOmegaInt.Kinding.Canonical}
\item\texttt{FOmegaInt.Kinding.Canonical.HereditarySubstitution}
\item\texttt{FOmegaInt.Kinding.Canonical.Validity}
\item\texttt{FOmegaInt.Kinding.Canonical.Inversion}
\end{itemize}
Lifting of weak (untyped) kind and type equality to canonical kind
and type equality.
\begin{itemize}
\item\texttt{FOmegaInt.Kinding.Canonical.WeakEquality}
\end{itemize}
Equivalence of canonical and declarative kinding.
\begin{itemize}
\item\texttt{FOmegaInt.Kinding.Canonical.Equivalence}
\end{itemize}
Generation of typing and inversion of declarative subtyping in the
empty context.
\begin{itemize}
\item\texttt{FOmegaInt.Typing.Inversion}
\end{itemize}
Type safety (preservation and progress).
\begin{itemize}
\item\texttt{FOmegaInt.Typing.Preservation}
\item\texttt{FOmegaInt.Typing.Progress}
\end{itemize}

\subsection{Undecidability of Subtyping}

A reduced variant of the canonical system.
\begin{itemize}
\item\texttt{FOmegaInt.Kinding.Canonical.Reduced}
\end{itemize}
Setup for the undecidability proof:
syntax and lemmas for the SK combinator calculus, and support for
encoding/decoding SK terms and equality proofs into types and
subtyping derivations.
\begin{itemize}
\item\texttt{FOmegaInt.Undecidable.SK}
\item\texttt{FOmegaInt.Undecidable.Encoding}
\item\texttt{FOmegaInt.Undecidable.Decoding}
\end{itemize}
Undecidability of subtyping
\begin{itemize}
\item\texttt{FOmegaInt.Undecidable}
\end{itemize}

\subsection{Auxiliary Modules Providing Generic Functionality}

Generic support for typing contexts over abstract bindings.
\begin{itemize}
\item\texttt{Data.Context}
\item\texttt{Data.Context.WellFormed}
\item\texttt{Data.Context.Properties}
\end{itemize}
Extra lemmas that are derivable in the substitution framework of
the Agda standard library, as well as support for binary (term)
relations lifted to substitutions, typed substitutions, and typed
relations lifted to substitutions.
\begin{itemize}
\item\texttt{Data.Fin.Substitution.Extra}
\item\texttt{Data.Fin.Substitution.ExtraLemmas}
\item\texttt{Data.Fin.Substitution.Relation}
\item\texttt{Data.Fin.Substitution.Typed}
\item\texttt{Data.Fin.Substitution.TypedRelation}
\end{itemize}
Support for generic reduction relations, and relational reasoning for
transitive relations.
\begin{itemize}
\item\texttt{Relation.Binary.Reduction}
\item\texttt{Relation.Binary.TransReasoning}
\end{itemize}

\section{Basic Metatheory and Admissible Rules of the Declarative
  System}
\label{sec:decl_details}

We establish basic meta-theoretic properties of the declarative system
and introduce a number of admissible rules, many of which are used in
proofs later on.
In particular, we establish a series of standard lemmas stating that the
declarative judgments are preserved under common operations on
contexts.
Next, we show that the usual order-theoretic
rules for subkinding and type and kind equality are admissible, as are
congruence rules (\wrt to all the type formers) for type equality.
We further introduce admissible typing, kinding, subtyping and
subkinding rules for the encoded higher-order extremal types, interval
kinds and bounded quantifiers described in \Sec{declarative} of the
paper, and discuss alternative encodings.
Finally, we state and prove a number of standard validity properties
for the various judgments.

The contents of this section are based on Chapter~3 of the first
author's PhD~dissertation.  We refer the interested reader to the
dissertation for the full details~\cite[][Ch.~3]{Stucki17thesis}.

\subsection{Basic Metatheoretic Properties}
\label{sec:decl_basics}

We start our metatheoretic development by showing that judgments can
only be derived in well-formed contexts.
\begin{lemma}[context validity]\label{lem:decl_ctx_valid}
If $\; \judgD{\Gamma}$ for any of the judgments defined above, then
  $\ctxD{\Gamma}$.
\end{lemma}
\begin{proof}
  By (simultaneous) induction on the derivations of the various
  judgments.
The cases for context formation judgments and rules that contain
  $\ctxD{\Gamma}$ as a premise are trivial.
For other rules, the result follows by applying the IH to any of the
  premises that do not extend the context.  There is always at least
  one such premise.
\end{proof}

\noindent Next, we establish a series of standard lemmas stating that the
declarative judgments are preserved under common operations on
contexts,
namely, addition and narrowing of bindings, and substitutions.
\begin{lemma}[weakening]\label{lem:decl_weaken}
  If $\; \judgD{\Gamma, \Delta}$, then
  \begin{enumerate}[nosep]
  \item for any type $A$ and $x \notin \dom(\Gamma, \Delta)$, if
    $\; \typeD{\Gamma}{A}$, then $\judgD{\Gamma, x \tas A, \Delta}$;
  \item for any kind $K$ and $X \notin \dom(\Gamma, \Delta)$, if
    $\; \kindD{\Gamma}{K}$, then $\judgD{\Gamma, X \kas K, \Delta}$.
  \end{enumerate}
\end{lemma}
\begin{corollary}[iterated weakening]\label{cor:decl_iter_weaken}
  If $\; \ctxD{\Gamma, \Delta}$ and $\judgD{\Gamma}$, then
  $\judgD{\Gamma, \Delta}$.
\end{corollary}
\begin{lemma}[substitution]\label{lem:decl_subst}
  ~
  \begin{enumerate}[nosep]
  \item If $\; \Gamma \ts t \tin A$ and
    $\judgD{\Gamma, x \tas A, \Delta}$, then
$\judgD[\subst{\genjudg}{x}{t}]{\Gamma, \Delta}$.
  \item If $\; \Gamma \ts A \kin K$ and
    $\judgD{\Gamma, X \kas K, \Delta}$, then
    $\judgD[\subst{\genjudg}{X}{A}]{\Gamma, \subst{\Delta}{X}{A}}$.
  \end{enumerate}
\end{lemma}
\begin{lemma}[context narrowing]\label{lem:decl_weak_narrow}
  ~
  \begin{enumerate}[nosep]
  \item If $\; \typeD{\Gamma}{A}$, $\Gamma \ts A \tsub B \kin \kstar$
    and $\judgD{\Gamma, x \tas B, \Delta}$, then
    $\judgD{\Gamma, x \tas A, \Delta}$.
  \item If $\; \kindD{\Gamma}{J}$, $\Gamma \ts J \ksub K$ and
    $\judgD{\Gamma, X \kas K, \Delta}$, then
    $\judgD{\Gamma, X \kas J, \Delta}$.
  \end{enumerate}
\end{lemma}\noindent
The context narrowing lemma is a bit weaker than one might expect;
the premises $\typeD{\Gamma}{A}$ and $\kindD{\Gamma}{J}$ seem
redundant.
Surely, if $\Gamma \ts A \tsub B \kin \kstar$ then $A$ and $B$ ought
to be proper types.
This property --~called \emph{subtyping validity}~-- does indeed hold,
but we are not yet ready to prove it.
Indeed, one of the prerequisites is the context narrowing lemma itself.

\Lem[Lemmas]{decl_weaken}, \ref{lem:decl_subst}
and~\ref{lem:decl_weak_narrow} are proven in that order, each by
simultaneous induction on the derivations of the various judgments.
The proofs of \Lem[Lemmas]{decl_subst} and~\ref{lem:decl_weak_narrow}
rely on \Cor{decl_iter_weaken} for the variable cases \ruleref{T-Var}
and \ruleref{K-Var}.
All three proofs are entirely standard, so we only present an excerpt
from the proof of \Lem{decl_subst} to illustrate the basic strategy.
In it, we make use of the following helper lemma about substitutions.
\begin{lemma}[substitutions commute]\label{lem:sub_commutes}
  Let $e$ be some arbitrary expression, $A$, $B$ types and $X$, $Y$
  distinct type variables such that $X \notin \fv(B)$.  Then
  $\subst{\subst{e}{X}{A}}{Y}{B} \alEq
  \subst{\subst{e}{Y}{B}}{X}{\subst{A}{Y}{B}}$
\end{lemma}
\begin{proof}
  By induction on the structure of $e$.
\end{proof}

\begin{proof}[Proof of \Lem{decl_subst}]
  The two parts are proven separately, each by induction on the
  derivation of the second premise ($\judgD{\Gamma, x \tas A, \Delta}$
  for the first part, $\judgD{\Gamma, X \kas K, \Delta}$ for the
  second).
For the context formation judgment, the proofs proceed by a local
  induction on the structure of $\Delta$.
We show the cases for \ruleref{K-Var} and
  \rulerefN{ST-Beta1}{ST-$\beta_1$} for the second part of the lemma.
The other cases are similar.
\begin{itemize}
  \item \emph{Case \ruleref{K-Var}.}
$\; \genjudg$ is $X \kin J$ and we have
    $\ctxD{\Gamma, X \kas K, \Delta}$ and
    $(\Gamma, X \kas K, \Delta)(Y) = J$.
By the IH, we get $\ctxD{\Gamma, \subst{\Delta}{X}{A}}$.
We distinguish two sub-cases based on whether $Y = X$.
\begin{itemize}
    \item \emph{Sub-case $Y = X$.}
We want to show that
      $\Gamma, \subst{\Delta}{X}{A} \ts A \kin \subst{K}{X}{A}$.
By well-scopedness, $X$ does not occur freely in $K$, so we have
      $\subst{K}{X}{A} \alEq K$.
The desired result follows from applying iterated
      weakening~(\Cor{decl_iter_weaken}) to the premise
      $\Gamma \ts A \kin K$.
    \item \emph{Sub-case $Y \neq X$.}
We want to show that
      $\Gamma, \subst{\Delta}{X}{A} \ts Y \kin \subst{J}{X}{A}$.
Since the domains of $\Gamma$ and $\Delta$ are disjoint, $Y$
      must appear either in $\Gamma$ or $\Delta$ but not in both.
If $Y \in \dom(\Gamma)$, then $X$ does not occur freely in
      $J = \Gamma(Y)$ and hence
      $(\Gamma, \subst{\Delta}{X}{A})(Y) = J \alEq \subst{J}{X}{A}$.
Otherwise
      $(\Gamma, \subst{\Delta}{X}{A})(Y) = (\subst{\Delta}{X}{A})(Y) =
      \subst{J}{X}{A}$.
In either case we conclude by~\ruleref{K-Var}.
    \end{itemize}

  \item \emph{Case \rulerefN{ST-Beta1}{ST-$\beta_1$}.}
$\; \genjudg$ is
    $\app{(\lam{Y}{J_1}{B_1})}{B_2} \tsub \subst{B_1}{X}{B_2} \kin
    \subst{J_2}{X}{B_2}$ and we have
    $\Gamma, X \kas K, \Delta, Y \kas J_1 \ts B_1 \kin J_2$ and
    $\Gamma, X \kas K, \Delta \ts B_2 \kin J_1$.
By the IH we get
    \begin{alignat*}{3}
      \Gamma, \subst{\Delta}{X}{A}&, Y \kas \subst{J_1}{X}{A} &\;&\ts \;
        \subst{B_1}{X}{A} \; \kin \; \subst{J_2}{X}{A} &&\text{ and}\\
      \Gamma, \subst{\Delta}{X}{A}& &\;&\ts \; \subst{B_2}{X}{A} \;
        \kin \; \subst{J_1}{X}{A}.
    \end{alignat*}

    Applying \rulerefN{ST-Beta1}{ST-$\beta_1$}, we obtain
    \begin{align*}
    \Gamma, \subst{\Delta}{X}{A} \quad \ts \quad
      \subst{(\app{(\lam{Y}{J_1}{B_1})}{B_2})}{X}{A} \quad \tsub &\quad
      \subst{\subst{B_1}{X}{A}}{Y}{\subst{B_2}{X}{A}} \\
      \kin \,{} &\quad \subst{\subst{K}{X}{A}}{Y}{\subst{B_2}{X}{A}}
    \end{align*}
    and using \Lem{sub_commutes} twice, it follows that
    \[\Gamma, \subst{\Delta}{X}{A} \ts
      \subst{(\app{(\lam{Y}{J_1}{B_1})}{B_2})}{X}{A} \tsub
      \subst{\subst{B_1}{Y}{B_2}}{X}{A}
      \kin  \subst{\subst{K}{Y}{B_2}}{X}{A}\]
    which concludes the case.\qedhere
  \end{itemize}
\end{proof}

\subsection{Admissible Order-Theoretic Rules}
\label{sec:decl_order}

The rules~\ruleref{ST-Refl} and \ruleref{ST-Trans} establish that
subtyping is a preorder.
Via~\ruleref{ST-AntiSym}, we can lift these
properties to type equality, and show that the latter is symmetric.
Together, these properties make type equality an equivalence relation.
\begin{corollary} Type equality is an equivalence, \ie the following
  equality rules are admissible.
\end{corollary}
\begin{center}
  \infruleSimp[TEq-Refl]
  {\Gamma \ts A \kin K}
  {\Gamma \ts A \teq A \kin K}
  \hspace{4em}
  \infruleSimp[TEq-Sym]
  {\Gamma \ts A \teq B \kin K}
  {\Gamma \ts B \teq A \kin K}
  \medskip

  \infruleSimp[TEq-Trans]
  {\Gamma \ts A \teq B \kin K  \andalso  B \teq C \kin K}
  {\Gamma \ts A \teq C \kin K}
  \medskip
\end{center}
The same is true of kind equality, though we first have to establish
that subkinding is a preorder.
\begin{lemma} Subkinding is a preorder, \ie the following
  subkinding rules are admissible.
\end{lemma}
\begin{center}
  \infruleSimp[KS-Refl]
  {\kindD{\Gamma}{K}}
  {\Gamma \ts K \tsub K}
  \hspace{3em}
  \infruleSimp[KS-Trans]
  {\Gamma \ts K \tsub J  \andalso  J \tsub L}
  {\Gamma \ts K \tsub L}
\end{center}
\begin{proof}
  Separately for each rule, by structural induction on $K$ for
  \ruleref{KS-Refl} and on $J$ for \ruleref{KS-Trans}.  For the type
  interval cases we use the corresponding order-theoretic properties
  of subtyping.  The proof of transitivity uses context narrowing
  (\Lem{decl_weak_narrow}) and context validity (\Lem{decl_ctx_valid})
  for subkinding in the case where $J = \dfun{X}{J_1}{J_2}$.
\end{proof}

\begin{corollary} Kind equality is an equivalence, \ie the following
  equality rules are admissible.
\end{corollary}
\begin{center}
  \infruleSimp[KEq-Refl]
  {\kindD{\Gamma}{K}}
  {\Gamma \ts K \teq K}
  \hspace{2em}
  \infruleSimp[KEq-Trans]
  {\Gamma \ts K \teq J  \andalso  J \teq L}
  {\Gamma \ts K \teq L}
  \hspace{2em}
  \infruleSimp[KEq-Sym]
  {\Gamma \ts K \teq J}
  {\Gamma \ts J \teq K}
  \medskip
\end{center}
In addition, the following variants of subtyping and subkinding
reflexivity are also admissible, which makes the subtyping and
subkinding relations partial orders \wrt type and kind equality.
\begin{center}
  \medskip
  \infruleSimp[ST-Refl-TEq]
  {\Gamma \ts A \teq B \kin K}
  {\Gamma \ts A \tsub B \kin K}
  \hspace{3em}
  \infruleSimp[SK-Refl-KEq]
  {\Gamma \ts J \keq K}
  {\Gamma \ts J \tsub K}
  \medskip
\end{center}
Another consequence of these rules is that we can treat well-typed
terms and well-kinded types up to type and kind equality,
respectively.
\begin{corollary}[conversion] The following are admissible.
\end{corollary}
\begin{multicols}{2}
  \typicallabel{K-Conv}

  \infrule[\ruledef{K-Conv}]
  {\Gamma \ts A \kin J  \andalso  \Gamma \ts J \keq K}
  {\Gamma \ts A \kin K}

  \infrule[\ruledef{ST-Conv}]
  {\Gamma \ts A \tsub B \kin J  \andalso  \Gamma \ts J \keq K}
  {\Gamma \ts A \tsub B \kin K}

  \infrule[\ruledef{T-Conv}]
  {\Gamma \ts t \tin A  \andalso  \Gamma \ts A \teq B \kin \kstar}
  {\Gamma \ts t \tin B}

  \infrule[\ruledef{TEq-Conv}]
  {\Gamma \ts A \teq B \kin J  \andalso  \Gamma \ts J \keq K}
  {\Gamma \ts A \teq B \kin K}
\end{multicols}

In light of their order-theoretic properties, we often call kind and
type equality judgments \emph{equations} and subkinding and subtyping
judgments \emph{inequations}.
We sometimes use equational reasoning notation in proofs, \ie we write
\[
  \Gamma \; \ts \; A_1 \; \teq \; A_2 \; \teq \; \cdots \; \teq \; A_i
  \; \tsub \; A_{i+1} \; \tsub \; \cdots \; \teq \; A_n \; \alEq \;
  A_{n+1} \; \alEq \; \cdots \; \teq \; A_m \; \kin \; K
\]
to denote chains of (in)equations where the use of the corresponding
transitivity rules is left implicit.
In so doing, we may freely mix the relations $\tsub$, $\teq$ and
$\alEq$ provided that they are defined on the same sort (\ie kinds or
types).
Such chains are always interpreted as judgments of the weakest
relation they contain.

\subsection{Validity of the Declarative System}
\label{sec:decl_validity_full}

In this section, we state and prove a number of \emph{validity
  properties} for the various judgments defined in
\Fig[Figs.]{decl_rules1} and~\ref{fig:decl_rules2} of the paper
(page~\pageref{fig:decl_rules1}).
Roughly, we say that a judgment is valid if all its parts are
well-formed.  For example, \emph{subkinding validity} states that, if
$\Gamma \ts J \ksub K$, then both $J$ and $K$ are actually well-formed
kinds.  We saw another example earlier: \emph{context
  validity}~(\Lem{decl_ctx_valid}) states that the context $\Gamma$ of
any judgment $\judgD{\Gamma}$ is well-formed.  Here is a summary of
the validity properties that remain to be proven.

\begin{lemma}[validity]\label{lem:decl_validity_full}
  The judgments defined in \Fig[Figs.]{decl_rules1}
  and~\ref{fig:decl_rules2} enjoy the following validity properties.
\begin{enumerate}[nosep,leftmargin=11.5em,labelsep=1em,topsep=1ex]
\item[(kinding validity)]
If $\; \Gamma \ts A \kin K$, then $\kindD{\Gamma}{K}$.

\item[(typing validity)]
If $\; \Gamma \ts t \tin A$, then $\typeD{\Gamma}{A}$.

\item[(subkinding validity)]
If $\; \Gamma \ts J \ksub K$, then $\kindD{\Gamma}{J}$ and
    $\kindD{\Gamma}{K}$.

\item[(subtyping validity)]
If $\; \Gamma \ts A \tsub B \kin K$, then $\Gamma \ts A \kin K$ and
    $\Gamma \ts B \kin K$.

\item[(kind equation validity)]
If $\; \Gamma \ts J \keq K$, then $\kindD{\Gamma}{J}$ and
    $\kindD{\Gamma}{K}$.

\item[(type equation validity)]
If $\; \Gamma \ts A \teq B \kin K$, then $\Gamma \ts A \kin K$ and
    $\Gamma \ts B \kin K$.
  \end{enumerate}
\end{lemma}
These validity properties provide a ``sanity check'' for the static
semantics developed in this section, but they also play a crucial role
in the proofs of other important properties, such as subject
reduction, soundness of type normalization and and type safety.

Unfortunately, the validity properties are harder to prove than one
might expect.  The proofs of kinding, subkinding and subtyping
validity require the following \emph{functionality} lemma for
the case of \ruleref{ST-App}.
\begin{lemma}[functionality]\label{lem:decl_funct_full}
  Let $\Gamma \ts A_1 \teq A_2 \kin K$.
  \begin{enumerate}
  \item If $\; \kindD{\Gamma, X \kas K, \Delta}{J}$, then
    $\Gamma, \subst{\Delta}{X}{A_1} \ts \subst{J}{X}{A_1} \keq
    \subst{J}{X}{A_2}$.
  \item If $\; \Gamma, X \kas K, \Delta \ts B \kin J$, then
    $\Gamma, \subst{\Delta}{X}{A_1} \ts \subst{B}{X}{A_1} \keq
    \subst{B}{X}{A_2} \kin \subst{J}{X}{A_1}$.
  \end{enumerate}
\end{lemma}
But a naive attempt at proving this lemma directly leads to a circular
dependency on kinding and subtyping validity, in a way that is not
easily resolved.
In particular, it is not sufficient to simply prove the two statements
simultaneously.

It is instructive to play through the critical cases encountered when
attempting to prove \Lem[Lemmas]{decl_validity_full}
and~\ref{lem:decl_funct_full} directly to see where things go wrong
and to better understand the solution described in the next section.
We start with subtyping validity, attempting a proof by induction on
subtyping derivations.  For the case of the application
rule~\ruleref{ST-App}, we are given
$\Gamma \ts A_1 \tsub A_2 \kin \dfun{X}{J}{K}$ and
$\Gamma \ts B_1 \teq B_2 \kin J$, and we would like to show that
$\Gamma \ts \app{A_1}{B_1} \kin \subst{K}{X}{B_1}$ and
$\Gamma \ts \app{A_2}{B_2} \kin \subst{K}{X}{B_1}$.  We can already
spot the source of trouble: the type $B_1$ that is being substituted
for $X$ in the kind of the second type application differs from the
argument type $B_2$.  By the IH, we get
$\Gamma \ts A_2 \kin \dfun{X}{J}{K}$ and $\Gamma \ts B_2 \kin J$, and
applying~\ruleref{K-App} we obtain
$\Gamma \ts \app{A_2}{B_2} \kin \subst{K}{X}{B_2}$ but, as expected,
the kinds do not match up.  If we could show that
$\Gamma \ts \subst{K}{X}{B_2} \keq \subst{K}{X}{B_1}$, then
by~\ruleref{K-Conv}, we would be done.  Enter functionality of kind
formation.

For the functionality lemma, we attempt a proof by simultaneous
induction on kind formation and kinding derivations and consider the
case where the current kinding derivation ends in an instance of the
operator abstraction rule~\ruleref{K-Abs}.  In addition to the premise
$\Gamma \ts B_1 \teq B_2 \kin J$, we are given derivations for
$\kindD{\Gamma, X \kas J}{K_1}$ and
$\Gamma, X \kas J, Y \kas K_1 \ts A \kin K_2$, and we want to show
that
\[
  \Gamma \ts \; \subst{(\lam{Y}{K_1}{A})}{X}{B_1} \; \teq \;
  \subst{(\lam{Y}{K_1}{A})}{X}{B_2} \; \kin \;
  \subst{(\dfun{Y}{K_1}{K_2})}{X}{B_1}.
\]
To do so, we would like to use the rule~\ruleref{ST-Abs} together
with~\ruleref{ST-AntiSym} but we first need to establish the
right-hand validity of the above equation, \ie that
\[
  \Gamma \ts \subst{(\lam{Y}{K_1}{A})}{X}{B_2} \kin
  \subst{(\dfun{Y}{K_1}{K_2})}{X}{B_1}.
\]
Clearly, \Lem{decl_validity_full} would be helpful here: equation
validity would give us $\Gamma \ts B_2 \kin J$, from which we could
obtain
$\Gamma \ts \subst{(\lam{Y}{K_1}{A})}{X}{B_2} \kin
\subst{(\dfun{Y}{K_1}{K_2})}{X}{B_2}$ by~\Lem{decl_subst}.  This is
almost what we need.  Again, we face a mismatch in kinds that could,
in principle, be remedied by using functionality of kind formation
together with~\ruleref{K-Conv}.  Concretely, we would like to invoke
the IH to derive
$\Gamma \ts \subst{(\dfun{Y}{K_1}{K_2})}{X}{B_1} \keq
\subst{(\dfun{Y}{K_1}{K_2})}{X}{B_2}$.  Alas, we do not have a
suitable sub-derivation to do so.  Although
$\kindD{\Gamma}{\dfun{Y}{K_1}{K_2}}$ follows from kinding validity, we
cannot apply the IH to this result because it is not a
sub-derivation of our overall premise.  Indeed, none of the
sub-derivations we are given are sufficient to derive the required
kind equation.

Note that we could finish the proof of this case if only (1)~the
rule~\ruleref{K-Abs} had an additional premise
$\kind{\Gamma, X \kas J, Y \kas K_1}{K_2}$ and (2)~the IH was a bit
stronger, so that we could use it to derive
\[
  \Gamma, Y \kas \subst{K_1}{X}{B_1} \ts \subst{K_2}{X}{B_1} \keq
  \subst{K_2}{X}{B_2}.
\]
This, together with a similar use of the IH on the first premise
of~\ruleref{ST-Abs} and some uses of \Lem{decl_subst},
\ruleref{ST-Refl-TEq} and~\ruleref{SK-DArr}, would be enough to derive
$\Gamma \ts \subst{(\dfun{Y}{K_1}{K_2})}{X}{B_2} \ksub
\subst{(\dfun{Y}{K_1}{K_2})}{X}{B_1}$,
which we could then put to use with~\ruleref{K-Sub}.  Indeed, these
are the basic ideas that will allow us to resolve the circular
dependency between~\Lem{decl_validity_full} and~\Lem{decl_funct_full}.
We start by addressing point~(1).

\subsubsection{The Extended System}

We define a pair of \emph{extended} kinding and subtyping judgments
where some rules have been endowed with additional premises.
These are precisely the premises highlighted in gray in
\Fig[Figs.]{decl_rules1} and~\ref{fig:decl_rules2}.
We call these extra premises \emph{validity conditions}.
Crucially, the validity conditions of an extended rule are redundant
in the sense that they follow (more or less) directly from the
remaining premises of the rule via \Lem{decl_validity_full}.
For example, the extended rule~\ruleref{K-Abs} carries the extra
premise $\kindD{\Gamma, X \kas J}{K}$ which follows directly from
applying kinding validity to the rule's second premise
$\Gamma, X \kas J \ts A \kin K$.  Thanks to this invariant, the two
sets of rules are in fact equivalent -- every derivation of an
extended kinding or subtyping judgment has a corresponding derivation
that uses only original rules, and vice-versa.  We give a formal
equivalence proof in \Sec{decl_equivalence}.

Since kinding and subkinding are defined mutually with all the other
judgments of the declarative system (except typing), the extension
indirectly affects those judgments as well.  We call the entire set of
extended judgments the \emph{extended (declarative) system}, as
opposed to the \emph{original (declarative) system}.  We will
sometimes distinguish the two systems by writing $\judgDD{\Gamma}$ for
judgments of the original system and $\judgEE{\Gamma}$ for those of
the extended system.  Since the two systems are equivalent, this
distinction only matters in a few key situations -- notably the
development in the remainder of this section.  When we refer to the
``declarative system'' in the following sections, we will always mean
the original declarative system, unless otherwise noted.

The idea of extending a set of inference rules with redundant premises
in order to simplify metatheoretic proofs is not new.
For example, \citet{HarperP05tocl} use similar premises to establish
validity properties for the typing judgments of a variant of LF.
Furthermore, some readers will have noticed that a few of the original
declarative rules already carry redundant premises.  For example, the
first premise $\kindD{\Gamma}{J}$ of the rule \ruleref{K-Abs} could
easily be reconstructed from its second premise
$\Gamma, X \kas J \ts A \kin K$ via context validity
(\Lem{decl_ctx_valid}).  The rules~\ruleref{Wf-DArr}, \ruleref{K-All},
\ruleref{T-Abs}, and~\ruleref{T-TAbs} carry similar validity
conditions.  We include these premises primarily because their
presence simplifies the proof of the substitution lemma
(\Lem{decl_subst}).  Context validity, on the other hand, remains
easily provable without them.

To prove the validity properties stated in~\Lem{decl_validity_full}
for both the original and extended systems, we use the following
strategy:
\begin{enumerate}
\item\label{item:val-s1} prove that the validity properties hold for the
  extended judgments;

\item prove that the two systems are equivalent, \ie that
  \begin{enumerate}
  \item the extended rules are \emph{sound} \wrt to the original
    ones --~we can drop the validity conditions without affecting the
    conclusions of any derivations~-- and that
  \item the extended rules are \emph{complete} \wrt to the original
    ones -- the additional validity conditions follow from the remaining
    premises of the extended rules via the validity properties proved
    in step~\Item{val-s1};
  \end{enumerate}

\item prove that the validity properties hold for the original system
  via the equivalence -- convert original derivations to extended
  derivations (via completeness), derive the property in question,
  convert the conclusion back (via soundness).
\end{enumerate}

Before we continue, we should point out that some of the validity
conditions of the extended system are not actually necessary for the
proof of~\Lem[Lemmas]{decl_validity_full}
and~\ref{lem:decl_funct_full} -- some even complicate the proofs.
However, we will face a similar cyclic dependency later on when
attempting to prove the equivalence of the (original) declarative
system and the \emph{canonical system} of judgments introduced
in~\Sec{canonical}.  Rather than introducing yet another extension to
the declarative system later, we opt for a combined system containing
all of the extra conditions.

We start the development set out above by noting that all the basic
metatheoretic properties established in \Sec{decl_basics} still hold
for the extended system.  The proofs carry over with minor adjustments
to deal with the additional premises.
Next, we prove a variant of the functionality lemma discussed above
but using the extended kinding and subtyping rules.
To do so, we need the following auxiliary definition of \emph{context
  equality} and an associated corollary of context weakening.
\begin{center}
  \infruleSimp{\quad}{\ctxEqD{\cempty}{\cempty}}
  \hspace{3em}
  \infruleSimp
  {\ctxEqD{\Gamma}{\Delta} \andalso \Gamma \ts J \keq K}
  {\ctxEqD{\Gamma, X \kas J}{\Delta, X \kas K}}
  \hspace{3em}
  \infruleSimp
  {\ctxEqD{\Gamma}{\Delta} \andalso \Gamma \ts A \teq B \kin \kstar}
  {\ctxEqD{\Gamma, x \tas A}{\Delta, x \tas B}}
\end{center}
\begin{corollary}[context conversion]\label{cor:decl_ctx_conv}
  If $\; \ctxD{\Gamma}$ and $\ctxEqD{\Gamma}{\Delta}$ and $\judgD{\Delta}$,
  then $\judgD{\Gamma}$.
\end{corollary}
Context equality $\ctxEqD{\Gamma}{\Delta}$ is simply the pointwise
lifting of type and kind equality to contexts.
It allows us to relate the bindings appearing in two syntactically
different contexts $\Gamma$ and $\Delta$, as illustrated in our
extended functionality lemma.
\begin{lemma}[functionality -- extended version]
  \label{lem:decl_funct_ext}
  Substitutions of equal types in well-formed expressions result in
  well-formed equations.  Let $\Gamma$, $\Delta$, $\Sigma$ be
  contexts, $K$ a kind and $A_1$, $A_2$ types, such that
  $\Gamma \ts A_1 \kin K$ and $\Gamma \ts A_2 \kin K$, the context
  $\Gamma, \Sigma$ is well-formed and the following equations hold:
  \begin{align*}
    &\Gamma \ts A_1 \; \teq \; A_2 \kin K
    &&\ctxD{\Gamma, \Sigma \; \keq \; \Gamma, \subst{\Delta}{X}{A_1}}
    &&\ctxD{\Gamma, \Sigma \; \keq \; \Gamma, \subst{\Delta}{X}{A_2}}.
  \end{align*}
  \begin{enumerate}
  \item If $\; \kindD{\Gamma, X \kas K, \Delta}{J}$, then
    $\Gamma, \Sigma \ts \subst{J}{X}{A_1} \keq \subst{J}{X}{A_2}$.
  \item If $\; \Gamma, X \kas K, \Delta \ts B \kin J$, then
    $\Gamma, \Sigma \ts \subst{B}{X}{A_1} \keq \subst{B}{X}{A_2} \kin
    \subst{J}{X}{A_1}$.
  \end{enumerate}
\end{lemma}
Compared to \Lem{decl_funct_full}, the lemma has been strengthened
--~so that it is applicable to any type variable binding in a context,
not just the last one~-- and simultaneously weakened -- by adding
extra conditions on $A_1$, $A_2$ and the target context
$\Gamma, \Sigma$.
The latter are effectively validity conditions ensuring that the proof
of the lemma does not depend on \Lem{decl_validity_full}.
The separate target context~$\Sigma$ is used to symmetrize the
treatment of context extensions, which is helpful when dealing with
kind annotations in contravariant positions.
\begin{proof}
  The two parts are proven simultaneously, by induction on extended
  kind formation and kinding derivations, respectively.  The proof of
  the first part is relatively straightforward, while the proof of
  the second part deserves some attention.  We present a few key
  cases, the others are similar.
  \begin{itemize}
  \item \emph{Case~\ruleref{K-Var}.} We have $B = Y$, $J = \Gamma(Y)$
    and $\ctxD{\Gamma, X \kas K, \Delta}$.  We distinguish two cases
    based on $Y$.
    \begin{itemize}
    \item \emph{Sub-case $Y = X$.}  We have
      $J = K \alEq \subst{K}{X}{A_1}$ since $X \notin \fv(K)$.  By
      iterated weakening, we get
      $\Gamma, \Sigma \ts A_1 \teq A_2 \kin \subst{K}{X}{A_1}$ and we
      are done.
    \item \emph{Sub-case $Y \neq X$.}  We have
      $(\Gamma, \subst{\Delta}{X}{A_1})(Y) \alEq \subst{J}{X}{A_1}$,
      either because $Y \in \dom(\Gamma)$ and $X \notin \fv(J)$, or
      because $Y \in \dom(\Delta)$ and
      $(\subst{\Delta}{X}{A_1})(Y) = \subst{J}{X}{A_1}$.  Furthermore,
      since $\ctxEqD{\Gamma, \Sigma}{\Gamma, \subst{\Delta}{X}{A_1}}$
      we have
      \[
        \Gamma, \Sigma \; \ts \; (\Gamma, \Sigma)(Y) \; \keq
        \; (\Gamma, \subst{\Delta}{X}{A_1})(Y) \; \alEq \;
        \subst{J}{X}{A_1}.
      \]
      We conclude by~\ruleref{K-Var}, \ruleref{K-Conv},
      and~\ruleref{TEq-Refl}.
    \end{itemize}

  \item \emph{Case~\ruleref{K-All}.} We have $B = \all{Y}{J_1}{B_1}$
    and $J = \kstar$ for some kind $J_1$ and type $B_1$, as well as
    $\kindD{\Gamma, X \kas K, \Delta}{J_1}$ and
    $\typeD{\Gamma, X \kas K, \Delta, Y \kas J_1}{B_1}$.  We want to
    show that $\subst{(\all{Y}{J_1}{B_1})}{X}{A_1}$ and
    $\subst{(\all{Y}{J_1}{B_1})}{X}{A_2}$ are mutual subtypes.  To do
    so, we first prove that
    \begin{alignat*}{2}
      &\typeD{\Gamma, \Sigma \;}{\; \subst{(\all{Y}{J_1}{B_1})}{X}{A_1} \;}
      && \quad
      \typeD{\Gamma, \Sigma \;}{\; \subst{(\all{Y}{J_1}{B_1})}{X}{A_2} \;}\\
      &\Gamma, \Sigma \; \ts \; \subst{J_1}{X}{A_2} \; \ksub \;
        \subst{J_1}{X}{A_1} && \quad
      \Gamma, \Sigma \; \ts \; \subst{J_1}{X}{A_1} \; \ksub \;
        \subst{J_1}{X}{A_2} \\
      &\Gamma, \Sigma, X \kas \subst{J_1}{X}{A_2} \; \ts \;
        \subst{B_1}{X}{A_1} \; \tsub \; \subst{B_1}{&&X}{A_2} \; \kin \;
        \kstar \\
      &\Gamma, \Sigma, X \kas \subst{J_1}{X}{A_1} \; \ts \;
        \subst{B_1}{X}{A_2} \; \tsub \; \subst{B_1}{&&X}{A_1} \; \kin \;
        \kstar
    \end{alignat*}
    then apply~\ruleref{ST-All} twice, and conclude
    with~\ruleref{ST-AntiSym}.  The two kinding judgments follow from
    the premises by the substitution lemma~(\Lem{decl_subst}), the two
    subkinding judgments by the IH.  The last two subtyping judgments
    require some extra work.

    Note that the additional type variable bindings in the two
    judgments differ syntactically, so we will have to use the IH
    twice, with different target contexts.  In each case we need to
    show that the kind of the additional binding is well-formed and
    equal to $\subst{J_1}{X}{A_1}$ and $\subst{J_1}{X}{A_2}$,
    respectively.  Concretely, we need to show that
    \begin{gather*}
      \kindD{\Gamma, \Sigma}{ \subst{J_1}{X}{A_2} \keq \subst{J_1}{X}{A_1}}
      \hspace{4em} \kindD{\Gamma, \Sigma}{\subst{J_1}{X}{A_2}} \\
      \kindD{\Gamma, \Sigma}{ \subst{J_1}{X}{A_2} \keq \subst{J_1}{X}{A_2}}
    \end{gather*}
    for the first invocation of the IH, and three analogous statements
    for the second.  The first equation follows from the two
    subkinding judgments above via~\ruleref{ST-AntiSym}.  The context
    formation judgment and the first equation follow from the
    substitution lemma and~\ruleref{TEq-Refl}.  This is sufficient to
    apply the IH and obtain the first of the two remaining subtyping
    judgments via~\ruleref{ST-Refl-TEq}.  The proof of the second one
    is similar.

  \item \emph{Case~\ruleref{K-App}.} We have $B = \app{B_1}{B_2}$
    and $J = \subst{J_2}{Y}{B_2}$ for some $B_1$, $B_2$, $J_2$, as
    well as
    \begin{align*}
      &\Gamma \ts A_1 \kin K, && \Gamma \ts A_2 \kin K \\
      &\Gamma \ts A_1 \teq A_2 \kin K, && \ctxD{\Gamma, \Sigma} \\
      &\ctxD{\Gamma, \Sigma \keq \Gamma, \subst{\Delta}{X}{A_1}}, &
      &\ctxD{\Gamma, \Sigma \keq \Gamma, \subst{\Delta}{X}{A_2}}, \\
      &\Gamma, X \kas K, \Delta \ts B_1 \kin \dfun{Y}{J_1}{J_2}, &
      &\Gamma, X \kas K, \Delta \ts B_2 \kin J_1, \\
      &\kindD{\Gamma, X \kas K, \Delta, Y \kas J_1}{J_2}, &
      &\kindD{\Gamma, X \kas K, \Delta}{\subst{J_2}{Y}{B_2}}
    \end{align*}
    for some $J_1$.  We want to establish that that
    $\subst{(\app{B_1}{B_2})}{X}{A_1}$ and
    $\subst{(\app{B_1}{B_2})}{X}{A_2}$ are mutual subtypes in
    $\subst{\subst{J_2}{Y}{B_2}}{X}{A_1}$, \ie that
    \begin{align}
      &\Gamma, \Sigma \; \ts \; \subst{(\app{B_1}{B_2})}{X}{A_1} \; \tsub \;
      \subst{(\app{B_1}{B_2})}{X}{A_2} \; \kin \;
      \subst{\subst{J_2}{Y}{B_2}}{X}{A_1}, \quad \text{and}
        \label{eq:funct1}\\
      & \Gamma, \Sigma \; \ts \; \subst{(\app{B_1}{B_2})}{X}{A_2}
      \; \tsub \; \subst{(\app{B_1}{B_2})}{X}{A_1} \; \kin \;
      \subst{\subst{J_2}{Y}{B_2}}{X}{A_1}.\label{eq:funct2}
    \end{align}

    The first half is fairly straightforward.  Applying the IH to the
    first two premises of~\ruleref{K-App} yields
    corresponding equations, the first of which we turn into an
    inequation via~\ruleref{ST-Refl-TEq}.
    \begin{align*}
      &\Gamma, \Sigma \; \ts \; \subst{B_1}{X}{A_1} \; \tsub \;
        \subst{B_1}{X}{A_2} \; \kin \; \subst{(\dfun{Y}{J_1}{J_2})}{X}{A_1},\\
      &\Gamma, \Sigma \; \ts \; \subst{B_2}{X}{A_1} \; \teq \;
        \subst{B_2}{X}{A_2} \; \kin \; \subst{J_1}{X}{A_1}.
    \end{align*}

    In order to apply~\ruleref{ST-App} we also need to derive the
    following validity conditions:
    \begin{align*}
      & \Gamma, \Sigma \ts \subst{B_2}{X}{A_1} \kin \subst{J_1}{X}{A_1}, &&
      \kindD{\Gamma, \Sigma, Y \kas \subst{J_1}{X}{A_1}}{
        \subst{J_2}{X}{A_1}}, \\
      & \kindD{\Gamma, \Sigma}{
        \subst{\subst{J_2}{X}{A_1}}{Y}{\subst{B_2}{X}{A_1}}}.
    \end{align*}
    All three follow from premises of~\ruleref{K-App} and the
    substitution lemma~(\Lem{decl_subst}), followed by a use
    of~\Cor{decl_ctx_conv} to adjust the contexts.  Adjusting the
    context of
    $\kindD{\Gamma, \subst{\Delta}{X}{A_1}, Y \kas
      \subst{J_1}{X}{A_1}}{ \subst{J_2}{X}{A_1}}$,
    requires a bit more work because we need to prove that the kind of
    the extra binding $Y \kas \subst{J_1}{X}{A_1}$ is well-formed.  To
    do so, we first invoke context validity~(\Lem{decl_ctx_valid}) on
    the second validity condition of~\ruleref{K-App}, which gives
    us $\kindD{\Gamma, X \kas K, \Delta}{J_1}$.  Form this, we derive
    the desired well-formedness proof via the substitution lemma.
    By~\ruleref{ST-App} and~\Lem{sub_commutes}, we arrive
    at~\eqref{eq:funct1}.

    We have to work a bit harder to prove~\eqref{eq:funct2}.
    Again, we want to apply~\ruleref{ST-App}, and again, the first
    two premises follow from the IH --~this time followed by a use
    of~\ruleref{TEq-Sym} to adjust the direction~--
    and~\ruleref{ST-Refl-TEq} to turn the first equation into a
    subtyping statement.  The validity conditions are
    \begin{align*}
      & \Gamma, \Sigma \ts \subst{B_2}{X}{A_2} \kin \subst{J_1}{X}{A_1}, &&
      \kindD{\Gamma, \Sigma, Y \kas \subst{J_1}{X}{A_1}}{
        \subst{J_2}{X}{A_1}}, \\
      & \kindD{\Gamma, \Sigma}{
        \subst{\subst{J_2}{X}{A_1}}{Y}{\subst{B_2}{X}{A_2}}}.
    \end{align*}
    We have already established the second condition; the third one
    follows from applying the substitution lemma to the first two.  So
    it remains to prove the first.

    We start by deriving
    $\Gamma, \Sigma \ts \subst{B_2}{X}{A_2} \kin \subst{J_1}{X}{A_2}$
    via the substitution lemma and context conversion.  Next, we would
    like to use the IH to show that
    $\Gamma, \Sigma \ts \subst{J_1}{X}{A_2} \keq \subst{J_1}{X}{A_1}$
    in order to adjust the kind of the previous judgment
    via~\ruleref{K-Conv}.  But to do so, we need to find a derivation
    of $\kindD{\Gamma, X \kas K, \Delta}{J_1}$ that is a strict
    sub-derivation of our current instance of \ruleref{K-App}.

    Fortunately, this is always possible, thanks to the validity
    condition $\kindD{\Gamma, X \kas K, \Delta, Y \kas J_1}{J_2}$.
    Since the contexts of kind formation judgments are always
    well-formed themselves (see~\Lem{decl_ctx_valid}), it suffices to
    traverse the derivation tree of this judgment upwards along kind
    formation and kinding rules until one arrives at a ``leaf'' --~an
    instance of~\ruleref{K-Var}, \ruleref{K-Top} or~\ruleref{K-Bot}~--
    which holds a well-formedness derivation for the current context.
    That context formation derivation, in turn, contains a
    sub-derivation of the desired kind formation judgment.  Readers
    who are skeptical of this somewhat informal argument are
    encouraged to state and prove a helper lemma that combines the IH
    with the ``lookup procedure'' just described.  The lemma is proven
    simultaneously with the main lemma, by induction on kind formation
    and kinding derivations.

    With all the validity conditions in place, we
    apply~\ruleref{ST-App} to obtain
    \[
      \Gamma, \Sigma \; \ts \; \subst{(\app{B_1}{B_2})}{X}{A_2} \;
      \tsub \; \subst{(\app{B_1}{B_2})}{X}{A_1} \; \kin \;
      \subst{\subst{J_2}{X}{A_1}}{Y}{\subst{B_2}{X}{A_2}}.
    \]
    To complete the proof of~\ref{eq:funct2}, we
    use~\ruleref{ST-Conv} and
    \begin{align*}
      \Gamma, \Sigma &{} \; \ts \;
      \subst{\subst{J_2}{X}{A_1}}{Y}{\subst{B_2}{X}{A_2}}\\
      &{} \; \keq \;
      \subst{\subst{J_2}{X}{A_2}}{Y}{\subst{B_2}{X}{A_2}}
      \tag{by~\Lem{decl_subst}}\\
      &{} \; \alEq \;
      \subst{\subst{J_2}{Y}{B_2}}{X}{A_2} \tag{by~\Lem{sub_commutes}}\\
      &{} \; \keq \;
      \subst{\subst{J_2}{Y}{B_2}}{X}{A_1} \tag{by the IH}
    \end{align*}
    using our earlier result
    $\Gamma, \Sigma \ts \subst{J_1}{X}{A_2} \keq \subst{J_1}{X}{A_1}$
    in the first step and the final validity condition
    of~\ruleref{K-App} in the last.  \qedhere

\end{itemize}
\end{proof}

We are now ready to prove~\Lem{decl_validity_full} in the extended
system, simultaneously with the following lemma.
\begin{lemma}\label{lem:ST-Intv-Star}
  Subtypes inhabiting interval kinds are proper subtypes.  If
  $\; \Gamma \ts A \tsub B \kin C \intv D$, then also
  $\Gamma \ts A \tsub B \kin \kstar$.
\end{lemma}
\begin{proof}[Proof of \Lem{decl_validity_full} and
  \Lem{ST-Intv-Star} -- extended version]\label{pf:decl_validity_ext}
  All the validity properties are proven simultaneously with
  \Lem{ST-Intv-Star}, by induction on the derivations of the
  respective premises.  The proof is now mostly routine, thanks to the
  validity conditions.  The only interesting cases are those of
  of~\ruleref{ST-App}, where we use the functionality lemma to adjust
  the kind of the right-hand validity proof, and~\ruleref{ST-Intv},
  where we use~\Lem{ST-Intv-Star}.  The proof of \Lem{ST-Intv-Star}
  uses subtyping validity in turn.
\end{proof}

\begin{corollary}\label{lem:TEq-Intv-Star}
  Equal types in intervals are equal as proper types.  If
  $\; \Gamma \ts A \teq B \kin C \intv D$, then also
  $\Gamma \ts A \teq B \kin \kstar$.
\end{corollary}

\subsubsection{Equivalence}
\label{sec:decl_equivalence}

The next and final step in our program for
proving~\Lem{decl_validity_full} is to establish the equivalence of
the two declarative systems.
\begin{lemma}\label{decl_equiv}
  The original and extended declarative systems are equivalent:
  $\judgDD{\Gamma}$ iff\, $\judgEE{\Gamma}$.
\end{lemma}

Thanks to this equivalence, all the validity properties laid out
in~\Lem{decl_validity_full} also hold for the original judgments of
the declarative system.
Our original functionality lemma (\Lem{decl_funct_full}) and the
following strengthened version of \Lem{decl_weak_narrow} follow as
corollaries of validity and \Lem[Lemmas]{decl_funct_ext} and
\ref{lem:decl_weak_narrow}, respectively.
\begin{corollary}[context narrowing -- strong version]
  \label{cor:decl_narrow}
  ~
  \begin{enumerate}
  \item If $\; \Gamma \ts A \tsub B \kin \kstar$
    and $\judgD{\Gamma, x \tas B, \Delta}$, then
    $\judgD{\Gamma, x \tas A, \Delta}$.
  \item If $\; \Gamma \ts J \ksub K$ and
    $\judgD{\Gamma, X \kas K, \Delta}$, then
    $\judgD{\Gamma, X \kas J, \Delta}$,
  \end{enumerate}
\end{corollary}
 
\subsection{Admissible Congruence Rules for Type and Kind Equality}
\label{sec:decl_cong}

Thanks to the validity properties established in
\Sec{decl_validity_full},
we are able to prove a number of admissible \emph{congruence rules}
for kind and type equality.
These follow the same structure as the corresponding subkinding and
subtyping rules but are generally a bit simpler.
First, we no longer need to pay attention to the variance (or
polarity) of constructor arguments because equality is symmetric.
Second, the left-hand validity conditions present in the
rules~\ruleref{SK-DArr} and~\ruleref{ST-All} become redundant in the
corresponding equality rules because the kind annotations in the left-
and right-hand sides are convertible.
Finally, thanks to symmetry, only one rule is needed for
$\beta$-conversion, and likewise for $\eta$-conversion.
\begin{lemma}\label{lem:keq_cong}
  Kind equality is a congruence with respect to the interval and
  dependent arrow kind formers, \ie the following kind equality
  rules are admissible.
\end{lemma}
\begin{multicols}{2}
\infrule[\ruledef{KEq-Intv}]
  {\Gamma \ts A_1 \teq A_2 \kin \kstar  \andalso
   \Gamma \ts B_1 \teq B_2 \kin \kstar}
  {\Gamma \ts A_1 \intv B_1 \keq A_2 \intv B_2}

\infrule[\ruledef{KEq-DArr}]
  {\Gamma \ts J_1 \keq J_2  \andalso
   \Gamma, X \kas J_1 \ts K_1 \keq K_2}
  {\Gamma \ts \dfun{X}{J_1}{K_1} \keq \dfun{X}{J_2}{K_2}}
\end{multicols}

\begin{lemma}\label{lem:teq_cong}
  Type equality is a congruence with respect to the various type
  formers and includes $\beta$~and $\eta$-conversion, \ie the
  following type equality rules are admissible.
\end{lemma}
\begin{multicols}{2}
\infrule[\ruledef{TEq-All}]
  {\Gamma \ts K_1 \keq K_2  \andalso
   \Gamma, X \kas K_1 \ts A_1 \teq A_2 \kin \kstar}
  {\Gamma \ts \all{X}{K_1}{A_1} \teq \all{X}{K_2}{A_2} \kin \kstar}

\infrule[\ruledef{TEq-Arr}]
  {\Gamma \ts A_1 \teq A_2 \kin \kstar  \andalso
   \Gamma \ts B_1 \teq B_2 \kin \kstar}
  {\Gamma \ts \fun{A_1}{B_1} \teq \fun{A_2}{B_2} \kin \kstar}

\infrule[\ruledef{TEq-Abs}]
  {\Gamma \ts \lam{X}{J_1}{A_1} \kin \dfun{X}{J}{K}\\
   \Gamma \ts \lam{X}{J_2}{A_2} \kin \dfun{X}{J}{K}\\
   \Gamma, X \kas J \ts A_1 \teq A_2 \kin K}
  {\Gamma \ts \lam{X}{J_1}{A_1} \teq \lam{X}{J_2}{A_2} \kin \dfun{X}{J}{K}}
\end{multicols}
\begin{multicols}{2}
\infrule[\ruledef{TEq-App}]
  {\Gamma \ts A_1 \teq A_2 \kin \dfun{X}{J}{K}  \andalso
   \Gamma \ts B_1 \teq B_2 \kin J}
  {\Gamma \ts \app{A_1}{B_1} \teq \app{A_2}{B_2} \kin \subst{K}{X}{B_1}}

\infrule[\ruledef{TEq-Sing}]
  {\Gamma \ts A_1 \teq A_2 \kin B \intv C}
  {\Gamma \ts A_1 \teq A_2 \kin A_1 \intv A_1}
\end{multicols}
\begin{multicols}{2}
\infrule[\ruledefN{TEq-Beta}{TEq-$\beta$}]
  {\Gamma, X \kas J \ts A \kin K  \andalso  \Gamma \ts B \kin J}
  {\Gamma \ts \app{(\lam{X}{J}{A})}{B} \teq \subst{A}{X}{B} \kin
    \subst{K}{X}{B}}

\infrule[\ruledefN{TEq-Eta}{TEq-$\eta$}]
  {\Gamma \ts A \kin \dfun{X}{J}{K}  \andalso  X \notin \fv(A)}
  {\Gamma \ts \lam{X}{J}{\app{A}{X}} \teq A \kin \dfun{X}{J}{K}}
\end{multicols}
\begin{proof}
  The admissibility proofs of the above rules all follow the same
  basic pattern.
We want to show that the left- and right-hand sides of the
  conclusions are mutual subkinds or subtypes, respectively.
To do so, we employ the respective subkinding and subtyping rules,
  adjusting the kinds of additional bindings and subtyping judgments
  using context narrowing (\Cor{decl_narrow}),
  subsumption~\ruleref{ST-Sub}, conversion~\ruleref{ST-Conv} and
  functionality~(\Lem{decl_funct}) where necessary.
When additional validity properties are required,
  \Lem{decl_validity} delivers the required well-formedness or
  well-kindedness proofs.

  For example, the proof of \ruleref{TEq-Sing} proceeds as follows.
We note that the premise must have been derived using
  \ruleref{ST-AntiSym}, hence we have
$\Gamma \ts A_1 \tsub A_2 \kin B \intv C$ and
  $\Gamma \ts A_2 \tsub A_1 \kin B \intv C$.
By~\ruleref{ST-Intv}, \Lem{ST-Intv-Star} and subtyping validity
  \begin{align*}
    &\Gamma \ts A_1 \tsub A_2 \kin A_1 \intv A_2 \quad \text{(1a)} &
    &\Gamma \ts A_1 \tsub A_2 \kin \kstar \quad \text{(1b)} &
    &\Gamma \ts A_1 \kin \kstar \quad \text{(1c)} \\
    &\Gamma \ts A_2 \tsub A_1 \kin A_2 \intv A_1 \quad \text{(2a)} &
    &\Gamma \ts A_2 \tsub A_1 \kin \kstar \quad \text{(2b)} &
    &\Gamma \ts A_2 \kin \kstar \quad \text{(2c)}
  \end{align*}
  From (1a), (1c) and (2b) we derive
\begin{prooftree}
    \AxiomC{$\Gamma \ts A_1 \tsub A_2 \kin A_1 \intv A_2$}
      \AxiomC{$\Gamma \ts A_1 \kin \kstar$}
      \LeftLabel{\rulerefP{ST-Refl}}
      \UnaryInfC{$\Gamma \ts A_1 \tsub A_1 \kin \kstar$}
        \AxiomC{$\Gamma \ts A_2 \tsub A_1 \kin \kstar$}
      \RightLabel{\rulerefP{SK-Intv}}
      \BinaryInfC{$\Gamma \ts A_1 \intv A_2 \ksub A_1 \intv A_1$}
    \RightLabel{\rulerefP{ST-Sub}}
    \BinaryInfC{$\Gamma \ts A_1 \tsub A_2 \kin A_1 \intv A_1$}
  \end{prooftree}
and similarly $\Gamma \ts A_2 \tsub A_1 \kin A_1 \intv A_1$ from
  (2a), (2b) and (1c).
We conclude with \ruleref{ST-AntiSym}.
\end{proof}

Note that the rule~\ruleref{TEq-Sing} (and by transitivity, the
rule~\ruleref{ST-Intv}) plays an important role in the proof of
subject reduction for types (\Thm{beta_red_teq}).
It is used in the case for \ruleref{K-Sing}, where it allows us to
relate $\beta$-equal types inhabiting singleton kinds.

\subsection{Admissible Rules for Higher-Order Extrema and Intervals}
\label{sec:decl_horules}

In this section, we state and prove admissible rules that justify the
encodings of higher-order extremal types and interval kinds given
in~\Sec{encodings} of the paper (page~\pageref{sec:encodings}).
Many of these rules are straightforward generalizations of the
corresponding rules for the types $\Top$, $\Bot$ and for proper type
intervals $A \intv B$.
The remaining rules and lemmas mostly deal with the family of kinds
$\kmax{K}$, which plays a crucial role in the other encodings and the
proofs of their respective properties.

We start by stating and proving a formation rule for
$\kmax{K}$.
\begin{lemma} The kind $\kmax{K}$ is well-formed whenever $K$ is,
  \ie the following is admissible.
\end{lemma}
\begin{center}
  \infruleSimp[Wf-KMax]{\kindD{\Gamma}{K}}{\kindD{\Gamma}{\kmax{K}}}
\end{center}
\begin{proof}
  By induction on the structure of $K$.
The base case uses~\ruleref{K-Bot}, \ruleref{K-Top}
  and~\ruleref{Wf-Intv} to show that $\kstar$ is well-formed.
\end{proof}

The kind $\kmax{K}$ is a widened version of $K$, \ie the latter is
always a subkind of the former.  As a consequence, any type of kind
$K$ is also of kind $\kmax{K}$.
\begin{lemma}\label{lem:SK-KMax}
  Any well-formed kind $K$ is a subkind of $\kmax{K}$.
\end{lemma}
\begin{center}
  \infruleSimp[SK-KMax]
  {\kindD{\Gamma}{K}}
  {\Gamma \ts K \ksub \kmax{K}}
\end{center}
\begin{proof}
  By straightforward induction on the structure of $K$.
\end{proof}
\begin{corollary}\label{cor:K-KMax}
  If $\; \Gamma \ts A \kin K$, then also $\Gamma \ts A \kin \kmax{K}$.
\end{corollary}

The following two lemmas introduce admissible kinding rules for the
higher-order extremal types, and prove that $\tmin{K}$ and $\tmax{K}$
are in fact extrema in~$\kmax{K}$, \ie they are the least and
greatest inhabitants of $\kmax{K}$, respectively.
\begin{lemma} Higher-order extremal types are well-formed if their
  index kind is.
\end{lemma}
\begin{center}
  \infruleSimp[K-TMax]{\kindD{\Gamma}{K}}{\Gamma \ts \tmax{K} \kin \kmax{K}}
  \hspace{3em}
  \infruleSimp[K-TMin]{\kindD{\Gamma}{K}}{\Gamma \ts \tmin{K} \kin \kmax{K}}
\end{center}
\begin{proof}
  Separately, by induction on the structure of $K$.  The cases for
  dependent arrow kinds use \ruleref{Wf-KMax}.
\end{proof}

\begin{lemma} The types $\tmax{K}$ and $\tmin{K}$ are the maximal and
  minimal elements of $\kmax{K}$, respectively.
\end{lemma}
\begin{center}
  \infruleSimp[ST-TMax]{\Gamma \ts A \kin K}
  {\Gamma \ts A \tsub \tmax{K} \kin \kmax{K}}
  \hspace{3em}
  \infruleSimp[ST-TMin]{\Gamma \ts A \kin K}
  {\Gamma \ts \tmin{K} \tsub A \kin \kmax{K}}
\end{center}
\begin{proof}
  Separately, by induction on the structure of $K$.  \Cor{K-KMax} is
  used to adjust the kind of the premises where necessary.  In the
  inductive step, we use~\ruleref{ST-Abs} and the
  $\eta$-rules~\rulerefN{ST-Eta1}{ST-$\eta_{1,2}$}.  For example, for
  $K = \dfun{X}{K_1}{K_2}$ we have
  \begin{alignat}{2}
    \Gamma \; \ts \; A \; &\tsub \; \lam{X}{K_1}{\app{A}{X}}
      &&\text{for $X \notin \fv(A)$}
      \tag{by \rulerefN{ST-Eta2}{ST-$\eta_2$}}\\
    &\tsub \; \lam{X}{K_1}{\tmax{K_2}} \tag{by the IH and~\ruleref{ST-Abs}}\\
    &\alEq \; \tmax{\dfun{X}{K_1}{K_2}} &\;\; \kin \; \kmax{\dfun{X}{K_1}{K_2}}.
      \tag{by definition}
  \end{alignat}
\end{proof}

Having generalized the properties of the extremal types to their
higher-order counterparts, we now turn to interval kinds.  We start
with an admissible formation rule for higher-order intervals.
\begin{lemma} Higher-order interval kinds are well-formed if their
  bounds are.
\end{lemma}
\begin{center}
  \infruleSimp[Wf-HoIntv]
  {\Gamma \ts A \kin K  \andalso  \Gamma \ts B \kin K}
  {\kindD{\Gamma}{A \hointv{K} B}}
\end{center}
\begin{proof}
  By induction on the structure of $K$.  The inductive step uses
  kinding validity, \ruleref{K-Var} and~\ruleref{K-App} to expand the
  bounds.
\end{proof}

The subkinding rule~\ruleref{SK-Intv} for proper type intervals also
generalizes straightforwardly to intervals over arbitrary type
operators.
\begin{lemma} Higher-order interval kinds are widened in accordance
  with their bounds.
\end{lemma}
\begin{center}
  \infruleSimp[SK-HoIntv]
  {\Gamma \ts A_2 \tsub A_1 \kin K  \andalso
   \Gamma \ts B_1 \tsub B_2 \kin K}
  {\Gamma \ts A_1 \hointv{K} B_1 \ksub A_2 \hointv{K} B_2}
\end{center}
\begin{proof}
  By induction on the structure of $K$.  The inductive step uses
  subtyping validity, \ruleref{K-Var}, \ruleref{TEq-Refl}
  and~\ruleref{ST-App} to expand the bounds, and~\ruleref{Wf-HoIntv}
  to establish well-formedness of the left-hand side.
\end{proof}

Next, we would like to prove an admissible higher-order singleton
introduction rule that generalizes~\ruleref{K-Sing}.  Ideally, we
would like to show that any well-kinded type $\Gamma \ts A \kin K$
inhabits its corresponding singleton kind
$\ksing{A}{K} = A \hointv{K} A$.
This is not necessarily true, however.
Consider the case of an operator variable
$X$ with declared type $\Gamma(X) = \fun{\kstar}{\kstar}$.  The
singleton kind corresponding to $X$ is
$\ksing{X}{\fun{\kstar}{\kstar}} = \dfun{Y}{\kstar}{\app{X}{Y} \!
  \intv \app{X}{Y}}$,
so we would like to prove that
$\Gamma \ts X \kin \dfun{Y}{\kstar}{\app{X}{Y} \! \intv \app{X}{Y}}$.
Which kinding rules could we use to adjust the kind of $X$ to the
desired singleton kind?  Since~\ruleref{K-Sing} can only be applied to
proper types, our only option is to use the subsumption
rule~\ruleref{K-Sub}.  But unfortunately, the declared kind
$\fun{\kstar}{\kstar}$ of $X$ is a \emph{strict supertype} of the
singleton kind $\dfun{Y}{\kstar}{\app{X}{Y} \! \intv \app{X}{Y}}$, so
this cannot work.

We can, however, assign the desired singleton kind to the
\emph{$\eta$-expansion} of $X$, \ie to
$\lam{Y}{\kstar}{\app{X}{Y}}$.  Unlike $X$, the application
$\app{X}{Y}$ in the body of the $\eta$-expansion is a proper type, so
we can use~\ruleref{K-Sing} to narrow its kind.  The full derivation
is
\begin{prooftree}
  \AxiomC{$\kindD{\Gamma}{\kstar}$}
    \AxiomC{$\Gamma, Y \kas \kstar \ts
      X \kin \fun{\kstar}{\kstar}$}
      \AxiomC{$\Gamma, Y \kas \kstar \ts
        Y \kin \kstar$}
    \RightLabel{\rulerefP{K-App}}
    \BinaryInfC{$\Gamma, Y \kas \kstar \ts
      \app{X}{Y} \kin \kstar$}
    \RightLabel{\rulerefP{K-Sing}}
    \UnaryInfC{$\Gamma, Y \kas \kstar \ts
      \app{X}{Y} \kin \app{X}{Y} \! \intv \app{X}{Y}$}
  \RightLabel{\rulerefP{K-Abs}}
  \BinaryInfC{$\Gamma \ts \lam{Y}{\kstar}{\app{X}{Y}} \kin
    \dfun{Y}{\kstar}{\app{X}{Y} \! \intv \app{X}{Y}}$}
\end{prooftree}
This principle generalizes to arbitrary well-kinded types: the
$\eta$-expansion of a well-kinded type $\Gamma \ts A \kin K$ always
inhabits the corresponding singleton kind $\ksing{A}{K}$.

Given a type $A$, we define the \emph{weak $\eta$-expansion}
$\weakEtaExp{K}{A}$ of $A$ as $\weakEtaExp{B \intv C}{A} = A$ and
$\weakEtaExp{\dfun{X}{J}{K}}{A} =
\lam{X}{J}{\weakEtaExp{K}{\app{A}{X}}}$
where, as usual, we assume that $X \notin \fv(A)$.  We call this
expansion ``weak'' because the argument $X$ in the definition
$\lam{X}{J}{\weakEtaExp{K}{\app{A}{X}}}$ of the arrow case is not
$\eta$-expanded further.  This means that the result is not
$\eta$-long.  This is sufficient for the purpose of this section; we
will define a stronger version in the next section.

As expected, a type of kind $K$ is equal to its weak $\eta$-expansion
in $K$.
\begin{lemma}
  Weak $\eta$-expansion is sound, \ie if $\; \Gamma \ts A \kin K$,
  then $\Gamma \ts A \teq \weakEtaExp{K}{A} \kin K$.
\end{lemma}
\begin{proof}
  By induction on the structure of $K$,
  using~\rulerefN{TEq-Eta}{TEq-$\eta$} and~\ruleref{TEq-Abs} in the
  inductive case.
\end{proof}
\vfill

\begin{lemma}
  The $\eta$-expansions of type operators inhabit their higher-order
  singleton intervals.
  \begin{center}\normalfont
    \infruleSimp[K-HoSing]
    {\Gamma \ts A \kin K}
    {\Gamma \ts \weakEtaExp{K}{A} \kin \ksing{A}{K}}
  \end{center}
\end{lemma}
\begin{proof}
  By induction on the structure of $K$.  The base case follows
  from~\ruleref{K-Sing}, the inductive step from the usual combination
  of kinding validity, \ruleref{K-Var}, \ruleref{K-App}
  and~\ruleref{Wf-HoIntv}.
\end{proof}
\begin{corollary}
  If $\Gamma \ts B_1 \tsub A \kin K$ and
  $\Gamma \ts A \tsub B_2 \kin K$, then
  $\Gamma \ts \weakEtaExp{K}{A} \kin B_1 \hointv{K} B_2$.
\end{corollary}

Having found ways to form, widen and populate higher-order intervals,
we still need a way to put their bounds to use.  To this end, we
introduce two \emph{higher-order bound projection} rules, which
generalize the corresponding rules~\rulerefN{ST-Bnd1}{ST-Bnd$_1$}
and~\rulerefN{ST-Bnd2}{ST-Bnd$_2$} for proper type intervals.
\begin{lemma}[higher-order bound projection]\label{lem:ho_bnd_proj}
  Inhabitants of a higher-order interval are supertypes of its lower
  bound and subtypes of its upper bound.
\end{lemma}
\begin{center}
  \infruleSimpN{ST-HoBnd1}{ST-HoBnd$_1$}
  {\Gamma \ts A \kin K  \andalso
   \Gamma \ts B_1 \kin K \\
   \Gamma \ts A \kin B_1 \hointv{K} B_2}
  {\Gamma \ts B_1 \tsub A \kin K}
  \hspace{3em}
  \infruleSimpN{ST-HoBnd2}{ST-HoBnd$_2$}
  {\Gamma \ts A \kin K  \andalso
   \Gamma \ts B_2 \kin K \\
   \Gamma \ts A \kin B_1 \hointv{K} B_2}
  {\Gamma \ts A \tsub B_2 \kin K}
\end{center}
These rules are a bit weaker than one might expect.  In particular,
the additional premises $\Gamma \ts A \kin K$, $\Gamma \ts B_1 \kin K$
and $\Gamma \ts B_2 \kin K$, might seem redundant.  They are necessary
because we cannot, in general, invert well-formedness judgments about
higher-order intervals.  That is, $\kindD{\Gamma}{B_1 \hointv{K} B_2}$
does \emph{not} imply $\Gamma \ts B_1 \kin K$ and
$\Gamma \ts B_2 \kin K$, nor does
$\Gamma \ts A \kin B_1 \hointv{K} B_2$ imply $\Gamma \ts A \kin K$.
To see this, consider the kind $K = \Bot \hointv{\kempty} \Top$, where
$\kempty = \Top \intv \Bot$ is the empty interval (note the absurd
bounds).  The kind $K$ is well-formed and inhabited by both $\Bot$ and
$\Top$, yet clearly $\Bot, \Top$ are not inhabitants of $\kempty$.
Note that the formation rule~\ruleref{Wf-HoIntv} for higher-order
intervals is not to blame: although $K$ is well-formed, we cannot
prove this fact using~\ruleref{Wf-HoIntv}.  There are simply more
well-formed higher-order intervals than can be derived
using~\ruleref{Wf-HoIntv}.
\begin{proof}[Proof of \Lem{ho_bnd_proj}]
  Separately, by induction on the structure of $K$.  In the base case,
  we use the interval projection
  rules~\rulerefN{ST-Bnd1}{ST-Bnd$_{1, 2}$} as well
  as~\ruleref{ST-Intv} and~\ruleref{ST-Sub} to adjust the kinds of the
  resulting inequations.  In the inductive step, we
  use~\ruleref{ST-Abs} and the
  $\eta$-rules~\rulerefN{ST-Eta1}{ST-$\eta_{1,2}$}.  For example, for
  the left-hand case and $K = \dfun{X}{K_1}{K_2}$ we have
  \begin{alignat}{2}
    \Gamma \; \ts \; B_1 \; &\tsub \; \lam{X}{K_1}{\app{B_1}{X}}
      \tag{by~\rulerefN{ST-Eta2}{ST-$\eta_2$}}\\
    &\tsub \; \lam{X}{K_1}{\app{A}{\,X}} \tag{by the IH and~\ruleref{ST-Abs}}\\
    &\tsub \; A &\;\; \kin \; \dfun{X}{K_1}{K_2}.
      \tag{by~\rulerefN{ST-Eta1}{ST-$\eta_1$}}
  \end{alignat}
\end{proof}

Thanks to the admissible kinding and subtyping rules for higher-order
intervals and extrema, we can now easily derive judgments for forming,
introducing or eliminating bounded universal quantifiers over
arbitrary type operators.

For example, well-formedness of the
higher-order universal quantifier $\all{X \tsub A}{K}{B}$ can be
derived as
\begin{prooftree}
  \AxiomC{$\Gamma \ts A \kin K$}
  \LeftLabel{(kinding validity)}
  \UnaryInfC{$\kindD{\Gamma}{K}$}
  \LeftLabel{\rulerefP{K-TMin}}
  \UnaryInfC{$\Gamma \ts \tmin{K} \kin \kmax{K}$}
    \AxiomC{$\Gamma \ts A \kin K$}
    \RightLabel{(\Cor{K-KMax})}
    \UnaryInfC{$\Gamma \ts A \kin \kmax{K}$}
  \LeftLabel{\rulerefP{Wf-HoIntv}}
  \BinaryInfC{$\kindD{\Gamma}{\tmin{K} \hointv{\kmax{K}} A}$}
  \LeftLabel{(\Lem{re-idx_kmax})}
  \UnaryInfC{$\kindD{\Gamma}{\tmin{K} \hointv{K} A}$}
    \AxiomC{$\typeD{\Gamma, X \kas \tmin{K} \hointv{K} A}{B}$}
  \LeftLabel{\rulerefP{K-All}}
  \BinaryInfC{$\typeD{\Gamma}{\all{X \tsub A}{K}{B}}$}
\end{prooftree}
The derivation uses the following lemma for simplifying interval
kinds; its proof is by structural induction on the index $K$.
\begin{lemma}\label{lem:re-idx_kmax}
  Let $A$, $B$ be types and $K$ a kind. Then
  $A \hointv{\kmax{K}} B \alEq A \hointv{K} B$
\end{lemma}

Similar derivations exist for the introduction and elimination rules.
\begin{corollary}[bounded quantification] The following rules for the
  formation, introduction and elimination of bounded universal
  quantifiers are admissible.
\end{corollary}
\begin{center}
  \infruleSimp[K-AllBnd]
  {\Gamma \ts A \kin K  \andalso
   \typeD{\Gamma, X \kas \tmin{K} \hointv{K} A}{B}}
  {\typeD{\Gamma}{\all{X \tsub A}{K}{B}}}

  \infruleSimp[T-TAbsBnd]
  {\Gamma \ts A \kin K  \andalso
   \Gamma, X \kas \tmin{K} \hointv{K} A \ts t \tin B}
  {\Gamma \ts \Lam{X \tsub A}{K}{t} \tin \all{X \tsub A}{K}{B}}

  \infruleSimp[T-TAppBnd]
  {\Gamma \ts t \tin \all{X \tsub A}{K}{B}  \andalso
   \Gamma \ts C \tsub A \kin K}
  {\Gamma \ts \app{t}{(\weakEtaExp{K}{C})} \tin
     \subst{B}{X}{\weakEtaExp{K}{C}}}
\end{center}
Similar rules for the formation, abstraction and elimination of
bounded operators are also admissible.

Note that we need to $\eta$-expand the type argument $C$ in the
elimination rule~\ruleref{T-TAppBnd}
before it can by applied to $t$.  This is because $C$ has kind $K$,
while the polymorphic expression $t$ expects an argument of kind
$\tmin{K} \hointv{K} C$.  As discussed earlier, $C$ is not guaranteed
to inhabit that kind but its $\eta$-expansion is --
via~\ruleref{K-HoSing} and a subsequent widening of its kind from
$C \hointv{K} C$ to $\tmin{K} \hointv{K} C$.

There is an alternative encoding of higher-order bounded
quantification (and bounded type operators) that
separates the declaration of type variables from that of the subtyping
constraints imposed by their bounds, at the cost of using an auxiliary
type variable with potentially inconsistent bounds.  Assume a
partition of the set of type variable names into two distinct sets of
\emph{operator names} denoted by $\op{X}, \op{Y}, \dotsc$ and
\emph{constraint names} denoted by $\cst{X}, \cst{Y}, \dotsc$
We may then encode an upper-bounded type variable binding
$X \tsub A \kin K$ as a pair of bindings $\op{X} \kin K$,
$\cst{X} \kin \op{X} \hointv{K} A$, separating the declaration of the
operator name $X$ from the subtyping constraint $X \tsub A$.  For
example, the encoding of bounded universal quantifiers according to
this scheme would be
$\all{X \tsub A}{K}{B} = \all{\op{X}}{K}{\all{\cst{X}}{\op{X}
    \hointv{K} A}{B}}$
where $\cst{X} \notin \fv(B)$.

The advantage of this encoding is a
cleaner separation between the uses of bounded variable bindings in
kinding and subtyping.  Whenever we want to refer to the original type
variable $X$ or its kind, we simply use $\op{X}$.  When we require a
proof of the fact that $X \tsub A$ we obtain one from $\cst{X}$
via~\rulerefN{ST-HoBnd1}{ST-HoBind$_1$},
\rulerefN{ST-HoBnd1}{ST-HoBind$_2$}, and~\ruleref{ST-Trans}.  The same
is true when we instantiate type parameters.  For example, a type
application $\app{t}{C}$, where $t$ has type $\all{X \tsub A}{K}{B}$
and $C \tsub A$, is now desugared to
$\app{\app{t}{C}}{(\weakEtaExp{K}{C})}$, \ie only the second type
argument, which corresponds to the constraint parameter $\cst{X}$,
needs to be $\eta$-expanded, while the argument $C$ for the parameter
$\op{X}$ can be left as is.  Since $\cst{X}$ does not occur freely in
the codomain $B$ of the desugared universal type, the overall type of
the desugared application is just $\subst{B}{\op{X}}{C}$.
A clear drawback of this encoding is the necessary duplication of
bindings and the corresponding introduction and elimination forms
(abstraction, application).  In addition, the kind
$\op{X} \hointv{K} A$ of the constraint $\cst{X}$ has inconsistent
bounds in general, which can be problematic.

There is an obvious alternative definition for the family of kinds
$\kmax{K}$, namely $\kmax{K} = (\tmin{K}) \hointvP{K} (\tmax{K})$.
The original definition, given in~\Fig{encodings} of the paper
(page~\pageref{fig:encodings}), has the advantage of being independent
of the definition of the higher-order extrema $\tmax{K}$ and
$\tmin{K}$.
This allowed us to prove properties such as~\ruleref{Wf-KMax}
and~\Lem{SK-KMax} admissible without appealing to any of the
properties of higher-order extrema, and thereby avoid some cyclic
dependencies in the proofs of the latter.  The alternative definition,
on the other hand, seems more intuitive.  To conclude the section, we
show that the two definitions are equal for well-formed kinds $K$.
\begin{lemma}\label{lem:kmax_hointv}
  The kind $\kmax{K}$ is equal to the higher-order interval bounded by
  $\tmin{K}$ and $\tmax{K}$.  If $\; \kindD{\Gamma}{K}$, then
  $\Gamma \ts \kmax{K} \keq (\tmin{K}) \hointvP{K} (\tmax{K})$.
\end{lemma}
\begin{proof}[Proof of \Lem{kmax_hointv}]
  By induction on the structure of $K$.  The base case is immediate.
  In the inductive step, we use~\ruleref{SK-HoIntv} and the
  $\beta$-rule~\rulerefN{TEq-Beta}{TEq-$\beta$}.

  Let $K = \dfun{X}{K_1}{K_2}$.  We want to show that
  \[
    \Gamma \; \ts \; \dfun{X}{K_1}{\kmax{K_2}} \; \keq \;
    \dfun{Y}{K_1}{
      (\app{\tmin{\dfun{X}{K_1}{K_2}}}{Y}) \hointvP{K_2}
      (\app{\tmax{\dfun{X}{K_1}{K_2}}}{Y})}
  \]
  for some $Y$ that does not occur freely in
  $\fv(\tmin{\dfun{X}{K_1}{K_2}})$ or
  $\fv(\tmax{\dfun{X}{K_1}{K_2}})$.  By the IH, we have
  $\Gamma \ts \kmax{K_2} \keq (\tmin{K_2}) \hointvP{K_2} (\tmax{K_2})$ but
  we need to adjust the bounds of the right-hand interval.  We use the
  following equation for the lower bound, and an similar one for the
  upper bound.
  \begin{align}
    \Gamma \; \ts \; \app{\tmin{\dfun{X}{K_1}{K_2}}}{Y} \; &\alEq \;
      \app{\lam{X}{K_1}{\tmin{K_2}}}{Y}
      \tag{by definition}\\
    &\teq \; \subst{\tmin{K_2}}{X}{Y}
      \tag{by \ruleref{K-TMin} and \rulerefN{TEq-Beta}{TEq-$\beta$}}\\
    &\alEq \; \tmin{K_2} \hspace{5em}\kin \; \kmax{K_2}. \tag{$\alpha$-renaming}
  \end{align}
  By~\ruleref{SK-Refl-KEq}, \ruleref{SK-HoIntv}
  and~\ruleref{SK-AntiSym} we obtain
  \[
    \Gamma \; \ts \; (\tmin{K_2}) \hointvP{\kmax{K_2}} (\tmax{K_2})
    \; \keq \;
    (\app{\tmin{\dfun{X}{K_1}{K_2}}}{Y}) \hointvP{\kmax{K_2}}
    (\app{\tmax{\dfun{X}{K_1}{K_2}}}{Y})
  \]
  which we re-index using~\Lem{re-idx_kmax}.  We conclude
  by~\ruleref{KEq-Refl} and~\ruleref{KEq-DArr}.
\end{proof}

\section{Simple Kinding of Normal Types}
\label{sec:simple}

This section introduces a system of \emph{simplified kinding}
judgments which provide a syntactic characterization of normal types
and allow us to establish important properties about hereditary
substitutions and normal forms, notably a pair of commutativity lemmas
(\Lem[Lemmas]{hsubst-comm} and~\ref{lem:weq-nf-hsubst}) that play an
important role in the development of \Sec{canonical} in the paper.

The contents of this section are based on Chapter~4 of the first
author's PhD~dissertation.  We refer the interested reader to the
dissertation for the full details~\cite[][Ch.~4]{Stucki17thesis}.

\subsection{Preliminaries}

Before we can introduce the simplified system, we need to introduce a
few auxiliary definitions and lemmas that were omitted from
\Sec{normalization} in the paper.

We start with the definition of \emph{shape contexts} -- the type of
context used in our simplified kinding judgments.
Shape contexts $\gamma$, $\delta$ are best thought of as typing
contexts consisting exclusively of type variable bindings~$X \kas k\;$
with shape annotations~$k$ (as opposed to full kind annotations~$K$).
Their grammar is defined as follows. \begin{align*}
  \gamma, \delta &::= \cempty \orElse \gamma, X \kas k
  \tag{Shape context}
\end{align*}
As for full contexts, we assume that the variables bound in a shape
context are all distinct.
We write $\dom(\gamma)$ for the set of variables bound in $\gamma$ and
$(\gamma, \delta)$ for the concatenation of two shape contexts with
disjoint domains.

\subsubsection{Weak Equality}
\label{sec:weak_eq}

Recall that \emph{weak equality} is an equivalence $A \wkEq B$ on
types and kinds that identifies operator abstractions up to the shape
of their domain annotations.
\begin{center}
 \infruleSimpBase{(\ruledef{WEq-Abs})}
 {\simpEq{J}{K}  \andalso  A \wkEq B}
 {\lam{X}{J}{A} \wkEq \lam{X}{K}{B}}
\end{center}
Formally, weak equality is defined as the smallest congruence (\wrt to
all the type and kind formers) that includes both syntactic equality
and the above rule~\ruleref{WEq-Abs}.
It is easy to verify that weak equality is a congruence \wrt
ordinary substitution.
In the following sections, we will show that weak equality is also a
congruence \wrt hereditary substitution, $\eta$-expansion and
normalization, and that weakly equal kinds $K_1 \wkEq K_2$ have equal
shapes $\simpEq{K_1}{K_2}$.

\subsubsection{Properties of Hereditary Substitution}

Since it is defined pointwise on spines, hereditary substitution
commutes with spine concatenation.
\begin{lemma}\label{lem:hsubst-concat}
$\hsubst{(\scons{\vec{D}_1}{\vec{D}_2})}{X}{k}{E} \; \alEq \;
  (\scons{\hsubst{\vec{D}_1}{X}{k}{E}}{\hsubst{\vec{D}_2}{X}{k}{E}})$.
\end{lemma}
Just as for ordinary substitution, shapes are stable under hereditary
substitution.
\begin{lemma}[stability of shapes under hereditary
  substitution]\label{lem:simp-stable-hsubst}
$\ksimp{\hsubst{J}{X}{k}{A}} \alEq \ksimp{J}$.
\end{lemma}
\begin{proof}
  By straightforward induction on the structure of $J$.
\end{proof}
As a consequence, weak equality is a congruence \wrt hereditary
substitution and reducing application.
\begin{lemma}\label{lem:weq-par-hsubst}
  Weak equality is a congruence w.r.t.\ hereditary substitution and
  reducing application.  Let $E_1 \wkEq E_2$,
  \begin{enumerate}
  \item\label{item:weq-hsubst-kind} if $K_1 \wkEq K_2$, then
    $\hsubst{K_1}{X}{k}{E_1} \wkEq \hsubst{K_2}{X}{k}{E_2}$;
  \item\label{item:weq-hsubst-elim} if $D_1 \wkEq D_2$, then
    $\hsubst{D_1}{X}{k}{E_1} \wkEq \hsubst{D_2}{X}{k}{E_2}$;
  \item\label{item:weq-hsubst-sp} if $\vec{D}_1 \wkEq \vec{D}_2$, then
    $\hsubst{\vec{D}_1}{X}{k}{E_1} \wkEq
    \hsubst{\vec{D}_2}{X}{k}{E_2}$;
  \item\label{item:weq-rapp-elim} if $D_1 \wkEq D_2$, then
    $\rapp{E_1}{k}{D_1} \wkEq \rapp{E_2}{k}{D_2}$;
  \item\label{item:weq-rapp-sp} if $\vec{D}_1 \wkEq \vec{D}_2$, then
    $\rapp{E_1}{k}{\vec{D}_1} \wkEq \rapp{E_2}{k}{\vec{D}_2}$.
  \end{enumerate}
\end{lemma}
\begin{proof}
  The structure of the proof mirrors that of the recursive definitions
  of hereditary substitution and reducing application.  All five parts
  are proven simultaneously, by induction on the structure of $k$.
  Parts~\Item{weq-hsubst-kind}--\Item{weq-hsubst-sp} proceed by an
  inner induction on the derivations of $K_1 \wkEq K_2$,
  $D_1 \wkEq D_2$ and $\vec{D}_1 \wkEq \vec{D}_2$, respectively.
  Parts~\Item{weq-rapp-elim} and~\Item{weq-rapp-sp} proceed by a case
  analysis on the final rules used to derive $E_1 \wkEq E_2$ and
  $\vec{D}_1 \wkEq \vec{D}_2$, respectively.  Part~\Item{weq-hsubst-elim} proceeds by a case analysis on the final
  rule used to derive $F_1 \wkEq F_2$, where $F_1$ and $F_2$ are the
  heads, respectively, of $E_1 = F_1 \vec{D}_1$ and
  $E_2 = F_2 \vec{D}_2$.  In the case for~\ruleref{WEq-Abs}, we use
  stability of kind simplification under hereditary substitution
  (\Lem{simp-stable-hsubst}).
In the variable case,
  we use the IH twice: first for part~\Item{weq-hsubst-sp} to derive
  $\hsubst{\vec{D}_1}{X}{k}{E_1} \wkEq \hsubst{\vec{D}_2}{X}{k}{E_2}$,
  then for part~\Item{weq-rapp-sp}, to derive
  $\rapp{E_1}{k}{(\hsubst{\vec{D}_1}{X}{k}{E_1})} \wkEq
  \rapp{E_2}{k}{(\hsubst{\vec{D}_2}{X}{k}{E_2})}$.  In the second
  instance, $k$ does not decrease nor is
  $\hsubst{\vec{D}_1}{X}{k}{E_1} \wkEq \hsubst{\vec{D}_2}{X}{k}{E_2}$
  a sub-derivation of the current premise.  This use of the IH
is nevertheless justified because any subsequent use of the IH for
  part~\Item{weq-hsubst-elim} in the proof of part~\Item{weq-rapp-sp}
  must occur after the use of the IH for part~\Item{weq-rapp-elim}, at
  which point $k$ has necessarily decreased.
\end{proof}

\subsubsection{Properties of $\eta$-Expansion and Normalization}

Unsurprisingly, shapes are stable under normalization.
\begin{lemma}[stability of shapes under normalization]
  \label{lem:simp-stable-nf}
$\ksimp{\nf{\Gamma}{K}} \alEq \ksimp{K}$.
\end{lemma}
\begin{proof}
  By straightforward induction on the structure of $K$.
\end{proof}

Extending kind simplification pointwise to contexts, we define
$\ksimp{\Gamma}$ as
\begin{align*}
  \ksimp{\cempty} &= \cempty &
  \ksimp{\Gamma, x \tas A} &= \ksimp{\Gamma} &
  \ksimp{\Gamma, X \kas K} &= \ksimp{\Gamma}, X \kas \ksimp{K}
\end{align*}
It is easy to see that context lookup commutes with simplification,
\ie $\ksimp{\Gamma}(X) = \ksimp{\Gamma(X)}$, and that simplified
contexts are stable under (hereditary) substitution and
normalization, \ie
\begin{align*}
  \ksimp{\subst{\Gamma}{X}{A}} &= \ksimp{\Gamma} &
  \ksimp{\hsubst{\Gamma}{X}{k}{A}} &= \ksimp{\Gamma} &
  \ksimp{\nfCtx{\Gamma}} &= \ksimp{\Gamma}
\end{align*}

The following two lemmas show that $\eta$-expansion and normalization
preserve weak equality.  Importantly, this is true even if the
corresponding kinds and contexts, respectively, are not themselves
weakly equal but have equal shapes -- a much weaker requirement.

\begin{lemma}\label{lem:weq-etaexp}
  Weak equality is preserved by $\eta$-expansion along kinds of equal
  shape.  If $\simpEq{J}{K}$ and $D \wkEq E$, then
  $\etaExp{J}{D} \wkEq \etaExp{K}{E}$.
\end{lemma}
\begin{proof}
  By induction on the structure of $J$ and case analysis on the final
  rule used to derive $D \wkEq E$.
\end{proof}

\begin{lemma}\label{lem:weq-nf}
  Kinds and types normalize weakly equally in contexts that simplify
  equally.  Let $\Gamma$ and $\Delta$ be contexts such that
  $\simpEq{\Gamma}{\Delta}$.  Then
  \begin{enumerate}
  \item $\nf{\Gamma}{K} \wkEq \nf{\Delta}{K}$ for any kind $K$, and
  \item $\nf{\Gamma}{A} \wkEq \nf{\Delta}{A}$ for any type $A$.
  \end{enumerate}
\end{lemma}
\begin{proof}
  Simultaneously, by induction on the structure of $K$ and $A$,
  respectively.  In the type variable case $A = X$ we use
  \Lem{weq-etaexp}; in the operator application case
  $A = \app{A_1}{A_2}$ we use \Lem{weq-par-hsubst}.\Item{weq-hsubst-elim};
  in the cases for dependent operator kinds, universal types and
  operator abstraction, we use \Lem{weq-simpeq}.
\end{proof}

\subsection{The Simplified System}

The function $\nfRaw$ assigns to each raw type $A$ in a given context
$\Gamma$ a unique type $\nf{\Gamma}{A}$.  But as we have seen, the
type $\nf{\Gamma}{A}$ may not be $\beta\eta$-normal if $A$ is
ill-kinded.  In this section, we prove the converse: whenever $A$ is
well-kinded in $\Gamma$, the type $\nf{\nfCtx{\Gamma}}{A}$ is a
$\eta$-long $\beta$-normal form.  To do so, we first introduce a set
of \emph{simplified kinding} judgments.  Roughly, a simplified kinding
judgment $\gamma \ts E \kin k$ establishes that the type $E$ is a
normal form of \emph{shape} $k$ in the \emph{shape context} $\gamma$.
Given $\gamma \ts E \kin k$, we say that $E$ is a \emph{well-shaped}
or \emph{simply (well-)kinded normal form}, or just that $E$ is
\emph{simply kinded}.  As we are about to show, every well-kinded type
$\Gamma \ts A \kin K$ has a well-shaped normal form
$\ksimp{\Gamma} \ts E \kin \ksimp{K}$, namely
$E = \nf{\nfCtx{\Gamma}}{A}$ (see \Lem{nf-simp} below).

It is important to note that the converse is not true: not every
simply kinded type is well-kinded.  Because kind simplification
forgets dependencies, there are necessarily some ill-kinded types that
are considered simply well-kinded according to the judgments we are
about to introduce.  However, every simply kinded type is guaranteed
to be an $\eta$-long $\beta$-normal form and, as we will see in this
section, simple kinding is preserved by operations such as hereditary
substitution and $\eta$-expansion.  Hence simple kinding allows us to
prove important properties about these operations on
$\beta\eta$-normal forms without subjecting ourselves to the
complexity of fully dependent kinds.

To enhance readability, we use the following naming conventions for
normal forms: the metavariables $U$, $V$, $W$ denote normal types,
while $M$ and $N$ denote neutral types.\footnote{This is just
  notation.  We do not consider normal forms a separate syntactic
  category, \eg the letters $U$, $V$, $W$ are metavariables
  denoting types (typically in elimination form) rather than
  non-terminals in some grammar of normal forms.}  No special notation
is used for normal kinds.

\paragraph{Judgments.}  \Fig{simple_rules} defines the following
judgments by mutual induction.
\[\begin{array}{l@{\gap}l}
  \kindS{\gamma}{K} &
  \text{the kind $K$ is simply well-formed and normal in $\gamma$}\\
  \gamma \ts V \kin k &
  \text{the type $V$ is a normal form of shape $k$ in $\gamma$}\\
  \gamma \tsNe N \kin k &
  \text{the type $N$ is a neutral form of shape $k$ in $\gamma$}\\
  \spS{\gamma}{j}{\vec{V}}{k} &
  \text{applying an operator of shape $j$ to the normal spine
                                $\vec{V}$}\\
  &\text{yields a type of shape $k$ in $\gamma$.}
\end{array}\]

\fig[tb]{fig:simple_rules}{Simplified kinding}{
\judgment{Simplified well-formedness of kinds}{\fbox{$\kindS{\gamma}{K}$}}
\begin{center}
  \begin{minipage}[t]{0.44\linewidth}
\infruleLeft{SWf-Intv}
    {\typeS{\gamma}{U} \andalso \typeS{\gamma}{V}}
    {\kindS{\gamma}{U \intv V}}
  \end{minipage}\hspace{1ex}
  \begin{minipage}[t]{0.54\linewidth}
\infruleLeft{SWf-DArr}
    {\kindS{\gamma}{J} \andalso \kindS{\gamma, X \kas \ksimp{J}}{K}}
    {\kindS{\gamma}{\dfun{X}{J}{K}}}
  \end{minipage}
\end{center}

\judgment{Kinding of neutral types}{\fbox{$\gamma \tsNe N \kin k$}}

\infruleLeft{SK-VarApp}
{\gamma(X) = j  \andalso  \gamma \ts j \kin \vec{V} \kin k}
{\gamma \tsNe \app{X}{\vec{V}} \kin k}

\judgment{Simple spine kinding}{\fbox{$\spS{\gamma}{j}{\vec{V}}{k}$}}
\begin{center}
  \begin{minipage}[t]{0.42\linewidth}
\infruleLeft{SK-Empty}{\quad}{\spS{\gamma}{k}{\sempty}{k}}
  \end{minipage}\hspace{1ex}
  \begin{minipage}[t]{0.56\linewidth}
\infruleLeft{SK-Cons}
    {\gamma \ts U \kin j \andalso \spS{\gamma}{k}{\vec{V}}{l}}
    {\spS{\gamma}{\fun{j}{k}}{\scons{U}{\vec V}}{l}}
  \end{minipage}
\end{center}\medskip

\judgment{Simple kinding of normal types}{\fbox{$\gamma \ts V \kin k$}}
\begin{multicols}{2}
  \typicallabel{SK-Sing}
\infrule[\ruledef{SK-Top}]
  {\,}
  {\typeS{\gamma}{\Top}}

\infrule[\ruledef{SK-Arr}]
  {\typeS{\gamma}{U}  \andalso  \typeS{\gamma}{V}}
  {\typeS{\gamma}{\fun{U}{V}}}

\infrule[\ruledef{SK-Abs}]
  {\kindS{\gamma}{J}  \andalso  \gamma, X \kas \ksimp{J} \ts V \kin k}
  {\gamma \ts \lam{X}{J}{V} \kin \fun{\ksimp{J}}{k}}

\infrule[\ruledef{SK-Bot}]
  {\,}
  {\typeS{\gamma}{\Bot}}

\infrule[\ruledef{SK-All}]
  {\kindS{\gamma}{K}  \andalso  \typeS{\gamma, X \kas \ksimp{K}}{V}}
  {\typeS{\gamma}{\all{X}{K}{V}}}

\infrule[\ruledef{SK-Ne}]
  {\gamma \tsNe N \kin \kstar}
  {\typeS{\gamma}{N}}
\end{multicols}}
 
The judgments for simple kind formation and kinding follow the
syntactic structure of normal kinds and types.  A type $V$ is a
$\beta\eta$-normal form $\gamma \ts V \kin k$ of shape $k$ if it
is either a proper type introduced by one of the basic type formers
applied to normal arguments (rules~\ruleref{SK-Top}, \ruleref{SK-Bot},
\ruleref{SK-Arr}, and~\ruleref{SK-All}), an operator abstraction with
a normal body (rule~\ruleref{SK-Abs}), or a simply kinded
neutral type (rule~\ruleref{SK-Ne}).  Simply kinded neutral forms
$\gamma \tsNe N \kin k$ are eliminations headed by an abstract type
operator, \ie a type variable $X$, which is applied to a spine of
normal types $\vec{V}$ (rule~\ruleref{SK-VarApp}).  Finally, a simply
well-formed normal kind $\kindS{\gamma}{K}$ is either a type interval
bounded by normal types (rule~\ruleref{SWf-Intv}) or a dependent arrow
with normal domain and codomain (rule~\ruleref{SWf-DArr}).  Note that
type operator abstractions are the only normal forms of arrow shape.
This ensures that normal types are always $\eta$-long.

The simple spine kinding judgment $\spS{\gamma}{j}{\vec{V}}{k}$ is
different from the other judgment forms in that it is a quaternary
rather than a ternary relation.  The shapes $j$ and $k$ should be read
as inputs and outputs, respectively, of such judgments: when a type of
shape $j$ is applied to the spine $\vec{V}$ (the subject of the judgment), the resulting type
is of shape $k$ -- as exemplified by the rule~\ruleref{SK-VarApp}.

There is no formation judgment for shape contexts $\gamma$ since such
contexts only contain bindings assigning shapes to type variables, and there is no such
thing as an ill-formed simple kind.

Because kinding is simplified, there is no notion of subkinding or
kind equality, and hence no need for a subsumption rule.  As a
consequence, the simple kind formation and kinding rules are
syntax-directed.
Another important property of simplified kinding is that none of the
rules involve substitutions in kinds.
This substantially simplifies the proofs of key lemmas about
hereditary substitutions, in particular that of \Lem{simp-hsubst}
which states that hereditary substitutions preserve simple kinding
(and thus normal forms).
It is also important in establishing admissibility of the following
simple rules about spines and neutral types.

\begin{lemma}
  The following simple kinding rules for spine concatenation and
  application of neutrals are admissible.
\end{lemma}
\begin{center}
  \infruleSimp[SK-Concat]
  {\spS{\gamma}{j}{\vec{U}}{k}  \andalso  \spS{\gamma}{k}{\vec{V}}{l}}
  {\spS{\gamma}{j}{\scons{\vec U}{\vec V}}{l}}
  \hspace{2em}
  \infruleSimp[SK-Snoc]
  {\spS{\gamma}{j}{\vec{U}}{\fun{k}{l}}  \andalso  \gamma \ts V \kin k}
  {\spS{\gamma}{j}{\scons{\vec U}{V}}{l}}
  \hspace{2em}
  \infruleSimp[SK-NeApp]
  {\gamma \tsNe N \kin \fun{j}{k}  \andalso  \gamma \ts V \kin j}
  {\gamma \tsNe \app{N}{V} \kin k}
\end{center}
\begin{proof}
  The proofs are done separately for each of the three rules in the
  order the rules are listed.  The proof for~\ruleref{SK-Concat} is by
  induction on the derivation of the first premise. The
  rule~\ruleref{SK-Cons} is derivable from~\ruleref{SK-Snoc} as a
  special case where $\vec{V} = \scons{V}{\sempty}$,
  using~\ruleref{SK-Empty} and~\ruleref{SK-Cons}.  The proof
  of~\ruleref{SK-NeApp} starts with a case analysis on the final rule
  used to derive $\Gamma \tsNe N \kin \fun{j}{k}$.  The only rule for
  deriving such judgments is~\ruleref{SK-VarApp}, hence $N$ must be of
  the form $N = \app{X}{\vec{U}}$ with $\gamma(X) = l$ and
  $\gamma \ts l \kin \vec{U} \kin \fun{j}{k}$.  We conclude
  by~\ruleref{SK-Snoc} and~\ruleref{SK-VarApp}.
\end{proof}

\subsubsection{Simply-Kinded Hereditary Substitution}
\label{sec:simple_hsubst}

Before we can prove that hereditary substitutions preserve simple
kinding, we first need to establish the usual weakening properties for
simple kind formation and kinding.

\begin{lemma}[weakening]\label{lem:simp-weaken}
  A simple judgment remains true if its context is extended by an
  additional binding.  Let $\gamma$, $\delta$ be shape contexts, $k$ a
  shape and $X \notin \dom(\gamma, \delta)$.  If
  $\; \judgS{\gamma, \delta}$ for any of the simple judgments defined
  above, then $\judgS{\gamma, X \kas k, \delta}$.
\end{lemma}
\begin{proof}
  Simultaneously for all four judgments, by induction on the
  derivation of $\; \judgS{\gamma, \delta}$.
\end{proof}
\begin{corollary}[Iterated weakening]\label{cor:simp-iter-weaken}
  Given a pair $\gamma$, $\delta$ of disjoint shape contexts, if
  $\; \judgS{\gamma}$, then $\judgS{\gamma, \delta}$.
\end{corollary}

\begin{lemma}[hereditary substitution]\label{lem:simp-hsubst}
  Hereditary substitutions and reducing applications preserve the
  shapes of types and simple well-formedness of kinds.  Let $\gamma$,
  $\delta$ be shape contexts and $X$ such that
  $X \notin \dom(\gamma, \delta)$.  Assume further that
  $\gamma \ts V \kin k$ for some $V$ and $k$.  Then
  \begin{enumerate}
  \item\label{item:simp-hsubst-kds} if
    $\; \kindS{\gamma, X \kas k, \delta}{J}$, then
    $\kindS{\gamma, \delta}{\hsubst{J}{X}{k}{V}}$;
  \item\label{item:simp-hsubst-nf} if
    $\; \gamma, X \kas k, \delta \ts U \kin j$, then
    $\gamma, \delta \ts \hsubst{U}{X}{k}{V} \kin j$;
  \item\label{item:simp-hsubst-ne} if
    $\; \gamma, X \kas k, \delta \tsNe N \kin j$, then
    $\gamma, \delta \ts \hsubst{N}{X}{k}{V} \kin j$ as a normal form;
  \item\label{item:simp-hsubst-sp} if
    $\; \spS{\gamma, X \kas k, \delta}{j}{\vec{U}}{l}$, then
    $\spS{\gamma, \delta}{j}{\hsubst{\vec{U}}{X}{k}{V}}{l}$;
  \item\label{item:simp-rapp-nf} if $k = \fun{k_1}{k_2}$ and
    $\; \gamma \ts U \kin k_1$, then
    $\gamma \ts \rapp{V}{\fun{k_1}{k_2}}{U} \kin k_2$;
  \item\label{item:simp-rapp-sp} if
    $\; \spS{\gamma}{k}{\vec{U}}{j}$, then
    $\gamma \ts \rapp{V}{k}{\vec{U}} \kin j$.
  \end{enumerate}
\end{lemma}
Note that hereditary substitutions preserve the shapes of neutral
types but not neutrality itself.
\begin{proof}
  All six parts are proven simultaneously by induction on the
  structure of $k$.
  Parts~\Item{simp-hsubst-kds}--\Item{simp-hsubst-sp} proceed by an
  inner induction on the simple formation or kinding derivations for
  $J$, $U$, $N$ and $\vec{U}$, respectively.
  Parts~\Item{simp-rapp-nf} and~\Item{simp-rapp-sp} proceed by a case
  analysis on the final rules used to derive
  $\gamma \ts V \kin \fun{k_1}{k_2}$ and
  $\spS{\gamma}{k}{\vec{V}}{j}$, respectively; for
  part~\Item{simp-rapp-nf}, the only applicable rule
  is~\ruleref{SK-Abs}.  For part~\Item{simp-hsubst-ne}, in the case
  for~\ruleref{SK-VarApp} when $N = \app{X}{\vec{U}}$, we use iterated
  weakening (\Cor{simp-iter-weaken}) and the IH
  (for~\Item{simp-hsubst-sp}), respectively, to obtain
  $\gamma, \delta \ts V \kin k$ and
  $\spS{\gamma, \delta}{k}{\hsubst{\vec{U}}{X}{k}{V}}{j}$.  To
  conclude the case, we apply the IH again (for~\Item{simp-rapp-sp}).
  In this second use of the IH, $k$ does not decrease nor is
  $\spS{\gamma, \delta}{k}{\hsubst{\vec{U}}{X}{k}{V}}{j}$ a strict
  sub-derivation of the current premise.  However, in order to use the
  IH for part~\Item{simp-hsubst-ne} again from within the proof of
  part~\Item{simp-rapp-sp}, we must go through
  part~\Item{simp-rapp-nf}, at which point $k$ necessarily decreases.
  Again, the structure of the proof mirrors that of the mutually
  recursive definitions of hereditary substitution and reducing
  application.
\end{proof}

Thanks to \Lem{simp-hsubst}, we can now prove the following
commutativity lemma about hereditary substitutions, which will play an
important role in the proof of \Lem{weq-nf-hsubst} below and in
\Sec{canonical} of the paper.

\begin{lemma}[commutativity of hereditary
  substitutions]\label{lem:hsubst-comm}
  Hereditary substitutions of simply kinded types commute; hereditary
  substitutions of simply kinded types commute with simply kinded
  reducing applications.  Let $\gamma_1 \ts U \kin j$ and
  $\gamma_1, X \kas j, \gamma_2 \ts V \kin k$.  Then
  \begin{enumerate}[nosep]
  \item\label{item:hsubst-comm-kds} if
    $\; \kindS{\gamma_1, X \kas j, \gamma_2, Y \kas k, \gamma_3}{J}$, then
    \[
      \hsubst{\hsubst{J}{Y}{k}{V}}{X}{j}{U} \; \alEq \;
      \hsubst{\hsubst{J}{X}{j}{U}}{Y}{k}{\hsubst{V}{X}{j}{U}};
    \]
  \item\label{item:hsubst-comm-nf} if
    $\; \gamma_1, X \kas j, \gamma_2, Y \kas k, \gamma_3 \ts W \kin l$, then
    \[
      \hsubst{\hsubst{W}{Y}{k}{V}}{X}{j}{U} \; \alEq \;
      \hsubst{\hsubst{W}{X}{j}{U}}{Y}{k}{\hsubst{V}{X}{j}{U}};
    \]
  \item\label{item:hsubst-comm-ne} if
    $\; \gamma_1, X \kas j, \gamma_2, Y \kas k, \gamma_3 \tsNe N \kin l$, then
    \[
      \hsubst{\hsubst{N}{Y}{k}{V}}{X}{j}{U} \; \alEq \;
      \hsubst{\hsubst{N}{X}{j}{U}}{Y}{k}{\hsubst{V}{X}{j}{U}};
    \]
  \item\label{item:hsubst-comm-sp} if
    $\; \spS{\gamma_1, X \kas j, \gamma_2, Y \kas k,
      \gamma_3}{l_1}{\vec{W}}{l_2}$,
    then
    \[
      \hsubst{\hsubst{\vec{W}}{Y}{k}{V}}{X}{j}{U} \; \alEq \;
      \hsubst{\hsubst{\vec{W}}{X}{j}{U}}{Y}{k}{\hsubst{V}{X}{j}{U}};
    \]
  \item\label{item:rapp-comm-nf} if $k = \fun{k_1}{k_2}$ and
    $\; \gamma_1, X \kas j, \gamma_2 \ts W \kin k_1$, then
    \[
      \hsubst{(\rapp{V}{\fun{k_1}{k_2}}{W})}{X}{j}{U} \; \alEq \;
      \rapp{(\hsubst{V}{X}{j}{U})}{\fun{k_1}{k_2}}{(\hsubst{W}{X}{j}{U})};
    \]
  \item\label{item:rapp-comm-sp} if
    $\; \spS{\gamma_1, X \kas j, \gamma_2}{k}{\vec{W}}{l}$, then
    $\hsubst{(\rapp{V}{k}{\vec{W}})}{X}{j}{U} \; \alEq \;
    \rapp{(\hsubst{V}{X}{j}{U})}{k}{(\hsubst{\vec{W}}{X}{j}{U})}$.
  \end{enumerate}
\end{lemma}
\begin{proof}All six parts are proven simultaneously by simultaneous induction on
  the structures of $j$ and $k$.  Simultaneous structural induction on
  $j$ and $k$ means roughly that it is sufficient for either one of
  $j$ or $k$ to decrease in an induction step.  More formally, denote
  by $\sqsubseteq$ the sub-expression order on shapes, then the
  simultaneous induction order $<$ on unordered pairs $\set{ j, k }$
  of shapes is defined as $\set{ j_1, j_2 } < \set{ k_1, k_2 }$ if
  $j_1 \sqsubset k_1$ and $j_2 \sqsubseteq k_2$.
Importantly, $<$ is defined over \emph{unordered} pairs which allows
  us to exchange $j$ and $k$ in an induction step.
Parts~\Item{hsubst-comm-kds}--\Item{hsubst-comm-sp}
  proceed by an inner induction on the simple formation or kinding
  derivations for $J$, $W$, $N$ and $\vec{W}$, respectively.

  As usual, the interesting cases are those for
  part~\Item{hsubst-comm-ne}, when $N = \app{Y}{\vec{W}}$ and
  $N = \app{X}{\vec{W}}$.
  \begin{itemize}
  \item \emph{Case \ruleref{SK-VarApp}, $N = \app{Y}{\vec{W}}$.}  We
    have
    $\spS{\gamma_1, X \kas j, \gamma_2, Y \kas k,
      \gamma_3}{k}{\vec{W}}{l}$.
    By \Lem{simp-hsubst}.\Item{simp-hsubst-sp}, we obtain
    $\spS{\gamma_1, X \kas j, \gamma_2,
      \gamma_3}{k}{\hsubst{\vec{W}}{Y}{k}{V}}{l}$, and hence we have
    \begin{align*}\hspace{-2em}
      \hsubst{\hsubst{(\app{Y}{\vec{W}})&}{Y}{k}{V}}{X}{j}{U}\\
      &\alEq \;
        \hsubst{(\rapp{V}{k}{(\hsubst{\vec{W}}{Y}{k}{V})})}{X}{j}{U}
        \tag{by definition}\\
      &\alEq \;
        \rapp{(\hsubst{V}{X}{j}{U})}{k}{(
          \hsubst{\hsubst{\vec{W}}{Y}{k}{V}}{X}{j}{U})}
        \tag{by the IH for~\Item{rapp-comm-sp}}\\
      &\alEq \;
        \rapp{(\hsubst{V}{X}{j}{U})}{k}{(
          \hsubst{\hsubst{\vec{W}}{X}{j}{U}}{Y}{k}{\hsubst{V}{X}{j}{U}})}
        \tag{by the IH for~\Item{hsubst-comm-sp}}\\
      &\alEq \;
        \hsubst{\hsubst{(\app{Y}{\vec{W}})}{X}{j}{U}}{Y}{k}{\hsubst{V}{X}{j}{U}}.
        \tag{by definition}
    \end{align*}
  \item \emph{Case \ruleref{SK-VarApp}, $N = \app{X}{\vec{W}}$.}  We
    have
    $\spS{\gamma_1, X \kas j, \gamma_2, Y \kas k,
      \gamma_3}{j}{\vec{W}}{l}$.
    By \Lem{simp-hsubst}.\Item{simp-hsubst-sp}, we obtain
    $\spS{\gamma_1, X \kas j, \gamma_2,
      \gamma_3}{j}{\hsubst{\vec{W}}{Y}{k}{V}}{l}$, and hence we have
    \begin{align*}\hspace{-2em}
      \hsubst{\hsubst{(\app{X}{\vec{W}})&}{Y}{k}{V}}{X}{j}{U}\\
      &\alEq \;
        \hsubst{(\app{X}{(\hsubst{\vec{W}}{Y}{k}{V})})}{X}{j}{U}
        \tag{by definition}\\
      &\alEq \;
        \rapp{U}{j}{(
          \hsubst{\hsubst{\vec{W}}{Y}{k}{V}}{X}{j}{U})}
        \tag{by definition}\\
      &\alEq \;
        \rapp{U}{j}{(
          \hsubst{\hsubst{\vec{W}}{X}{j}{U}}{Y}{k}{\hsubst{V}{X}{j}{U}})}
        \tag{by the IH for~\Item{hsubst-comm-sp}}\\
      &\alEq \;
        \rapp{(\hsubst{U}{Y}{k}{\hsubst{V}{X}{j}{U}})}{j}{(
          \hsubst{\hsubst{\vec{W}}{X}{j}{U}}{Y}{k}{\hsubst{V}{X}{j}{U}})}
        \tag{as $Y \notin \fv(U)$}\\
      &\alEq \;
        \hsubst{(\rapp{U}{j}{(
          \hsubst{\vec{W}}{X}{j}{U})})}{Y}{k}{\hsubst{V}{X}{j}{U}}
        \tag{by the IH for~\Item{rapp-comm-sp}}\\
      &\alEq \;
        \hsubst{\hsubst{(\app{X}{\vec{W}})}{X}{j}{U}}{Y}{k}{\hsubst{V}{X}{j}{U}}.
        \tag{by definition}
    \end{align*}
  \end{itemize}
  Note that, in the second case, we switched the roles of the shape
  $j$ and $k$ when invoking the IH for part~\Item{rapp-comm-sp}.
\end{proof}

\subsubsection{Simplification and Normalization of Kinding}
\label{sec:simple_nf}

Thanks to \Lem{nf-sound} we know that the definition of the
normalization function $\nfRaw$ is sound, \ie that well-formed kinds
$\kindD{\Gamma}{K}$ and well-kinded types $\Gamma \ts A \kin K$ are
convertible with $\nf{\nfCtx{\Gamma}}{K}$ and
$\nf{\nfCtx{\Gamma}}{A}$, respectively.  But we have yet to establish
that $\nf{\nfCtx{\Gamma}}{K}$ and $\nf{\nfCtx{\Gamma}}{A}$ are
actually normal forms.  In this section, we prove a more general
result, namely that, whenever $\kindD{\Gamma}{K}$ and
$\Gamma \ts A \kin K$,
it follows that $\nf{\nfCtx{\Gamma}}{K}$ is a simply well-formed
normal kind and and $\nf{\nfCtx{\Gamma}}{A}$ is a simply well-kinded
normal type.

As a first step, we show that the shapes of variables and, more
generally, of neutral types are preserved by $\eta$-expansion.
\begin{lemma}\label{lem:sk-etaexp}
  $\eta$-expansion preserves the shapes of neutral types.
  Assume $\kindS{\gamma}{K}$ and $\gamma \tsNe N \kin \ksimp{K}$.
  Then $\gamma \ts \etaExp{K}{N} \kin \ksimp{K}$.
\end{lemma}
\begin{proof}
  By induction on the structure of $K$.  The case fore
  $K = \dfun{X}{K_1}{K_2}$ proceeds by case analysis on the final
  rules used to derive $\kindS{\gamma}{\dfun{X}{K_1}{K_2}}$ and
  $\gamma \tsNe N \kin \fun{\ksimp{K_1}}{\ksimp{K_2}}$ and uses the
  weakening lemma (\Lem{simp-weaken}) as well as~\ruleref{SK-NeApp}.
\end{proof}

Next, we require a syntactic notion of normal contexts.  We define the
\emph{simple context formation} judgment $\ctxS{\Gamma}$ as the
pointwise lifting of simple kind formation and kinding to bindings:
\begin{center}
  \infruleSimp {\quad}{\ctxS{\cempty}}
  \hspace{3em}
  \infruleSimp {\ctxS{\Gamma}  \andalso  \kindS{\ksimp{\Gamma}}{K}}
  {\ctxS{\Gamma, X \kas K}}
  \hspace{3em}
  \infruleSimp {\ctxS{\Gamma}  \andalso  \ksimp{\Gamma} \ts V \kin \kstar}
  {\ctxS{\Gamma, x \kas V}}
\end{center}

Since simple kind formation and kinding is defined on normal kinds and
types, a simply well-formed context $\Gamma$ is also normal.
Conversely, if we lookup the declared kind or type of a variable in a
simply well-formed context, the result is guaranteed to be a normal
form.

\begin{lemma}\label{lem:sc-lookup}
  The declared kinds and types of variables in a simply well-formed
  context $\Gamma$ are simply well-formed and well-kinded,
  respectively, in $\ksimp{\Gamma}$, i.e
\end{lemma}
\begin{center}
  \infruleSimp[SC-TpLookup]
  {\ctxS{\Gamma, X \kas K, \Delta}}
  {\kindS{\ksimp{\Gamma, X \kas K, \Delta}}{K}}
  \hspace{3em}
  \infruleSimp[SC-TmLookup]
  {\ctxS{\Gamma, x \tas V, \Delta}}
  {\typeS{\ksimp{\Gamma, x \tas V, \Delta}}{V}}
\end{center}
\begin{proof}
  Both parts are proven separately by structural induction on $\Delta$
  and case analysis on the final rule used to derive the premise.  In
  the inductive case, we use the weakening lemma for simple kind
  formation.
\end{proof}

With \Lem[Lemmas]{sk-etaexp} and~\ref{lem:sc-lookup} at hand, it is
easy to show that $\nfRaw$ does indeed produce normal forms.
\begin{lemma}[normalization and simplification]\label{lem:nf-simp}
  Well-formed kinds and well-kinded types have simply well-formed and
  simply kinded normal forms, respectively.
  \begin{enumerate}
  \item\label{item:nf-simp-kd} If $\; \kindD{\Gamma}{K}$, then
    $\kindS{\ksimp{\nfCtx{\Gamma}}}{\nf{\nfCtx{\Gamma}}{K}}$.
  \item\label{item:nf-simp-tp} If $\; \Gamma \ts A \kin K$, then
    $\ksimp{\nfCtx{\Gamma}} \ts \nf{\nfCtx{\Gamma}}{A} \kin \ksimp{K}$.
  \item\label{item:nf-simp-ctx} If $\; \ctxD{\Gamma}$, then
    $\ctxS{\nfCtx{\Gamma}}$.
  \end{enumerate}
\end{lemma}
The proof uses the following helper lemma about the shapes of
subkinds, which is proven by straightforward induction on subkinding
derivations.
\begin{lemma}\label{lem:simp-decl-sk}
  Subkinds have equal shapes.  If $\; \Gamma \ts J \ksub K$, then
  $\ksimp{J} \alEq \ksimp{K}$.
\end{lemma}

\begin{proof}[Proof of \Lem{nf-simp}]
  Simultaneously by induction on declarative kind formation, kinding,
  and context formation derivations.  The only interesting cases
  are~\ruleref{K-Sub} (where we use \Lem{simp-decl-sk}),
  \ruleref{K-Var} and \ruleref{K-App}.  In the case for
  \ruleref{K-Var}, where $A = X$, we use the IH for
  part~\Item{nf-simp-ctx} and \ruleref{SC-TpLookup} to obtain
  $\kindS{\ksimp{\nfCtx{\Gamma}}}{K}$ for $K = \nfCtx{\Gamma}(X)$, and
  we conclude by \Lem{sk-etaexp}.  In the case for~\ruleref{K-App},
  where $A = \app{A_1}{A_2}$, we start by applying the IH to obtain
$\ksimp{\nfCtx{\Gamma}} \ts \nf{\nfCtx{\Gamma}}{A_1} \kin
  \fun{\ksimp{K_1}}{\ksimp{K_2}}$, and
$\ksimp{\nfCtx{\Gamma}} \ts \nf{\nfCtx{\Gamma}}{A_2} \kin
  \ksimp{K_1}$.
The first of these judgments must be derived using~\ruleref{SK-Abs}
  because that is the only simple kinding rule assigning an arrow kind
  to a type.  Hence $\nf{\nfCtx{\Gamma}}{A_1} = \lam{X}{J}{V}$ for
  some $J$ and $V$ such that $\simpEq{J}{K_1}$ and
  $\ksimp{\nfCtx{\Gamma}}, X \kas \ksimp{K_1} \ts V \kin \ksimp{K_2}$.
  We conclude by the hereditary substitution lemma for normal types
  (\Lem{simp-hsubst}.\Item{simp-hsubst-nf}) and
  \Lem{simp-stable-hsubst}.
\end{proof}

\subsubsection{Commutativity of Normalization and Substitution}
\label{sec:comm-nf-subst}

Our final task in this section is to establish another commutativity
property that will play a crucial role in proving equality of
declarative and canonical subtyping: the fact that normalization
commutes with substitution.

In the past few sections, we have seen that well-formed kinds and
well-kinded types have normal forms (\Lem{nf-simp}) and that these
normal forms are convertible to the kinds and types they were computed
from (\Lem{nf-sound}).  By validity, context conversion and kind
conversion, this means that every declarative subtyping judgment
$\Gamma \ts A \tsub B \kin K$ has an associated judgment
$\nfCtx{\Gamma} \ts \nf{}{A} \tsub \nf{}{B} \kin \nf{}{K}$ relating
the normal forms of the original expressions.

In \Sec{canonical} of the paper, we describe a system of canonical
rules for deriving such judgments which are defined directly on normal
forms -- similar to the simple kinding and kind formation judgments
introduced in this section.
The proof of equivalence of the two systems requires one canonical
rule
--~possibly a derivable or admissible one~--
for every declarative rule.
But some of the declarative rules, such as the subtyping
rules~\rulerefN{ST-Beta1}{ST-$\beta_{1,2}$} for $\beta$-conversions,
or the kinding rule~\ruleref{K-App} for applications, involve
substitutions, which do not preserve normal forms.
To see why this is
a problem, consider the declarative rule~\ruleref{K-App}:
\begin{center}
  \infruleSimp
  {\Gamma \ts A \kin \dfun{X}{J}{K}  \andalso  \Gamma \ts B \kin J}
  {\Gamma \ts \app{A}{B} \kin \subst{K}{X}{B}}
\end{center}
By soundness of normalization (\Lem{nf-sound}), type equation validity
(\Lem{decl_validity}), and context conversion (\Cor{decl_ctx_conv}),
we know that the following is also admissible:
\begin{center}
  \infruleSimp
  {\nfCtx{\Gamma} \ts U \kin \dfun{X}{J'}{K'}  \andalso
   \nfCtx{\Gamma} \ts V \kin J'}
  {\nfCtx{\Gamma} \ts \nf{}{\app{A}{B}} \kin \nf{}{\subst{K}{X}{B}}}
\end{center}
where $U = \nf{}{A}$, $V = \nf{}{B}$,
$J' = \nf{}{J}$ and $K' = \nf{}{K}$.  By \Lem{nf-simp}, we know that
$U$ and $V$ are simply well-kinded normal types, and that $J'$ and
$K'$ are simply well-formed normal kinds.  For our canonical
application rule, we would like to express the type
$\nf{}{\app{A}{B}}$ and the kind $\nf{}{\subst{K}{X}{B}}$ in the
conclusion directly using $U$, $V$, $J'$ and $K'$.  This is relatively
straightforward for the application $\nf{}{\app{A}{B}}$ because we
know that $U$ must be an operator abstraction $U = \lam{X}{L}{W}$;
after all, $U$ has shape $\fun{\ksimp{J}}{\ksimp{K}}$ and normal
forms are $\eta$-long.  We also know that
$\ksimp{L} \alEq \ksimp{J} \alEq \ksimp{J'}$ (see the proof of
\Lem{nf-simp} for details).  Hence
$\nf{}{\app{A}{B}} \alEq \hsubst{W}{X}{\ksimp{J'}}{V}$ by definition
of $\nfRaw$, and we are done.

Things are more complicated for the normal kind
$\nf{}{\subst{K}{X}{B}}$.  The definition of the normalization
function $\nfRaw$ does not tell us anything immediately useful about
substitutions.  Indeed, we know that substitutions do not preserve
normal forms, \eg $\subst{(\app{Y}{V})}{Y}{\lam{X}{J'}{W}}$ is not a
normal form, even if $\app{Y}{V}$ and $\lam{X}{J'}{W}$ are.  However,
\Cor{hsubst-sound}.\Item{hsubst-snd-kind} tells us that substitutions
in kinds are judgmentally equal to hereditary substitutions, \ie
$\Gamma \ts \subst{K}{X}{B} \keq \hsubst{K}{X}{\ksimp{J'}}{B}$, and
\Lem{simp-hsubst}.\Item{simp-hsubst-kds} tells us that hereditary
substitutions preserve normal forms, all of which suggests that
$\nf{}{\subst{K}{X}{B}}$ should be equal to
$\hsubst{K'}{X}{\ksimp{J'}}{V}$.  This is indeed the case; one can
show that
$\Gamma \ts \nf{}{\subst{K}{X}{B}} \keq
\hsubst{K'}{X}{\ksimp{J'}}{V}$.
But there is a caveat: the two normal forms are not
\emph{syntactically} equal, \ie
$\nf{}{\subst{K}{X}{B}} \, \not\alEq \, \hsubst{K'}{X}{\ksimp{J'}}{V}$.
Similarly,
$\nf{}{\subst{A}{X}{B}} \, \not\alEq \, \hsubst{\nf{}{A}}{X}{k}{\nf{}{B}}$
for types $A$ and $B$ in general.

This fact is best illustrated through the case of type variables,
\ie when $A = X$ and we have
$\nf{\Gamma}{\subst{X}{X}{B}} \alEq \nf{\Gamma}{B}$ and
$\hsubst{(\nf{\Gamma}{X})}{X}{k}{\nf{}{B}} \alEq
\hsubst{(\etaExp{\Gamma(X)}{X})}{X}{k}{\nf{}{B}}$.
We would like to show that $\hsubst{(\etaExp{K}{X})}{X}{k}{V}$ is
syntactically equal to $V$ at least when all the involved types and
kinds are well-kinded and well-formed, \ie when
$\Gamma \ts X \kin K$, $\; \Gamma \ts V \kin K$ and $k = \ksimp{K}$.
But this is not the case.  The culprit is a mismatch of kind
annotations in operator abstractions, as illustrated by the following
counterexample.

Let $J_1 = \Top \intv \Top$ and $J_2 = \kstar$ so that
$\Gamma \ts J_1 \ksub J_2$ for any context $\Gamma$.  Let $U = \Top$,
$V = \lam{X}{J_2}{U}$ and $K = \dfun{X}{J_1}{\kstar}$ so that
$\Gamma \ts V \kin \dfun{X}{J_2}{\kstar}$,
$\Gamma \ts \dfun{X}{J_2}{\kstar} \ksub K$ and hence
$\Gamma \ts V \kin K$. Then
\begin{alignat*}{2}
  \etaExp{K}{Y} &&& \; \alEq \; \lam{Z}{J_1}{\app{Y}{Z}}\\
  \hsubst{\etaExp{K}{Y}&}{Y}{\ksimp{K}}{V} && \; \alEq \;
    \hsubst{(\lam{Z}{J_1}{\app{Y}{Z}})}{Y}{\ksimp{K}}{V} \\
  &&& \; \alEq \; \lam{Z}{J_1}{\hsubst{(\app{Y}{Z})}{Y}{\ksimp{K}}{V}}
    \tag{because $Y \notin \fv(J_1)$}\\
  &&& \; \alEq \; \lam{Z}{J_1}{\rapp{V}{\ksimp{K}}{Z}} \\
  &&& \; \alEq \;
    \lam{Z}{J_1}{\rapp{\lam{X}{J_2}{U}}{\fun{\ksimp{J_1}}{\kstar}}{Z}} \\
  &&& \; \alEq \; \lam{Z}{J_1}{\hsubst{U}{X}{\ksimp{J_1}}{Z}} \\
&&& \; \alEq \; \lam{X}{J_1}{U}  \tag{as $X, Z \notin \fv(U)$} \\
  &&& \; \not\alEq \; \lam{X}{J_2}{U} \; \alEq \; V.
    \tag{because $J_1 \not\alEq J_2$}
\end{alignat*}
So we are forced to conclude that
$\hsubst{\etaExp{K}{Y}}{Y}{\ksimp{K}}{V} \not\alEq V$ in general.  The
problem, as illustrated by this example, is that the domain annotation
$J_2$ of the type operator abstraction $V = \lam{X}{J_2}{U}$ is not
necessarily preserved by the hereditary substitution.  It is replaced
by the domain $J_1$ of the declared kind $K = \dfun{X}{J_1}{U}$ of
$Y$, which need not be \emph{syntactically} equal to $J_1$.

However, we do have $\Gamma \ts J_1 \ksub J_2$ and thus
$\simpEq{J_1}{J_2}$.  The solution, therefore, is to be more lenient
when comparing domain annotations in operator abstractions: the
\emph{weak} equation $\hsubst{\etaExp{K}{Y}}{Y}{\ksimp{K}}{V} \wkEq V$
does hold.  In fact, it holds for any simply well-formed kind $K$ and
simply well-kinded type $U$, as the following lemma shows.

\begin{lemma}\label{lem:etaexp-hsubst}
  ~
  \begin{enumerate}[nosep]
  \item\label{item:etaexp-ne-hsubst} Let
    $\; \kindS{\gamma, X \kas j, \delta}{K}$, $\; \gamma \ts U \kin j$
    and
    $\; \gamma, X \kas j, \delta \tsNe \app{X}{\vec{V}} \kin
    \ksimp{K}$.
    Then
    \[
      \hsubst{(\etaExp{K}{\app{X}{\vec{V}}})}{X}{j}{U} \; \wkEq \;
      \hsubst{(\app{X}{\vec{V}})}{X}{j}{U}.
\]
  \end{enumerate}
  Let $\kindS{\gamma}{K}$, then
  \begin{enumerate}[resume,nosep]
  \item\label{item:hsubst-kds-etaexp} if
    $\; \kindS{\gamma, X \kas \ksimp{K}, \delta}{J}$, then
    $\hsubst{J}{X}{\ksimp{K}}{\etaExp{K}{X}} \; \wkEq \; J$;
  \item\label{item:hsubst-nf-etaexp} if
    $\; \gamma, X \kas \ksimp{K}, \delta \ts V \kin j$, then
    $\hsubst{V}{X}{\ksimp{K}}{\etaExp{K}{X}} \; \wkEq \; V$;
  \item\label{item:hsubst-sp-etaexp} if
    $\; \spS{\gamma, X \kas \ksimp{K}, \delta}{j}{\vec{V}}{l}$, then
    $\hsubst{\vec{V}}{X}{\ksimp{K}}{\etaExp{K}{X}} \; \wkEq \; \vec{V}$;
  \item\label{item:etaexp-rapp-nf} if $\; K = \dfun{X}{K_1}{K_2}$,
    $\; \gamma \tsNe N \kin \ksimp{K}$ and
    $\; \gamma \ts V \kin \ksimp{K_1}$, then
    $\; \rapp{\etaExp{K}{N}}{\ksimp{K}}{V} \; \wkEq \;
    \etaExp{\hsubst{K_2}{X}{\ksimp{K_1}}{V}}{\app{N}{V}}$;
  \item\label{item:etaexp-rapp-sp} if $\; \gamma \tsNe N \kin \ksimp{K}$
    and $\spS{\gamma}{\ksimp{K}}{\vec{V}}{\kstar}$, then
    $\rapp{\etaExp{K}{N}}{\ksimp{K}}{\vec{V}} \; \wkEq \; \app{N}{\vec{V}}$.
  \end{enumerate}
\end{lemma}
\begin{corollary}\label{cor:etaexp-var-hsubst}
  If $\; \kindS{\gamma, X \kas j, \delta}{K}$ and
  $\; \gamma \ts U \kin j$, then
  $\hsubst{(\etaExp{K}{X})}{X}{j}{U} \; \wkEq \; U$.
\end{corollary}

It is in this lemma that we see the true usefulness of weak equality.
While syntactic equality is too strict for this particular
commutativity property, using judgmental type and kind equality would
have forced us to formulate its premises in terms of declarative
kinding.  This would have resulted in a weaker lemma with a more
complicated proof.
In \Sec{weq-canon} we show that weak equations can be converted
into judgmental ones provided the related types or kinds are
well-kinded or well-formed, respectively.  Hence, weak equality
affords us a relatively straightforward proof of this lemma (and the
next) with a minimal overhead in complexity.

In the proof of \Lem{etaexp-hsubst}, we employ the following helper
lemmas.
They are proven separately by three easy inductions:
the first on the derivation of $K_1 \wkEq K_2$, the second on the
structure of the kind $K$, and the third on the derivation of
$\; \spS{\gamma}{j}{\vec{V}}{k}$.
\begin{lemma}\label{lem:weq-simpeq}
  Weakly equal kinds have equal shapes.  If $K_1 \wkEq K_2$, then
  $\simpEq{K_1}{K_2}$.
\end{lemma}
\begin{lemma}\label{lem:etaexp-ne-hsubst-miss}
Let $X \neq Y$, then
  $
    \hsubst{(\etaExp{K}{\app{X}{\vec{D}}})}{Y}{j}{E} \; \alEq \;
    \etaExp{\hsubst{K}{Y}{j}{E}}{\app{X}{(\hsubst{\vec{D}}{Y}{j}{E})}}
  $
  for any $K$, $\vec{D}$, $j$ and $E$.
\end{lemma}
\begin{lemma}\label{lem:rapp-sp-concat}
  Let $\; \gamma \ts U \kin j$, $\spS{\gamma}{j}{\vec{V}}{k}$ and
  $\spS{\gamma}{k}{\vec{W}}{l}$, then
  $
    \; \rapp{U}{j}{(\scons{\vec{V}}{\vec{W}})} \; \alEq \;
    \rapp{(\rapp{U}{j}{\vec{V}})}{k}{\vec{W}}.
  $
\end{lemma}

\begin{proof}[Proof of \Lem{etaexp-hsubst}]
  All six parts are proven simultaneously by induction on the
  structure of~$K$.
  Parts~\Item{hsubst-kds-etaexp}--\Item{hsubst-sp-etaexp} proceed by
  an inner induction on the simple formation and kinding derivations
  for $J$, $V$ and $\vec{V}$, respectively.  We show a few key cases,
  the remainder of the proof is routine.
  \begin{itemize}
  \item \emph{Part~\Item{etaexp-ne-hsubst}, $K = \dfun{Y}{K_1}{K_2}$.}
    By inspection of the formation and kinding rules, we must have
    $\kindS{\gamma, X \kas j, \delta}{K_1}$,
    $\kindS{\gamma, X \kas j, \delta, Y \kas \ksimp{K_1}}{K_2}$ and
    $\spS{\gamma, X \kas j, \delta}{j}{\vec{V}}{\ksimp{K}}$.
    By \Lem{simp-hsubst} we have
    \[ \gamma, \delta \ts \hsubst{(\app{X}{\vec{V}})}{X}{j}{U} \kin
      \fun{\ksimp{K_1}}{\ksimp{K_2}} \qquad \text{and} \qquad
      \kindS{\gamma, \delta}{\hsubst{K_1}{X}{j}{U}}.
    \]
    The final kinding rule used to derive the first of these judgments
    must be~\ruleref{SK-Abs} since that is the only rule assigning
    arrow shapes to normal types.  Therefore, the following must
    hold for some $J$ and $W$:
    \begin{align}
      \hsubst{(\app{X}{\vec{V}})}{X}{j}{U} &= \lam{Y}{J}{W}
        \label{eq:kind-gen1}\\
      \ksimp{K_1} &= \ksimp{J} \label{eq:kind-gen2} \\
      \gamma, \delta, Y \kas \ksimp{J} &\ts W \kin \ksimp{K_2}.
        \label{eq:kind-gen3}
    \end{align}
    By weakening (\Lem{simp-weaken}), \Lem{sk-etaexp}
    and~\ruleref{SK-NeApp} we also have
    \[
      \gamma, X \kas j, \delta, Y \kas \ksimp{K_1} \tsNe
      \appp{X}{\vec{V}}{(\etaExp{K_1}{Y})} \kin \ksimp{K_2}
    \] and hence
    \begin{align*}
      \hsubst{(\etaExp{K_2}{\appp{X}{\vec{V}}{(\etaExp{K_1}{Y})}})}{X}{j}{U}
      \hspace{-10em}& \\
      &\wkEq \;
        \hsubst{(\appp{X}{\vec{V}}{(\etaExp{K_1}{Y})})}{X}{j}{U}
        \tag{by IH for~\Item{etaexp-ne-hsubst}}\\
      &\alEq \;
        \rapp{U}{j}{(\hsubst{(\scons{\vec{V}}{(\etaExp{K_1}{Y})})}{X}{j}{U})}
        \tag{by definition}\\
      &\alEq \;
        \rapp{(\rapp{U}{j}{(\hsubst{\vec{V}}{X}{j}{U})})}{
          \ksimp{K}}{(\hsubst{(\etaExp{K_1}{Y})}{X}{j}{U})}
        \tag{by \Lem[Lemmas]{hsubst-concat} and \ref{lem:rapp-sp-concat}}\\
      &\alEq \;
        \rapp{(\rapp{U}{j}{(\hsubst{\vec{V}}{X}{j}{U})})}{
          \ksimp{K}}{(\etaExp{\hsubst{K_1}{X}{j}{U}}{Y})}
        \tag{by \Lem{etaexp-ne-hsubst-miss}}\\
      &\alEq \;
        \rapp{(\hsubst{(\app{X}{\vec{V}})}{X}{j}{U})}{
          \fun{\ksimp{K_1}}{\ksimp{K_2}}}{(\etaExp{\hsubst{K_1}{X}{j}{U}}{Y})}
        \tag{by definition}\\
      &\alEq \;
        \rapp{(\lam{Y}{J}{W})}{
          \fun{\ksimp{K_1}}{\ksimp{K_2}}}{(\etaExp{\hsubst{K_1}{X}{j}{U}}{Y})}
        \tag{by \eqref{eq:kind-gen1}}\\
      &\alEq \;
        \hsubst{W}{Y}{\ksimp{K_1}}{(\etaExp{\hsubst{K_1}{X}{j}{U}}{Y})}
        \tag{by definition}\\
      &\wkEq \;
        W
        \tag{by~\Lem{weq-simpeq}, \eqref{eq:kind-gen3} and
          IH for~\Item{hsubst-nf-etaexp}}
    \end{align*}
    We conclude that
    \begin{align*}
      \hsubst{(\etaExp{K}{\app{X}{\vec{V}}})}{X}{j}{U} \hspace{-6em}& \\
      &\alEq \;
        \lam{Y}{\hsubst{K_1}{X}{j}{U}}{
           \hsubst{(\etaExp{K_2}{\appp{X}{\vec{V}}{
             (\etaExp{K_1}{Y})}})}{X}{j}{U}}
        \tag{by definition}\\
      &\wkEq \; \lam{Y}{J}{W}
        \tag{by~\Lem{weq-simpeq}, \eqref{eq:kind-gen2},
          the above and~\ruleref{WEq-Abs}}\\
      &\alEq \; \hsubst{(\app{X}{\vec{V}})}{X}{j}{U}.
        \tag{by~\eqref{eq:kind-gen1}}
    \end{align*}

  \item \emph{Part~\Item{hsubst-nf-etaexp}, case \ruleref{SK-Ne}.}
    The rule~\ruleref{SK-Ne} has only one premise which must have been
    derived using~\ruleref{SK-VarApp}, so $V = N = \app{Y}{\vec{V}}$
    and we have
    $\spS{\gamma, X \kas \ksimp{K}, \delta}{j}{\vec{V}}{\kstar}$ with
    $(\gamma, X \kas \ksimp{K}, \delta)(Y) = j$.  We distinguish two
    cases: $Y = X$ and $Y \neq X$ but consider only the first case
    here; the second case is simpler.  Since $Y = X$, we have
    $j = \ksimp{K}$, and
\begin{align*}
      \hsubst{V&}{X}{\ksimp{K}}{\etaExp{K}{X}}\\
      &\alEq \;
        \rapp{\etaExp{K}{X}}{\ksimp{K}}{
          (\hsubst{\vec{V}}{X}{\ksimp{K}}{\etaExp{K}{X}})}
        \tag{by definition}\\
      &\wkEq \;
        \rapp{\etaExp{K}{X}}{\ksimp{K}}{\vec{V}}
        \tag{by~reflexivity of $\wkEq$, IH for~\Item{hsubst-sp-etaexp}
          and \Lem{weq-par-hsubst}.\Item{weq-rapp-sp}}\\
      &\wkEq \;
        \app{X}{\vec{V}}.
        \tag{by \ruleref{SK-VarApp} and the IH for~\Item{etaexp-rapp-sp}}
    \end{align*}
  \item \emph{Part~\Item{etaexp-rapp-nf}, $K = \dfun{X}{K_1}{K_2}$.}
    By inspection of the formation and kinding rules, we must have
    $\kindS{\gamma}{K_1}$ and $N = \app{Y}{\vec{U}}$ with
    $\spS{\gamma}{\gamma(Y)}{\vec{U}}{\ksimp{K}}$.
    \begin{align*}
      \rapp{\etaExp{K}{N}}{\ksimp{K}}{V} &\alEq \;
        \rapp{\lam{X}{K_1}{\etaExp{K_2}{\appp{Y}{\vec{U}}{(\etaExp{K_1}{X})}}}}{
          \ksimp{K}}{V}
        \tag{by definition}\\
      &\alEq \;
        \hsubst{\etaExp{K_2}{\appp{Y}{\vec{U}}{(\etaExp{K_1}{X})}}}{X}{
          \ksimp{K_1}}{V}
        \tag{by definition}\\
      &\alEq \;
        \etaExp{\hsubst{K_2}{X}{\ksimp{K_1}}{V}}{
          \app{Y}{(\hsubst{(\scons{\vec{U}}{(\etaExp{K_1}{X})})}{X}{
          \ksimp{K_1}}{V})}}
        \tag{by \Lem{etaexp-ne-hsubst-miss}}\\
      &\alEq \;
        \etaExp{\hsubst{K_2}{X}{\ksimp{K_1}}{V}}{
          \appp{Y}{\vec{U}}{(\hsubst{(\etaExp{K_1}{X})}{X}{
          \ksimp{K_1}}{V})}}
        \tag{as $X \notin \fv(\vec{U})$}\\
      &\wkEq \;
        \etaExp{\hsubst{K_2}{X}{\ksimp{K_1}}{V}}{
          \appp{Y}{\vec{U}}{V}}.
        \tag{by IH for \Item{etaexp-ne-hsubst} and \Lem{weq-etaexp}}
    \end{align*}
  \end{itemize}
  The use of the IH in the last step corresponds to
  \Cor{etaexp-var-hsubst}.
\end{proof}

With
\Lem{etaexp-hsubst} in place, we are ready to prove that normalization
weakly commutes with substitution.  In the following,
$\Gamma \wkEq \Delta$ denotes the pointwise lifting of weak equality
to contexts.
\begin{lemma}\label{lem:weq-nf-hsubst}
  Substitution weakly commutes with normalization of well-formed kinds
  and well-kinded types.  Let $\; \Gamma \ts A \kin J$ and
  $V = \nf{\nfCtx{\Gamma}}{A}$, then
  \begin{enumerate}
  \item\label{item:weq-nf-hsubst-kd} if
    $\; \kindD{\Gamma, X \kas J, \Delta}{K}$, then
    $\nf{\nfCtx{\Gamma,\subst{\Delta}{X}{A}}}{\subst{K}{X}{A}} \;
    \wkEq \;
    \hsubst{(\nf{\nfCtx{\Gamma, X \kas J, \Delta}}{K})}{X}{\ksimp{J}}{V}$;
  \item\label{item:weq-nf-hsubst-tp} if
    $\; \Gamma, X \kas J, \Delta \ts B \kin K$, then
    $\nf{\nfCtx{\Gamma,\subst{\Delta}{X}{A}}}{\subst{B}{X}{A}} \;
    \wkEq \;
    \hsubst{(\nf{\nfCtx{\Gamma, X \kas J, \Delta}}{B})}{X}{\ksimp{J}}{V}$;
  \end{enumerate}
\end{lemma}
\begin{proof}
  Simultaneously by induction on declarative kind formation and
  kinding derivations.
In the case for~\ruleref{K-Var} where $B = Y$, we use
  \Cor{etaexp-var-hsubst} if $Y = X$.
Otherwise, let $\Sigma = \Gamma, X \kas J, \Delta$ and
  $\Sigma' = \Gamma, \subst{\Delta}{X}{A}$.
By stability of simplified contexts under normalization and
  hereditary substitution, and by commutation of simplification and
  context lookup, we have
\begin{align*}
    \ksimp{\nfCtx{\Sigma'}(Y)}
    &\; \alEq \;
    \ksimp{\nfCtx{\Sigma'}}(Y)
    \; \alEq \;
    \ksimp{\Sigma'}(Y) \\
    &\; \alEq \;
    (\ksimp{\Gamma}, \ksimp{\subst{\Delta}{X}{A}})(Y)
    \; \alEq \;
    (\ksimp{\Gamma}, \ksimp{\Delta})(Y)
    \; \alEq \;
    (\ksimp{\Gamma}, X \kas \ksimp{J}, \ksimp{\Delta})(Y)
    \; \alEq \;
    \ksimp{\Sigma}(Y) \\
    &\; \alEq \;
    \ksimp{\nfCtx{\Sigma}}(Y)
    \; \alEq \;
    \ksimp{\nfCtx{\Sigma}(Y)}
    \; \alEq \;
    \ksimp{\nf{\Sigma}{\Sigma(Y)}}
    \; \alEq \;
    \ksimp{\hsubst{(\nf{\Sigma}{\Sigma(Y)})}{X}{\ksimp{J}}{V}}
  \end{align*}
Hence, by \Lem{weq-etaexp} and \Lem{etaexp-ne-hsubst-miss} we have
  \[
    \etaExp{\nfCtx{\Sigma'}(Y)}{Y}
    \; \wkEq \;
    \etaExp{\hsubst{(\nf{\Sigma}{\Sigma(Y)})}{X}{\ksimp{J}}{V}}{Y}
    \; \alEq \;
    \hsubst{(\etaExp{\nf{\Sigma}{\Sigma(Y)}}{Y})}{X}{\ksimp{J}}{V}.
  \]
  In the case for~\ruleref{K-App}, we use
  \Lem{weq-par-hsubst}.\Item{weq-hsubst-elim} and
  \Lem{hsubst-comm}.\Item{hsubst-comm-nf}.
\end{proof}

The very last lemma of this section will be used in our equivalence
proof in \Sec{canon-complete-full} to show that subtyping rules for
$\eta$-conversion of normal operators are admissible in canonical
kinding.
\begin{lemma}\label{lem:weq-nf-etaexp}
  If $\; \Gamma \ts A \kin \dfun{X}{J}{K}$ with $X \notin \fv(A)$,
  then
  $\nf{\nfCtx{\Gamma}}{\lam{X}{J}{\app{A}{X}}} \wkEq
  \nf{\nfCtx{\Gamma}}{A}$.
\end{lemma}
\begin{proof}
  The proof uses \Lem{nf-simp} to obtain
  $\; \ksimp{\nfCtx{\Gamma}} \ts \nf{\nfCtx{\Gamma}}{A} \kin
  \fun{\ksimp{J}}{\ksimp{K}}$
  and proceeds by case analysis on the final rule used to derive this
  simple kinding judgment; the only applicable rule
  is~\ruleref{SK-Abs}.  The remainder of the proof uses equational
  reasoning very similar to that used in the proof of
  \Lem{etaexp-hsubst}.\Item{etaexp-ne-hsubst}.
\end{proof}

\section{Properties of the Canonical System}
\label{sec:canon_basics}

This section states and proves some basic yet important metatheoretic
properties of the canonical system that have been omitted from
\Sec{canonical} of the paper.
As a first step, we prove soundness of the canonical
system in \Sec{canon-sound-full}.
Next, we state and prove a \emph{hereditary substitution lemma} in
\Sec{canon-hsubst-full}, which establishes that canonical judgments
are preserved by hereditary substitutions.
This lemma is a key ingredient in proving completeness of the
canonical system.
We prove completeness in \Sec{canon-complete-full}, after showing that
canonical subtyping can be inverted at the top-level in
\Sec{canon-inversion-full}.
\Fig{canon_rules_sup} shows the canonical judgments that have been
omitted from the paper: context and kind formation, subkinding, and
kind and type equality.

\fig{fig:canon_rules_sup}{Canonical presentation of \FOmegaInt{} -- supplement}{
\judgment{Context formation}{\fbox{$\ctxC{\Gamma}$}}

\bigskip

\begin{minipage}[t]{0.21\linewidth}
  \infrule[\ruledef{CC-Empty}]{\quad}{\ctxC{\cempty}}
\end{minipage}\hspace{1ex}
\begin{minipage}[t]{0.35\linewidth}
  \infrule[\ruledef{CC-TmBind}]
  {\ctxC{\Gamma} \quad \kindC{\Gamma}{K}}
  {\ctxC{\Gamma, X \kas K}}
\end{minipage}\hfill
\begin{minipage}[t]{0.40\linewidth}
  \infrule[\ruledef{CC-TpBind}]
  {\ctxC{\Gamma} \quad \typeC{\Gamma}{V}}
  {\ctxC{\Gamma, x \tas V}}
\end{minipage}

\bigskip

\judgment{Kind formation}{\fbox{$\kindC{\Gamma}{K}$}}

\bigskip

\begin{minipage}[t]{0.52\linewidth}
\infruleLeft{CWf-Intv}
  {\typeC{\Gamma}{U} \andalso \typeC{\Gamma}{V}}
  {\kindC{\Gamma}{U \intv V}}
\end{minipage}\hfill
\begin{minipage}[t]{0.46\linewidth}
\infruleLeft{CWf-DArr}
  {\kindC{\Gamma}{J} \andalso \kindC{\Gamma, X \kas J}{K}}
  {\kindC{\Gamma}{\dfun{X}{J}{K}}}
\end{minipage}

\bigskip

\judgment{Subkinding}{\fbox{$\Gamma \ts J \ksub K$}}

\bigskip

\begin{minipage}[t]{0.44\linewidth}
\infruleLeft{CSK-Intv}
  {\quad\\
   \Gamma \ts U_2 \tsub U_1  \andalso  \Gamma \ts V_1 \tsub V_2}
  {\Gamma \ts U_1 \intv V_1 \ksub U_2 \intv V_2}
\end{minipage}\hspace{1ex}
\begin{minipage}[t]{0.54\linewidth}
\infruleLeft{CSK-DArr}
  {\kind{\Gamma}{\dfun{X}{J_1}{K_1}}\\
   \Gamma \ts J_2 \ksub J_1  \andalso
   \Gamma, X \kas J_2 \ts K_1 \ksub K_2}
  {\Gamma \ts \dfun{X}{J_1}{K_1} \ksub \dfun{X}{J_2}{K_2}}
\end{minipage}

\bigskip

\begin{minipage}[2]{0.4\linewidth}
  \judgment{Kind equality}{\fbox{$\Gamma \ts J \keq K$}}

\infruleLeft{CSK-AntiSym}
  {\kind{\Gamma}{J}      \andalso  \kind{\Gamma}{K}\\
   \Gamma \ts J \tsub K  \andalso  \Gamma \ts K \tsub J}
  {\Gamma \ts J \teq K}
\end{minipage}\hspace{0.07\linewidth}
\begin{minipage}[2]{0.5\linewidth}
  \judgment{Type equality}{\fbox{$\Gamma \ts U \keq V \kck K$}}

\infruleLeft{CST-AntiSym}
  {\kind{\Gamma}{K}\\
   \Gamma \ts U \tsub V \kck K  \andalso  \Gamma \ts V \tsub U \kck K}
  {\Gamma \ts U \teq V \kck K}
\end{minipage}}
 
The contents of this section are based on Chapter~5 of the first
author's PhD~dissertation.  We refer the interested reader to the
dissertation for the full details~\cite[][Ch.~5]{Stucki17thesis}.

\subsection{Soundness and Basic Properties}
\label{sec:canon-sound-full}

Before we establish any other metatheoretic properties of the
canonical system, let us prove its soundness with respect to the
declarative presentation.  To avoid confusion, we mark canonical
judgments with the subscript ``$\mathsf{c}$'' and declarative ones
with ``$\mathsf{d}$'' in the following lemma.

\begin{lemma}[soundness of the canonical rules]
  \label{lem:canon-sound-full}
  ~
  \begin{enumerate}[nosep]
  \item\label{item:can-snd-ctx} If $\; \ctxCC{\Gamma}$, then
    $\; \ctxDD{\toType{\Gamma}}$.
  \item\label{item:can-snd-kd} If $\; \kindCC{\Gamma}{K}$, then
    $\; \kindDD{\toType{\Gamma}}{\toType{K}}$.
  \item\label{item:can-snd-sk} If $\; \Gamma \tsC J \ksub K$, then
    $\; \toType{\Gamma} \tsD \toType{J} \ksub \toType{K}$.
  \item\label{item:can-snd-keq} If $\; \Gamma \tsC J \keq K$, then
    $\; \toType{\Gamma} \tsD \toType{J} \keq \toType{K}$.
  \item\label{item:can-snd-var} If $\; \Gamma \tsVar X \kin K$, then
    $\; \toType{\Gamma} \tsD X \kin \toType{K}$.
  \item\label{item:can-snd-ne} If $\; \Gamma \tsNe N \kin K$, then
    $\; \toType{\Gamma} \tsD \toType{N} \kin \toType{K}$.
  \item\label{item:can-snd-nfsn} If $\; \Gamma \tsC V \ksn K$, then
    $\; \toType{\Gamma} \tsD \toType{V} \kin \toType{K}$.
  \item\label{item:can-snd-nfck} If $\; \Gamma \tsC V \kck K$, then
    $\; \toType{\Gamma} \tsD \toType{V} \kin \toType{K}$.
  \item\label{item:can-snd-st} If $\; \Gamma \tsC U \tsub V$, then
    $\; \toType{\Gamma} \tsD \toType{U} \tsub \toType{V} \kin \kstar$.
  \item\label{item:can-snd-stck} If $\; \Gamma \tsC U \tsub V \kin K$, then
    $\; \toType{\Gamma} \tsD \toType{U} \tsub \toType{V} \kin
    \toType{K}$.
  \item\label{item:can-snd-teq} If $\; \Gamma \tsC U \teq V \kin K$, then
    $\; \toType{\Gamma} \tsD \toType{U} \tsub \toType{V} \kin
    \toType{K}$.
  \item\label{item:can-snd-sp} If
    $\; \toType{\Gamma} \tsD A \kin \toType{J}$ and
    $\; \spCC{\Gamma}{J}{\vec{V}}{K}$, then
    $\; \toType{\Gamma} \tsD \toType{\app{\toElim{A}}{\vec{V}}} \kin
    \toType{K}$.
  \item\label{item:can-snd-speq} If
    $\; \toType{\Gamma} \tsD A \tsub B \kin \toType{J}$ and
    $\; \spCC{\Gamma}{J}{\vec{U} \teq \vec{V}}{K}$, then
    $\; \toType{\Gamma} \tsD \toType{\app{\toElim{A}}{\vec{U}}} \tsub
    \toType{\app{\toElim{B}}{\vec{V}}} \kin \toType{K}$.
  \end{enumerate}
\end{lemma}
The proof uses the following two helper lemmas.
The first is a derived rule, the second is proven by easy case
analysis of the rules for kind synthesis.
\begin{lemma}\label{lem:K-Intv-Star}
If $\; \Gamma \ts A \kin B \intv C$, then also
  $\Gamma \ts A \kin \kstar$.
\end{lemma}\noindent
\begin{lemma}\label{lem:nfsn-intv-sing}
  If $\; \Gamma \ts U \ksn V \intv W$, then $V \alEq U$ and
  $W \alEq U$.
\end{lemma}\noindent
\begin{proof}[Proof of \Lem{canon-sound-full}]
  By induction on derivations of the various canonical judgments; for
  parts~\Item{can-snd-ctx}--\Item{can-snd-teq}, on the derivation of
  the first premise, for parts~\Item{can-snd-sp}
  and~\Item{can-snd-speq}, on that of the second premise.  Most cases
  are straightforward.  In cases involving synthesized kinding of
  proper types, we use~\ruleref{K-Sing} and \Lem[Lemmas]{K-Intv-Star},
  \ref{lem:ST-Intv-Star} and~\ref{lem:nfsn-intv-sing} to adjust kinds
  where necessary.  In cases involving hereditary substitutions, i.e.\
  those for~\ruleref{CK-Cons} and~\ruleref{SpEq-Cons}, we use
  soundness of hereditary substitution in kinds
  (\Cor{hsubst-sound}.\Item{hsubst-snd-kind}).
\end{proof}

As for the declarative system, the contexts of most canonical
judgments are well-formed.  There are two exceptions: kinding and
equality of spines. The rules~\ruleref{CK-Empty}
and~\ruleref{SpEq-Empty} for empty spines offer no guarantee that the
enclosing context is well-formed.  This is not a problem in practice,
since well-kinded spines only appear in judgments about neutral types,
the heads of which must be kinded in a well-formed context.
\begin{lemma}[context validity]\label{lem:canon_ctx_valid}
  Assume $\; \judgD{\Gamma}$ for any canonical judgment except spine
  kinding or equality.  Then $\; \ctxD{\Gamma}$.
\end{lemma}
\begin{proof}
  By simultaneous induction on the derivations of the various
  judgments.
\end{proof}
We can prove a few more validity properties at this point.
First, since synthesized kinds are singletons, validity of synthesized
kinding judgments follows by~\ruleref{CWf-Intv} for proper types and
by an easy induction for type operators.
Second, validity of kind equality as well as the checked subtyping and
type equality judgments follows immediately from the validity
conditions included in the rules~\ruleref{CSK-AntiSym},
\ruleref{CST-Intv}, \ruleref{CST-Abs} and~\ruleref{CST-AntiSym}.
\begin{lemma}[canonical validity -- part 1]\label{lem:canon-valid1}
  ~
\begin{enumerate}[nosep,leftmargin=14.85em,labelsep=.9em]
\item[(synthesized kinding validity)]
If $\; \Gamma \ts V \ksn K$, then $\kindC{\Gamma}{K}$.

\item[(checked subtyping validity)]
If $\; \Gamma \ts U \tsub V \kck K$, then $\Gamma \ts U \kck K$ and
    $\Gamma \ts V \kck K$.

\item[(kind equation validity)]
If $\; \Gamma \ts J \keq K$, then $\kindC{\Gamma}{J}$ and
    $\kindC{\Gamma}{K}$.

\item[(type equation validity)]
If $\; \Gamma \ts U \teq V \kck K$, then $\Gamma \ts U \kck K$,
    $\; \Gamma \ts V \kck K$ and $\kindC{\Gamma}{K}$.
  \end{enumerate}
\end{lemma}
We prove the remaining validity properties of the canonical system
in~\Sec{canon-valid}, once we have shown that hereditary substitutions
preserve well-formedness of kinds.

Before we can do so, we need to establish the usual weakening and
context narrowing lemmas for canonical typing.
\begin{lemma}[weakening]\label{lem:canon-weaken}
  Assume $\; \judgC{\Gamma, \Delta}$ for any of the canonical
  judgments.
  \begin{enumerate}
  \item If $\; \typeC{\Gamma}{A}$, then
    $\; \judgC{\Gamma, x \tas A, \Delta}$.
  \item If $\; \kindC{\Gamma}{K}$, then
    $\; \judgC{\Gamma, X \kas K, \Delta}$.
  \end{enumerate}
\end{lemma}

\begin{corollary}[iterated weakening]\label{cor:canon-iter-weaken}
  If $\; \ctxC{\Gamma, \Delta}$ and $\judgC{\Gamma}$, then
  $\judgC{\Gamma, \Delta}$.
\end{corollary}

\begin{lemma}[context narrowing -- weak version]\label{lem:canon-weak-narrow}
~
  \begin{enumerate}
  \item If $\; \typeC{\Gamma}{A}$, $\; \Gamma \ts A \tsub B \kin \kstar$
    and $\; \judgC{\Gamma, x \tas B, \Delta}$, then
    $\; \judgC{\Gamma, x \tas A, \Delta}$.
  \item If $\; \kindC{\Gamma}{J}$, $\; \Gamma \ts J \ksub K$ and
    $\; \judgC{\Gamma, X \kas K, \Delta}$, then
    $\; \judgC{\Gamma, X \kas J, \Delta}$.
  \end{enumerate}
\end{lemma}

The proofs of both lemmas are routine inductions on the derivations of
the respective judgments.  The proof of context narrowing is only easy
thanks to the rule~\ruleref{CV-Sub}.  Without this rule, the canonical
kinding judgments for variables and neutral types would become
synthesis judgments and context narrowing would no longer hold in its
present form for these two judgments.  To see this, consider the
variable kinding judgment $\Gamma, X \kas K, \Delta \tsVar X \kin K$,
which, after eliminating~\ruleref{CV-Sub}, could only be derived
using~\ruleref{CV-Var}.  If we were to narrow the context by changing
the declared kind $K$ of $X$ to some $J$ such that
$\Gamma \ts J \ksub K$, then the synthesized kind of $X$ would
necessarily change to $J$ too.

For neutral kinding, we would be in a similar situation as the new
synthesized kind of the head would have to be propagated through the
kinding derivation of the spine.  Along the way, the new kind $J$
would have to be unraveled by repeatedly separating
$J = \dfun{X}{J_1}{J_2}$ into $J_1$, $J_2$ and hereditarily
substituting the next element of the spine into the codomain $J_2$.
To prove context narrowing for the remainder of the canonical
judgments, we would have to maintain the invariant that the new
synthesized kind of a neutral type is a subtype of the original one,
i.e. if $\; \Gamma, X \kas K_1, \Delta \tsVar N \ksn K_2$ and
$\Gamma \ts J_1 \ksub K_1$ then
$\Gamma, X \kas J_1, \Delta \tsVar N \ksn J_2$ and
$\Gamma \ts J_2 \ksub K_2$.  Because spine kinding involves hereditary
substitution, a proof involving this invariant would require a
hereditary substitution lemma for subtyping, i.e.\ a proof that
hereditary substitutions preserve subtyping.  We give such a proof in
\Sec{canon-hsubst-full} and, as we will see there, it makes crucial use of
context narrowing itself.  By allowing us to establish
\Lem{canon-weak-narrow} independently, the rule~\ruleref{CV-Sub} thus
simplifies the proof of an otherwise rather complicated lemma (see
\Lem{canon-hsubst-full} for details).

\subsubsection{Order-Theoretic Properties}
\label{sec:canon-order}

Having established context narrowing, we can prove the usual
order-theoretic properties of canonical subkinding, subtyping, as well
as kind and type equality.  We start by stating an proving the various
reflexivity properties which, unlike those for the declarative
relations, have to be proven simultaneously for the canonical
variants.

\begin{lemma}\label{lem:canon-refl}
  The following reflexivity rules are all admissible.
\end{lemma}
\begin{multicols}{2}
  \typicallabel{CSK-ReflSub}
\infrule[\ruledef{CSK-Refl}]
  {\kindC{\Gamma}{K}}
  {\Gamma \ts K \ksub K}

\infrule[\ruledef{CST-ReflSyn}]
  {\Gamma \ts U \ksn V \intv W}
  {\Gamma \ts U \tsub U}

\infrule[\ruledef{CKEq-Refl}]
  {\kindC{\Gamma}{K}}
  {\Gamma \ts K \keq K}

\infrule[\ruledef{CST-ReflCk}]
  {\Gamma \ts V \kck K  \andalso  \kindC{\Gamma}{K}}
  {\Gamma \ts V \tsub V \kck K}
\end{multicols}
\vspace{.1pt}

\infrule[\ruledef{CST-ReflSub}]
{\Gamma \ts V \ksn J  \andalso  \Gamma \ts J \ksub K  \andalso
 \kindC{\Gamma}{K}}
{\Gamma \ts V \tsub V \kck K}
\vspace{.1pt}

\begin{multicols}{2}
\infrule[\ruledef{SpEq-Refl}]
  {\spC{\Gamma}{J}{\vec{V}}{K}}
  {\spC{\Gamma}{J}{\vec{V} \keq \vec{V}}{K}}

\infrule[\ruledef{CTEq-Refl}]
  {\Gamma \ts V \kck K  \andalso  \kindC{\Gamma}{K}}
  {\Gamma \ts V \teq V \kck K}
\end{multicols}
\begin{proof}
  The proof is by mutual induction in the structure of the kinds and
  types being related to themselves, and then by case-analysis on the
  final rules of the corresponding kind formation and kinding
  derivations.  The proof for~\ruleref{CST-ReflSub} proceeds by an
  inner induction on the derivation of $\Gamma \ts J \ksub K$.  In the
  case for~\ruleref{CSK-DArr} where $V = \lam{X}{J_1}{U}$ and we have
$\Gamma \ts K_1 \ksub J_1$, $\Gamma, X \kas K_1 \ts J_2 \ksub K_2$
  and $\Gamma, X \kas J_1 \ts U \ksn J_2$, we use context narrowing to
  adjust the declared kind of $X$ from $J_1$ to $K_1$ in
  $\; \Gamma, X \kas J_1 \ts U \ksn J_2$ before applying the IH to
  obtain $\Gamma, X \kas K_1 \ts U \tsub U \kck K_2$.
\end{proof}
Transitivity of the various relations and symmetry of the equalities
are more easily established, thanks to the rule~\ruleref{CST-Trans} on
the one hand, and to the structure of equality on the other.
\begin{lemma}\label{lem:canon-trans}
  Canonical subkinding, subtyping, kind and type equality are
  transitive.
\end{lemma}
\begin{multicols}{2}
  \typicallabel{CKEq-Trans}
\infrule[\ruledef{CSK-Trans}]
  {\Gamma \ts J \ksub K  \andalso  \Gamma \ts K \ksub L}
  {\Gamma \ts J \ksub L}

\infrule[\ruledef{CKEq-Trans}]
  {\Gamma \ts J \keq K  \andalso  \Gamma \ts K \keq L}
  {\Gamma \ts J \keq L}
\end{multicols}
\begin{multicols}{2}
\infrule[\ruledef{CST-TransCk}]
  {\Gamma \ts U \tsub V \kck K  \andalso  \Gamma \ts V \tsub W \kck K}
  {\Gamma \ts U \tsub W \kck K}

\infrule[\ruledef{CTEq-Trans}]
  {\Gamma \ts U \teq V \kck K  \andalso  \Gamma \ts V \teq W \kck K}
  {\Gamma \ts U \teq W \kck K}
\end{multicols}
\begin{proof}
  The proof of~\ruleref{CSK-Trans} is by induction on the structure of
  $K$ and case analysis on the final rules used to derive the
  premises.  In the case for $K = \dfun{X}{K_1}{K_2}$ we use context
  narrowing.  The proof of~\ruleref{CST-Trans} is by induction on the
  derivation of the first premise.  The proofs of~\ruleref{CKEq-Trans}
  and~\ruleref{CTEq-Trans} are by inspection of the equality rules and
  use~\ruleref{CSK-Trans} and~\ruleref{CST-Trans}, respectively.
\end{proof}

\begin{lemma}\label{lem:canon-sym}
  Canonical kind and type equality are symmetric.
  \begin{multicols}{2}\normalfont
    \typicallabel{CKEq-Sym}
\infrule[\ruledef{CKEq-Sym}]
    {\Gamma \ts J \keq K}
    {\Gamma \ts K \keq J}

\infrule[\ruledef{CTEq-Sym}]
    {\Gamma \ts U \keq V \kck K}
    {\Gamma \ts V \keq U \kck K}
  \end{multicols}
\end{lemma}
\begin{proof}
  By inspection of the equality rules, \ruleref{CSK-AntiSym}
  and~\ruleref{CST-AntiSym}.
\end{proof}

Thanks to context narrowing and subkinding transitivity, we can prove
admissibility of the following subsumption rules for the three checked
judgments.
\begin{lemma}
  Kind subsumption is admissible in the checked judgments.
\end{lemma}
\begin{center}
  \infruleSimp[CK-SubCk]
  {\Gamma \ts U \kck J  \andalso  \Gamma \ts J \ksub K}
  {\Gamma \ts U \kck K}
  \hspace{2em}
  \infruleSimp[CST-SubCk]
  {\Gamma \ts U \tsub V \kck J  \andalso  \Gamma \ts J \ksub K}
  {\Gamma \ts U \tsub V \kck K}

  \infruleSimp[CTEq-SubCk]
  {\Gamma \ts U \teq V \kck J  \andalso  \Gamma \ts J \ksub K}
  {\Gamma \ts U \teq V \kck K}
\end{center}
\begin{proof}
  Admissibility is proven separately for the three rules, in the order
  they are listed.  The proof of~\ruleref{CK-SubCk} is by inspection
  of the kind checking rules
and uses~\ruleref{CSK-Trans}.  The proof of~\ruleref{CST-SubCk} is
  by induction on the derivation of the second premise
  $\Gamma \ts J \ksub K$ and uses~\ruleref{CK-SubCk} as well as context
  narrowing for the inductive step.  The proof of~\ruleref{CTEq-SubCk}
  is by inspection of the equality rules and uses~\ruleref{CST-SubCk}.
\end{proof}

Kind subsumption subsumes kind conversion thanks to the first of the
following three rules, all of which follow immediately by inspection
of the equality and checked subtyping rules.
\begin{center}
  \infruleSimp[CSK-Refl-KEq]
  {\Gamma \ts J \keq K}
  {\Gamma \ts J \tsub K}
  \hspace{2em}
  \infruleSimp[CST-Refl-TEq]
  {\Gamma \ts U \teq V \kck K}
  {\Gamma \ts U \tsub V \kck K}

  \infruleSimp[CST-Refl-TEq']
  {\Gamma \ts U \teq V \kck W \intv W'}
  {\Gamma \ts U \tsub V}
\end{center}

\subsubsection{Canonical Replacements for Declarative Rules}
\label{sec:canon-admissible}

Our completeness proof for the canonical system relies on the fact
that the normal form $\nf{}{A}$ of any declaratively well-kinded type
$\Gamma \ts A \kin K$ kind checks against the normal form $\nf{}{K}$
of the corresponding kind $K$, i.e.\ that we have
$\nfCtx{\Gamma} \ts \nf{}{A} \kck \nf{}{K}$.  To simplify the proof of
this fact, we establish a set of admissible kind checking rules below
that mirror the corresponding declarative rules.

We begin by proving that normal forms with synthesized or checked
interval kinds also check against~$\kstar$.
\begin{lemma}\label{lem:CK-Intv-Star}
  Types inhabiting interval kinds are proper types.  If
  $\; \Gamma \ts U \ksn V \intv W$ or
  $\; \Gamma \ts U \kck V \intv W$, then also
  $\Gamma \ts U \kck \kstar$.
\end{lemma}
\begin{proof}
  If $\; \Gamma \ts U \ksn V \intv W$, then by~\Lem{nfsn-intv-sing},
  $\; \Gamma \ts U \ksn U \intv U$.  From this we derive the result
  by~\ruleref{CST-Bot}, \ruleref{CST-Top}, \ruleref{CSK-Intv}
  and~\ruleref{CK-Sub}.  If $\; \Gamma \ts U \kck V \intv W$, then
  this must have been derived using~\ruleref{CK-Sub} and hence
  $\; \Gamma \ts U \ksn K$ and $\Gamma \ts K \ksub V \intv W$.  By
  inspection of the subkinding rules, $K = V' \intv W'$ for some $V'$
  and $W'$.  The result then follows by the first part of the lemma.
\end{proof}

The second part of the proof follows a pattern that is is quite
typical for proofs in the remainder of this section.  Thanks to the
division of kinding into kind synthesis and checking, and thanks to
the simplicity of both kind checking and subkinding derivations, we
can often prove properties of kind checking judgments
$\Gamma \ts U \kck K$ by appealing to similar properties of kind
synthesis judgments $\Gamma \ts U \ksn K'$ where $K$ and $K'$ have the
same shape, i.e. where $K$ and $K'$ are either both intervals or both
arrows.  Two more examples of this pattern appear in the following
lemma, where the admissibility proofs of the
rules~\rulerefN{CST-CkBnd1}{CST-CkBnd$_1$}
and~\rulerefN{CST-CkBnd2}{CST-CkBnd$_2$}, which have kind checking
judgments as their premises, appeal to instances
of~\rulerefN{CST-SynBnd1}{CST-SynBnd$_1$}
and~\rulerefN{CST-SynBnd2}{CST-SynBnd$_2$}, respectively, which have
similar kind synthesis judgments as their premises.

\begin{lemma}\label{lem:canon-admissible}
  All of the following are admissible.
\end{lemma}
\begin{multicols}{3}
  \typicallabel{}
\infrule[\ruledef{CK-Sing'}]
  {\Gamma \ts U \kck V \intv W}
  {\Gamma \ts U \ksn U \intv U}

\infrule[\ruledef{CK-SynCk}]
  {\Gamma \ts V \ksn K}
  {\Gamma \ts V \kck K}

\infrule[\ruledef{CK-NeCk}]
  {\Gamma \tsNe N \kin U \intv V}
  {\Gamma \ts N \kck U \intv V}
\end{multicols}
\begin{multicols}{2}
\typicallabel{CST-SynBnd$_1$}
\infrule[\ruledef{CK-Arr'}]
  {\typeCCk{\Gamma}{U}  \andalso  \typeCCk{\Gamma}{V}}
  {\typeCCk{\Gamma}{\fun{U}{V}}}

\infrule[\ruledef{CK-All'}]
  {\kindC{\Gamma}{K}  \andalso  \typeCCk{\Gamma, X \kas K}{V}}
  {\typeCCk{\Gamma}{\all{X}{K}{V}}}
\end{multicols}
\begin{multicols}{2}
  \typicallabel{CST-SynBnd$_1$}
\infrule[\ruledef{CK-Abs'}]
  {\kindC{\Gamma}{J}  \andalso  \Gamma, X \kas J \ts V \kck K}
  {\Gamma \ts \lam{X}{J}{V} \kck \dfun{X}{J}{K}}

\infrule[\ruledef{CST-Intv'}]
  {\Gamma \ts V_1 \tsub V_2  \kck U \intv W}
  {\Gamma \ts V_1 \tsub V_2 \kck V_1 \intv V_2}
\end{multicols}
\begin{multicols}{2}
  \typicallabel{CST-SynBnd$_1$}
\infrule[\ruledefN{CST-SynBnd1}{CST-SynBnd$_1$}]
  {\Gamma \ts U \ksn V_1 \intv V_2}
  {\Gamma \ts V_1 \tsub U}

\infrule[\ruledefN{CST-CkBnd1}{CST-CkBnd$_1$}]
  {\Gamma \ts U \kck V_1 \intv V_2}
  {\Gamma \ts V_1 \tsub U}

\infrule[\ruledefN{CST-SynBnd2}{CST-SynBnd$_2$}]
  {\Gamma \ts U \ksn V_1 \intv V_2}
  {\Gamma \ts U \tsub V_2}

\infrule[\ruledefN{CST-CkBnd2}{CST-CkBnd$_2$}]
  {\Gamma \ts U \kck V_1 \intv V_2}
  {\Gamma \ts U \tsub V_2}
\end{multicols}
\vspace{.1pt}

\typicallabel{SpEq-Snoc}
\infrule[\ruledef{CK-Snoc}]
{\spC{\Gamma}{J}{\vec{U}}{\dfun{X}{K}{L}}  \andalso
 \Gamma \ts V \kck K  \andalso  \kindC{\Gamma}{K}}
{\spC{\Gamma}{J}{\scons{\vec{U}}{V}}{\hsubst{L}{X}{\ksimp{K}}{V}}}

\infrule[\ruledef{SpEq-Snoc}]
{\spC{\Gamma}{J}{\vec{U}_1 \keq \vec{U}_2}{\dfun{X}{K}{L}}  \andalso
 \Gamma \ts V_1 \keq V_2 \kck K}
{\spC{\Gamma}{J}{\scons{\vec{U}_1}{V_1} \keq \scons{\vec{U}_2}{V_2}}{\hsubst{L}{X}{\ksimp{K}}{V_1}}}

\begin{proof}
  Some of the rules are derivable others are admissible; most of the
  proofs are straightforward, so we omit the details.  The proofs of
  the alternate type formation rules use~\ruleref{CK-Sing'} and
  \Lem{CK-Intv-Star}.  The proofs of the alternate projection
  rules~\rulerefN{CST-SynBnd1}{CST-SynBnd$_1$}
  and~\rulerefN{CST-SynBnd2}{CST-SynBnd$_2$} use \Lem{nfsn-intv-sing}
  and reflexivity; those of~\rulerefN{CST-CkBnd1}{CST-CkBnd$_1$}
  and~\rulerefN{CST-CkBnd2}{CST-CkBnd$_2$} are by inspection of kind
  checking and subkinding and
  use~\rulerefN{CST-SynBnd1}{CST-SynBnd$_1$}
  and~\rulerefN{CST-SynBnd2}{CST-SynBnd$_2$}, respectively.  The
  proofs of the last two rules are by induction on the derivations of
  the respective first premises.
\end{proof}

As in the declarative system, we define \emph{canonical context
  equality} $\ctxEqC{\Gamma}{\Delta}$ as the pointwise lifting of
canonical type and kind equality to context bindings:
\begin{center}
  \infruleSimp{\quad}{\ctxEqD{\cempty}{\cempty}}
  \hspace{3em}
  \infruleSimp
  {\ctxEqD{\Gamma}{\Delta} \andalso \Gamma \ts J \keq K}
  {\ctxEqD{\Gamma, X \kas J}{\Delta, X \kas K}}
  \hspace{3em}
  \infruleSimp
  {\ctxEqD{\Gamma}{\Delta} \andalso \Gamma \ts U \teq V \kck W \intv W'}
  {\ctxEqD{\Gamma, x \tas U}{\Delta, x \tas V}}
\end{center}

\subsubsection{Simplification of Canonical Kinding}

In \Sec{simple_nf} of the previous section, we showed that every
well-formed kind and well-kinded type has a simply well-formed or
simply well-kinded normal form, respectively (see \Lem{nf-simp}).
Canonically well-formed kinds and canonically well-kinded types are
already in normal form, but we can still simplify their kind formation
and kinding derivations, as the following pair of lemma shows.  In the
statement of the second lemma, we use the subscript ``$\mathsf{nes}$''
to mark simple kinding judgments for neutral types and
``$\mathsf{ne}$'' to mark their canonical counterparts.
\begin{lemma}\label{lem:st-teq-simp}
  Canonical subkinds and equal kinds simplify equally.  If
  $\; \Gamma \ts J \tsub K$ or $\; \Gamma \ts J \teq K$ then
  $\simpEq{J}{K}$.
\end{lemma}
\begin{proof}
  Separately, by induction on subkinding and kind equality
  derivations, respectively.
\end{proof}
\begin{lemma}[simplification]\label{lem:canon-simp}
  Well-formed kinds and well-kinded normal forms, neutrals and spines
  are also simply well-formed and well-kinded, respectively.
  \begin{enumerate}[nosep]
  \item\label{item:canon-simp-var} If $\; \Gamma \tsVar X \kin K$, then
    $\; \ksimp{\Gamma} \ts[nes] X \kin \ksimp{K}$.
  \item\label{item:canon-simp-kd} If $\; \kindC{\Gamma}{K}$, then
    $\; \kindS{\ksimp{\Gamma}}{K}$.
  \item\label{item:canon-simp-ne} If $\; \Gamma \tsNe N \kin K$, then
    $\; \ksimp{\Gamma} \ts[nes] N \kin \ksimp{K}$.
  \item\label{item:canon-simp-sn} If $\; \Gamma \ts V \ksn K$, then
    $\; \ksimp{\Gamma} \ts V \kin \ksimp{K}$.
  \item\label{item:canon-simp-ck} If $\; \Gamma \ts V \kck K$, then
    $\; \ksimp{\Gamma} \ts V \kin \ksimp{K}$.
  \item\label{item:canon-simp-sp} If $\; \spC{\Gamma}{J}{\vec{V}}{K}$,
    then $\; \spS{\ksimp{\Gamma}}{\ksimp{J}}{\vec{V}}{\ksimp{K}}$.
  \end{enumerate}
\end{lemma}
\begin{proof}
  Part~\Item{canon-simp-var} is proven separately,
  parts~\Item{canon-simp-kd}--\Item{canon-simp-sp} are proven
  simultaneously, all by induction on the derivations of the
  respective premises.  The cases of~\ruleref{CV-Sub}
  and~\ruleref{CK-Sub} use of \Lem{st-teq-simp}.
\end{proof}
Thanks to \Lem{canon-simp}, properties of simply kinded normal forms
still hold for canonically kinded normal forms.  For example,
by~\Lem[Lemmas]{canon-simp} and~\ref{lem:hsubst-comm}, hereditary
substitutions in canonically kinded types commute.

\subsection{The Hereditary Substitution Lemma}
\label{sec:canon-hsubst-full}

We have arrived at
the core of the technical development of this section: the proof of
the \emph{hereditary substitution lemma}.  The hereditary substitution
lemma states, roughly, that canonical judgments are preserved by
hereditary substitutions of canonically well-kinded types.  Just as
the ordinary substitution lemma for the declarative system
(\Lem{decl_subst})
played a key role in the proofs of several metatheoretic properties
in \Sec{declarative} of the paper,
the hereditary substitution lemma is key to
proving important metatheoretic properties of the canonical system.
But unlike that of its ordinary counterpart, the proof of the
hereditary substitution lemma is rather challenging.  This is
reflected already in the statement of the lemma, which
features \Item{can-rapp-teq} separate parts, all of which have to be
proven simultaneously (see \Lem{canon-hsubst-full} below).

One reason for the large number of parts is simply that there are more
judgment forms in the canonical system than there are in the
declarative system.  But the foremost reason is that the proof of the
hereditary substitution lemma circularly depends on
\emph{functionality of the canonical judgments}, i.e.\ on the fact
that hereditarily substituting canonically equal types in normal forms
yields canonically equal normal forms.
This also renders the proof more challenging since both properties
have to be established at the same time.

The main source of complexity is the subtyping rule~\ruleref{CST-Ne}.
It is because of this rule that we have to prove the hereditary
substitution and functionality lemmas simultaneously.

To illustrate this, consider the neutral types $N = \app{X}{\vec{V}}$
and $M = \app{X}{\vec{W}}$ with $\vec{V} = V_1, V_2$,
$\vec{W} = W_1, W_2$ such that
$X \notin \fv(\vec{V}) \cup \fv(\vec{W})$, and assume some $\Gamma$,
$\Delta$ and $U$ such that
\begin{align*}
&\Gamma, \Delta \; \ts \;
    \lam{Y_1}{J_1}{\lam{Y_2}{J_2}{U}} \; \kck \; J &
  &\text{and}&
  &\spC{\Gamma, X \kas J, \Delta \;}{\; J \;}{\; \vec{V}  \; \teq  \;
    \vec{W} \;}{\; \kstar}
\end{align*}
for $J = \dfun{Y_1}{J_1}{\dfun{Y_2}{J_2}{J_3}}$.  Then,
by~\ruleref{CST-Ne}, we have $\Gamma, X \kas J, \Delta \ts N \tsub M$.

We would like to hereditarily substitute
$U' = \lam{Y_1}{J_1}{\lam{Y_2}{J_2}{U}}$ for $X$ in $N$ and $M$ and
show that the resulting types remain subtypes, i.e. that
\[
  \Gamma, \hsubst{\Delta}{X}{\ksimp{J}}{U'} \; \ts \;
  \hsubst{N}{X}{\ksimp{J}}{U'} \; \tsub \;
  \hsubst{M}{X}{\ksimp{J}}{U'}.
\]
By the definition of hereditary substitution, we have
\begin{align*}
  \hsubst{N}{X}{\ksimp{J}}{U'} \; &\alEq \;
  \hsubst{(\app{X}{\vec{V}})}{X}{\ksimp{J}}{U'} \; \alEq \;
  \rapp{U'}{\ksimp{J}}{(\hsubst{\vec{V}}{X}{\ksimp{J}}{U'})}\\
  &\alEq \;
  \rapp{\rapp{U'}{\ksimp{J}}{V_1}}{\fun{\ksimp{J_2}}{\ksimp{J_3}}}{V_2}\\
  &\alEq \;
  \rapp{(\hsubst{(\lam{Y_2}{J_2}{U})}{Y_1}{\ksimp{J_1}}{V_1})}{\fun{\ksimp{J_2}}{\ksimp{J_3}}}{V_2}\\
  &\alEq \;
  \rapp{(\lam{Y_2}{\hsubst{J_2}{Y_1}{\ksimp{J_1}}{V_1}}{\hsubst{U}{Y_1}{\ksimp{J_1}}{V_1}})}{\fun{\ksimp{J_2}}{\ksimp{J_3}}}{V_2} \\
  &\alEq \; \hsubst{V'}{Y_2}{\ksimp{J_2}}{V_2}
\end{align*}
where $V' = \hsubst{U}{Y_1}{\ksimp{J_1}}{V_1}$.  Similarly,
$\hsubst{M}{X}{\ksimp{J}}{U'} \alEq
\hsubst{W'}{Y_2}{\ksimp{J_2}}{W_2}$
for $W' = \hsubst{U}{Y_1}{\ksimp{J_1}}{W_1}$.  Hence, we need to show
that
\begin{alignat*}{3}
  \Gamma, \hsubst{\Delta}{X}{\ksimp{J}}{U'} \; &\ts \;{}&
  \hsubst{U\,}{Y_1}{\ksimp{J_1}}{V_1}  &{}\; \tsub \;{}&
  \hsubst{U\,}{Y_1}{\ksimp{J_1}}{W_1}& \qquad \text{and} \\
  \Gamma, \hsubst{\Delta}{X}{\ksimp{J}}{U'} \; &\ts \;{}&
  \hsubst{V'}{Y_2}{\ksimp{J_2}}{V_2} &{}\; \tsub \;{}&
  \hsubst{W'}{Y_2}{\ksimp{J_2}}{W_2}&.
\end{alignat*}
Since they belong to equal spines, $V_1$ and $W_1$ are judgmentally
equal as types, and so are $V_2$ and $W_2$.  But in general, neither
of these pairs of types are syntactically equal, i.e.\
$V_1 \not\alEq W_1$, $V_2 \not\alEq W_2$ and $V' \not\alEq W'$.  To
establish the above inequations, we therefore need to show that
simultaneous hereditary substitutions of judgmentally equal types
preserve inequations.

The example illustrates a second point, namely that, in order to prove
that hereditary substitutions preserve canonical kinding and
subtyping, we need to prove that kinding and subtyping of reducing
applications is admissible.
Our hereditary substitution lemma must cover all of these properties,
leading to the aforementioned grand total
of \Item{can-rapp-teq}~parts.

\begin{lemma}[hereditary substitution]\label{lem:canon-hsubst-full}
  Hereditary substitutions of canonically kind-checked types preserve
  the canonical judgments; substitutions of canonically equal
  types in canonically well-formed and well-kinded expressions result
  in canonical equations; substitutions of canonically equal
  types preserve canonical (in)equations; kinding and subtyping of
  reducing applications is admissible.  Assume that the following equations hold
  \begin{align*}
    &\Gamma \ts U_1 \; \teq \; U_2 \kck K
    &&\ctxD{\Gamma, \Sigma \; \keq \;
       \Gamma, \hsubst{\Delta}{X}{\ksimp{K}}{U_1}}
    &&\ctxD{\Gamma, \Sigma \; \keq \;
       \Gamma, \hsubst{\Delta}{X}{\ksimp{K}}{U_2}}
  \end{align*}
  for given $\Gamma$, $\Delta$, $\Sigma$, $U_1$, $U_2$, $K$ and
  $X \notin \dom(\Gamma, \Delta, \Sigma)$.
  \begin{enumerate}[nosep]
\item\label{item:can-hsubst-kd} If
    $\; \kindC{\Gamma, X \kas K, \Delta}{J}$, then
    $\; \kindC{\Gamma, \Sigma}{\hsubst{J}{X}{\ksimp{K}}{U_1}}$.
\item\label{item:can-hsubst-var-hit} If
    $\; \Gamma, X \kas K, \Delta \tsVar X \kin J$, then
    $\; \Gamma, \Sigma \ts U_1 \kck \hsubst{J}{X}{\ksimp{K}}{U_1}$.
\item\label{item:can-hsubst-var-miss} If
    $\; \Gamma, X \kas K, \Delta \tsVar Y \kin J$ and $Y \neq X$, then
    $\; \Gamma, \Sigma \tsVar Y \kin \hsubst{J}{X}{\ksimp{K}}{U_1}$.
\item\label{item:can-hsubst-ne} If
    $\; \Gamma, X \kas K, \Delta \tsNe N \kin V \intv W$, then
    $\; \Gamma, \Sigma \ts \hsubst{N}{X}{\ksimp{K}}{U_1} \kck
    \hsubst{(V \intv W)}{X}{\ksimp{K}}{U_1}$.
\item\label{item:can-hsubst-sp} If $\; \kindS{\ksimp{\Gamma}}{J}$
    and $\; \spC{\Gamma, X \kas K, \Delta}{J}{\vec{V}}{L}$, then
    \[
      \spC{\Gamma, \Sigma}{\hsubst{J}{X}{\ksimp{K}}{U_1}}{\hsubst{\vec{V}}{X}{\ksimp{K}}{U_1}}{\hsubst{L}{X}{\ksimp{K}}{U_1}}.
    \]
\item\label{item:can-hsubst-sn} If
    $\; \Gamma, X \kas K, \Delta \ts V \ksn J$, then
    $\; \Gamma, \Sigma \ts \hsubst{V}{X}{\ksimp{K}}{U_1} \ksn
    \hsubst{J}{X}{\ksimp{K}}{U_1}$.
\item\label{item:can-hsubst-ck} If
    $\; \Gamma, X \kas K, \Delta \ts V \kck J$, then
    $\; \Gamma, \Sigma \ts \hsubst{V}{X}{\ksimp{K}}{U_1} \kck
    \hsubst{J}{X}{\ksimp{K}}{U_1}$.
\item\label{item:can-par-hsubst-kd-sk} If
    $\; \kindC{\Gamma, X \kas K, \Delta}{J}$, then
    $\; \Gamma, \Sigma \ts \hsubst{J}{X}{\ksimp{K}}{U_1} \ksub
    \hsubst{J}{X}{\ksimp{K}}{U_2}$.
\item\label{item:can-par-hsubst-kd-keq} If
    $\; \kindC{\Gamma, X \kas K, \Delta}{J}$, then
    $\; \Gamma, \Sigma \ts \hsubst{J}{X}{\ksimp{K}}{U_1} \keq
    \hsubst{J}{X}{\ksimp{K}}{U_2}$.
\item\label{item:can-par-hsubst-var-hit} If
    $\; \Gamma, X \kas K, \Delta \tsVar X \kin J$, then
    $\; \Gamma, \Sigma \ts U_1 \teq U_2 \kck
    \hsubst{J}{X}{\ksimp{K}}{U_1}$.
\item\label{item:can-par-hsubst-ne} If
    $\; \Gamma, X \kas K, \Delta \tsNe N \kin V \intv W$, then
    $\Gamma, \Sigma \ts \hsubst{N}{X}{\ksimp{K}}{U_1} \tsub
    \hsubst{N}{X}{\ksimp{K}}{U_2}$.
\item\label{item:can-par-hsubst-sp} If $\; \kindS{\ksimp{\Gamma}}{J}$
    and $\; \spC{\Gamma, X \kas K, \Delta}{J}{\vec{V}}{L}$, then
    \[
      \spC{\Gamma, \Sigma}{\hsubst{J}{X}{\ksimp{K}}{U_1}}{\hsubst{\vec{V}}{X}{\ksimp{K}}{U_1} \teq
        \hsubst{\vec{V}}{X}{\ksimp{K}}{U_2}}{\hsubst{L}{X}{\ksimp{K}}{U_1}}.
    \]
\item\label{item:can-par-hsubst-sn-intv} If
    $\; \Gamma, X \kas K, \Delta \ts V \ksn W \intv W'$, then
    $\Gamma, \Sigma \ts \hsubst{V}{X}{\ksimp{K}}{U_1} \tsub
    \hsubst{V}{X}{\ksimp{K}}{U_2}$.
\item\label{item:can-par-hsubst-sn} If
    $\; \Gamma, X \kas K, \Delta \ts V \ksn J$ and
    $\; \kindC{\Gamma, X \kas K, \Delta}{J}$, then
    \[
      \Gamma, \Sigma \ts \hsubst{V}{X}{\ksimp{K}}{U_1} \tsub
      \hsubst{V}{X}{\ksimp{K}}{U_2} \kck
      \hsubst{J}{X}{\ksimp{K}}{U_1}.
    \]
\item\label{item:can-par-hsubst-ck} If
    $\; \Gamma, X \kas K, \Delta \ts V \kck J$ and
    $\; \kindC{\Gamma, X \kas K, \Delta}{J}$, then
    \[
      \Gamma, \Sigma \ts \hsubst{V}{X}{\ksimp{K}}{U_1} \tsub
      \hsubst{V}{X}{\ksimp{K}}{U_2} \kck
      \hsubst{J}{X}{\ksimp{K}}{U_1}.
    \]
\item\label{item:can-par-hsubst-sk} If
    $\; \Gamma, X \kas K, \Delta \ts J_1 \ksub J_2$, then
    $\; \Gamma, \Sigma \ts \hsubst{J_1}{X}{\ksimp{K}}{U_1} \ksub
    \hsubst{J_2}{X}{\ksimp{K}}{U_2}$.
\item\label{item:can-par-hsubst-speq} If
    $\; \kindS{\ksimp{\Gamma}}{J}$ and
    $\; \spC{\Gamma, X \kas K, \Delta}{J}{\vec{V}_1 \teq \vec{V}_2}{L}$,
    then
    \[
      \spC{\Gamma, \Sigma}{\hsubst{J}{X}{\ksimp{K}}{U_1}}{\hsubst{\vec{V}_1}{X}{\ksimp{K}}{U_1} \teq
        \hsubst{\vec{V}_2}{X}{\ksimp{K}}{U_2}}{\hsubst{L}{X}{\ksimp{K}}{U_1}}.
    \]
\item\label{item:can-par-hsubst-st} If
    $\; \Gamma, X \kas K, \Delta \ts V_1 \tsub V_2$, then
    $\; \Gamma, \Sigma \ts \hsubst{V_1}{X}{\ksimp{K}}{U_1} \tsub
    \hsubst{V_2}{X}{\ksimp{K}}{U_2}$.
\item\label{item:can-par-hsubst-stck} If
    $\; \Gamma, X \kas K, \Delta \ts V_1 \tsub V_2 \kck J$ and
    $\kindC{\Gamma, X \kas K, \Delta}{J}$, then
    \[
      \Gamma, \Sigma \ts \hsubst{V_1}{X}{\ksimp{K}}{U_1} \tsub
      \hsubst{V_2}{X}{\ksimp{K}}{U_2} \kck
      \hsubst{J}{X}{\ksimp{K}}{U_1}.
    \]
\item\label{item:can-par-hsubst-teq} If
    $\; \Gamma, X \kas K, \Delta \ts V_1 \teq V_2 \kck J$, then
    $\; \Gamma, \Sigma \ts \hsubst{V_1}{X}{\ksimp{K}}{U_1} \teq
    \hsubst{V_2}{X}{\ksimp{K}}{U_2} \kck
    \hsubst{J}{X}{\ksimp{K}}{U_1}$.
\item\label{item:can-rapp-sp} If $\; \spC{\Gamma}{K}{\vec{V}}{J}$,
    then $\; \Gamma \ts \rapp{U_1}{\ksimp{K}}{\vec{V}} \kck J$.
\item\label{item:can-rapp-ck} If $\; K = \dfun{X}{K_1}{K_2}$,
    $\; \Gamma \ts V \kck K_1$ and $\; \kindC{\Gamma}{K_1}$, then
    $\; \Gamma \ts \rapp{U_1}{\ksimp{K}}{V} \kck
    \hsubst{K_2}{X}{\ksimp{K_1}}{V}$.
\item\label{item:can-rapp-speq} If
    $\; \spC{\Gamma}{K}{\vec{V}_1 \teq \vec{V}_2}{J}$, then
    $\; \Gamma \ts \rapp{U_1}{\ksimp{K}}{\vec{V}_1} \teq
    \rapp{U_2}{\ksimp{K}}{\vec{V}_2} \kck J$.
\item\label{item:can-rapp-teq} If $\; K = \dfun{X}{K_1}{K_2}$ and
    $\; \Gamma \ts V_1 \teq V_2 \kck K_1$, then
    \[
      \Gamma \ts \rapp{U_1}{\ksimp{K}}{V_1} \teq
      \rapp{U_1}{\ksimp{K}}{V_2} \kck
      \hsubst{K_2}{X}{\ksimp{K_1}}{V_1}.
    \]
  \end{enumerate}
\end{lemma}
\begin{proof}
As for the proof of~\Lem{simp-hsubst}, the structure of the proof
  mirrors that of the recursive definition of hereditary substitution
  itself.  All \Item{can-rapp-teq}~parts are proven simultaneously by
  induction in the structure of the simple kind $\ksimp{K}$.
  Parts~\Item{can-hsubst-kd}--\Item{can-par-hsubst-teq} proceed by an
  inner induction on the respective formation, kinding, subkinding,
  subtyping or equality derivations of the expressions in which $U_1$
  and $U_2$ are being substituted for $X$.
  Parts~\Item{can-rapp-sp}--\Item{can-rapp-teq} proceed by a case
  analysis on the final rule used to derive
  $\; \spC{\Gamma}{K}{\vec{V}}{J}$, $\; \Gamma \ts U_1 \kck K$,
  $\; \spC{\Gamma}{K}{\vec{V}_1 \teq \vec{V}_2}{J}$ and
  $\; \Gamma \ts U_1 \teq U_2 \kck K$, respectively.

  The proofs of parts~\Item{can-hsubst-kd}--\Item{can-hsubst-ck} are
  similar to that of the declarative substitution lemma
  (\Lem{decl_subst}), while those of
  parts~\Item{can-par-hsubst-kd-sk}--\Item{can-par-hsubst-teq}
  resemble the proof of the extended functionality lemma
  (\Lem{decl_funct_ext}).  In cases like \ruleref{CWf-DArr} or
  \ruleref{CK-All}, where the context is extended by an additional
  binding, we use the IH together with context narrowing
  (\Lem{canon-weak-narrow}) to maintain the invariants
  $\; \ctxD{\Gamma, \Sigma \; \keq \; \Gamma,
    \hsubst{\Delta}{X}{\ksimp{K}}{U_1}}$
  and
  $\; \ctxD{\Gamma, \Sigma \; \keq \; \Gamma,
    \hsubst{\Delta}{X}{\ksimp{K}}{U_2}}$.

  The cases where the proofs of
  parts~\Item{can-hsubst-kd}--\Item{can-par-hsubst-teq} differ most
  substantially from those of \Lem{decl_subst} and
  \Lem{decl_funct_ext} are
parts~\Item{can-hsubst-ne}, \Item{can-par-hsubst-ne} and the case
  for~\ruleref{CST-Ne} of part~\Item{can-par-hsubst-st}, which deal
  with neutral types.
There, we proceed by case distinction on $X = Y$, where $Y$ is the
  head of the corresponding neutral types.  If $X = Y$, then we
  proceed using either part~\Item{can-rapp-sp}, or
  part~\Item{can-rapp-speq} followed by~\ruleref{CST-Refl-TEq'}.  If
  $X \neq Y$, then we use part~\Item{can-hsubst-var-miss} and proceed
  with either part~\Item{can-hsubst-sp} followed by~\ruleref{CK-NeCk},
  or with parts~\Item{can-par-hsubst-sp} or~\Item{can-par-hsubst-speq}
  followed by~\ruleref{CST-Ne}.  In the cases
  for~\rulerefN{CST-Bnd1}{CST-Bnd$_1$}
  and~\rulerefN{CST-Bnd2}{CST-Bnd$_2$}, we use
  part~\Item{can-hsubst-ne} followed
  by~\rulerefN{CST-CkBnd1}{CST-CkBnd$_1$}
  or~\rulerefN{CST-CkBnd1}{CST-CkBnd$_1$}.

  In the cases for~\ruleref{CK-Cons} and~\ruleref{SpEq-Cons} of
  parts~\Item{can-hsubst-sp}, \Item{can-par-hsubst-sp}
  and \Item{can-par-hsubst-speq}, respectively, where
  $J = \dfun{Y}{J_1}{J_2}$, we use
  \Lem{hsubst-comm}.\Item{hsubst-comm-kds} to show that hereditary
  substitutions in kinds commute, \ie that
  \[
    \hsubst{\hsubst{J_2}{Y}{\ksimp{J_1}}{V_1}}{X}{\ksimp{K}}{U_1} \alEq
    \hsubst{\hsubst{J_2}{X}{\ksimp{K}}{U_1}}{Y}{\ksimp{J_1}}{\hsubst{V_1}{X}{\ksimp{K}}{U_1}}.
  \]
  The necessary simple kinding derivations are provided by case
  analysis of the final rule used to derive the premise
  $\; \kindS{\ksimp{\Gamma}}{J}$ and
  \Lem{canon-simp}.\Item{canon-simp-ck}.

  The proofs of parts~\Item{can-rapp-ck} and~\Item{can-rapp-teq}
  resemble that of \Lem{simp-hsubst}.\Item{simp-rapp-nf} but are
  complicated slightly by the presence of subkinding.  We show the
  proof of part~\Item{can-rapp-ck}, that of part~\Item{can-rapp-teq}
  is similar.  We have $K = \dfun{X}{K_1}{K_2}$,
  $\; \Gamma \ts U_1 \kck K$, $\; \Gamma \ts V \kck K_1$ and
  $\; \kindC{\Gamma}{K_1}$, and we want to show that
  $\; \Gamma \ts \rapp{U_1}{\ksimp{K}}{V} \kck
  \hsubst{K_2}{X}{\ksimp{K_1}}{V}$.
  By inspection of the kind checking and subkinding rules, we must
  have $U_1 = \lam{X}{J_1}{U}$ such that
  $\; \Gamma, X \kas J_1 \ts U \ksn J_2$,
  $\; \Gamma \ts K_1 \ksub J_1$ and
  $\; \Gamma, X \kas K_1 \ts J_2 \ksub K_2$, and by the definition of
  reducing application,
  $\rapp{U_1}{\ksimp{K}}{V} \alEq \hsubst{U}{X}{\ksimp{K_1}}{V}$.
  Using context narrowing and the IH for part~\Item{can-hsubst-ck}, we
  obtain
  $\Gamma \ts \hsubst{U}{X}{\ksimp{K_1}}{V} \ksn
  \hsubst{J_2}{X}{\ksimp{K_1}}{V}$.
  By~\ruleref{TEq-Refl} and the IH for part~\Item{can-par-hsubst-sk},
  we have
  $\Gamma \ts \hsubst{J_2}{X}{\ksimp{K_1}}{V} \ksub
  \hsubst{K_2}{X}{\ksimp{K_1}}{V}$.  We conclude by~\ruleref{CK-Sub}.
\end{proof}

\subsubsection{Validity}
\label{sec:canon-valid}

With the hereditary substitution lemma in place, we can now prove the
remaining validity properties of the canonical judgments.  The most
intricate cases are those for spine kinding and equality, which is
where we use \Lem{canon-hsubst-full}.

\begin{lemma}[canonical validity -- part 2]\label{lem:canon-valid2}
  ~
\begin{enumerate}[nosep,leftmargin=13.7em,labelsep=1em]
\item[(spine kinding validity)]
If $\; \kindC{\Gamma}{J}$ and $\; \spC{\Gamma}{J}{\vec{U}}{K}$,
    then $\kindC{\Gamma}{K}$.

\item[(spine equation validity)]
If $\; \Gamma \ts J_1 \ksub J_2$ and
    $\; \spC{\Gamma}{J_2}{\vec{U} \teq \vec{V}}{K_2}$, then
    $\spC{\Gamma}{J_2}{\vec{U}}{K_2}$, $\;
    \spC{\Gamma}{J_1}{\vec{V}}{K_1}$ and $\Gamma \ts K_1 \ksub K_2$
    for some $K_1$.

\item[(neutral kinding validity)]
If $\; \Gamma \ts N \kin K$, then $\kindC{\Gamma}{K}$.

\item[(subkinding validity)]
If $\; \Gamma \ts J \ksub K$, then $\kindC{\Gamma}{J}$ and
    $\kindC{\Gamma}{K}$.

\item[(proper subtyping validity)]
If $\; \Gamma \ts U \tsub V$, then $\typeC{\Gamma}{U}$ and
    $\typeC{\Gamma}{V}$.

\item[(checked kinding validity)]
If $\; \Gamma \ts V \kck K$, then $\kindC{\Gamma}{K}$.
  \end{enumerate}
\end{lemma}
\begin{proof}
  Subkinding and proper subtyping validity are proven simultaneously,
  the remaining parts are proven separately, in the order they are
  listed.  All parts are proven by induction on derivations of the
  judgments they are named after: spine kinding and equation validity
  are proven by induction on their respective second premises, the
  remaining parts on their respective first premises.  In inductive
  steps of the proofs of spine kinding and equation validity, we use
  the hereditary substitution lemmas to derive suitable first premises
  for applying the IH.  The proof of spine equation validity relies on
  checked equation validity from \Lem{canon-valid2}.  The proof of
  neutral kinding validity relies on spine kinding validity.  In the
  proof of proper subtyping validity, we use neutral kinding validity
  in the cases for the bound projection
  rules~\rulerefN{CST-Bnd1}{CST-Bnd$_{1,2}$}.  The proof of checked
  kinding validity relies on proper subtyping validity.
\end{proof}

\subsubsection{Lifting of Weak Equality to Canonical Equality}
\label{sec:weq-canon}

In \Sec{simple_nf} of the previous section, we established a number of
weak commutativity properties (see \Lem[Lemmas]{etaexp-hsubst}
and~\ref{lem:weq-nf-hsubst}).  Among others, we showed that
normalization weakly commutes with hereditary substitution.  But up
until now, we do not have any effective means to put these properties
to use -- we have yet to establish a relationship between weak
equality and the equality judgments of the declarative and canonical
systems.

To remedy this situation, we prove that a weak equation $U \wkEq V$
can be lifted to canonical equation $\Gamma \ts U \teq V \kck K$,
provided the left- and right-hand sides $U$, $V$ are well-kinded, i.e.\
$\Gamma \ts U \kck K$ and $\Gamma \ts V \kck K$.  Similarly, we show
that weakly equal kinds are canonically equal if they are well-formed.

\begin{lemma}\label{lem:weq-canon}
  Weakly equal canonically well-formed kinds and well-kinded types are
  canonically equal.
  \begin{enumerate}[nosep]
  \item\label{item:weq-can-sk} If $\; \kindC{\Gamma}{J}$, $\;
    \kindC{\Gamma}{K}$ and $J \wkEq K$, then $\Gamma \ts J \ksub K$.
  \item\label{item:weq-can-st} If $\; \typeC{\Gamma}{U}$, $\;
    \typeC{\Gamma}{V}$ and $U \wkEq V$, then $\Gamma \ts U \tsub V$.
  \item\label{item:weq-can-stck} If $\; \kindC{\Gamma}{K}$,
    $\; \Gamma \ts U \kck K$, $\; \Gamma \ts V \kck K$ and
    $U \wkEq V$, then $\Gamma \ts U \tsub V \kck K$.
  \item\label{item:weq-can-speq} If $\; \kindC{\Gamma}{J}$,
    $\; \Gamma \ts J \ksub K_1$, $\; \Gamma \ts J \ksub K_2$,
    $\; \spC{\Gamma}{K_1}{\vec{U}_1}{V_1 \intv W_1}$,
    $\; \spC{\Gamma}{K_2}{\vec{U}_2}{V_2 \intv W_2}$, and
    $\vec{U}_1 \wkEq \vec{U}_2$, then
    $\spC{\Gamma}{J}{\vec{U}_1 \teq \vec{U}_2}{V \intv W}$ for some
    $V$ and $W$.
  \item\label{item:weq-can-keq} If $\; \kindC{\Gamma}{J}$, $\;
    \kindC{\Gamma}{K}$ and $J \wkEq K$, then $\Gamma \ts J \keq K$.
  \item\label{item:weq-can-teq} If $\; \kindC{\Gamma}{K}$,
    $\; \Gamma \ts U \kck K$, $\; \Gamma \ts V \kck K$ and
    $U \wkEq V$, then $\Gamma \ts U \teq V \kck K$.
  \end{enumerate}
\end{lemma}

\begin{proof}
  Simultaneously for all \Item{weq-can-teq}~parts by simultaneous
  induction on the corresponding pairs of kinds, types or spines being
  related, then by case analysis on the final rules used to derive the
  corresponding formation, kinding and weak equality judgments.  In
  the proof of part~\Item{weq-can-keq}, we apply the IH for
  part~\Item{weq-can-sk} directly to the equations $J \wkEq K$ and
  $K \wkEq J$, where the latter is derived using symmetry of weak
  equality. Neither $J$ nor $K$ decrease in this
  step, but the proof of part~\Item{weq-can-keq} does not make any
  further use of the IH and could therefore be inlined in the proofs
  of the other parts.  The proof of part~\Item{weq-can-keq} is
  similar.

  The proofs of the remaining parts are largely routine.
The most interesting case is
the inductive one in
  part~\Item{weq-can-speq}, where we have
  $\vec{U}_1 = (\scons{U_1}{\vec{U}_1'})$,
  $\vec{U}_2 = (\scons{U_2}{\vec{U}_2'})$,
  $K_1 = \dfun{X}{K_{11}}{K_{12}}$, $K_2 = \dfun{X}{K_{21}}{K_{22}}$,
$\; U_1 \wkEq U_2$ and $\vec{U}_1' \wkEq \vec{U}_2'$.  Analyzing the
  derivations of the remaining premises, we have
  \begin{align*}
    \Gamma \ts K_{11} &\ksub J_1 &
    \Gamma, X \kas K_{11} \ts J_2 &\ksub K_{12} &
    \Gamma \ts K_{21} &\ksub J_1 &
    \Gamma, X \kas K_{21} \ts J_2 &\ksub K_{22}
  \end{align*}
  such that $J = \dfun{X}{J_1}{J_2}$, as well as
  \begin{align*}
  \Gamma \ts U_1 &\kck K_{11} &&
  \spC{\Gamma}{\hsubst{K_{12}}{X}{\ksimp{K_{11}}}{U_1}}{\vec{U}_1'}{V_1 \intv W_1}\\
  \Gamma \ts U_2 &\kck K_{21} &&
  \spC{\Gamma}{\hsubst{K_{22}}{X}{\ksimp{K_{21}}}{U_2}}{\vec{U}_2'}{V_2 \intv W_2}
  \end{align*}
  We use~\ruleref{CK-SubCk} and the IH for part~\Item{weq-can-teq} to
  derive $\Gamma \ts U_1 \teq U_2 \kck J_1$, then we use hereditary
  substitution (\Lem[Lemmas]{canon-hsubst-full}.\Item{can-par-hsubst-sk}
  and~\ref{lem:canon-hsubst-full}.\Item{can-par-hsubst-kd-sk}) and
  \Lem{st-teq-simp} to derive
  \begin{alignat*}{2}
    \Gamma \; \ts \; \hsubst{J_2}{X}{\ksimp{J_1}}{U_1} \;
    & \ksub \; &\hsubst{K_{12}}{X}{\ksimp{K_{11}}}{U_1}\\
    \Gamma \; \ts \; \hsubst{J_2}{X}{\ksimp{J_1}}{U_1} \;
    & \ksub \; \hsubst{J_2}{X}{\ksimp{J_1}}{U_2} \;
    \ksub \; &\hsubst{K_{22}}{X}{\ksimp{K_{21}}}{U_2}
  \end{alignat*}
  We conclude the case by the IH for part~\Item{weq-can-speq}
  and~\ruleref{SpEq-Cons}.
\end{proof}

\subsection{Completeness of Canonical Kinding}
\label{sec:canon-complete-full}

In the previous section, we saw that every declaratively well-formed
kind or well-kinded type has a judgmentally equal $\beta\eta$-normal
form (\Lem[Lemmas]{nf-sound} and~\ref{lem:nf-simp}).  In this section,
we prove that every declarative judgment has a canonical counterpart
where the expressions related by the original judgment have been
normalized.  Roughly, whenever $\judgDD{\Gamma}$ holds, we also have
$\judgCC[\nf{\nfCtx{\Gamma}}{\genjudg}]{\nfCtx{\Gamma}}$.  Since
normalization does not change the meaning of an expression, this
result establishes completeness of the canonical system w.r.t.\ to the
declarative one.

There are several judgments for kinding types in the canonical system,
but only one in the declarative system.  To establish completeness, we
show that, if a type $A$ is of kind $K$ according to declarative
kinding, then the normal form $\nf{}{A}$ kind checks against the
normal form $\nf{}{K}$, i.e.\ if $\Gamma \tsD A \kin K$, then
$\nfCtx{\Gamma} \tsC \nf{}{A} \kck \nf{}{K}$.

When $A$ is a variable $A = X$, the normal form $\nf{}{A}$ is its
$\eta$-expansion $\nf{}{A} = \etaExp{\nf{}{K}}{X}$ and we use the
following lemma to prove that it kind checks against $\nf{}{K}$.

\begin{lemma}[$\eta$-expansion]\label{lem:canon-etaexp}
  $\eta$-expansion preserves the canonical kinds of neutral types.  If
  $\; \Gamma \tsNe N \kin K$, then $\Gamma \ts \etaExp{K}{N} \kck K$.
\end{lemma}
Instead of proving the lemma directly, we first prove the following
helper lemma.
\begin{lemma}\label{lem:canon-etaexp-helper}
  ~
  \begin{enumerate}[nosep]
  \item\label{item:canon-hsubst-etaexp-kd} If
    $\; \kindC{\Gamma, X \kas J}{K}$ and
    $\; \Gamma, X \kas J \ts \etaExp{J}{X} \kck J$, then
    $\; \Gamma, X \kas J \ts \hsubst{K}{X}{\ksimp{J}}{\etaExp{J}{X}}
    \keq K$.
  \item\label{item:canon-etaexp-ksub} If
    $\; \Gamma \tsNe N \kin J$ and $\; \Gamma \ts J \ksub K$, then
    $\Gamma \ts \etaExp{K}{N} \kck K$.
  \end{enumerate}
\end{lemma}
The first part says that hereditary substitutions of $\eta$-expanded
variables in kinds vanish, while the second part is a strengthened
version of \Lem{canon-etaexp}.
\begin{proof}
  The two parts are proven separately.  For the first part, we use
  simplification of canonical kinding
  (\Lem{canon-simp}.\Item{canon-simp-kd}) and
  \Lem{etaexp-hsubst}.\Item{hsubst-kds-etaexp} to derive
  $\hsubst{K}{X}{\ksimp{J}}{\etaExp{J}{X}} \wkEq K$.  By weakening and
  the hereditary substitution lemma
  (\Lem{canon-hsubst-full}.\Item{can-hsubst-kd}), we have
  $\; \kindC{\Gamma, X \kas
    J}{\hsubst{K}{X}{\ksimp{J}}{\etaExp{J}{X}}}$.
  The conclusion of the first part then follows by
  \Lem{weq-canon}.\Item{weq-can-keq}.

  The proof of the second part is by induction on the structure of $K$
  and case analysis on the final rule used to derive
  $\Gamma \ts J \ksub K$.  In the base case, we use~\ruleref{CK-NeCk}
  and~\ruleref{CK-SubCk}.  In the inductive case, we have
  $K = \dfun{X}{K_1}{K_2}$, $J = \dfun{X}{J_1}{J_2}$ such that
  $\Gamma \ts K_1 \ksub J_1$ and
  $\Gamma, X \kas K_1 \ts J_2 \ksub K_2$.  By subkinding validity, we
  further have $\kindC{\Gamma}{K_1}$ and
  $\kindC{\Gamma, X \kas K_1}{K_2}$.  By weakening, the IH
  and~\ruleref{CSK-Refl}, we obtain first
  $\Gamma, X \kas K_1 \ts \etaExp{K_1}{X} \kck K_1$, then
  $\Gamma, X \kas K_1 \ts \etaExp{K_1}{X} \kck J_1$ by weakening and
  \ruleref{CK-SubCk}.  By inspection of neutral kinding, we know that
  $N = \app{Y}{\vec{V}}$ for some $Y$ and $\vec{V}$ and that
  $\Gamma \tsVar Y \kin L$ and
  $\spC{\Gamma}{L}{\vec{V}}{\dfun{X}{J_1}{J_2}}$.  We use weakening,
  \ruleref{CK-Snoc} and~\ruleref{CK-Ne} to derive
  $\Gamma, X \kas K_1 \tsNe \appp{Y}{\vec{V}}{(\etaExp{K_1}{X})} \kin
  \hsubst{J_2}{X}{\ksimp{J_1}}{\etaExp{K_1}{X}}$.

  Now we see why it was necessary to strengthen the IH: the body of
  the $\eta$-expansion of $N$ has kind
  $\hsubst{J_2}{X}{\ksimp{J_1}}{\etaExp{K_1}{X}}$ rather than $K_2$ as
  required by \Lem{canon-etaexp}.  In order to apply the IH, we show
  that
  \begin{align*}
    \Gamma, X \kas K_1 \; \ts \;
    \hsubst{J_2}{X}{\ksimp{J_1}}{\etaExp{K_1}{X}} \; &\alEq \;
    \hsubst{J_2}{X}{\ksimp{K_1}}{\etaExp{K_1}{X}}
    \tag{by \Lem{st-teq-simp}}\\
    &\ksub \; \hsubst{K_2}{X}{\ksimp{K_1}}{\etaExp{K_1}{X}}
    \tag{by \Lem{canon-hsubst-full}.\Item{can-par-hsubst-sk}}\\
    &\keq \; K_2.  \tag{by part~\Item{canon-hsubst-etaexp-kd}}
  \end{align*}
  We conclude the case by the IH and~\ruleref{CK-Abs'}.
\end{proof}
\Lem{canon-etaexp} as well as a strengthened version of
\Lem{canon-etaexp-helper}.\Item{canon-hsubst-etaexp-kd} now follow as
corollaries.
\begin{corollary}\label{cor:canon-hsubst-etaexp-kd}
  If $\; \kindC{\Gamma, X \kas J}{K}$, then
  $\; \Gamma, X \kas J \ts \hsubst{K}{X}{\ksimp{J}}{\etaExp{J}{X}}
  \keq K$.
\end{corollary}

To establish completeness of the canonical system w.r.t.\ the
declarative system, we show that every declarative judgment derived
using the \emph{extended} declarative rules, rather than the original
ones, has a canonical counterpart.  The proof makes crucial use of the
validity conditions present in the extended rules.
To avoid confusion, we again mark canonical judgments with the
subscript ``$\mathsf{c}$'' and extended declarative ones with
``$\mathsf{e}$''.  To enhance readability, we omit the subscript
$\nfCtx{\Gamma}$, writing e.g.\ $\nf{}{A}$ instead of
$\nf{\nfCtx{\Gamma}}{A}$.
\begin{lemma}[completeness of the canonical rules -- extended
  version]\label{lem:canon-complete-full}
  ~
  \begin{enumerate}[nosep]
  \item\label{item:can-cmp-ctx} If $\; \ctxE{\Gamma}$, then
    $\; \ctxCC{\nfCtx{\Gamma}}$.
  \item\label{item:can-cmp-kd} If $\; \kindE{\Gamma}{K}$, then
    $\; \kindCC{\nfCtx{\Gamma}}{\nf{}{K}}$.
  \item\label{item:can-cmp-tp} If $\; \Gamma \tsE A \kin K$, then
    $\; \nfCtx{\Gamma} \tsC \nf{}{A} \kck \nf{}{K}$.
  \item\label{item:can-cmp-sk} If $\; \Gamma \tsE J \ksub K$, then
    $\; \nfCtx{\Gamma} \tsC \nf{}{J} \ksub \nf{}{K}$.
  \item\label{item:can-cmp-st-intv} If
    $\; \Gamma \tsE A \ksub B \kin C \intv D$, then
    $\; \nfCtx{\Gamma} \tsC \nf{}{A} \ksub \nf{}{B}$.
  \item\label{item:can-cmp-st} If $\; \Gamma \tsE A \ksub B \kin K$, then
    $\; \nfCtx{\Gamma} \tsC \nf{}{A} \ksub \nf{}{B} \kck \nf{}{K}$.
  \item\label{item:can-cmp-keq} If $\; \Gamma \tsE J \keq K$, then
    $\; \nfCtx{\Gamma} \tsC \nf{}{J} \keq \nf{}{K}$.
  \item\label{item:can-cmp-teq} If $\; \Gamma \tsE A \keq B \kin K$, then
    $\; \nfCtx{\Gamma} \tsC \nf{}{A} \keq \nf{}{B} \kck \nf{}{K}$.
  \item\label{item:can-cmp-eta} If
    $\; \Gamma \tsE A \kin \dfun{X}{J}{K}$ and $\; X \notin \fv(A)$, then
    \[
      \nfCtx{\Gamma} \tsC \nf{}{\lam{X}{J}{\app{A}{X}}} \keq
      \nf{}{A} \kck \nf{}{\dfun{X}{J}{K}}.
    \]
  \item\label{item:can-cmp-hsubst-kd} If
    $\; \kindE{\Gamma, X \kas J}{K}$, $\; \Gamma \tsE A \kin J$ and
    $\; \kindE{\Gamma}{\subst{K}{X}{A}}$, then
    \[
      \nfCtx{\Gamma} \tsC
      \hsubst{\nf{}{K}}{X}{\ksimp{\nf{}{J}}}{\nf{}{A}} \keq
      \nf{}{\subst{K}{X}{A}}.
    \]
  \item\label{item:can-cmp-beta} If
    $\; \Gamma, X \kas J \tsE A \kin K$, $\; \Gamma \tsE B \kin J$,
    $\; \kindE{\Gamma, X \kas J}{K}$,
    $\; \Gamma \tsE \subst{A}{X}{B} \kin \subst{K}{X}{B}$\\ and
    $\; \kindE{\Gamma}{\subst{K}{X}{B}}$, then
    \[
      \nfCtx{\Gamma} \tsC \nf{}{\app{(\lam{X}{J}{A})}{B}} \keq
      \nf{}{\subst{A}{X}{B}} \kck \nf{}{\subst{K}{X}{B}}.
    \]
  \end{enumerate}
\end{lemma}
Equivalent statements w.r.t.\ the original declarative rules follow by
equivalence of the original and extended declarative systems.
\begin{proof}
  All parts are proven simultaneously, by induction on the derivations
  of the respective premises, except for
  parts~\Item{can-cmp-eta}--\Item{can-cmp-beta}, which are helper
  lemmas that apply the IH directly to all of their premises but could
  be inlined in the proofs of the other parts.

  Thanks to the admissible rules introduced in~\Lem{canon-admissible},
  the proofs of parts~\Item{can-cmp-ctx}--\Item{can-cmp-tp} are
  straightforward, except for the cases of~\ruleref{K-Var}, where we
  use \Lem{canon-etaexp}, and~\ruleref{K-App}, where we use the
  hereditary substitution lemma to normalize $A = \app{A_1}{A_2}$ if
  $\nf{}{A_1}$ is an abstraction, and the IH for
  part~\Item{can-cmp-hsubst-kd} together with~\ruleref{CK-SubCk} to
  adjust the kind of the result.  The validity conditions
  of~\ruleref{K-App} are crucial in this last step.

  Parts~\Item{can-cmp-st-intv} and~\Item{can-cmp-teq} follow almost
  immediately from part~\Item{can-cmp-st}, part~\Item{can-cmp-keq}
  from part~\Item{can-cmp-sk}.

  The proofs of parts~\Item{can-cmp-eta}--\Item{can-cmp-beta} all
  follow the same pattern.  First, we use the IH to normalize the
  premises and establish validity of the left- and right-hand sides of
  the respective equations.
Then we use \Lem{weq-nf-etaexp},
  \Lem{weq-nf-hsubst}.\Item{weq-nf-hsubst-kd} and
  \Lem{weq-nf-hsubst}.\Item{weq-nf-hsubst-tp}, respectively, to derive
  weak versions of these equations, and \Lem{weq-canon} to turn them
  into canonical equations.

  The remaining parts~\Item{can-cmp-sk} and~\Item{can-cmp-st} are the
  most difficult to prove.  The cases of the extended $\beta$~and
  $\eta$-conversion rules are covered by parts~\Item{can-cmp-eta}
  and~\Item{can-cmp-beta} thanks to the validity conditions in the
  extended rules.  The case of~\ruleref{ST-App} is similar to that
  of~\ruleref{K-App} -- again the validity conditions are crucial.
  Some of the remaining cases are covered by the admissible rules
  introduced in
\Sec[\S\S]{canon-order} and~\ref{sec:canon-admissible}.
The challenging cases are those where one of the premises of the
  corresponding rule extends the contexts, \ie those
  of~\ruleref{CSK-DArr}, \ruleref{CST-All} and~\ruleref{CST-Abs}.
We show the case for~\ruleref{CSK-DArr} here, the other two are
  similar.

  We are given $\kindE{\Gamma}{\dfun{X}{J_1}{J_2}}$,
  $\; \Gamma \tsE K_1 \ksub J_1$ and
  $\Gamma, X \kas K_1 \tsE J_2 \ksub K_2$
  with $J = \dfun{X}{J_1}{J_2}$
  and $K = \dfun{X}{K_1}{K_2}$.  We start by applying the IH to all
  the premises and analyze the first of the resulting derivations to
  obtain
  \begin{align*}
    \kindCC{\nfCtx{\Gamma}&}{\nf{}{J_1}}
    &\kindCC{\nfCtx{\Gamma}, X \kas \nf{}{J_1}&}{\nf{\nfCtx{\Gamma, X \kas J_1}}{J_2}}\\
\nfCtx{\Gamma} &\tsC \nf{}{K_1} \ksub \nf{}{J_1}
    &\nfCtx{\Gamma}, X \kas \nf{}{K_1} &\tsC
       \nf{\nfCtx{\Gamma, X \kas K_1}}{J_2} \ksub
       \nf{\nfCtx{\Gamma, X \kas K_1}}{K_2}.
\end{align*}
  Note the different contexts $\nfCtx{\Gamma, X \kas J_1}$ and
  $\nfCtx{\Gamma, X \kas K_1}$ used to normalize $J_2$ in the second
  and fourth of these judgments, respectively.  This leads to a
  syntactic difference in the resulting normal forms, i.e.\ we have
  $\nf{\nfCtx{\Gamma, X \kas J_1}}{J_2} \not\alEq \nf{\nfCtx{\Gamma, X
      \kas K_1}}{J_2}$.
  In order to apply~\ruleref{CSK-DArr}, we need to resolve this
  difference.

  We notice that
  $\simpEq{\nfCtx{\Gamma, X \kas J_1}}{\nfCtx{\Gamma, X \kas K_1}}$
  by~\Lem{st-teq-simp}. Hence, by~\Lem{weq-nf}, we have
  $\nf{\nfCtx{\Gamma, X \kas J_1}}{J_2} \wkEq \nf{\nfCtx{\Gamma, X
      \kas K_1}}{J_2}$.
  Using context narrowing and subkinding validity, we derive
  \begin{align*}
    \kindC{\nfCtx{\Gamma}, X \kas \nf{}{K_1}}{\nf{\nfCtx{\Gamma, X \kas
        J_1}}{J_2}}
    &&\text{and}&&
    \kindC{\nfCtx{\Gamma}, X \kas \nf{}{K_1}}{\nf{\nfCtx{\Gamma, X \kas
        K_1}}{J_2}},
  \end{align*}
  from which we obtain, by \Lem{weq-canon}.\Item{weq-can-sk},
  \[
    \nfCtx{\Gamma}, X \kas \nf{}{K_1} \; \tsC \; \nf{\nfCtx{\Gamma, X \kas
        J_1}}{J_2}
    \; \keq  \; \nf{\nfCtx{\Gamma, X \kas K_1}}{J_2}
    \; \ksub \; \nf{\nfCtx{\Gamma, X \kas K_1}}{K_2}
  \]
  We conclude the case by~\ruleref{CWf-DArr} and~\ruleref{CSK-DArr}.
\end{proof}

\subsection{Inversion of Subtyping}
\label{sec:canon-inversion-full}

As we saw in \Sec{declarative} of the paper, preservation of types under CBV
reduction does not hold in arbitrary contexts.  The culprit are type
variable bindings with inconsistent bounds.  Such bindings can inject
arbitrary inequations into the subtyping relation and hence break
putative properties that hold for subtyping of closed types.  For
example, the absurd assumption $X \kin \Top \intv \Bot$ trivializes
the subtyping relation under any context in which it appears.  To see
this, consider the following derivation, where
$\Gamma = X \kin \Top \intv \Bot$.
\vspace{-\baselineskip}
\begin{prooftree}
  \AxiomC{$\vdots$}
\UnaryInfC{$\Gamma \ts \fun{U}{V} \tsub \Top$}
    \AxiomC{$\vdots$}
    \UnaryInfC{$\Gamma \tsNe X \kin \Top \intv \Bot$}
    \LeftLabel{\rulerefPN{CST-Bnd1}{CST-Bnd$_1$}}
    \UnaryInfC{$\Gamma \ts \Top \tsub X$}
      \AxiomC{$\vdots$}
      \UnaryInfC{$\Gamma \tsNe X \kin \Top \intv \Bot$}
      \RightLabel{\rulerefPN{CST-Bnd2}{CST-Bnd$_2$}}
      \UnaryInfC{$\Gamma \ts X \tsub \Bot$}
    \RightLabel{\rulerefP{CST-Trans}}
    \BinaryInfC{$\Gamma \ts \Top \tsub \Bot$}
  \RightLabel{\rulerefP{CST-Trans}}
  \insertBetweenHyps{\hspace{-2em}}
  \BinaryInfC{$\Gamma \ts \fun{U}{V} \tsub \Bot$}
    \AxiomC{$\vdots$}
\UnaryInfC{$\Gamma \ts \Bot \tsub \all{X}{K}{W}$}
  \LeftLabel{\rulerefP{CST-Trans}}
  \insertBetweenHyps{\hspace{-3em}}
  \BinaryInfC{$\Gamma \ts \fun{U}{V} \tsub \all{X}{K}{W}$}
\end{prooftree}

Under such conditions, subtyping cannot be inverted in any meaningful
way.
We therefore consider inversion of canonical subtyping only in the
empty context, following the approach taken by Rompf and Amin in their
type safety proof for DOT~\cite{RompfA16oopsla}

\fig{fig:transfree_rules}{Top-level transitivity-free canonical subtyping}{
\judgment{Transitivity-free subtyping of closed proper types}{\fbox{$\tsTf U \ksub V$}}
\begin{multicols}{2}
  \typicallabel{TfST-Top}
\infrule[\ruledef{TfST-Top}]
  {\typeC{\cempty}{V}}
  {\tsTf V \tsub \Top}

\infrule[\ruledef{TfST-Arr}]
  {\quad\\
   \cempty \ts U_2 \tsub U_1  \andalso  \cempty \ts V_1 \tsub V_2}
  {\tsTf \fun{U_1}{V_1} \tsub \fun{U_2}{V_2}\\ \vphantom{TfST-All}}

\infrule[\ruledef{TfST-Bot}]
  {\typeC{\cempty}{V}}
  {\tsTf \Bot \tsub V}

\infruleLeft{TfST-All}
  {\typeC{\cempty}{\all{X}{K_1}{V_1}}\\
   \cempty \ts K_2 \ksub K_1  \andalso
   X \kas K_2 \ts V_1 \tsub V_2}
  {\tsTf \all{X}{K_1}{V_1} \tsub \all{X}{K_2}{V_2}}
\end{multicols}}
 
As a first step we show that any top-level uses of the transitivity
rule~\ruleref{CST-Trans} can be eliminated.  To do so, we introduce a
helper judgment $\tsTf U \tsub V$, which states that $U$ is a proper
subtype of $V$ in the empty context (see~\Fig{transfree_rules}).  It
is easy to see that this judgment is sound w.r.t.\ canonical subtyping
in the empty context (the proof is by routine induction on subtyping
derivations).
\begin{lemma}[soundness of top-level subtyping]
  \label{lem:transfree-sound}
  If $\; \tsTf U \tsub V$, then $\cempty \ts U \tsub V$.
\end{lemma}

Crucially, the inference rules for the judgment $\tsTf U \tsub V$ do
not include a transitivity rule, but the following variant of that
rule is admissible.
\begin{lemma}[top-level transitivity elimination]
  The following is admissible.
  \begin{center}\normalfont
    \infruleSimp[TfST-Trans]
    {\cempty \ts U \tsub V  \andalso  \tsTf V \tsub W}
    {\tsTf U \tsub W}
  \end{center}
\end{lemma}
\begin{proof}
  The proof is by induction on the derivation of the first premise and
  case analysis of the final rule used to derive the second.  In the
  case of~\ruleref{CST-Trans}, we use the IH twice.  In the case
  of~\ruleref{CST-Bot} where $\emptyset \ts \Bot \tsub U$ and
  $\tsTf U \tsub V$, we use \Lem{transfree-sound} and validity of
  canonical typing to derive $\typeC{\emptyset}{V}$ and conclude
  with~\ruleref{TfST-Bot}.  Similarly, in cases where the second
  premise was derived using~\ruleref{TfST-Top}, we use validity of
  canonical subtyping and~\ruleref{TfST-Top}.
\end{proof}

Thanks to~\ruleref{TfST-Trans}, it is straightforward to establish
completeness, and thus equivalence of the judgments
$\cempty \ts U \tsub V$ and $\tsTf U \tsub V$.
\begin{lemma}[equivalence of top-level subtyping]
  \label{lem:transfree-equiv}
  The two versions of canonical subtyping are equivalent in the empty
  context: $\cempty \ts U \tsub V$ iff $\; \tsTf U \tsub V$.
\end{lemma}
\begin{proof}
  We have already proven soundness~($\Leftarrow$).
  Completeness~($\Rightarrow$) is by induction on the derivations of
  $\cempty \ts U \tsub V$ and uses~\ruleref{TfST-Trans} in the case
  of~\ruleref{CST-Trans}.
\end{proof}

Inversion of the canonical subtyping relation in the empty context now
follows immediately by inspection of the transitivity-free subtyping
rules and
\Lem{transfree-equiv}.  We only state the relevant cases.
\begin{corollary}[inversion of canonical subtyping -- embedding]
  \label{cor:cst-inv}
  Let $\; \cempty \ts U_1 \tsub U_2$.
  \begin{enumerate}[nosep]
  \item\label{item:cst-inv-arr} If $\; U_1 = \fun{V_1}{W_1}$ and
    $U_2 = \fun{V_2}{W_2}$, then $\cempty \ts V_2 \tsub V_1$ and
    $\cempty \ts W_1 \tsub W_2$.
  \item\label{item:cst-inv-all} If $\; U_1 = \all{X}{K_1}{V_1}$ and
    $U_2 = \all{X}{K_2}{V_2}$, then $\cempty \ts K_2 \ksub K_1$ and
    $X \kas K_2 \ts V_1 \tsub V_2$.
  \end{enumerate}
\end{corollary}
\begin{corollary}[inversion of canonical subtyping -- contradiction]
  \label{cor:cst-inv-neg}
For any $U$, $V$, $W$ and $K$, \begin{enumerate}[nosep]
  \item\label{item:cst-inv-top-not-bot}
    $\cempty \nts \Top \tsub \Bot$,
  \item\label{item:cst-inv-arr-not-all}
    $\cempty \nts \fun{U}{V} \tsub \all{X}{K}{W}$, and
  \item\label{item:cst-inv-all-not-arr}
    $\cempty \nts \all{X}{K}{U} \tsub \fun{V}{W}$.
  \end{enumerate}
\end{corollary}\noindent
A bit more work is needed to also prove the declarative version of
subtyping inversion.
Again, we only state the relevant cases.
\begin{lemma}[inversion of declarative subtyping --
  embedding]\label{lem:st-inv-emb}
  Let $\; \cempty \ts A_1 \tsub A_2 \kin \kstar$.
  \begin{enumerate}[nosep]
  \item\label{item:st-inv-arr} If $\; A_1 = \fun{B_1}{C_1}$ and
    $A_2 = \fun{B_2}{C_2}$, then $\cempty \ts B_2 \tsub B_1 \kin \kstar$ and
    $\cempty \ts C_1 \tsub C_2 \kin \kstar$.
  \item\label{item:st-inv-all} If $\; A_1 = \all{X}{K_1}{B_1}$ and
    $A_2 = \all{X}{K_2}{B_2}$, then $\cempty \ts K_2 \ksub K_1$ and
    $X \kas K_2 \ts B_1 \tsub B_2 \kin \kstar$.
  \end{enumerate}
\end{lemma}\noindent

The proof makes use of the following generation lemma for well-kinded
arrow and universal types, which is proven by induction on kinding
derivations.
\begin{lemma}[generation of kinding for arrows and universals]
  \label{lem:decl-arr-all-gen}
  The following are admissible.
\end{lemma}
\begin{center}
  \infruleSimp
  {\Gamma \ts \fun{A}{B} \kin \kstar}
  {\Gamma \ts A \kin \kstar  \andalso  \Gamma \ts B \kin \kstar}
  \hspace{3em}
  \infruleSimp
  {\Gamma \ts \all{X}{K}{A} \kin \kstar}
  {\kindD{\Gamma}{K}  \andalso  \Gamma, X \kas K \ts A \kin \kstar}
\end{center}

\begin{proof}[Proof of \Lem{st-inv-emb}]
  By completeness of canonical subtyping and soundness of
  normalization.  We show only the first part, the second is
  analogous.  Assume
  $\cempty \ts \fun{B_1}{C_1} \tsub \fun{B_2}{C_2} \kin \kstar$.  Then by validity
  of declarative subtyping (\Lem{decl_validity}), generation of
  kinding for arrow types, soundness of normalization (\Lem{nf-sound})
  and completeness of canonical subtyping, we have
  \begin{gather*}
    \cempty \tsD B_1 \teq \nf{\nfCtx{\Gamma}}{B_1} \kin \kstar \qquad \qquad
    \cempty \tsD C_1 \teq \nf{\nfCtx{\Gamma}}{C_1} \kin \kstar \\
    \cempty \tsD B_2 \teq \nf{\nfCtx{\Gamma}}{B_2} \kin \kstar \qquad \qquad
    \cempty \tsD C_2 \teq \nf{\nfCtx{\Gamma}}{C_2} \kin \kstar \\
    \cempty \tsC
    \fun{\nf{\nfCtx{\Gamma}}{B_1}}{\nf{\nfCtx{\Gamma}}{C_1}}
    \tsub
    \fun{\nf{\nfCtx{\Gamma}}{B_2}}{\nf{\nfCtx{\Gamma}}{C_2}}
  \end{gather*}
  By inversion and soundness of canonical subtyping, it follows that
  \begin{gather*}
    \cempty \tsD \; B_2 \; \teq \; \nf{\nfCtx{\Gamma}}{B_2} \; \tsub \;
    \nf{\nfCtx{\Gamma}}{B_1} \; \teq \; B_1 \; \kin \; \kstar
    \quad \text{ and}\\
    \cempty \tsD \; C_1 \; \teq \; \nf{\nfCtx{\Gamma}}{C_1} \; \tsub \;
    \nf{\nfCtx{\Gamma}}{C_2} \; \teq \; C_2 \; \kin \; \kstar.
    \qedhere
  \end{gather*}
\end{proof}

We also prove a declarative counterpart of \Cor{cst-inv-neg}, which is
used in the proof of the progress theorem below.
\begin{lemma}[inversion of declarative subtyping --
  contradiction]\label{lem:st-inv-neg}
For any $A$, $B$, $C$ and $K$, \begin{enumerate}[nosep]
  \item\label{item:st-inv-top-not-bot}
    $\cempty \nts \Top \tsub \Bot$,
  \item\label{item:st-inv-arr-not-all}
    $\cempty \nts \fun{A}{B} \tsub \all{X}{K}{C}$, and
  \item\label{item:st-inv-all-not-arr}
    $\cempty \nts \all{X}{K}{A} \tsub \fun{B}{C}$.
  \end{enumerate}
\end{lemma}
\begin{proof}
  By completeness of canonical subtyping, then by contradiction using
  \Cor{cst-inv-neg}.
\end{proof}

\subsubsection{Type Safety Revisited}
\label{sec:type-safety-done}

We are finally ready to prove type safety of \FOmegaInt{}.
The proof of weak preservation requires a standard generation lemma
for term and type abstractions.
\begin{lemma}[generation of typing for term and type abstraction]
  \label{lem:decl_abs_tabs_gen}
  ~
  \begin{enumerate}
  \item\label{item:decl_abs_gen} If
    $\; \Gamma \ts \lam{x}{A}{t} \tin B$, then
    $\Gamma, x \tas A \ts t \tin C$ and
    $\, \Gamma \ts \fun{A}{C} \tsub B \kin \kstar$ for some $C$.
  \item\label{item:decl_tabs_gen} If
    $\; \Gamma \ts \Lam{X}{K}{t} \tin A$, then
    $\Gamma, X \kas K \ts t \tin B$ and
    $\, \Gamma \ts \all{X}{K}{B} \tsub A \kin \kstar$ for some $B$.
  \end{enumerate}
\end{lemma}\noindent
Recall that weak preservation (\Prop{weak-pres}) states that
CBV~reduction preserves the types of closed terms, i.e.\ if
$\ts t \tin A$ and $t \cbvstep t'$, then $\ts t' \tin A$.
\begin{proof}[Proof of \Prop{weak-pres}]
  The proof is by induction on typing derivations and case analysis on
  CBV reduction rules.
The interesting cases are those where $\beta$-contractions occur.
We describe the case of~\ruleref{T-App}.
The corresponding case for~\ruleref{T-TApp} is similar.
We have $t = \app{(\lam{x}{B}{s})}{v}$ with
  $\ts \lam{x}{B}{s} \tin \fun{C}{A}$ and $\ts v \tin C$ for some $B$
  and $C$.
By generation of term abstractions
  (\Lem{decl_abs_tabs_gen}.\Item{decl_abs_gen}),
  $x \tas B \ts s \tin D$ and
  $\ts \fun{B}{D} \tsub \fun{C}{A} \kin \kstar$ for some $D$.
By inversion of subtyping (\Lem{st-inv-emb}.\Item{st-inv-all}), we have
  $\ts C \tsub B \kin \kstar$ and $\ts D \tsub A \kin \kstar$ and
  hence $\ts v \tin B$ and $x \tas B \ts s \tin A$ by subsumption.
To conclude the proof we need to show that
  $\ts \subst{s}{x}{v} \tin A$, which follows from the substitution
  lemma~(\Lem{decl_subst}).
\end{proof}

This establishes the first half of type safety.  For the second half,
progress, we first need to prove a standard canonical forms lemma.
\begin{lemma}[canonical forms]\label{lem:can-forms}
  Let $v$ be a closed, well-typed value.
  \begin{enumerate}
  \item\label{item:can-arr} If $\; \cempty \ts v \in \fun{A}{B}$, then
    $v = \lam{x}{C}{t}$ for some $C$ and $t$.
  \item\label{item:can-all} If $\; \cempty \ts v \in \all{X}{K}{A}$,
    then $v = \Lam{X}{J}{t}$ for some $J$ and $t$.
  \end{enumerate}
\end{lemma}
\begin{proof}
  Separately for the two parts; each by case analysis, first on $v$,
  then on the final typing rule used to derive the respective premise.
  Since the only values are abstractions, the relevant sub-cases
  are~\ruleref{T-Abs}, \ruleref{T-TAbs} and~\ruleref{T-Sub}.  The
  sub-cases for~\ruleref{T-Abs} and~\ruleref{T-TAbs} are immediate.
  In the sub-cases for~\ruleref{T-Sub}, we first use the generation
  lemma for abstractions (\Lem{decl_abs_tabs_gen}), then dismiss
  impossible sub-cases using \Lem{st-inv-neg}.
\end{proof}

Thanks to the canonical forms lemma, the proof of the progress theorem
is now entirely standard.
\begin{theorem}[progress]\label{thm:progress}
  If $\; \ts t \tin A$, then either $t$ is a value, or $t \cbvstep t'$
  for some term $t'$.
\end{theorem}
\begin{proof}
  By routine induction on typing derivations.  The cases
  for~\ruleref{T-App} and~\ruleref{T-TApp} use the canonical forms
  lemma (\Lem{can-forms}).
\end{proof}

\section{Encoding Custom Subtyping Theories}
\label{sec:custom_sub_theories}

\newcommand{\cVar}[1]{\bar{#1}}
\newcommand{\rVar}[1]{Y_{#1}}
\newcommand{\andVar}{\cVar{\wedge}}
\newcommand{\fixVar}{\cVar{\mu}}

\emph{\textbf{Disclaimer.} The following examples have not been
  mechanized in Agda.}
\medskip

\noindent
Recall that bindings of the form $X \kin A \intv B$ represent
first-class \emph{type inequations} $A \tsub B$ in \FOmegaInt{}
because, in a context containing such a binding, we have $A \tsub X$
and $X \tsub B$, and hence --~by transitivity of subtyping~--
$A \tsub B$.
Interval kinds thus provide us with a mechanism for \emph{(in)equality
  reflection}, \ie a way to extend the subtyping relation via
assumptions made at the term- or type-level (via type abstractions).
Among other things, this allows us to postulate type operators with
associated subtyping rules through type variable bindings.
For example, we may postulate intersection types $A \wedge B$ by
assuming an abstract binary type operator
$\andVar \kin \fun{\kstar}{\fun{\kstar}{\kstar}}$ and the usual typing
rules for intersections as abstract type inequations:
\begin{align*}
  &\rVar{\tsub L} \kin \dfun{X_1, X_2}{\kstar}{(\appp{\andVar}{X_1}{X_2}) \intv X_1},
  \qquad \rVar{\tsub R} \kin \dfun{X_1, X_2}{\kstar}{(\appp{\andVar}{X_1}{X_2}) \intv X_2}, \\
  &\rVar{\tsub\wedge} \kin \dfun{Z, X_1, X_2}{\kstar}{\fun{Z \intv X_1}{\fun{Z \intv X_2}{Z \intvP (\appp{\andVar}{X_1}{X_2})}}}.
\end{align*}
where we abbreviated dependent arrow kinds
$\dfun{X_1}{J}{\dfun{X_2}{J}{\fun{\dotsm}{\dfun{X_n}{J}{K}}}}$ with
multiple parameters
as $\dfun{X_1, X_2, \dotsc, X_n}{J}{K}$ for readability.
The two variables $\rVar{\tsub L}$ and $\rVar{\tsub R}$ represent,
respectively, the left- and right-hand projection rules for
intersections ($A \wedge B \tsub A$ and $A \wedge B \tsub B$).
The assumption $\rVar{\tsub\wedge}$ encodes the fact that
intersections are greatest lower bounds, i.e.\ that
$A \tsub B \wedge C$
when $A \tsub B$ and $A \tsub C$.
To see how this last ``rule'' can be put to work, let $A$, $B$ and $C$
be proper types and assume that $A \tsub B$ and $A \tsub C$.  Then we
also have $A \kin A \intv B$ and $A \kin A \intv C$, and hence
$\appp{\apppp{\rVar{\tsub\wedge}}{A}{B}{C}}{A}{A} \kin A \intvP
(\appp{\andVar}{B}{C})$, from which we conclude
$A \tsub \appp{\andVar}{B}{C}$.

We can also postulate recursive inequations.  For example, the
following bindings encode an equi-recursive type constructor
$\mu \kin \fun{(\fun{\kstar}{\kstar})}{\kstar}$ that, when applied to
a unary operator $A$, represents the fixpoint $\app{\mu}{A}$ of $A$.
\begin{align*}
  &\fixVar \kin \fun{(\fun{\kstar}{\kstar})}{\kstar},\\
  &\rVar{\tsub L} \kin \dfun{X}{\fun{\kstar}{\kstar}}{(\app{\fixVar}{X}) \intvP (\app{X}{(\app{\fixVar}{X})})}, &
  &\rVar{\tsub R} \kin \dfun{X}{\fun{\kstar}{\kstar}}{(\app{X}{(\app{\fixVar}{X})}) \intvP (\app{\fixVar}{X})}.
\end{align*}
Together, the assumptions $\rVar{\tsub L}$ and $\rVar{\tsub R}$ say
that $\app{\mu}{A} = \app{A}{(\app{\mu}{A})}$, i.e.\ that
$\app{\mu}{A}$ is a fixpoint of $A$.

The above examples are only possible because we do not impose any
\emph{consistency constraints} on the bounds of intervals.
That is, an interval kind $A \intv B$ is well-formed, irrespective of
whether $A \tsub B$ is actually provable or not.
For example, the signature of the abstract intersection operator
$\andVar$ above does not tell us anything about how the type
application $\appp{\andVar}{A}{B}$ is related to its first argument
$A$, nor does the abstract left projection inequality $\rVar{\tsub L}$
impose any constraints on its parameters $X_1$ and $X_2$.
It is therefore impossible to say anything about the relationship of
the bounds $\appp{\andVar}{X_1}{X_2}$ and $X_1$ of the codomain of
$\rVar{\tsub L}$, other than that they are both proper types.
We can certainly not prove that $\appp{\andVar}{X_1}{X_2} \tsub X_1$
in general.

\section{Reported Scala~3 Issues}
\label{sec:scala_issues}

During the development of \FOmegaInt{} we discovered and reported the
following issues to the Scala~3 bug tracker.
\begin{itemize}
\item \url{https://github.com/lampepfl/dotty/issues/2887}
\item \url{https://github.com/lampepfl/dotty/issues/6320}
\item \url{https://github.com/lampepfl/dotty/issues/6499}
\item \url{https://github.com/lampepfl/dotty/issues/9691}
\item \url{https://github.com/lampepfl/dotty/issues/9695}
\item \url{https://github.com/lampepfl/dotty/issues/9697}
\end{itemize}
 
\bibliography{paper}

\end{document}